\documentclass[authorcolumns,numberwithinsect,a4paper]{no-lipics}
\usepackage{bm}

\usetikzlibrary{arrows,arrows.meta,shapes.misc, decorations.pathreplacing,decorations.pathmorphing,calligraphy, patterns.meta}

\usetikzlibrary{arrows,arrows.meta,shapes.misc, decorations.pathreplacing,decorations.pathmorphing,calligraphy, patterns.meta}
\tikzset{dotmark/.style={circle,fill,inner sep=1.5pt}}
\tikzset{emptymark/.style={circle,draw,fill=white,inner sep=1.5pt}}
\tikzset{crossmark/.style={thick,inner sep=1.5pt}}

\newcommand{\A}{\mathcal{A}}
\newcommand{\B}{\mathcal{B}}
\newcommand{\D}{\mathcal{D}}

\DeclareRobustCommand{\twoheadleadsto}{\tikz[baseline=(a.base)]{%
        \node at (.2,0) {\(\leadsto\)};%
        \fill[white] (-.1,-.1+.012) -- (.25,-.1+.012) -- (.345,0.011) -- (.25,.113) --
        (-.1,.113) --cycle;
        \node at (.125,0) {\(\leadsto\)};%
        \node (a) at (.4/2,-.0) {\phantom{\(\leadsto\)}};%
}}

\def\onto{\twoheadleadsto}
\def\ponto#1{\twoheadleadsto\kern-.3em{}_{#1}\kern.3em}

\def\intvl#1#2{\bm{[}\,#1\,\bm{,}\,#2\,\bm{]}}
\def\fragmentco#1#2{\bm{[}\,#1\,\bm{.\,.}\,#2\,\bm{)}}
\def\fragmentoc#1#2{\bm{(}\,#1\,\bm{.\,.}\,#2\,\bm{]}}
\def\fragmentoo#1#2{\bm{(}\,#1\,\bm{.\,.}\,#2\,\bm{)}}
\def\fragment#1#2{\bm{[}\,#1\,\bm{.\,.}\,#2\,\bm{]}}
\def\position#1{\bm{[}\,#1\,\bm{]}}

\def\mid{\ensuremath :}

\def\emptyset{\varnothing}

\def\Sb{\S_{\beta}}
\def\Sm{\S_{\mu}}
\def\Sf{\S_{\varphi}}

\newcommand{\Oh}{\mathcal{O}}
\newcommand{\tOh}{{\rlap{\raisebox{-0.2ex}{$\widetilde{\phantom{\Oh}}$}}\Oh}}
\newcommand{\cO}{\mathcal{O}}
\newcommand{\cOtilde}{\tOh}
\newcommand{\Ohtilde}{\tOh}

\newcommand{\per}{\operatorname{per}}
\newcommand{\rot}{\operatorname{rot}}

\newcommand{\OccEx}{\mathrm{Occ}}
\newcommand{\OccE}{\mathrm{Occ}}
\newcommand{\OccW}{\mathrm{Occ}^w}
\newcommand{\ed}{\mathsf{ed}}
\newcommand{\edw}{\mathsf{ed}^w}
\newcommand\selfed{\mathsf{self}\text{-}\mathsf{ed}}
\def\edu#1#2{\mathsf{ed}(#1,#2)}
\def\edw#1#2{\mathsf{ed}^w(#1 \rightarrow #2)}
\def\edwk#1#2#3{\mathsf{ed}^w_{\leq #1}(#2 \rightarrow #3)}
\def\edwa#1#2#3{\mathsf{ed}^w_{#1}(#2 \rightarrow #3)}
\def\w#1#2{\ensuremath w(#1 \mapsto #2)}

\newcommand{\Esigma}{\bar{\Sigma}}
\newcommand{\sqEsigma}{\bar{\Sigma}^2}

\newcommand{\edl}[2]{{\ed}(#1,{}^*\!#2^*)}

\def\edwl#1#2{\mathsf{ed}^w(#1 \rightarrow {}^*\!#2^*)}
\def\edws#1#2{\mathsf{ed}^w(#1 \rightarrow {}^*\!#2)}
\def\edwp#1#2{\mathsf{ed}^w(#1 \rightarrow #2^*)}

\newcommand{\eds}[2]{{\ed}(#1,{}^*\!#2)}
\newcommand{\edp}[2]{{\ed}(#1,#2^*)}

\def\pillar{{\tt PILLAR}\xspace}
\def\modelname{{\tt PILLAR}\xspace}
\def\lceOp#1#2{{\tt LCP}(#1, #2)}
\def\lcbOp#1#2{{\tt LCP}^R(#1, #2)}
\def\ipmOp#1#2{{\tt IPM}(#1, #2)}

\def\accOp#1#2{#1\position{#2}}
\def\ipmOpName{{\tt IPM}\xspace}
\def\accOpName{{\tt Access}\xspace}
\def\extractOpName{{\tt Extract}\xspace}
\def\lenOpName{{\tt Length}\xspace}

\def\cycEqOp#1#2{{\tt Rotations}(#1,#2)}
\newcommand{\ceil}[1]{\left\lceil #1 \right\rceil}
\newcommand{\floor}[1]{\left\lfloor #1 \right\rfloor}

\newcommand{\allOccs}{\textsc{\ref{keyxtssyqb}}\xspace}
\newcommand{\growFern}{\textsc{\ref{prob:grow_fern}}\xspace}
\newcommand{\verify}{\textsc{\ref{dputabwsur}}\xspace}
\newcommand{\algverify}{\texttt{Verify}}

\newcommand{\listalloccs}{\texttt{ListAllOccs}}
\newcommand{\apmsmallsed}{\texttt{GrowFern-SmallSED}}
\newcommand{\apmsmalled}{\texttt{GrowFern}}

\def\Q{\ensuremath Q^{\infty}}

\newcommand{\ktot}{\kappa}
\newcommand{\Z}{\mathbb{Z}}
\newcommand{\Zz}{\mathbb{Z}_{\ge 0}}
\newcommand{\Zp}{\mathbb{Z}_{> 0}}

\renewcommand{\S}{\mathcal{S}}
\def\Sb{\S_{\beta}}
\def\Sm{\S_{\mu}}
\def\Sf{\S_{\varphi}}

\newcommand{\Sr}{\hat{S}}

\def\torn{torsion\xspace}
\newcommand{\PDPM}{\textsc{\ref{tvssrtdhgs}}\xspace}
\newcommand{\DPM}{\textsc{\ref{kaqmrblnrj}}\xspace}
\newcommand{\BCDPM}{\textsc{\ref{pupolgyhot}}\xspace}
\newcommand{\I}{\mathcal{I}}
\newcommand{\J}{\mathcal{J}}

\newcommand{\tor}{{\sf tor}}
\newcommand{\shift}{{\sf shift}}

\renewcommand{\P}{\mathcal{P}}

\newcommand{\N}{\mathcal{N}}
\newcommand{\X}{\mathcal{X}}
\newcommand{\Y}{\mathcal{Y}}
\newcommand{\FF}{\mathcal{F}}
\newcommand{\QQ}{\mathcal{Q}}
\newcommand{\val}{\textsf{val}}

\def\sp#1{\textsf{Special}(#1)}
\def\spb#1{\textsf{Special}_{\beta\varphi}(#1)}
\def\rred#1{\textsf{Red}(#1)}

\newcommand{\Comp}{\textsf{Trim}}
\newcommand{\Comps}{\textsf{Tm}}
\newcommand{\comp}{\alpha}

\newcommand{\update}{\textsf{update}}
\newcommand{\rpl}{\textsf{Plain}}
\SetKwFunction{swaps}{TrimmedPM}
\SetKwFunction{boosted}{BoostedPM}
\newcommand{\unpack}{\textsf{Unpack}\xspace}
\newcommand{\blanked}{\textsf{Bleached}\xspace}
\newcommand{\blact}{\textsf{BleachAct}\xspace}

\newcommand{\SM}{\textsc{\ref{prob:sm}}\xspace}

\long\def\SMproblem{
    \begin{problem}[AlignedPeriodicMatches]{Aligned\-Periodic\-Matches{\tt($P$, $T$, $k$, $Q$, $\A_P$, $\A_T$, $w$)}}
    \label{prob:sm}

    \PInput{A pattern $P$ of length $m$,\\
    a positive integer $k$,\\
    a text~$T$ of length $n \in \fragmentco{m-k}{\ceil{\threehalfs m} + k}$,\\
    a primitive string $Q$ of length $q \coloneqq |Q| \le {m}/{128k}$,\\
    an edit-distance alignment
    $\A_P: P \onto Q^\infty\fragmentco{0}{y_P}$
        of cost $d_P\coloneqq \edl{P}{Q}=\ed(P,Q^*)\le 16k$,\\
        and an edit-distance alignment
        $\A_T : T \onto Q^\infty{\fragmentco{x_T}{y_T}}$ of cost
        $d_T \coloneqq \edl{T}{Q}\le 48k$, where $x_T \in \fragmentco{0}{q}$.\\
        Additionally, oracle access to a normalized weight function
    \(w : \sqEsigma \to \intvl{0}{W}\).}

    \POutput{The set $\OccW_k(P,T)$.} 
    \end{problem}}
    
\newcommand{\NPM}{\textsc{\ref{prob:npm}}\xspace}

    \long\def\NPMproblem{
    \begin{problem}[NewPeriodicMatches]{New\-Periodic\-Matches{\tt($P$, $T$, $k$, $d$, $Q$, $\A_P$, $\A_T$)}}
    \label{prob:npm}

    \PInput{A pattern $P$ of length $m$,
    an integer threshold $k \in \fragment{0}{m}$,
    a positive integer $d\ge 2k$,
    a text~$T$ of length $n \in \fragmentco{m-k}{\ceil{\threehalfs m} + k}$,
    a primitive string $Q$ of length $q \coloneqq |Q| \le {m}/{8d}$, an edit-distance alignment
    $\A_P: P \onto Q^\infty\fragmentco{0}{y_P}$
        of cost $d_P\coloneqq \edl{P}{Q}=\ed(P,Q^*)\le d$, and an edit-distance alignment
        $\A_T : T \onto Q^\infty{\fragmentco{x_T}{y_T}}$ of cost
        $d_T \coloneqq \edl{T}{Q}\le 3d$, where $x_T \in \fragmentco{0}{q}$.}
        
    \POutput{The set $\OccE_k(P,T)$ represented as $\cO(d^3)$ disjoint arithmetic progressions with difference $q$.}
    \end{problem}}
    
\usepackage{mathtools} 

\usepackage{pgf}
\def\alphav{128}
\def\betav{8}
\def\deltavN{3}
\def\deltavD{8}
\pgfmathsetmacro{\betavh}{int(\betav/2)}
\def\threehalfs{{}^3{\mskip -4mu/\mskip -3.5mu}_2\,}

\newcommand{\AG}{\mathsf{AG}}
\newcommand{\AGW}{\mathsf{AG}^w}
\newcommand{\AGw}{\AGW}

\newcommand{\oAGw}{\overline{\mathsf{AG}}^w}

\def\emptystring{\ensuremath\varepsilon}

\def\epsilon{\gamma}

\newcommand{\dist}{\mathsf{dist}}

\newcommand{\Als}{\mathbf{A}} 

\newcommand{\concat}{\ensuremath{\mathtt{concat}}}
\newcommand{\makestring}{\ensuremath{\mathtt{makestring}}}
\newcommand{\splitOp}{\ensuremath{\mathtt{split}}}

\def\pn#1{\textsc{#1}}

\newcommand{\PMED}{\textsc{\ref{bspmdzgntc}}\xspace}
\newcommand{\PMWED}{\textsc{\ref{erbgbxwrvp}}\xspace}
\newcommand{\meq}[1]{\mathop{\smash{\stackrel{#1}{=}}\vphantom{{}^{|^|}}}}

\def\dmat#1#2#3#4{\ensuremath D_{#1,#2}^{\mathstrut
#3}\smash{\raisebox{-.025ex}{\big|}}_{#4}^{\mathstrut}}

\def\fmats#1#2#3#4#5{\ensuremath \FF_{#1,#2}^{\mathstrut #3}\smash{\raisebox{-.025ex}{\big|}}_{#4}^{\mathstrut #5}}

\newcommand{\dP}{\dot{P}}
\newcommand{\dT}{\dot{T}}

\newcommand{\dmati}[4]{\prescript{\circ}{\circ}{\vphantom{D}}\dmat{#1}{#2}{#3}{#4}}
\newcommand{\dmatl}[4]{\prescript{}{\circ}{\vphantom{D}}\dmat{#1}{#2}{#3}{#4}}
\newcommand{\dmatt}[4]{\prescript{\circ}{}{\vphantom{D}}\dmat{#1}{#2}{#3}{#4}}

\newcommand{\fmatsi}[5]{\prescript{\circ}{\circ}{\vphantom{F}}\fmats{#1}{#2}{#3}{#4}{#5}}
\newcommand{\fmatsl}[5]{\prescript{}{\circ}{\vphantom{F}}\fmats{#1}{#2}{#3}{#4}{#5}}
\newcommand{\fmatst}[5]{\prescript{\circ}{}{\vphantom{F}}\fmats{#1}{#2}{#3}{#4}{#5}}
\newcommand{\fmi}{\prescript{\circ}{\circ}{F}}
\newcommand{\fml}{\prescript{}{\circ}{F}}
\newcommand{\fmt}{\prescript{\circ}{}{F}}
\newcommand{\talt}{\hat{t}}

\newcommand{\GPd}{\ensuremath G^{P,d}}
\newcommand{\GPSd}{\ensuremath G^{P,S,d}}
\newcommand{\VPSd}{\ensuremath V^{P,S,d}}
\newcommand{\VPd}{\ensuremath V^{P,d}}
\newcommand{\DPSd}{\ensuremath D^{P,S,d}}
\newcommand{\DPd}{\ensuremath D^{P,d}}

\newcommand{\F}{\mathcal{F}}

\newcommand{\on}[1]{\operatorname{#1}}
\newcommand{\core}[1]{\on{core}(#1)}

\newcommand{\dens}[1]{#1^{\square}}
\newcommand{\mds}{\mathsf{CMO}}

\makeatletter
\renewenvironment{cases}{%
    \matrix@check\cases\env@cases
}{%
    \endarray\right.%
}
\def\env@cases{%
    \let\@ifnextchar\new@ifnextchar
    \left\lbrace
        \def\arraystretch{1.1}%
        \array{@{\;}c@{\quad}l@{}}%
}
\makeatother

\SetKwFunction{locked}{Locked}
\newcommand\markf{\ensuremath \mathrm{mk}}
\newcommand\tradeoff{\ensuremath \eta}
\newcommand\ktotm{\ensuremath \hat{\ktot}}
\newcommand\lpref{\ensuremath \partial_P}
\newcommand{\Hv}{\mathbb{H}}
\newcommand{\light}{\mathbb{L}}
\newcommand{\LP}{\mathcal{L}^P}
\newcommand{\LT}{\mathcal{L}^T}
\newcommand{\res}[2]{#1_{#2}}
\newcommand{\Lck}{\mathcal{L}}

\newcommand{\tbd}{13}
\newcommand{\slack}{30}
\def\tbc{32}

\newcommand{\gen}{\mathit{gen}}
\newcommand{\G}{\mathcal{G}}
\newcommand{\Rp}{\mathbb{R}_{\ge 0}}
\newcommand{\Prp}[1]{P_{#1}}

\newcommand{\orig}{\textsf{orig}}
\newcommand{\lmore}{\mathcal{R}}
\newcommand{\mq}{\mathcal{F}}
\newcommand{\redsize}{\rho}

\title{Pattern Matching under Weighted Edit Distance}

\author{Panagiotis Charalampopoulos}{King's College London\\
 United Kingdom}{p.charalampopoulos@kcl.ac.uk}{https://orcid.org/0000-0002-6024-1557}{}
\author{Tomasz Kociumaka}{Max Planck Institute for Informatics\\Saarland Informatics
Campus\\Saarbrücken, Germany}{tomasz.kociumaka@mpi-inf.mpg.de}{https://orcid.org/0000-0002-2477-1702}{}
\author{Philip Wellnitz}{National Institute of Informatics\\The Graduate University for Advanced Studies, SOKENDAI\\Tokyo, Japan}{wellnitz@nii.ac.jp}{https://orcid.org/0000-0002-6482-8478}{}

\authorrunning{P. Charalampopoulos, T. Kociumaka, and P. Wellnitz}

\begin{document}
\pagenumbering{roman}
\maketitle
\begin{abstract}
    In \pn{Pattern Matching with \emph{Weighted} Edits}
    (\PMWED), we are given a pattern $P$ of length~\(m\), a text $T$ of length
    \(n\), a positive threshold $k$, and oracle access to a weight function that
    specifies the costs of edits (depending on the involved characters, and normalized
    so that the cost of each edit is at least $1$).
    The goal is to compute the starting positions of all fragments of $T$ that can be
    obtained from $P$ with edits of total cost at most $k$.
    \PMWED captures typical real-world applications more accurately than its
    unweighted variant (\PMED), where all edits have unit costs.

    Indeed, the textbook $\Oh(nm)$-time algorithm of Sellers [J. Algorithms'80], devised
    in the context of bioinformatics, already accounts for weights.
    Surprisingly, the understanding of \PMWED has not advanced in the last 45
    years.
    In contrast, significant milestones for \PMED include an $\Oh(nk)$-time
    algorithm by Landau and Vishkin [STOC'86, J.~Algorithms'89], an $\Oh(n+k^4 \cdot
    n/m)$-time algorithm by Cole and Hariharan [SODA'98, SICOMP'02], and a recent
    $\cOtilde(n+k^{3.5} \cdot n/m)$-time solution by Charalampopoulos, Kociumaka, and
    Wellnitz [FOCS'22].

    In this work, we examine whether these results can be lifted to \PMWED even though
    (1) the underlying algorithms rely on combinatorial properties specific
    to the unweighted edit distance, and (2) under standard fine-grained complexity
    assumptions, computing the weighted edit
    distance is strictly harder than computing the unweighted edit distance
    [Cassis, Kociumaka, and Wellnitz; FOCS'23]. We obtain three main results:
    \begin{itemize}
        \item a conceptually simple $\Ohtilde(nk)$-time algorithm for
            \PMWED, very different from that of Landau and Vishkin;

        \item a significantly more complicated $\cOtilde(n+k^{3.5} \cdot W^4\cdot
            n/m)$-time algorithm for \PMWED under the assumption that the
            weight function is a metric with integer values between $0$ and $W$; and

        \item an $\cOtilde(n+k^4 \cdot
            n/m)$-time algorithm for \PMWED for the case of arbitrary weights.
    \end{itemize}
    In the setting of metrics with small integer values, we nearly match the
    state of the art for \PMED where \(W=1\).\par
\end{abstract}
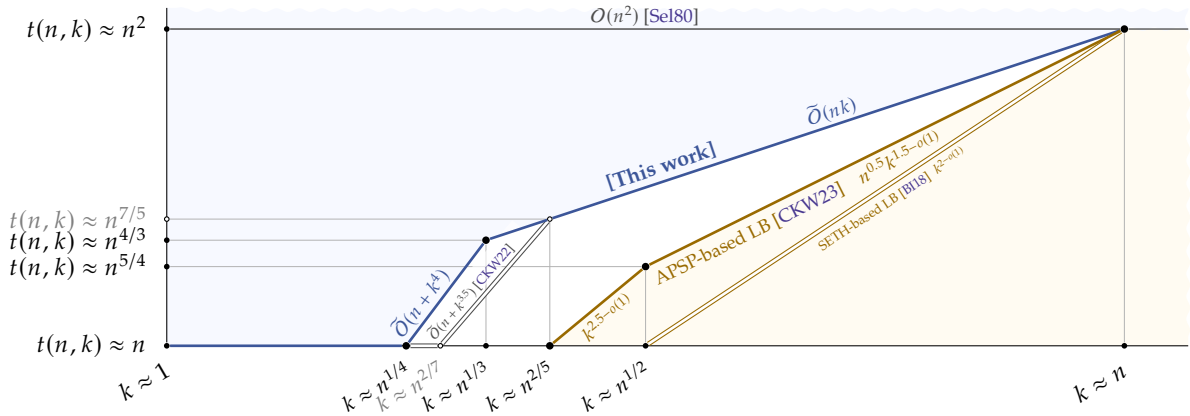
\begin{figure}[h!]
    \centering
    \scalebox{1}{

\begin{tikzpicture}

        \pgfmathsetmacro{\sx}{4.2}
        \pgfmathsetmacro{\sy}{1.4}

        \path[fill=color5!5] (\sx*1.2,0) -- (\sx*1.5, \sy*0.75) -- (\sx*3,\sy*3) --
        (\sx*3.2, \sy*3) decorate [decoration={snake,segment length=3mm, amplitude=.25mm}] {--(\sx*3.2,0)} -- cycle;
        \path[fill=color1!5] (0,0) -- (\sx*0.75,0) -- (\sx* 1,\sy*1) -- (\sx*3,\sy*3)
        decorate [decoration={snake,segment length=3mm,amplitude=.25mm}] {--  (\sx*3.2, \sy*3) --
        (\sx*3.2, \sy*3.2) -- (0,\sy*3.2)} -- cycle;

        \node[fill,circle, inner sep=.75pt] (o) at (0,0) {};
        \node[fill,circle, inner sep=1pt] (1) at (\sx*0.75,0) {};
        \node[fill=white,circle, inner sep=.6pt] (2) at  (\sx*0.85714285714,0) {};
        \node[fill,circle, inner sep=.75pt] (3) at  (\sx*1.5,0) {};
        \node[fill,circle, inner sep=1pt] (4) at  (\sx*1,\sy*1) {};
        \node[fill,circle, inner sep=.75pt] (4a) at  (\sx*1,0) {};
        \node[fill,circle, inner sep=.75pt] (4b) at  (0,\sy*1) {};
        \node[draw,fill=white,circle, inner sep=.6pt] (5) at  (\sx*1.2,\sy*1.2) {};
        \node[fill,circle, inner sep=1pt] (5a) at  (\sx*1.2,0) {};
        \node[fill,circle, inner sep=.6pt] (5b) at  (0,\sy*1.2) {};
        \node[fill,circle, inner sep=1pt] (6) at  (\sx*3,\sy*3) {};
        \node[fill,circle, inner sep=.75pt] (6a) at  (\sx*3,0) {};
        \node[fill,circle, inner sep=.75pt] (6b) at  (0,\sy*3) {};
        \node[fill,circle, inner sep=1pt] (7) at  (\sx*1.5,\sy*0.75) {};
        \node[fill,circle, inner sep=.75pt] (7b) at  (0,\sy*0.75) {};

        \draw[very thin,black!30] (3) -- (7) -- (7b);
        \draw[very thin,black!30] (6a) -- (6) -- (6b);
        \draw[very thin,black!30] (5a) -- (5) -- (5b);
        \draw[very thin,black!30] (4a) -- (4) -- (4b);

        \draw (0,\sy*3.2) -- (o) -- (\sx*3.2,0);

        \draw[line width=.5pt, black!70] (6b) -- node[midway, above, scale=0.7, inner sep=0.3ex]{$\Oh(n^2)$ \cite{S80}} (6) -- (\sx*3.2, \sy*3);
        \draw[thin, double=white, color=black!70, double distance=1pt] (1) -- (2)
        -- (5) node[midway, above=1pt, sloped,fill=white,fill
        opacity=.65,scale=.45,outer sep=.5ex] {\phantom{\(\tOh(n + k^{3.5})\)~\cite{ckw22}}} node[pos=.445, above=1pt,
        sloped,scale=.5,inner sep=0.3ex]{
        \(\tOh(n + k^{3.5})\)~\cite{ckw22}\qquad\mbox{}}; 
        \draw[color=c1, line width=1pt] (4) --
        (5) --
        node[pos=.2, above, sloped,
        scale=.8,inner sep=0.3ex] {\bf [This work]\vphantom{ \(()\)}}
        (6) node[midway, above, sloped,
        scale=.7,inner sep=0.3ex] {\(\Ohtilde(nk)\)}; 
        \draw[color=c1, line width=1pt] (o) -- (1)-- (4) node[pos=.318, above, sloped,scale=.7,inner sep=0.3ex] {\(\Ohtilde(n +
            k^4)\)}; 
        \draw[double=white, color=c5, double distance=1pt] (3) -- (6)
        node[midway, below=1pt,
        sloped,align=center,scale=.5,inner sep=0.5ex]
        {SETH-based LB  \cite{bi18} \quad \(k^{2-o(1)}\)}; 
        \draw[color=c5, line width=1pt] (5a) -- (7) node[pos=0.5, below,
        sloped,align=center,scale=.7,inner sep=0.5ex]
        {$k^{2.5-o(1)}$};
        \draw[color=c5, line width=1pt] (7) -- node[pos=0.3, below,
        sloped,align=center,scale=.7,inner sep=0.5ex]
        {APSP-based LB \cite{ckw23}\quad $n^{0.5}k^{1.5-o(1)}$} (6)
        ; 


        \def\slfrac#1#2{\ensuremath{}^{#1}\!/\!{}_{#2}}

        \node[anchor=north east,outer sep=.1em,scale=0.9,rotate=30] at (o)  { $k \approx 1$};
        \node[anchor=north east,outer sep=.1em,scale=0.8,rotate=30] at (1)  { $k \approx n^{1/4}$};
        \node[anchor=north east,outer sep=.1em,scale=0.8,rotate=30,black!50] at (2)  { $k \approx n^{2/7}$};
        \node[anchor=north east,outer sep=.1em,scale=0.8,rotate=30] at (4a)  { $k \approx n^{1/3}$};
        \node[anchor=north east,outer sep=.1em,scale=0.8,rotate=30] at (5a)  { $k \approx n^{2/5}$};
        \node[anchor=north east,outer sep=.1em,scale=0.8,rotate=30] at (3)  { $k \approx n^{1/2}$};
        \node[anchor=north east,outer sep=.1em,scale=0.9,rotate=30] at (6a)  { $k \approx n$};

        \node[anchor=east,outer sep=.5em,scale=0.9] at (o) { $t(n,k) \approx n$};
        \node[anchor=east,outer sep=.5em,scale=0.9] at (4b) { $t(n,k) \approx n^{4/3}$};
        \node[anchor=east,outer sep=.5em,scale=0.9,black!50] at (5b) { $t(n,k) \approx n^{7/5}$};
        \node[anchor=east,outer sep=.5em,scale=0.9] at (6b) { $t(n,k) \approx n^2$};
        \node[anchor=east,outer sep=.5em,scale=0.9] at (7b) { $t(n,k) \approx n^{5/4}$};

        \node[draw,fill=white,circle, inner sep=.6pt] (2) at  (\sx*0.85714285714,0) {};
        \node[draw,fill=white,circle, inner sep=.6pt] (5b) at  (0,\sy*1.2) {};
    \end{tikzpicture}

}
    \caption{The bounds on the running time $t(n, k)$ of algorithms for the \PMWED problem
        as a function of~$k$ for the important case when $m = \Theta(n)$.
        The scales are logarithmic and sub-polynomial factors are hidden.
        Doubled lines represent bounds for \PMED, included only when they differ by a
        polynomial factor from their \PMWED counterparts.
        The bounds for \PMED remain valid for \PMWED with integer edit weights between $0$ and
        $W=n^{o(1)}$.}\label{fig:plot}
\end{figure}
\clearpage
\thispagestyle{plain}
\tableofcontents

\newpage
\pagenumbering{arabic}

\section{Introduction}

\emph{Pattern matching} (string matching) models an action most of us perform daily when
we search for a term on the web or a word in a document: find all fragments of a
text $T$ that are identical to a pattern~$P$.
While time-optimal solutions for this  fundamental algorithmic problem date back to the
1970s~\cite{MP70,KMP77}, exact pattern matching remains very restrictive.
For example, a single typographical error in the pattern results in all occurrences
being missed.
Hence, both the theory and the practice communities have made great efforts to devise
algorithms that efficiently compute \emph{approximate occurrences}
of \(P\) in \(T\)---fragments of $T$ that are ``close'' to $P$.

Among many ways to quantify the notion of ``closeness'', one of the most natural and
well-studied is
to impose an upper bound $k$ on the edit distance between the pattern and its sought
approximate occurrences.
The (weighted) edit distance \(\edw{X}{Y}\) from a string \(X\) to a string \(Y\) is the
minimum cost of transforming \(X\) into \(Y\) by character edits (insertions, deletions,
and substitutions).
The cost of each edit depends on the edit type and the characters involved, and is
specified through
a weight function $w:\Esigma^2\to \mathbb{R}$, where $\Esigma$ denotes the union of the
alphabet $\Sigma$ and a special symbol $\varepsilon$ representing the empty string.
For $a,b\in \Sigma$, the cost of deleting $a$ is $\w{a}{\varepsilon}$, the cost of
inserting $b$ is $\w{\varepsilon}{b}$,
and the cost of substituting $a$ for $b$ is $\w{a}{b}$.

\begin{problem}[PMwWE]{Pattern Matching with Weighted Edits, PMwWE{\tt($P$, $T$, $k$, $w$)}}
    \label{erbgbxwrvp}
    \PInput{A pattern $P$ of length \(m\), a text $T$ of length \(n\), an
    integer $k > 0$, and oracle access to a weight function
    \(w : \sqEsigma \to \mathbb{R}\) normalized so that \(\w{a}{a} = 0\) and \(\w{a}{b}
    \geq 1\) for \(a \neq b \in \Esigma\).\footnotemark}
    \POutput{The set $\OccW_k(P,T)\coloneqq  \{i \mid \exists_{j \ge i} \;
    \edw{P}{T\fragmentco{i}{j}}\leq k\}$.}
\end{problem}

\footnotetext{Consistently with the literature~\cite{ColeH98,DGHKS23,ckw23,MBCT23},
    our normalization allows the complexity of algorithms to depend on~$k$.
Without this assumption, the weight function could be scaled arbitrarily.}

Applications of approximate pattern matching often incorporate domain-specific knowledge
in the weight function.
For example, in bioinformatics, scoring
matrices are used to
capture the frequency with which edits have been observed to happen
in nucleotide and peptide sequences (for instance, due to mutations).
Further scenarios where some edits are much more likely than others include optical
character recognition and accounting for spelling mistakes.

Unfortunately, apart from the initial \(\Oh(nm)\)-time algorithm for \PMWED~\cite{S80}, algorithmic
results and improvements were---perhaps surprisingly---limited to the important, but significantly
easier and somewhat theoretical special
case of the \emph{unweighted} edit distance \(\ed(X, Y)\), where the weight $\w{a}{b}$ for every two elements $a \neq b$ of $\Esigma$ is
uniformly set to 1.%
\footnote{This variant is often also named after Levenshtein~\cite{Lev66}.}

\begin{problem}[PMwE]{Pattern Matching with Edits, PMwE{\tt($P$, $T$, $k$)}}
    \label{bspmdzgntc}
    \PInput{A pattern $P$ of length \(m\), a text $T$ of length \(n\), and a positive
    integer $k$.}
    \POutput{The set $\OccE_k(P,T)\coloneqq  \{i \mid \exists_{j \ge i} \; \ed(P,
    T\fragmentco{i}{j})\leq k\}$.}
\end{problem}

Inspired by the landmark \(\Oh(nk)\) algorithm for \PMED~\cite{LandauV89}, we ask the following
natural question:

\begin{center}
    {{Does \PMWED admit an $\Ohtilde(nk)$-time solution?}\footnote{The $\cOtilde(\star)$
            notation suppresses factors that are polylogarithmic in
the length of the input strings.} }
\end{center}

Using a completely different and more complicated approach compared to~\cite{LandauV89},
we provide an affirmative answer.

\begin{restatable*}{mtheorem}{stalgmainnk}
    \dglabel{thm:stalgmainnk}[7.9.30-1,lem:PtoF,7.9.30-3,clm:preprocess,7.9.30-2](\PMWED is in $\Ohtilde(nk)$ time)
    \PMWED can be solved in $\cO((n\log m + m \log^2 m) \cdot k)$ time.
\end{restatable*}

While it may seem surprising that a result such as \cref{thm:stalgmainnk} is obtained only now,
good reasons prevented earlier, similar results.
Let us focus on the regime of fairly large
distances, $m^{0.5} \le k \le m$, where the $\Oh(nk)$-time complexity remains the state of
the art for \PMED.
Without much oversimplification, the Landau--Vishkin algorithm~\cite{LandauV89} can be
thought of as performing $\Oh(n/k)$ unweighted edit distance computations, in $\Oh(k^2)$
time each, between $P$ and various fragments of $T$.

To adapt this approach to \PMWED, one needs to compute weighted edit distances instead.
Applying the $\Oh(mk)$-time algorithm~\cite{Ukk85,Mye86}, which was the fastest known
until very recently, one merely recovers the $\Oh(nm)$ running time of \cite{S80}.
The $\Oh(m + k^5)$-time algorithm of Das, Gilbert, Hajiaghayi, Kociumaka, and
Saha~\cite{DGHKS23} is the earliest one that could be adapted to yield a running time of $o(nm)$ for \PMWED using this approach.
A subsequent $\Ohtilde(m + \sqrt{mk^3})$-time
solution by Cassis, Kociumaka, and
Wellnitz~\cite{ckw23} means that we can now expect to solve \PMWED in $\Ohtilde(n\sqrt{mk})$ time.
This complexity, however, constitutes a natural barrier due to the fine-grained
conditional lower bound of  $m^{0.5}k^{1.5-o(1)}$ for computing the weighted edit distance
when $m^{0.5} \le k \le m$~\cite{ckw23}.

Our key result behind \cref{thm:stalgmainnk} is an algorithm that, after $\Ohtilde(mk)$-time preprocessing
of one string ($P$ in our application), computes the weighted edit distance to any other string in $\Ohtilde(k^2)$ time.
Crucially, we use the SMAWK algorithm~\cite{SMAWK} for the
$(\min,+)$-multiplication of a Monge matrix with a vector as well as
the Multiple-Source Shortest Paths (MSSP) data structure of Klein~\cite{MSSP} for
efficient computations of distances in planar graphs.
Consult \cref{sec:techov_nk} for a more detailed description of our ideas and
\cref{sec:nk} for the full technical details.

A surprising consequence of \cref{thm:stalgmainnk} is that, for the regime of $m^{0.5} \le
k \le m$, the time complexity of \PMWED is strictly better understood than that of \PMED:
Whereas the state-of-the-art upper bounds essentially match, the fine-grained lower bound
for \PMWED exceeds its \PMED counterpart by a polynomial factor.
These two lower bounds, which we discuss next, are inherited from the problems of computing the
weighted and unweighted edit distance, respectively.

As proved in~\cite[Section 1]{ckw22}, the optimality of the $\Theta(m+k^2)$-time algorithm
by Landau and Vishkin~\cite{DBLP:journals/jcss/LandauV88} for unweighted edit distance (up
to sub-polynomial factors and conditioned on the Orthogonal Vector Hypothesis) implies an
$n+k^{2-o(1)}\cdot n/m$ lower bound for \PMED.
The $m^{0.5}k^{1.5-o(1)}$ lower bound for computing weighted edit distance~\cite{ckw23} is
conditioned on the All-Pairs Shortest Paths (APSP)
Hypothesis~\cite{DBLP:journals/jacm/WilliamsW18} and can be formally stated as
follows.

\begin{restatable*}[{\cite[Main~Theorem~2]{ckw23}}]{theoremq}{lbwed}
    \dglabel{thm:lb_wed}(Bounded weighted edit distance not in $m^{0.5}k^{1.5-o(1)}$ time, {\cite[Main~Theorem~2]{ckw23}})
    For real parameters $\delta > 0$ and $0.5 \le \kappa \le 1$,
    assuming the APSP Hypothesis, no algorithm,
    given strings $X,Y$ of length at most $m$, a real threshold $1 \le k \le m^{\kappa}$,
    and oracle access to a normalized weight function $w$,
    decides if $\edw{X}{Y}\le k$ in time $\Oh(m^{0.5+1.5\kappa - \delta})$.
\end{restatable*}

\Cref{thm:lb_wed} readily yields a $k^{1.5-o(1)}\cdot n/m^{0.5}$ lower
bound for \PMWED when $m^{0.5}\le k \le m$.

\begin{restatable*}{corollary}{lbpmwed}
    \dglabel{cor:lb_pmwed}[thm:lb_wed](\PMWED not in $\min(k^{1.5-o(1)}\cdot n/m^{0.5},
    k^{2.5-o(1)}\cdot n/m)$ time)
    For real parameters $\delta > 0$ and $0 < \kappa \le 1$,
    assuming the APSP Hypothesis, there is no algorithm that solves all instances of
    \PMWED satisfying $k \le m^{\kappa}$ in time
    $\Oh(n\cdot m^{\min(1.5\kappa-0.5,\,2.5\kappa-1)-\delta})$.
\end{restatable*}

In fact, the $\Oh(nk)$-time algorithm of Landau and Vishkin~\cite{LandauV89} remains the fastest known for \PMED
also in the case of $m^{0.4} \le k \le m^{0.5}$.
In this parameter range, \cref{cor:lb_pmwed} yields a non-trivial conditional lower bound
for \PMWED.
In contrast, no such non-trivial lower bound is known for \PMED{}---once again, the time complexity of \PMWED is
better understood than that of \PMED.
Naturally, we thus ask:
\begin{quote}
    \centering{Are there algorithms for \PMWED that nearly match the state of the art for
    \PMED when $k\le m^{0.4}$?}
\end{quote}

Before providing a (partially positive) answer to this question, let us highlight
milestone results for \PMED beyond the $\Oh(nk)$-time Landau--Vishkin
algorithm~\cite{LandauV89}.
We refer the interested reader to the thorough survey of Navarro~\cite{Navarro01} for a
review of other early results.
By exploiting the periodic structure of the involved strings and employing deterministic
coin tossing,
Sahinalp and Vishkin~\cite{SV96} developed an $\cOtilde(n
+
k^{25/3} \cdot n/ m^{1/3})$-time algorithm, which improves upon $\Oh(nk)$ for $k \ll m^{1/22}$.
Cole and Hariharan~\cite{ColeH98} extended this range to $k<m^{1/3}$ by  presenting an
$\cO(n+k^4 \cdot n/m)$-time solution shortly afterward.%
\footnote{
    Cole and Hariharan~\cite{ColeH98} incorrectly claimed that their result generalizes to
    \PMWED.
    Specifically, they wrote that ``the only change needed is to the Landau--Vishkin
    algorithm to take into account the differing costs''.
    However, in the presence of weights, there is no monotonicity in the diagonals of the
    dynamic programming table, and hence the Landau--Vishkin
    algorithm~\cite{LandauV86,LandauV89} does not work due to its greedy nature.
    This is confirmed by the lower bound of~\cite{ckw23} for the computation of bounded
    weighted edit distance, which implies that one cannot hope for a procedure that computes $J
    \cap \OccW_k(P,T)$ for a given interval $J$ of size $\cO(k)$ in time $\cO(k^2)$ after a
    linear-time preprocessing.}
The $\cO(nk)$-time algorithm of Landau and Viskhin~\cite{LandauV89} and the algorithm of
Cole and Hariharan~\cite{ColeH98} were the state of the art for over two decades, until
the recent breakthrough of Charalampopoulos, Kociumaka, and Wellnitz~\cite{ckw22}, who
presented an algorithm for \PMED that runs in time $\cOtilde(n + k^{3.5}\cdot n/m)$ and
remains the fastest known solution for \PMED when $k \ll m^{0.4}$.

The second main contribution of this work is an algorithm for \PMWED that matches the
$\Oh(n+k^{4}\cdot n/m)$ time complexity of~\cite{ColeH98} up to a small polylogarithmic
factor.
Consequently, our solutions for \PMWED essentially recover the state-of-the-art running
times for \PMED as they were before 2022.
Hence, $m^{0.25} \ll k \ll m^{0.4}$ is currently the only parameter regime when \PMED
has a significantly faster ($\cOtilde(n+k^{3.5} \cdot n/m)$-time) solution than \PMWED;
see \cref{fig:plot} for a visualization.

\begin{restatable*}{mtheoremq}{stalgmainkfour}
    \dglabel{thm:stalgmaink4}[rem:std,thm:main,thm:pilis](\PMWED in time $\Ohtilde(n + k^4\cdot n/m)$)
    \PMWED can be solved in $\cO(n + k^4\cdot n/m \cdot \log^2 (mk))$ time.
\end{restatable*}

In order to prove \cref{thm:stalgmaink4}, we exploit the structural characterization provided by
Charalampopoulos, Kociumaka, and Wellnitz~\cite{ckw20} for the approximate occurrences of
$P$ in $T$ under the \emph{unweighted} edit distance: $P$ has few approximate
occurrences or it is almost periodic.

As shown in \cite{ckw20}, in the first case, $\OccE_k(P,T)$ can be computed in $\cOtilde(n
+ k^4\cdot n/m)$ time and covered with $\Oh(k\cdot n/m)$ intervals of size $\Oh(k)$ each.
Due to the assumption that the weight function is normalized, we have $\OccW_k(P,T)
\subseteq \OccE_k(P,T)$, so it suffices to focus on the intervals covering
$\OccE_k(P,T)$.
One of the main results of Cassis, Kociumaka, and Wellnitz~\cite{ckw23} is that
$\edw{P}{T}\le k$ can be decided in $\Ohtilde(k^3)$ time assuming that one can check in
$\Ohtilde(1)$ time whether any two fragments of $P$ or $T$ match perfectly.
We generalize this result (leveraging the underlying structural insights) and show that,
if $|P|\le |T|+\Oh(k)$, then $\Ohtilde(k^3)$ time is sufficient to report all fragments of
$T$ at weighted edit distance at most $k$ from $P$.
For each interval $J$ in the cover of $\OccE_k(P,T)$, we apply this procedure for an
auxiliary text obtained by trimming $T$ to length $m+\Oh(k)$.
This way, after the standard $\Oh(n)$-time preprocessing for substring equality queries,
$\OccW_k(P,T)\cap J$ can be computed in $\Ohtilde(k^3)$ time, for a total of
$\Ohtilde(n+k^4\cdot n/m)$.

The main challenge then lies in the almost periodic case, which is also the bottleneck for
the \PMED problem.
The main contribution of \cite{ckw22} is how to compute
$\OccE_k(P,T)$ in that case in $\Ohtilde(n+k^{3.5}\cdot n/m)$ time using \emph{dynamic puzzle matching}.
We generalize an $\Ohtilde(n+k^4\cdot n/m)$-time algorithm provided in~\cite{ckw22} as a warm-up application of that technique.

The underlying idea is to preprocess a few subgraphs of the alignment graph (interpreted as \emph{puzzle
pieces}), and then stitch these subgraphs appropriately to reconstruct the distances
between $P$ and relevant fragments of~$T$.
Every piece can be associated with an $\Oh(k)\times \Oh(k)$ matrix so that stitching two
pieces corresponds to computing the $(\min,+)$-product of the underlying matrices.
Unweighted edit distance enjoys strong combinatorial properties that allow representing
each matrix in $\Oh(k)$ space and computing the $(\min,+)$-product in $\Ohtilde(k)$
time~\cite{Tis15}.
With weights, we have to resort to the naive $\Oh(k^2)$-space
representation and the $\Oh(k^2)$-time SMAWK~\cite{SMAWK} algorithm for computing the
$(\min,+)$-product of two Monge matrices.
In order to prevent the overall running time from increasing to $\Ohtilde(n+k^5\cdot
n/m)$, we refine the approach of \cite{ckw22} so that most matrix-matrix
multiplications are replaced by matrix-vector multiplications; the latter can be performed
in $\Oh(k)$ time using the SMAWK algorithm~\cite{SMAWK}.

Another challenge of a more technical nature is that the dynamic puzzle matching algorithm
of \cite{ckw22} is formalized in terms of restricted permutation matrices, meaningful only
for the unweighted edit distance.
Instead, we would ideally like to directly use distance matrices but, due to the bounded
width of the puzzle pieces, we must work with distances in various subgraphs that are difficult to describe explicitly.
In order to abstract away these technicalities, we introduce a notion of a \emph{fern
matrix}, which can be interpreted as storing distances in the entire alignment graph but
potentially failing to truthfully represent distances that are too large to be relevant
for our computations.
Further ideas behind our $\Ohtilde(n+k^4\cdot n/m)$-time algorithm are presented in
\cref{sec:techov_k4}.

A recent work of Gorbachev and Kociumaka~\cite{gk24} describes faster algorithms
for weighted edit distance with small integer weights:
they achieve a running time of $\Ohtilde(m+W\cdot k^2)$ if the cost of each edit belongs to $\fragment{1}{W}$.
For small integer weights, such as when $W=\Ohtilde(1)$, this complexity recovers (up to a polylogarithmic factor) the running time of the Landau--Vishkin algorithm~\cite{DBLP:journals/jcss/LandauV88} for unweighted edit distance.
This gives rise to the following question.

\begin{quote}
    \centering{Is it possible to solve \PMWED in $\cOtilde(n+k^{3.5}\cdot n/m)$ time when
    edit cost are integers between $0$ and $W=\Oh(1)$?}
\end{quote}

In \cref{sec:verify,sec:reduction,sec:dpm}, we incorporate the techniques of~\cite{gk24}
into a more complex version of our procedure behind \cref{thm:stalgmaink4} (building upon
the $\cOtilde(n+k^{3.5}\cdot n/m)$-time algorithm for \PMED~\cite{ckw22}, rather than its
simpler $\cOtilde(n+k^{4}\cdot n/m)$-time warm-up version) to derive the following result
under the extra assumption that $w$ is an integer metric weight function.%
\footnote{We did not try to optimize the exponent of $W$ in \cref{thm:stalgmaink35}.}

\begin{restatable*}{mtheoremq}{stalgmainkthreehalf}
    \dglabel{thm:stalgmaink35}[rem:std,thm:main,thm:pilis](\PMWED in time $\Ohtilde(n + k^{3.5} \cdot W^4 \cdot n/m)$ for integer metrics \(w:
    \sqEsigma \to \fragment{0}{W}\))
    \PMWED can be solved in $\Ohtilde(n + k^{3.5} \cdot W^4 \cdot n/m)$ time for metric weight
    functions $w$ with values in $\fragment{0}{W}$.
\end{restatable*}

It is worth noting that our solutions are meta-algorithms relying on a small set
of primitive operations.
Specifically, we work in the \pillar model of computation, introduced in \cite{ckw20} to
unify approximate pattern matching across several settings.
An efficient \pillar algorithm, together with an implementation of the primitive
operations in a given setting, yields an efficient algorithm for that setting.
\Cref{thm:stalgmaink4,thm:stalgmaink35} are corollaries of the following result.

\begin{restatable*}{mtheorem}{pillark}
    \dglabel"{thm:main}[def:puzzle,lem:red_to_SM,lem:solve_SM,lem:solve_int_SM,lem:verify](Efficient \pillar model algorithms for \PMWED)
    \PMWED on strings with $n < \threehalfs m + k$
    can be solved in the \pillar model
    \begin{itemize}
        \item in time $\cO(k^4 \log^2(mk))$ for general weights; and
        \item in time $\Ohtilde(k^{3.5}\cdot W^4)$ if $w$ is a metric weight function with values in $\fragment{0}{W}$.
    \end{itemize}
    The output set $\OccW_k(P,T)$ is represented as a collection of disjoint arithmetic progressions.
\end{restatable*}

In \cref{sec:concl}, we combine \cref{thm:main} with known efficient implementations of
the \pillar model to obtain efficient algorithms for approximate pattern matching in a
plethora of settings, such as: (1) the compressed setting, where $P$ and $T$ are given in
compressed form, and we wish to compute $\OccW_k(P,T)$ without decompressing them; (2) the
dynamic setting; and (3) the quantum setting.

\subsection*{Future Directions and Open Problems}

Let us briefly discuss potential future directions.
\begin{itemize}
    \item The most obvious yet the most important open problem stemming from this work is
        to narrow the gap between the upper and the lower bounds for \PMWED.
        In particular, it is unsatisfactory that the best lower bounds for both
        \PMED and \PMWED are direct corollaries of tight lower bounds for edit distance
        computation and weighted edit distance computation, respectively.
        It would be exciting to obtain a separation between \PMWED and the problem of
        weighted edit distance computation.
    \item Regarding faster algorithms, the most natural question is whether the
        $\Ohtilde(n+k^{3.5}\cdot n/m)$ time complexity of \cite{ckw22} for \PMED can be
        lifted to \PMWED.
        To derive this from an $\Ohtilde(k^{3.5}\cdot n/m)$-time \pillar algorithm (as in
        \cite{ckw22}), we would need a better combinatorial characterization of the set
        $\OccW_k(P,T)$:
        As shown in \cite{ckw20}, the set $\OccE_k(P,T)$ can be represented as the union
        of $\Oh(k^3\cdot n/m)$ arithmetic progressions; for $\OccW_k(P,T)$,
        \cref{thm:main} yields a representation using $\Oh(k^4\cdot n/m)$ progressions.
    \item Gorbachev and Kociumaka~\cite{gk24} also showed how to compute weighted edit distance
        in $\Ohtilde(m+k^{2.5})$ time for arbitrarily large integer weights.
        It remains open whether \PMWED can be solved faster in this regime compared to the
        general setting allowing fractional weights.
    \item The output of \PMWED, $\OccW_k(P,T)$, consists of the \emph{starting positions}
        of all fragments $T\fragmentco{i}{j}$ such that $\edw{P}{T\fragmentco{i}{j}} \leq k$.%
        \footnote{The algorithms underlying
            \cref{thm:stalgmainnk,thm:stalgmaink4,thm:stalgmaink35} can
            also compute, for each $i \in \OccW_k(P,T)$, the minimum distance $\min_j
            \edw{P}{T\fragmentco{i}{j}}$ and the position $j$ for which this minimum is attained.}
        A natural harder variant asks for all \emph{fragments}
        $T\fragmentco{i}{j}$ such that $\edw{P}{T\fragmentco{i}{j}} \leq k$.
        As noted in \cref{sec:verify}, a straightforward application of our tools yields
        an $\Ohtilde(nk^2)$-time solution of this problem.
        Can we improve this to $\Ohtilde(nk)$?
        For unweighted edit distance, the analogous variant can be solved in $\cO(nk)$
        time~\cite{DBLP:journals/siamcomp/LandauMS98}.
    \item Recently, Kociumaka, Nogler, and Wellnitz~\cite{KNW24,KNW25} investigated the
        communication complexity of \PMED.
        Their techniques also led to a quantum algorithm for \PMED with almost optimal query complexity
        and (for $k\ll m^{1/6}$) almost optimal running time.
        It would be interesting to investigate whether (any of) these results can be lifted to \PMWED.
        Our quantum algorithm for \PMWED (see~\cref{cor:quantum}) can be viewed as a first step in this direction.
\end{itemize}

\section{Technical Overview}
To describe our new algorithms, let us first briefly introduce crucial
notation.
A matrix $M\in \mathbb{R}^{p\times q}$ is \emph{Monge} if $M\position{i,j}+M\position{i + 1,j + 1} \le
M\position{i,j + 1}+M\position{i + 1,j}$ holds for all integers $0\le i <p-1$ and $0\le j <q-1$.
Given $\cO(1)$-time access to the entries of a matrix $M\in \mathbb{R}^{p\times q}$ and a
vector~$v\in  \mathbb{R}^q$,
the SMAWK algorithm~\cite{SMAWK} computes their $(\min,+)$-product $M\oplus v$ in $\cO(p+q)$ time.

For strings $X,Y\in \Sigma^*$ and a weight function $w: \sqEsigma \to \mathbb{R}_{\ge 0}$,
we define the \emph{alignment graph} $\AGw(X,Y)$ as a grid graph with vertices
$\fragment{0}{|X|}\times \fragment{0}{|Y|}$ and the following edges:
\begin{itemize}
    \item vertical edges $(x,y)\to (x+1,y)$ of cost $\w{X\position{x}}{\emptystring}$
        for $(x,y)\in \fragmentco{0}{|X|}\times \fragment{0}{|Y|}$;
    \item horizontal edges $(x,y)\to (x,y+1)$ of cost $\w{\emptystring}{Y\position{y}}$
        for $(x,y)\in \fragment{0}{|X|}\times \fragmentco{0}{|Y|}$;
    \item diagonal edges $(x,y)\to (x+1,y+1)$ of cost $\w{X\position{x}}{Y\position{y}}$
        for $(x,y)\in \fragmentco{0}{|X|}\times \fragmentco{0}{|Y|}$.
\end{itemize}
An \emph{alignment} of $X$ onto $Y$, denoted by $\A : X\onto Y$, is a path from $(0,0)$ to
$(|X|,|Y|)$ in $\AGw(X,Y)$.
In such a path, vertical edges correspond to deletions, horizontal edges correspond to
insertions, whereas diagonal edges correspond to matches and substitutions.
The weighted edit distance $\edw{X}{Y}$ is defined as the minimum cost (length) of an
alignment $\A : X\onto Y$.
We regard $(0,0)$ as the top-left vertex of the grid and
$(|X|,|Y|)$ as the bottom-right vertex of the grid, and we refer to the $|X|+1$ rows and the
$|Y|+1$ columns of the grid in a natural way.
As each insertion and deletion costs at least $1$, any path of cost at most $k$ intersects at most $k+1$ diagonals.

\subsection{The \texorpdfstring{\boldmath $\cOtilde(nk)$}{Õ(nk)}-time Algorithm}\label{sec:techov_nk}
Consider a collection $\S$ of $\cO(n/k)$ fragments of the form
$T\fragmentco{ik}{\max\{(i+2)k+m,n\}}$; observe that each $(k,w)$-error occurrences of $P$ and, in general, each fragment of $T$ of length at most
$m+k$ is fully contained in at least one fragment $S\in \mathcal{S}$.
Hence, it suffices to compute the
$(k,w)$-error occurrences of $P$ in all fragments $S\in \mathcal{S}$ and then merge the partial results.

Let us focus on computing $\OccW_k(P,S)$ for an arbitrary $S\in \mathcal{S}$.
One may use the MSSP data structure of Klein \cite{MSSP} on a width-\(\Oh(k)\)
diagonal band that is a subgraph of \(\AG^w(P, S)\) for this purpose (see
\cref{fig:techo_nk1}), that is, in order to compute each
top-to-bottom shortest-path distance in this band and check if it is at most
\(k\).\footnote{Given a planar graph with $N$ vertices and a distinguished infinite face,
the multiple-source shortest paths (MSSP) data structure~\cite{MSSP} can be constructed in
$\cO(N\log N)$ time and can report the distance from (or to) any vertex on the infinite
face to (or\ from) any other vertex in the graph in $\cO(\log N)$ time.}
However, this
diagonal band contains \(\Oh(mk)\) vertices---hence, this approach costs \(\Oh(mk \log m)\) time
per \(S \in \S\), or \(\Oh(mn \log m)\) time in total across all \(S \in \S\), which is prohibitive when $m \gg k$.

We improve upon this basic approach by making two observations.
\begin{itemize}
    \item Instead of computing the shortest-path distances all at once, we may
        equivalently split the (diagonal band subgraph of the) alignment graph along a
        separator $V$, compute
        shortest-path distances from the sources (top vertices of the band) to the
        separator and from the separator to the sinks (bottom vertices of the band), and
        then compose these
        distances using min-plus multiplication.
        Specifically, two instances of the the MSSP data structure provide $\cO(\log
        m)$-time access to, respectively, a distance matrix that captures all pairwise
        distances from the sources to the vertices of $V$ and a distance matrix that
        captures all pairwise distances from the vertices of $V$ to the sinks.
        Crucially, these distance matrices are \emph{Monge}, and this allows for min-plus
        multiplying them efficiently
        using the SMAWK algorithm~\cite{SMAWK}.%
        \footnote{As matrix-matrix multiplications are not fast enough for our
        purposes, we show how to make do with just matrix-vector multiplications which
        take time linear in the output size in our case.}
    \item In the non-trivial case when \(m \gg k\), any cost-\(k\) (weighted) alignment has to match exactly
        large parts of \(P\) and \(S\). We show that, with an \(\Oh(mk \log^2 m)\)-time preprocessing
        on \(P\), we may compute (a representation of) suitable distance matrices for any such pairs of equal parts for
        all \(S \in \S\) at once. In particular, for a single \(S \in \S\), the task then
        reduces to recomputing distances for graphs of total size \(\Oh(k^2)\)
        (the part corresponding to actual edits;  we can use Klein's algorithm \cite{MSSP}
        here, which takes time \(\Oh(k^2 \log m)\))
        plus querying \(\Oh(k)\) of our precomputed distance matrices \(\Oh(k)\) times
        each (the exactly matched
        parts between edits; where each such query costs \(\Oh(\log m)\)).
\end{itemize}

For a more detailed description of our algorithm, first observe that
if there is a $(k,w)$-error occurrence of $P$ in the fragment $S$, then $\edu{P}{S} \leq
4k$: We may extend an alignment $\A:P \onto S\fragmentco{f}{f'}$ with at most $k$ edits
to an alignment $\B:P \onto S$ by inserting at most~$3k$ excess characters, namely those comprising
$S\fragmentco{0}{f}$ and $S\fragmentco{f'}{|S|}$.

Now, after an $\Oh(m)$-time preprocessing of $P$ and~$T$, we check in $\Oh(k^2)$ time whether
$\edu{P}{S}\le 4k$ and, if so, compute an optimal unweighted alignment.
Let us henceforth assume that we have such an alignment $\B : P \onto S$ of unweighted
cost at most $4k$---which lets us
identify fragments of \(P\) that are matched exactly to fragments of \(S\).

For a simplified overview, let us assume that $\B:P \onto S$ consists only of
substitutions; this case illustrates the main ideas of the algorithm whilst avoiding
technicalities.

Now, observe that $\edw{P}{S\fragmentco{a}{b}}\leq k$ and $|S|=|P|$ imply $a \in \fragment{0}{k}$ and
$b \in \fragment{m-k}{m}$.
Hence, we consider $k+1$ source vertices and $k+1$ target vertices in
$\AG^w(P, S)$ as shown in~\cref{fig:techo_nk1}.
Moreover, the weighted edit distance of two strings $X$ and $Y$ is lower-bounded by
$\big||X|-|Y|\big|$,
so it suffices to consider paths confined in a band~$B$ comprised by $2k+1$ diagonals of
\(\AG^w(P, S)\) as shown in~\cref{fig:techo_nk1}.

Now, the alignment $\B : P \onto S$ yields at most $4k$ positions $i_0 < i_1 < \cdots <
i_{r-1}$
such that
\[
    S = P\fragmentco{0}{i_0} S\position{i_0} P\fragmentoo{i_0}{i_{1}} \cdots
S\position{i_r} P\fragmentoo{i_r}{m}.
\]
Let us also decompose the band $B$ into parts $B_0,B_1,\ldots, B_{s}$ so that $B_0$
corresponds to the first \(k\) columns (the left part) and
$B_{s}$ corresponds to the last $k$ columns (the right part),
whereas the intermediate parts correspond to substrings of $P$ and single characters
according to the decomposition of $S$; see~\cref{fig:techo_nk1}.
Further, let $V_j$ be the set of vertices shared by parts $B_{j}$ and $B_{j+1}$---these
vertices indeed form a separator of the band \(B\).

\begin{figure}[t!]
    \centering
    \includegraphics[page=4,height=.45\linewidth]{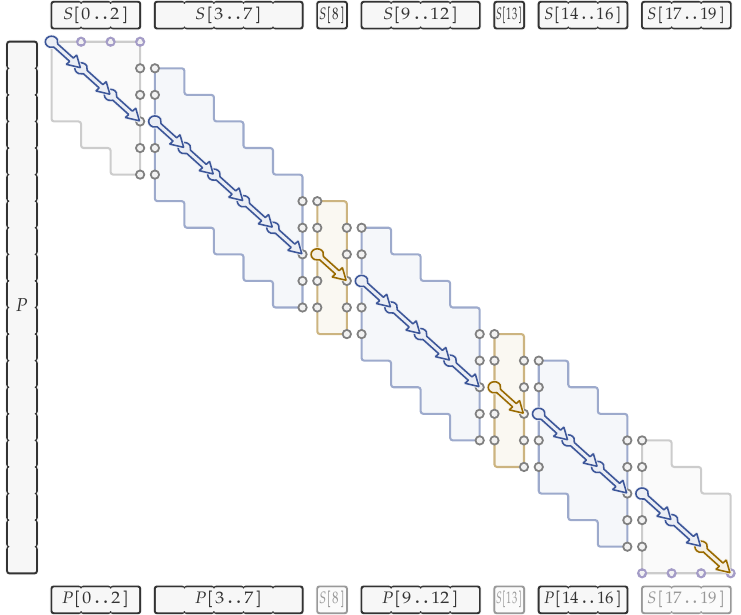}
    \caption{
        The band \(B\) of the alignment graph \(\AG^w(P, S)\) for a pattern $P$ of length $20$,
        a threshold $k=3$, and a fragment $S =
        P\fragmentco{0}{8}\texttt{a}P\fragmentco{9}{13}\texttt{c}P\fragmentco{14}{19}\texttt{c}$
        of $T$. The source vertices $(0,0), (0,1),
        (0,2), (0,3)$ in the top left and the target vertices $(20,17), (20,18), (20,19), (20,20)$
        in the bottom right are marked in light purple.
        The distinguished band~$B$ is decomposed to seven parts.
        The left and right parts are colored gray; \(B_i\) that also occur in
        \(\AG^w(P,P)\) are merged into \emph{pure parts} and are colored blue.
        We highlight the portal vertices on the boundaries of each part of \(B\).
    }\label{fig:techo_nk1}
    \vskip-4ex
\end{figure}

For a vertex $u$, let $d(u)$ denote the minimum distance (within the band $B$) from $u$ to
any target vertex.
We compute $d(u)$ for the vertices of the rightmost separator $V_{s-1}$ using the
algorithm of Klein \cite{MSSP} on \(B_s\). As \(B_s\) has a size of \(\Oh(k^2)\),
Klein's algorithm runs in time $\cO(k^2 \log m)$.

Next, we process the separators from right to left.
We compute the values $d(u)$ for all vertices $u \in V_{j-1}$ given $d(v)$ for all $v \in
V_{j}$ differently based on whether \(B_j\) is a subgraph
of \(\AG^w(P, P)\).

\begin{itemize}
    \item If \(B_j\) is not a subgraph of \(\AG^w(P, P)\), we can again use the algorithm of
        Klein \cite{MSSP} on \(B_j\) to obtain \(\Oh(\log n)\)-time random access to
        a Monge matrix \(D_{j-1,j}\) that stores the
        distances from the vertices of $V_{j-1}$ to the vertices of $V_j$ in \(B_j\).
        We then min-plus multiply \(D_{j-1,j}\) with a
        vector that stores $d(u)$ for $u \in V_j$ to obtain (a vector stores) \(d(u)\) for
        each \(u \in V_{j-1}\).
        Recall that, using the SMAWK algorithm, such a multiplication
        can be done in the time required for $\cO(k)$ accesses (at
        arbitrary positions) to the input matrices (interpreting the vector that stores
        $d(u)$ for $u\in V_{j}$ as an $m \times 1$ matrix).
    \item If $B_j$ is subgraph of \(\AG^w(P, P)\), we rely on a global $\cOtilde(mk)$-time preprocessing of $P$
        that constructs several instances of the MSSP data structure \cite{MSSP} for
        subgraphs of \(\AG^w(P, P)\). In particular, we obtain $\cO(\log m)$-time
        random access to two Monge matrices $D_{j-1, j-0.5}$ and  $D_{j-0.5,j}$
        with a min-plus product of \(D_{j-1,j}\).
        We first min-plus multiply \(D_{j-0.5,j}\) with a
        vector that stores $d(u)$ for $u \in V_j$ (to obtain an intermediate vector \(w\))
        and then min-plus multiply \(w\) with $D_{j-1, j-0.5}$ to obtain (a vector that
        stores) \(d(u)\) for each
        \(u \in V_{j-1}\).
        Recall that, using the SMAWK algorithm, such a multiplication
        can be done in the time required for $\cO(k)$ accesses (at
        arbitrary positions) to the input matrices.
\end{itemize}
Finally, we again rely in the algorithm of Klein \cite{MSSP} for the part between the
leftmost separator and the sources.

Over all $\cO(k)$ parts $B_{s-1},\ldots, B_1$, the above procedure takes $\cOtilde(k^2)$
time.
The final complexity is then $\Oh(n) + \cOtilde(mk) + |\mathcal{S}| \cdot \cOtilde(k^2) =
\cOtilde(n+mk + n/k \cdot k^2) = \cOtilde(nk)$.

\subsection{The \texorpdfstring{\boldmath $\cOtilde(n+k^4 \cdot n/m)$}{Õ(n+k⁴·n/m)}-time
Algorithm}\label{sec:techov_k4}
In this section, we provide an overview of our $\Ohtilde(n+k^4)$-time algorithm solving
\PMWED under the simplifying assumption that $n < \threehalfs m + k$ (as in
\cref{thm:main}).
The reduction from the general case relies on the \emph{standard trick}, which yields
$\Theta(n/m)$ independent instances of the restricted problem.

The starting point of our solution is the structural characterization of~\cite{ckw20} for
\PMED: unless~$P$ is almost periodic, it has few
approximate occurrences in $T$.

\paragraph*{The Non-Periodic Case}
In the first of the two cases of~\cite{ckw20}, the set $\{\lfloor{i/k}\rfloor : i \in
\OccE_k(P,T)\}$ is of size $\Oh(k)$ and can be computed in $\Oh(n+k^4)$ time.
This gives $\Oh(k)$ intervals of length $k$ whose union is a superset of
$\OccE_k(P,T)\supseteq \OccW_k(P,T)$.
The remaining task is to compute $\OccW_k(P,T)\cap I$ for each interval $I$ in this
collection.
For this, we design an efficient solution of the following \verify problem.

\begin{problem}[Verify]{Verify{\tt($P$, $T$, $k$, $I$, $w$)}}
    \label{dputabwsur}

    \PInput{A text $T$ of length $n$, a pattern $P$ of length $m$, an integer
    threshold $k >0$, an integer interval~$I$, and oracle access to a
    normalized weight function
    \(w:\sqEsigma \to \intvl{0}{W}\).}
    \POutput{$\OccW_k(P,T)\cap I$ and, for each $i\in \OccW_k(P,T)\cap I$, the distance $\min_{j\in
    \fragment{i}{n}}\edw{P}{T\fragmentco{i}{j}}$.}
\end{problem}

Our approach for \verify computes more than we need \emph{here}: it reports all relevant
$(k,w)$-error occurrences of $P$ in $T$ (rather than just their starting positions) and,
for each of them, outputs the corresponding weighted edit distance from $P$.
These extra features are useful later and come essentially~for~free.

In this overview, we present an $\Ohtilde(k^3)$-time \modelname model algorithm for \verify
for the case when $n\le m+2k$.
In general, we partition $I$ into $\ceil{|I|/k}$ pieces of size at most $k$, and, for each
piece $I'$, we trim the text to $T\fragment{\min I'}{\min\{n,\max I'+m-k\}}$.
Now, our problem generalizes the task considered in \cite{ckw23}: instead of computing the
edit distance from $P$ to the entire $T$, we need to consider $\Oh(k^2)$ fragments
$T\fragmentco{i}{j}$ with $j-i \ge m-k$.

In our solution, we apply the notion of \emph{self-edit distance}, introduced in
\cite{ckw23} to quantify the locality of edit distance.
The self-edit distance of a string $X$ is defined as the distance from $(0,0)$ to
$(|X|,|X|)$ in the \emph{unweighted} alignment graph $\AG(X,X)$ with edges on the main
diagonal removed.
The algorithm of \cite{ckw23} consists of a solution for the special case of
$\selfed(P)=\Oh(k)$ and a divide-and-conquer recursive procedure that constitutes a
reduction to this case.
The first of these two components generalizes to our setting quite easily, so
we focus on the case when $\selfed(P)=\omega(k)$.

In this case, we identify vertices $(p,t)$ and $(p',t')$ in $\AGW(P,T)$ with
$\selfed(P\fragmentco{0}{p})=\Oh(k)$ and
$\selfed(P\fragmentco{p'}{m})=\Oh(k)$
that lie on a $w$-optimal alignment $P\onto
T\fragmentco{i}{j}$ for every $(k,w)$-error occurrence of $P$ in $T$.
For every such $(k,w)$-error occurrence, the value $\edw{P}{T\fragmentco{i}{j}}$
decomposes to
\[\edw{P\fragmentco{0}{p}}{T\fragmentco{i}{t}}+\edw{P\fragmentco{p}{p'}}{T\fragmentco{t}{t'}}+\edw{P\fragmentco{p'}{m}}{T\fragmentco{t'}{j}}.\]
We compute the middle term using the algorithm of \cite{ckw23};
for the remaining two terms, we use our subroutine restricted to patterns of self-edit
distance $\Oh(k)$.

It is not obvious that such vertices \((p,t)\) and  \((p',t')\) exist.
Next, we show that the existence of such vertices follows from the central property of
self-edit distance: If alignments $\A,\B: X \onto Y$ of cost at most $k$ do not share
diagonal edges, then $\selfed(X)\le 2k$.
(The path in $\AG(X,X)$ witnessing $\selfed(X)\le 2k$ is obtained as the composition of
$\A$ and $\B^{-1}$; see~\cite{ckw23}.)

Fix an arbitrary $(k,w)$-error occurrence $T\fragmentco{i}{j}$ of $P$ in $T$ and a
$w$-optimal alignment $\A : P \onto T\fragmentco{i}{j}$.
We pick $(p,t),(p',t')\in \A$ so that
$\selfed(P\fragmentco{0}{p})=\selfed(P\fragmentco{p'}{m})=9k$.
We have $p \geq 9k-1$ as for any $q \in \fragment{1}{m}$, the $(0,0)\to (0,1)\to
(1,2)\to \cdots \to (q-1,q)\to (q,q)$ path in $\AG(P\fragmentco{0}{q},P\fragmentco{0}{q})$
has cost at most $q+1$.

To show that this selection is valid, consider another $(k,w)$-error occurrence
$T\fragmentco{i'}{j'}$ of $P$ in~$T$ and a $w$-optimal alignment $\A' : P \onto
T\fragmentco{i'}{j'}$.
We claim that $\A$ intersects $\A'$ both before $(p,t)$ and after $(p',t')$.
The subpaths of both $\A$ and $\A'$ are shortest paths between these two
intersection points,
so we may reroute $\A'$ (without increasing its cost) using $\A$ to pass through both
$(p,t)$ and $(p',t')$.

Due to $n \le m + 2k$, both of the (unique) extensions of $\A$ and $\A'$ to alignments $\B,\B' : P
\onto T$ have \emph{unweighted} cost at most $4k$.
Let $(\bar{p},\bar{t})$ be the first point in the intersection $\B\cap \B'$ with
$\bar{p}>0$.
If $\bar{p}=1$, then, in particular, $\bar{p}<9k-1\leq p$, and hence $(\bar{p},\bar{t})\in \A\cap \A'$
lies before $(p,t)$.
Otherwise, $\B$ and $\B'$ do not share any diagonal edge between $(0,0)$ and
$(\bar{p},\bar{t})$ and, restricted to this segment, still have unweighted costs of at
most $4k$ each.
Thus, $\selfed(P\fragmentco{0}{\bar{p}})\le 8k < 9k = \selfed(P\fragmentco{0}{p})$, which
implies $\bar{p}<p$ as $\selfed(\cdot)$ is monotone.

The construction in \cref{sec:verify} follows the same spirit but cannot rely on the
(unknown) arbitrary $(k,w)$-error occurrence $T\fragmentco{i}{j}$ of $P$ in $T$.
Instead, we define $(p,t)$ and $(p',t')$ using $(k,w)$-error occurrences
of a prefix $P_p$ and a suffix $P_s$ of $P$, respectively, with small self-edit distance; see
\cref{fig:crossing_alignments}.
These two $(k,w)$-error occurrences can be efficiently computed using the subroutine restricted to
patterns of self-edit distance $\Oh(k)$.

\paragraph*{The Periodic Case}
In the almost-periodic case, the characterization of \cite{ckw20} yields a primitive string
$Q$ of length $|Q|\ll m/k$ and an alignment from $P$ to a substring of $Q^\infty$
with unweighted cost $\Oh(k)$.
Further results of \cite{ckw20} let us trim $T$ (without losing $k$-error occurrences)
so that it admits a similar alignment.

To keep this overview simple, we assume that both substrings of $Q^\infty$ are integer
powers of $Q$, that is, $\edu{P}{Q^{p}}=\Oh(k)$ and $\edu{T}{Q^{t}}=\Oh(k)$ hold for some
integer exponents $p,t\gg k$.
We decompose $P = P_1\cdots P_p$ and $T = T_1 \cdots T_{t}$, where the $P_i$s and $T_j$s
are potentially edited copies of the string $Q$.
In this case, all elements of $\OccE_k(P,T)$ must be within distance $\cO(k)$ from the
start of some~$T_j$.
Following~\cite{ckw22}, we then think of shifting $P$ over $T$ in steps of roughly $|Q|$ positions at
a time, in order to align the starting position of $P$ with the start of each~$T_j$.
Intuitively, for subsequent $j\in \fragment{0}{t-p}$, we wish to compute
$\OccE_k(P_1P_2\cdots P_p, T_jT_{j+1} \cdots T_{j+p-1})$ (up to extending $T_j$ to the
left and~$T_{j+p-1}$ to the right by $\cO(k)$ positions).

\subparagraph*{Applying \DPM to \PMED.}
The solution of \cite{ckw22} is formalized through the \DPM problem, whose
(oversimplified) formulation asks to maintain a dynamic sequence $\I = (U_1,V_1)(U_2,V_2)
\cdots (U_z,V_z)$ of \emph{puzzle pieces} subject to updates (insertions, deletions, and
substitutions of pieces) so that, upon a query, one can efficiently compute $\OccE_k(U,V)$
for $U=U_1U_2\cdots U_z$ and $V=V_1V_2 \cdots V_z$.
In reality, the strings $V_i$ overlap, but we mostly ignore these overlaps in this overview.

A naive reduction to DPM constructs and queries $\I_j \coloneqq (P_1, T_{j+1}) (P_2, T_{j+2})
\cdots (P_p, T_{j+p})$ for subsequent integers $j\in \fragment{0}{t-p}$.
Since all but $\Oh(k)$ strings $P_i$ and $T_j$ are equal to $Q$, updating $\I_j$ to
$\I_{j+1}$ requires $\Oh(k)$ substitutions.
Each of these substitutions can be processed in $\Ohtilde(k)$ time each.
Unfortunately, this is too much as $t-p$ can be $\Theta(n/k)$.

To lay ground for further improvements, the algorithm of \cite{ckw22} issues different
updates, most of which increment or decrement the exponents of \emph{runs of plain pieces}
of the form $(Q,Q)^e$ for $e\in \mathbb{Z}_{\ge 0}$.

\begin{example}
    Suppose that $P=Q^{99}P_{100}Q^{99}$ and $T=Q^{149}T_{150}Q^{149}$.
    Then,
    \[
        \I_j = \begin{cases} (Q,Q)^{99}\cdot (P_{100},Q)\cdot (Q,Q)^{49-j}\cdot
            (Q,T_{150})\cdot (Q,Q)^{49+j} & \text{if }j\in \fragment{0}{49},\\
            (Q,Q)^{99}\cdot (P_{100},T_{150})\cdot (Q,Q)^{99} & \text{if
            }j=50,\\(Q,Q)^{149-j}\cdot (Q,T_{150})\cdot (Q,Q)^{j-51}\cdot (P_{100},Q)\cdot (Q,Q)^{99}
        & \text{if }j\in \fragment{51}{100}.\end{cases}
    \]
    For $j\in \fragment{0}{48}\cup \fragment{51}{99}$, one can update $\I_j$ to $\I_{j+1}$
    by incrementing or decrementing the exponents of two runs of plain pieces.
    Effectively, this allows moving the special piece $(Q,T_{150})$ to the left.
\end{example}

Formally, as we issue updates to iterate over $\I_j$ for subsequent indices $j\in
\fragment{0}{t-p}$, all but $\Oh(k^2)$ updates increment or decrement the exponent of
a run of plain pieces.
Furthermore, each of the $\Oh(k^2)$ runs existing throughout the lifetime of the algorithm
admits a simple structure of updates:
its exponent is first incremented for some subsequent indices $j$, then it remains fixed,
and then it is decremented.

This insight allows reducing the number of updates when combined with the following
observation:
For every run $(Q,Q)^e$ of plain pieces, the exponent can be capped to
$e=\min\{e,\alpha\}$ for $\alpha=\Oh(k)$ without affecting the output of DPM queries.
If, instead of $\I_j$, we maintain sequences $\I'_j$ with the exponents capped at
$\alpha$, each run needs to be updated only $\Oh(\alpha)$ times, for a total of $\Oh(k^2
\alpha)$ updates that can be processed in $\Ohtilde(k^3\alpha)=\Ohtilde(k^4)$~time in total.

The key insight to achieve $\Ohtilde(k^{3.5})$ time in \cite{ckw22} is that picking $\alpha\le k$ preserves answers
to \textsc{DPM} queries for all but $\Oh(k^2/\alpha)$ sequences $\I_j$.
Unfortunately, the underlying combinatorial arguments
do not seem to generalize to the case of arbitrary weight functions,
so we fix $\alpha=\Theta(k)$ for the remainder of
this discussion.

\subparagraph*{Applying \BCDPM to \PMWED.}
The aforementioned reduction to \DPM translates from \PMED to \PMWED with minor changes.
Unfortunately, the overall running time increases to $\Ohtilde(k^5)$ because each DPM update
takes $\Ohtilde(k^2)$ time instead of $\Ohtilde(k)$ time.

Intuitively, for each puzzle piece $(X,Y)$, the DPM data structure of \cite{ckw22} stores
a matrix $D_{X,Y}$ whose entries represent distances $\edu{X}{Y\fragmentco{i}{j}}$ for all
$k$-error occurrences $Y\fragmentco{i}{j}$ of $X$ in $Y$.
Since $|Y|\le |X|+\Oh(k)$, these are matrices of size $\Oh(k)\times \Oh(k)$.
The composition of two pieces $(X,Y)$ and $(X',Y')$ can be implemented by
$(\min,+)$-multiplying matrices $D_{X,Y}$ and $D_{X',Y'}$, due to the fact that $Y$ and
$Y'$ actually have a sufficiently large overlap.
Under mild technical assumptions, the matrix $D\coloneqq\bigoplus_{i=1}^{z} D_{U_i,V_i}$ can be
used to identify all $k$-error occurrences of $U=U_1\cdots U_z$ in $V=V_1\cdots V_z$.
For the starting positions, that is, $\OccE_k(U,V)$, it suffices to compute the
$(\min,+)$-product of $D$ and a $\mathbf{0}$-vector of appropriate dimension.

For unweighted edit distance (or, to be precise, for deletion distance, which admits a
simple isometric embedding from unweighted edit distance), all the matrices in question
are \emph{unit-Monge matrices} and the results of Tiskin~\cite{Tis15} allow storing each
of them in $\Ohtilde(k)$ space so that the $(\min,+)$-products can be computed in
$\Ohtilde(k)$ time.
Thus, the global product $D$ can be maintained in $\Ohtilde(k)$ time assuming the
individual matrices $D_{X,Y}$ have been precomputed for all plausible puzzle pieces $(X,Y)$.
In the presence of weights, the matrices are merely Monge matrices, which means that we
have to resort to an $\Oh(k^2)$-space representation and the $\Oh(k^2)$-time
$(\min,+)$-multiplication algorithm of~\cite{SMAWK}.

What comes to the rescue is that the queries require $D\oplus \mathbf{0}$ and the SMAWK
algorithm is able to compute a matrix-vector product in $\Oh(k)$ time.
We can offload some update cost to the query algorithm: if we lazily represent $D$ as a
product of $d$ matrices, the query time increases from $\Oh(k)$ to $\Oh(kd)$.
The challenge is to avoid matrix-matrix products for most updates without compromising the
query time.

For this, recall that all but $\Oh(k^2)$ updates adjust the exponent of a run of plain
pieces, and each such run is active (meaning that its exponent changes between subsequent
queries) for two blocks of $\alpha=\Oh(k)$ queries each.
Thus, we precompute the matrix for each possible exponent of a run of plain pieces and, in
the maintained decomposition of $D$, we make sure that each active run is represented with
a separate matrix that can be substituted in $\Ohtilde(1)$ time when the exponent changes.
The remaining $\Oh(k^2)$ updates take $\Ohtilde(k^2)$ time each.
The query time becomes $\Oh(k\cdot (1+r))$, where $r$ is the number of active runs
at the moment, but each of the $\Oh(k^2)$ runs is active during $\Oh(k)$ queries, so the
total query time is $\Oh(k^4)$.

In \cref{sec:reduction}, we state the \BCDPM problem (a variant of DPM with appropriately
extended interface) and provide an efficient reduction from \PMWED.
The solution to BCDPM is described in \cref{sec:dpm}.
The main challenge there is to formally define the matrices $D_{X,Y}$ so that they can be
efficiently constructed (using the techniques behind \verify) and so that their $(\min,+)$-product is
meaningful.
In~\cite{ckw22}, a custom \emph{restriction} operation is introduced so that $D_{X,Y}$ is
uniquely defined yet computable in $\Oh(k^2)$ time even if $X$ and $Y$ are very long.
Unfortunately, this operation relies on the \emph{seaweed monoid} of Tiskin~\cite{Tis15},
which loses its meaning in the presence of weights.
In this work, we instead introduce a notion of a \emph{fern matrix} that leaves a lot of
freedom in the particular choice of $D_{X,Y}$ (in particular, for the values that do not
correspond to $(k,w)$-error occurrences of $X$ in $Y$) and is nevertheless sufficient for us to
recover the $(k,w)$-error occurrences of $U=U_1\cdots U_z$ in $V=V_1\cdots V_z$ from the
$(\min,+)$-product of fern matrices \(D_{U_i,V_i}\).

\subsection{The \texorpdfstring{\boldmath $\cOtilde(n+k^{3.5}\cdot W^4 \cdot n/m)$}{Õ(n+k³˙⁵·W⁴·n/m)}-time
Algorithm}\label{sec:techov_k35}

We obtain \cref{thm:stalgmaink35}
by combining the approach outlined in~\cref{sec:techov_k4} with a generalization of
techniques from \cite{ckw22} and several results and insights from the recent (static and
dynamic) algorithms for edit distance with small integer weights~\cite{gk24}.
On the way to achieving this result, we make the following technical contributions.
First, we devise an $\cOtilde(k^2W)$-time \pillar algorithm for \verify, generalizing both
the conditionally optimal $\cO(k^2)$-time \pillar algorithm for \verify in the unweighted
setting and the aforementioned algorithm of Gorbachev and Kociumaka~\cite{gk24}.
This solution to \verify is instrumental in both the non-periodic case and the periodic
case of the algorithm underlying \cref{thm:stalgmaink35}.
For instance, in the periodic case, it allows us to compute optimal weighted alignments
from $P$ and $T$ to substrings of $Q^\infty$; the algorithm outlined in
\cref{sec:techov_k4} utilizes their unweighted counterparts.
Second, we adapt an algorithm from \cite{ckw20} for the computation of so-called
\emph{locked fragments} to the weighted setting;
this might be of independent interest as locked fragments have also recently been used in
obtaining solutions for a circular variant of \PMED~\cite{kCPM}.
Third, we lift combinatorial insights for \PMED from~\cite{ckw22} to \PMWED---our
assumption that the weight function is a metric is required in this part of the solution
since it ensures that the weighted edit distance satisfies the triangle inequality.

\section{Preliminaries}
\dglabel{sec:prelim}[7-16-1,fct:ali,def:alignment-graph,lem:oagw,7-16-3,lem:MSSP,thm:smawk,7-16-2,fact:simple](Definitions and generally useful facts
and results)

\paragraph*{Sets and Arithmetic Progressions}\label{sec:prel}
For integers \(i\) and \(j\),
we write $\fragment{i}{j}$ to denote the set $\{i, \dots, j\}$, and we write
$\fragmentco{i}{j}$ to denote the set $\{i ,\dots, j - 1\}$.
Similarly, we define $\fragmentoc{i}{j}$ and $\fragmentoo{i}{j}$.

For integers \(a\), \(d\), and \(\ell\) with $d\ne 0$ and $\ell\ge 0$,
the set $\{ a + j \cdot d \mid j \in \fragmentco{0}{\ell}\}$
is an \emph{arithmetic progression} with starting value $a$, difference $d$, and
length~$\ell$.
Whenever we use arithmetic progressions in an algorithm, we store them
as triples $(a,d,\ell)$ consisting of their starting value, difference, and length.

For a set $X\subseteq \mathbb{Z}$ and an integer $s \in \mathbb{Z}$, we write
$s+X$ and $X+s$ for the set $\{s+x : x\in X\}$ of all elements of $X$ incremented by $s$.

\paragraph*{Strings}

We write $T=T\position{0}\, T\position{1}\cdots T\position{n-1}$ to denote a \textit{string} of
length $|T|=n$ over an alphabet $\Sigma$. The elements of~$\Sigma$ are called \textit{characters}.
We write $\varepsilon$ to denote the \emph{empty string}.
In a slight abuse of notation, we write $\Esigma \coloneqq \Sigma
\cup \{\emptystring\}$ for the alphabet $\Sigma$ that
is extended by $\emptystring$ (which we use to represent the lack of a character).

A string $P$ is a \emph{substring} of a string~$T$
 if we have $P=T\position{i}\cdots T\position{j-1}$ for some integers $i,j$ with $0\le i \le j \le |T|$.
In this case, we say that there is an \emph{exact occurrence} of~$P$ at position $i$
in~$T$, or, more simply, that $P$ \emph{exactly occurs in} $T$.
We write $T\fragmentco{i}{j}$ for this particular occurrence of \(P\) in \(T\),
which is formally a \emph{fragment} of $T$ specified by the two endpoints $i,j$.
For notational convenience, we may also refer to this fragment as $T\fragment{i}{j-1}$,
$T\fragmentoc{i-1}{j-1}$, or $T\fragmentoo{i-1}{j}$.
Two fragments (perhaps of different strings) \emph{match} if they are occurrences of the same strings.

A \emph{prefix} of~a string $T$ is a fragment that starts at position~$0$ (that is, a
prefix is a fragment of~the form $T\fragmentco{0}{j}$ for some $j\in \fragment{0}{|T|}$).
A \emph{suffix} of~a string $T$ is a fragment that ends at position ${|T|-1}$ (that is,
a suffix is a fragment of~the form $T\fragmentco{i}{|T|}$ for some $i\in \fragment{0}{|T|}$).

For two strings $U$ and $V$, we write $UV$ or $U\odot V$ to denote their concatenation;
we also write $\bigodot_{i=1}^{m} X_i$ for $X_1 \odot \cdots \odot X_{m}$ and set $U^k
\coloneqq \bigodot_{i=1}^{k} U$.
Similarly, $U^\infty$ denotes an infinite string obtained by concatenating
infinitely many copies of~$U$. At certain points, we may access
such an infinite repetition of \(U\) also at negative positions; hence for an integer \(j \in
\fragmentco{0}{|U|}\) and a (possibly negative) integer \(i\), we formally set \(U^{\infty}\position{i
\cdot |U| + j} \coloneqq U\position{j}\).

A string $T$ is \emph{primitive} if it cannot be expressed as $T=U^k$
for any string~$U$ and any integer~$k > 1$.
A positive integer $p$ is a \emph{period} of~a string $T$ if $T\position{i} =
T\position{i + p}$ for
all $i \in \fragmentco{0}{|T|-p}$. We refer to the smallest
period as \emph{the period} $\per(T)$ of~the string.

For a string $T$ of length $n$, we define the following \emph{rotation} operations.
The operation $\rot(\star)$ takes as input a string, and moves its last character to the
front; that is,~$\rot(T) \coloneqq  T\position{n-1}T\fragment{0}{n-2}$.
The inverse operation $\rot^{-1}(\star)$ takes as input a string and
moves its initial character to the end; that is,~$\rot^{-1}(T) \coloneqq  T\fragment{1}{n-1}T\position{0}$.
Observe that a primitive string $T$ does not match any of~its non-trivial rotations,
that is, we have $T=\rot^j(T)$ if and only if $j \equiv 0 \pmod{n}$.

\paragraph*{Alignment Graphs and (Weighted) Edit Distances}
A \emph{weight function} $w : \sqEsigma \to \mathbb{R}$ is a (not necessarily symmetric)
function that assigns cost to character edits.
A weight function $w$ is \emph{normalized} if $\w{a}{a}=0$ and $\w{a}{b} \ge 1$ holds for all $a,b\in \Esigma$;
as in previous works, we work only with normalized weight functions.
Moreover, we assume that $w : \sqEsigma \to \intvl{0}{W}$, where $W$ is some upper bound
known to the algorithms.

\begin{figure}[t]
    \begin{subfigure}[t]{.675\linewidth}
        \centering
        \includegraphics[page=1,width=\linewidth]{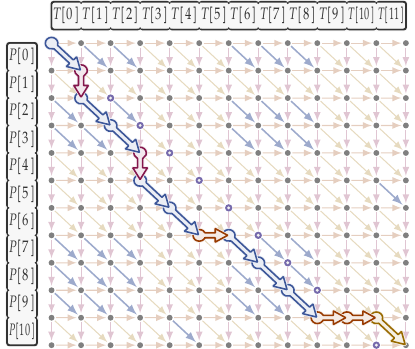}
        \caption{An example of an alignment graph of strings \(P\) and \(T\).
            Blue edges correspond to edges of weight \(0\), all other edges have a nonzero
            weight.
            We highlight an alignment
            \(\A : P \onto T\) (a path from \((0,0)\) to \((|P|, |T|)\))
            that makes \(6\) edits to transform \(P\) into \(T\). This alignment
            corresponds to the first part of the alignment in \cref{7.10.12-1}.
            The vertices of the \(0\)-th diagonal are depicted purple (unless they are
            part of \(\A\)).}
            \label{sfqpycukkb}
    \end{subfigure}\hfill%
    \begin{minipage}[b]{.31\linewidth}
     \begin{subfigure}[t]{\linewidth}
        \centering
        \includegraphics[page=2,scale=1.28]{figs/g02}
        \caption{A detailed view of a deletion and the associated
        cost.}
    \end{subfigure}\vskip.5ex
      \begin{subfigure}[t]{\linewidth}
        \centering
        \includegraphics[page=3,scale=1.28]{figs/g02}
        \caption{A detailed view of an insertion and the associated
        cost.}
    \end{subfigure}\vskip.5ex
       \begin{subfigure}[t]{\linewidth}
        \centering
        \includegraphics[page=4,scale=1.28]{figs/g02}
        \caption{A detailed view of a substitution (yellow; if \(P\position{p} \neq T\position{t}\)) and a
            match (blue; if \(P\position{p} = T\position{t}\)) and the associated costs.}
    \end{subfigure}
    \end{minipage}
    \vskip-1ex
    \caption{The alignment graph of two strings \(P\) and \(T\); as well as a detailed
        view of the edge costs. To obtain the corresponding augmented alignment graph, we add
        a back-edge of weight \(W\) for every depicted edge.\vskip-4ex}
    \label{oiuoqannlf}
\end{figure}

Observe that \(w\) does not need to satisfy the triangle inequality nor does \(w\) need to be
symmetric.

\begin{definition}[The alignment graph \(\AGw(P, T)\) of a weight function \(w\) and
    strings \(P\), \(T\),~{\cite{ckw23}}]
    \dglabel{def:alignment-graph}
    For strings $P, T\in \Sigma^*$ and a weight function $w: \sqEsigma \to \intvl{0}{W}$,
    we define the \emph{alignment graph} $\AGw(P, T)$ as a grid graph with vertices
    $\fragment{0}{|P|}\times \fragment{0}{|T|}$ and the following edges:
    \begin{itemize}
        \item vertical edges $(p,t)\to (p+1,t)$ of cost $\w{P\position{p}}{\emptystring}$
            for $(p,t)\in \fragmentco{0}{|P|}\times \fragment{0}{|T|}$,
        \item horizontal edges $(p,t)\to (p,t+1)$ of cost $\w{\emptystring}{T\position{t}}$
            for $(p,t)\in \fragment{0}{|P|}\times \fragmentco{0}{|T|}$, and
        \item diagonal edges $(p,t)\to (p+1,t+1)$ of cost $\w{P\position{p}}{T\position{t}}$
            for $(p,t)\in \fragmentco{0}{|P|}\times \fragmentco{0}{|T|}$.
            \qedhere
    \end{itemize}
\end{definition}

We visualize the alignment graph $\AG^w(P, T)$ as a grid graph with $|T|+1$ columns
and $|P|+1$ rows that grows down and to the right, that is,
the vertex \((0,0)\) is at the top-left, the vertex \((0,|T|)\) is at the top-right,
and the vertex \((|P|,|T|)\) is at the bottom-right.
Consult \cref{oiuoqannlf} for a visualization.

A \emph{diagonal} is a subgraph of the form \((i,j) \to (i+1,j+1) \to (i + 2, j + 2) \to
\cdots\) that starts at the top-left boundary of \(\AGw(P, T)\) and extends to the bottom-right
boundary of \(\AGw(P, T)\); we typically refer to this diagonal as the $(j-i)$-th diagonal.
\begin{example}
    The 0-th diagonal is the subgraph \((0,0)\to (1,1) \to (2,2) \to \cdots\),
    the $1$-st diagonal is the subgraph \((0,1)\to (1,2) \to (2,3) \to \cdots\), and
    the \((-1)\)-st diagonal is the subgraph \((1,0)\to (2,1) \to (3,2) \to \cdots\).
    Also consider \cref{sfqpycukkb} for a visualization.
\end{example}

Observe that every vertex $(p,t)$ of \(\AGw(P, T)\) lies on exactly one diagonal---the
diagonal $t-p$; hence, we may refer to \emph{the} diagonal passing through a vertex.
We say that a path in \(\AGw(P, T)\)
\emph{intersects} a diagonal if the path shares a vertex with the diagonal.

The alignment graph allows for a concise definition of an \emph{alignment}.

\begin{definition}
    \dglabel{7-16-1}(An alignment $\A: P\fragmentco{p}{p'} \onto T\fragmentco{t}{t'}$ for
    strings \(P\) and \(T\))
    For strings $P, T\in \Sigma^*$ and a weight function $w: \sqEsigma \to \intvl{0}{W}$,
    an \emph{alignment} of $P\fragmentco{p}{p'}$ onto
    $T\fragmentco{t}{t'}$, denoted by $\A: P\fragmentco{p}{p'} \onto T\fragmentco{t}{t'}$,
    is a path from \((p,t)\) to \((p',t')\) in \(\AGw(P, T)\), interpreted as a sequence of vertices.
    The \emph{cost} $\edwa{\A}{P\fragmentco{p}{p'}}{T\fragmentco{t}{t'}}$ of an alignment
    \(\A\) is the cost of the corresponding path in
    \(\AGw(P, T)\).

    We write
    $\Als(P\fragmentco{p}{p'}, T\fragmentco{t}{t'})$
    for the set of all alignments of $P\fragmentco{p}{p'}$ onto $T\fragmentco{t}{t'}$.
\end{definition}

For an alignment $\A = (p_j, t_j)_{j=0}^r\in \Als(P\fragmentco{p}{p'},
T\fragmentco{t}{t'})$ and an index $j \in \fragmentco{0}{r}$, we say that
\begin{itemize}
    \item $\A$ \emph{deletes} $P\position{p_j}$,
        denoted by \(P\position{p_j} \ponto{\A} \varepsilon\),
        if $(p_{j+1}, t_{j+1}) = (p_j+1, t_j)$;
    \item $\A$ \emph{inserts} $T\position{t_j}$,
        denoted by \(\varepsilon \ponto{\A} T\position{t_i}\),
        if $(p_{j+1}, t_{j+1}) = (p_j, t_j+1)$;
    \item $\A$ \emph{aligns} $P\position{p_j}$ to $T\position{t_j}$, denoted by
        $P\position{p_j} \ponto{\A} T\position{t_j}$
    if $(p_{j+1}, t_{j+1}) = (p_j+1, t_j+1)$;
    \item $\A$ \emph{matches} $P\position{p_j}$ with $T\position{t_j}$
        if $P\position{p_j}
        \ponto{\A} T\position{t_j}$ and
    $P\position{p_j} = T\position{t_j}$;
    \item $\A$ \emph{substitutes} $P\position{p_j}$ with $T\position{t_j}$
        if $P\position{p_j}
        \ponto{\A} T\position{t_j}$ but
    $P\position{p_j} \neq T\position{t_j}$.
\end{itemize}
Insertions, deletions, and substitutions are jointly called (character) \emph{edits}.

The \emph{breakpoint representation} of an alignment \[
    \A=((0,\varepsilon), (0,\varepsilon)), ((p_j,P\position{p_j}),(t_j, T\position{t_j}))_{t=0}^r\in \Als(P\fragmentco{p}{p'},
T\fragmentco{t}{t'}), ((|P|, \varepsilon), (|T|, \varepsilon)) \]
is the subsequence of $\A$ consisting of pairs $((p_j,
P\position{p_j}),(t_j,T\position{t_j}))$ such that
$\A$ does not match $P\position{p_j}$ with $T\position{t_j}$, padded with dummy entries at
the front and at the back.
Observe that the size of the breakpoint representation is $2+\ed_\A(P,T)$ and that it
can be used to retrieve the entire alignment:
for any two consecutive elements $((x', \cdot), (y', \cdot)),((x,\cdot),(y, \cdot))$ of
the breakpoint representation, it
suffices to add $(x-\delta,y-\delta)$ for $\delta \in \fragmentoo{0}{\max(x-x',y-y')}$.
Consult \cref{byutjsdglx} for a visualization of an example.
Unless stated otherwise, throughout this work, we represent alignments via their
breakpoint representations without always explicitly specifying it.

\begin{figure}[t]
    \begin{subfigure}[t]{\linewidth}
        \centering
        \includegraphics[page=1,width=\textwidth]{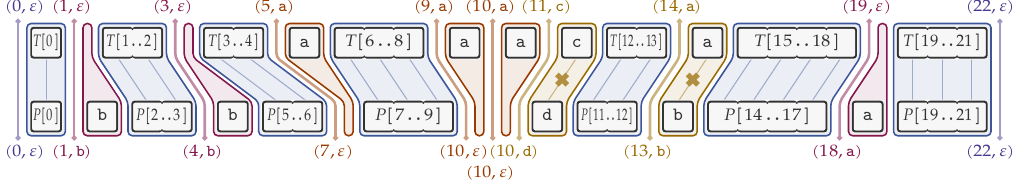}
        \caption{ An alignment \(\A: P \onto T\) of cost \(8\). Observe that each
        breakpoint (except for the first and last) corresponds to an edit operation that
        immediately follows.  }
        \label{7.10.12-1}
    \end{subfigure}
    \begin{subfigure}[t]{\linewidth}
        \centering
        \includegraphics[page=2,scale=1]{figs/g01}
        \includegraphics[page=3,scale=1]{figs/g01}
        \caption{The alignment \(\A\) translates a partition of
            \(P = P\fragmentco{0}{6} P\fragmentco{6}{10} P\fragmentco{10}{13}
            P\fragmentco{13}{22}\)
            into a partition
            of \(T = \A(P) =
            \A(P\fragmentco{0}{6}) \A( P\fragmentco{6}{10}) \A( P\fragmentco{10}{13} )
            \A( P\fragmentco{13}{22} ) =
            T\fragmentco{0}{4} T\fragmentco{4}{9} T\fragmentco{9}{14}
            T\fragmentco{14}{22}\).}
    \end{subfigure}%
    \caption{The breakpoint representation of an alignment \(\A : P \onto T\) and the
        corresponding strings, as well as an example how \(\A\) translates fragments (and
        partitions) of \(P\). The (edit operations of the) alignment are depicted as
        colored boxes; breakpoints are depicted as vertical lines separating the boxes of
        the alignment. For edit operations, we depict the characters involved; otherwise
        we depict the fragments of \(P\) and \(T\) that are matched.
        Consult \cref{sfqpycukkb} for a visualization of the alignment graph of
        \(P\fragmentco{0}{11}\) and \(T\fragmentco{0}{12}\), as well as the corresponding
        part of \(\A\).
    }\label{byutjsdglx}
\end{figure}

Given $\A = (p_j, t_j)_{j=0}^r \in \Als(P, T)$, we define the \emph{inverse alignment}
as $\A^{-1} \coloneqq (t_j, p_j)_{j=0}^r \in \Als(T, P)$.

We define the \emph{weighted edit distance} of strings $P, T \in \Sigma^*$
with respect to a weight function $w$ as
\[\edw{P}{T} \coloneqq \min_{\A \in \Als(P, T)} \edwa{\A}{P}{T}; \]
that is, as the length of a shortest path from \((p,t)\) to \((p',t')\) in \(\AGw(P, T)\).

\begin{fact}[{\cite[Fact 2.5]{DGHKS23}}]
    \dglabel{7-16-2}(The weighted edit distance is a metric for metric weight functions,~\cite[Fact 2.5]{DGHKS23})
    If $w$ is a metric on $\Esigma$, then $\mathsf{ed}^w$ is a metric on $\Esigma^*$.
\end{fact}

For an integer $k \geq 0$, we also define a capped version
\[
    \edwk{k}{P}{T} \coloneqq
        \begin{cases}
            \edw{P}{T} & \text{if } \edw{P}{T} \leq k,\\
            \infty & \text{otherwise}.
        \end{cases}
\]
We say that an alignment $\A \in \Als(P, T)$ is $w$-optimal (or just optimal if there is
no ambiguity) if $\edwa{\A}{P}{T} = \edw{P}{T}$.

For an alignment $\A:X\fragmentco{x}{x'}\onto Y\fragmentco{y}{y'}$ and a fragment
$X\fragmentco{\bar{x}}{\bar{x}'}$ that is contained in $X\fragmentco{x}{x'}$,
we write $\A(X\fragmentco{\bar{x}}{\bar{x}'})$ for the fragment
$Y\fragmentco{\bar{y}}{\bar{y}'}$ that is contained in $Y\fragmentco{y}{y'}$ which $\A$
aligns against $X\fragmentco{\bar{x}}{\bar{x}'}$.
To avoid ambiguities, we formally set
\[\bar{y} \coloneqq \min\{\hat{y} : (\bar{x},\hat{y})\in \A\}\quad\text{and}\quad
    \bar{y}' \coloneqq \left\{\begin{array}{c l}
            y' & \text{if }\bar{x}' = x',\\
            \min\{\hat{y}' : (\bar{x}',\hat{y}')\in \A\} & \text{otherwise}.
    \end{array}\right.\]
This particular choice satisfies the following decomposition property.
\begin{fact}[Alignments transfer decomposistions, {\cite{ckw22,ckw23}}]
    \dglabel{fct:ali}
    For any alignment $\A$ of $X$ onto $Y$ and a decomposition $X=X_1\cdots X_t$ into $t$
    fragments, $Y=\A(X_1)\cdots \A(X_t)$ is a decomposition into $t$ fragments with
    \[\edwa{\A}{X}{Y}\ge \sum_{i=1}^t \edw{X_i}{\A(X_i)}.\]

    Further, if \(\A\) is an optimal alignment, then we have equality:
    $\edwa{\A}{X}{Y} = \sum_{i=1}^t \edw{X_i}{\A(X_i)}$.
\end{fact}

For a text $T$, a pattern $P$, a weight function $w$, and an integer threshold $k$,
a \emph{\((k,w)\)-error occurrence of \(P\) in \(T\)} is a position \(i\) such that
for some \(j \ge i\) we have \(\edw{P}{T\fragmentco{i}{j}}\leq k\).%
\footnote{We assume that $k$ is integer as we often
wish to define objects in terms of $k$ and the usage of $\lceil k \rceil$ would create
unnecessary clutter. Our algorithms may be adapted to handle any $k \in \mathbb{R}_{\ge 0}$.}
We write \[
    \OccW_k(P,T)\coloneqq  \{i \mid \exists_{j \ge i} \; \edw{P}{T\fragmentco{i}{j}}\leq k\}.
\] for the set of all $(k,w)$-error occurrence of \(P\) in \(T\).

\begin{remark}
    If for every $a, b \in \Esigma$ we have that $\w{a}{b} = 1$ if $a \neq b$
    and $\w{a}{b} = 0$ otherwise, then $\edw{X}{Y}$ corresponds
    to the standard \emph{unweighted} edit distance (also known as Levenshtein distance~\cite{Lev66}).
    For this case, we drop the superscript $w$ in
    $\AGw(X,Y)$,  $\ed^w$, $\ed^w_\A$, $\ed^w_{\le k}$, and $\OccW_k(P,T)$.

    Additionally, since the unweighted edit distance is symmetric, we may write $\ed(X,Y)$
    instead of $\ed(X \to Y)$, and, similarly for $\ed_\A$ and $\ed_{\leq k}$.
\end{remark}

For any two strings $X,Y$, any normalized weight function $w$, and any alignment $\A: X
\onto Y$, we have $\ed^w_\A(X \to Y) \geq \ed_{\A}(X,Y)$.
This simple observation yields the following fact.

\begin{fact}
    \dglabel{fact:simple}(For normalized weight functions, weighted occurrences are unweighted occurrences)
    For any strings $P$, $T$, any integer threshold $k$, and any normalized weight function
    $w$, we have $\OccW_k(P,T) \subseteq \OccE_k(P,T)$.
\end{fact}

We write \[
    \edp{S}{T} \coloneqq  \min \{\ed(S,T^\infty\fragmentco{0}{j}) : j \in \mathbb{Z}_{\ge
0}\}\]  for the minimum edit
distance between a string $S$ and any prefix of~a string $T^\infty$.
Further, we write
\[\edl{S}{T}\coloneqq  \min\{\ed(S,T^\infty\fragmentco{i}{j}) : i, j \in \mathbb{Z}_{\ge
0}, i \le j\}\]
for the minimum edit distance between $S$ and any substring of~$T^\infty$,
and we set
\[\eds{S}{T} \coloneqq   \min\{\ed(S,T^\infty\fragmentco{i}{j|T|}) : i, j \in
\mathbb{Z}_{\ge 0}, i \le j|T|\}.\]

Finally, we discuss a notion for the similarity of strings from a family of strings.
\begin{definition}
    \dglabel{def:median_ed}(The median edit distance \(\ed(\S)\) of a family of strings
    \(\S\))
    The \emph{median edit distance}
    of a family $\S$ of strings over an alphabet $\Sigma$
    is \[
        \ed(\S) \coloneqq \min_{\Sr\in \Sigma^*} \sum_{S\in \S}\ed(S,\Sr).
        \qedhere
    \]
\end{definition}
\begin{remark}[Bounded median edit distance implies small size and small pairwise
    edit distance]
    \dglabel{dzefaaseda}[def:median_ed]
    As the unweighted edit distance has a triangle inequality, a bound \(\ed(\S) \le d\)
    implies \(||S|-|T|| \le \ed(S,T) \le d\) for any \(S, T \in \S\).
    Further, as two non-equal strings have an edit distance of at least \(1\), we also have
    \(|\S| = \Oh(d)\).
\end{remark}

\paragraph*{Planar Graphs, the Monge Property, and \boldmath $(\min,+)$-Multiplication}

Alignments graphs are \emph{planar}. The \emph{multiple-source shortest paths} (MSSP) data
structure of Klein~\cite{MSSP} represents all shortest path trees rooted at the vertices
of a single face $f$ in a planar graph $G$ of size $n$ using a persistent dynamic tree.

\begin{theoremq}[Efficient MSSP on planar graphs, {\cite{MSSP}}]
    \dglabel{lem:MSSP}
    Given a directed planar graph $G$ of size $n$ with a distinguished face $f$ and
    non-negative edge weights, we can construct in $\cO(n \log n)$ time an $\cO(n\log
    n)$-space data structure that, for any two vertices $u,v\in V(G)$, at least one of which
    is lies on $f$, computes the distance $\dist_G(u,v)$ in $\Oh(\log n)$ time.
    Moreover, the shortest $u\leadsto v$ path $P$ can be reported in $\Oh(|P|\log
    \log\Delta(G))$ time, where $\Delta(G)$ is the maximum degree of \(G\).
\end{theoremq}

Pairwise distances of vertices that lie on the outer face of a strongly connected planar
graph satisfy the \emph{Monge property}.

\begin{fact}[Distance matrices of planar graphs are Monge, {\cite[Section 2.3]{FR06}}]
    \dglabel{fct:monge}
    Consider a strongly connected directed planar graph $G$ with non-negative edge
    weights. For vertices
    $u_0,\ldots,u_{p-1},v_{q-1},\ldots,v_0$ that (in this cyclic order) lie on the outer
    face
    of $G$, define a $p\times q$ matrix $D$ with $D\position{i,j}=\dist_G(u_i,v_j)$.

    Then, $D$ satisfies the \emph{Monge property}, that is,
    \[D\position{i,j}+D\position{i',j'} \le D\position{i,j'}+D\position{i',j}\]
    holds for all integers $0\le i\le i' <p$ and $0\le j\le j'<q$.
\end{fact}

The composition of shortest paths in graphs can be interpreted as the
$(\min,+)$-multiplication of the corresponding distance matrices.
Formally, given an $p \times q$ matrix $A$ and an $q \times r$ matrix $B$, their $(\min,+)$-product is
the $p \times r$ matrix $C=A\oplus B$ is defined as
\[C\position{i,j} \coloneqq \min_{k\in \fragmentco{0}{q}} \big( A\position{i,k} + B\position{k, j}\big).\]
By $A^{\oplus (t)}$,
we denote the matrix equal to the $(\min,+)$-product of $t$ copies of a square matrix $A$.

In a seminal work~\cite{SMAWK}, Aggarwal, Klawe, Moran, Shor, and Wilber presented an
efficient algorithm to compute the $(\min,+)$-product of Monge matrices given
$\Oh(1)$-time random access to \mbox{their entries}.

\begin{theorem}[\cite{SMAWK}]
    \dglabel{thm:smawk}(SMAWK algorithm for efficient $(\min,+)$-product computation,~\cite{SMAWK})
    We can compute the $(\min,+)$-product of an $p\times q$ Monge matrix
    and an $q\times r$ Monge matrix using $\Oh(p\cdot r\cdot (1+\log(q/\max(p,r))))$
    queries to single positions of the input matrices.
    \lipicsEnd
\end{theorem}

Furthermore, if matrices $A$ and $B$ are Monge, then their $(\min,+)$-product $A \oplus B$
is also Monge.

\begin{fact}[{\cite[Theorem 2]{Tis15}}]
    \dglabel{fct:monge-product-is-monge}(The $(\min,+)$-product of Monge matrices is Monge, \cite[Theorem 2]{Tis15})
    Let $A, B$, and $C$ be matrices, such that $A \oplus B = C$.
    If $A$ and $B$ are Monge, then $C$ is also Monge.
\end{fact}

\paragraph*{Augmented Alignment Graphs}

Although the alignment graph $\AGw(P,T)$ is planar, it is not strongly connected, so we
cannot apply \cref{fct:monge}.
To circumvent this issue, we augment it with back edges of sufficiently large weight.

\begin{definition}[The augmented alignment graph $\oAGw(P, T)$ for a weight function
    \(w\) and strings \(P\) and \(T\), {\cite{gk24}}]
    \dglabel{7-16-3}
    For strings $P, T\in \Sigma^*$ and a weight function $w: \sqEsigma \to \intvl{0}{W}$,
    we define the \emph{augmented alignment graph} $\oAGw(P, T)$ obtained from $\AGw(X,Y)$
    by adding, for every
    edge of $\AGw(P, T)$, a back edge of weight $W+1$.
\end{definition}

\begin{remark}
    For our purpose of turning the alignment graph into a strongly connected graph, we
    technically do not need the diagonal back edges: insisting on both horizontal and
    vertical edges elegantly captures the special cases of either string being empty, but
    thereby also renders the diagonal edges superfluous.
    Nevertheless, we keep the diagonal back edges to stay in line with \cite{gk24}; which
    in particular allows us to use their results in an opaque manner.
\end{remark}

We proceed with a couple of useful properties of (shortest paths in) the augmented
alignment graph.

\begin{lemmaq}[Basic properties of the alignment graph,~{\cite[Lemma 5.2]{gk24}}]
    \dglabel{lem:oagw}
    Consider strings $P, T \in \Sigma^{*}$ and a weight function $w : \sqEsigma \to \intvl{0}{W}$.
    Any two vertices $(p, t)$ and $(p', t')$ of the graph $\oAGw(P,T)$ satisfy the
    following properties.
    \begin{description}
        \item[Monotonicity.] Shortest paths $(p, t) \leadsto (p', t')$ in
            $\overline{\AG}^w(P, T)$ are (non-strictly) monotone in both coordinates.
        \item[Distance preservation.] If $p \le p'$ and $t \le t'$, then
            \[\dist_{\AG^w(P, T)}((p, t), (p', t')) = \dist_{\overline{\AG}^w(P, T)}((p, t),
            (p', t')).\]
        \item[Path irrelevance.] If ($p \le p'$ and $t \ge t'$) or ($p \ge p'$ and $t \le
            t'$), every path from $(p, t)$ to $(p', t')$ in $\oAGw(P, T)$, that is monotone in
            both coordinates, is a shortest path between these two vertices.
            \qedhere
    \end{description}
\end{lemmaq}

\begin{remark}
    \label{rem:basic}
    Consider a path \(\P\) of cost at most \(k\) in the augmented alignment graph
    \(\oAGw(P, T)\) that originates at a vertex \((0,i)\) and
    terminates at a vertex \((|P|, j)\) for some $i,j\in \fragment{0}{|T|}$.

    Since $w$ is normalized, non-diagonal edges of \(\oAGw(P, T)\) are of cost at least
    $1$ each and connect adjacent diagonals.
    Consequently, \(\P\) intersects at most \(k+1\) subsequent diagonals: it may only
    intersect the diagonals that are at most \(k\) to the left or to
    the right of both the \(i\)-th diagonal of the origin \((0,i)\) and the \((j-|P|)\)-th
    diagonal of the target \((|P|,j)\).
    Thus, every vertex \((p,t)\) of the path \(\P\) satisfies
    \(t-p \in \fragment{i - k}{i + k} \cap \fragment{j -|P| - k}{j-|P| + k}\),%
    \footnote{We could potentially restrict the interval for \(t-p\) even
    further: transferring from diagonal \(i\) to diagonal \(j-|P|\) incurs a cost of at
    least \(|j -|P|- i|\); deviating to outside the diagonals between \(i\) and \(j-|P|\)
    requires a return to the said diagonals. Hence, \(\P\) may in fact deviate by only at
    most \(\lfloor (k - |j-|P|-i|)/2 \rfloor \le k/2\) from the diagonals between \(i\) and
    \(j-|P|\).
    However, to decrease the technical complexity a tiny little bit, we do not use
    this further refinement here.}
    and there is at least one such vertex for each \(p\in \fragment{0}{|P|}\).
\end{remark}

We summarize \cref{rem:basic} in \cref{obs:band}.

\begin{lemmaq}[Paths of cost \(k\) that connect the top and bottom of the alignment graph
    may use at most \(k\) non-diagonal edges]
    \dglabel{obs:band}
    Consider strings $P\in \Sigma^m$ and $T\in \Sigma^n$, and a threshold $k\in \Zz$.
    Let $\mathcal{P}$ denote a path of cost at most $k$ that connects the top boundary to
    the bottom boundary of $\oAGw(P,T)$.

    For every vertex $(p,t)\in \mathcal{P}$, we have $-k \le t-p \le n-m+k$.
\end{lemmaq}

We may use \cref{lem:MSSP} to compute boundary-to-boundary paths in the alignment graph.

\begin{corollary}
    \dglabel{lem:MSSPverify}[lem:MSSP](Using MSSP to compute boundary-to-boundary paths in
    the alignment graph)
    Given a string $P$ of length $m$, a string $T$ of length $n$, and oracle access to a
    normalized weight function $w: \sqEsigma \to \intvl{0}{W}$,
    after an $\cO(nm \log (nm))$-time preprocessing, we can compute the distance between
    any two boundary vertices of $\oAGw(X,Y)$ in $\Oh(\log(nm))$ time.
\end{corollary}
\begin{proof}
    At preprocessing, we build an MSSP data structure over $\oAGw(P,T)$, with the infinite
    face being the distinguished face.
\end{proof}

\paragraph*{Our Model of Computation and the \pillar Model}\label{sec:pillar}

Throughout this work, we consider the Real RAM model of computation that, besides typical
word RAM operations, allows us to store, add, subtract, and compare real numbers in
constant time (we do not use multiplication and other advanced operations)---this is a
standard assumption for algorithms computing distances in graphs
with real weights.
If all weights admit an $\Oh(\log n)$-bit fixed-point arithmetic representation, standard
word RAM is sufficient to implement our algorithms.

We use the  \pillar model of \cite{ckw20} to
express our algorithms for the \verify problem and our algorithm for \PMWED in
\cref{sec:reduction}.
In particular, we bound the running times of said algorithms in terms of the number of
calls to a small set of very common operations (the primitive \pillar operations) on
strings and calls to standard operations supported constant time in the Real RAM model of computation.

Combining \cref{thm:main} with efficient implementations of said primitive
operations, we then obtain (fast) algorithms for \PMWED in a
plethora of settings: the standard setting, the compressed
setting, the dynamic setting, the quantum setting, the packed setting, and the read-only
setting.
See \cref{sec:concl} for details.

We continue with a brief summary of the \pillar model.
In the \pillar model, we maintain a collection of
strings~\(\X\); the operations in the \pillar model work on fragments
\(X\fragmentco{\ell}{r}\) of \(X \in \X\), which are represented via a
\emph{handle}.\footnote{The implementation details depend on the specific setting.
For example, in the standard setting, a fragment $X\fragmentco{\ell}{r}$ is represented by
a reference to $X$ and the endpoints $\ell,r$.}
At the start of the computation, the \pillar model provides a handle to each \(X \in
\X\), which represents \(X\fragmentco{0}{|X|}\). Using an \extractOpName operation, we may
obtain handles to other fragments of the strings in \(\X\)~\cite{ckw20}.
\begin{itemize}
    \item $\extractOpName(S,\ell,r)$: Given a fragment $S$ and positions $0 \le \ell \le r
        \le |S|$, extract the (sub)fragment $S\fragmentco{\ell}{r}$. If
        $S=X\fragmentco{\ell'}{r'}$ for $X\in \X$, then $S\fragmentco{\ell}{r}$ is defined
        as $X\fragmentco{\ell'+\ell}{\ell'+r}$.
\end{itemize}
The other primitive \pillar model operations read as follows~\cite{ckw20}.
\begin{itemize}
    \item $\lceOp{S}{T}$: Compute the length of~the longest common prefix of~$S$ and $T$.
    \item $\lcbOp{S}{T}$: Compute the length of~the longest common suffix of~$S$ and $T$.
    \item $\ipmOp{P}{T}$: Assuming that $|T|\le 2|P|$, compute $\OccEx(P,T)$ (represented
        as an arithmetic progression with difference $\per(P)$).
    \item $\accOpName(S,i)$: Assuming $i\in \fragmentco{0}{|S|}$, retrieve the character $\accOp{S}{i}$.
    \item $\lenOpName(S)$: Retrieve the length $|S|$ of~the string $S$.
\end{itemize}

\SetKwFunction{alignment}{Alignment}
\SetKwFunction{weighteded}{WeightedED}

The Landau--Vishkin algorithm for computing the edit distance of two strings
\cite{DBLP:journals/jcss/LandauV88} can be reformulated as follows in the \pillar model
(see \cite[Lemma 6.1]{ckw20}).

\begin{lemmaq}[{\tt \protect\alignment($S$, $T$)}, {\cite{DBLP:journals/jcss/LandauV88}}]
    \dglabel{fct:LValignment}(Computing an optimal alignment of two strings of cost at most \(d\)
    (or deciding that no such alignment exists) takes
    $\Oh(1+d^2)$ time in the \pillar model, {\tt \protect\alignment($S$, $T$)},
{\cite{DBLP:journals/jcss/LandauV88}})
    Given strings $S$ and $T$ and a threshold \(d\), we can construct (the breakpoint representation of) an
    alignment $\A: S\onto T$ of optimal cost $\ed(S,T) \le d$ in $\Oh(1+d^2)$ time in
    the \pillar model (or decide in the same time that no such alignment may exist).
\end{lemmaq}

An efficient procedure in the \pillar for computing an optimal alignment under a
normalized weight function was presented in~\cite{ckw23}.

\begin{lemmaq}[{\tt\protect\weighteded($S$, $T$, $w$)}, {\cite[Theorem 4.1]{ckw23}}]
    \dglabel{lem:wed}(Computing the weighted edit distance of two strings takes
    $\cO(1+k^3 \log^2 n)$ time in the \pillar model, {\tt\protect\weighteded($S$, $T$,
    $w$)}, {\cite[Theorem 4.1]{ckw23}})
    Given strings $S$ and $T$ of length at most $n$ and oracle access to a normalized weight
    function $w$, we can compute the value $k = \edw{S}{T}$ in
    time $\cO(1+k^3 \log^2 n)$ in the \pillar model.
\end{lemmaq}

Finally, the adaptation of the Landau--Vishkin algorithm~\cite{LandauV89}
already provided in~\cite{ckw22} lets us efficiently compute an optimal alignment of a string $S$
onto a substring of $Q^\infty$ starting at a given position~$x$.

\begin{lemma}[{\tt\protect\alignment($S$, $Q$, $x$)}, {\cite[Lemma 2.10]{ckw22}}]
    \dglabel{fct:alignment}(Computing an optimal alignment of a string to an infinitely repeating
    string takes $\Oh(1+d^2)$ time in the \pillar model, {\tt\protect\alignment($S$, $Q$,
    $x$)}, {\cite[Lemma 2.10]{ckw22}})
    Consider non-empty strings $S,Q$ and an integer $x\in \Z$.
    We can construct (the breakpoint representation of) an alignment
    $\A: S\onto Q^\infty\fragmentco{x}{y}$ of optimal cost $d \coloneqq \edp{S}{\rot^{-x}(Q)}$
    in $\Oh(1+d^2)$ time in the \pillar model.\lipicsEnd
\end{lemma}

\section{Pattern Matching with Weighted Edits in \texorpdfstring{\boldmath
$\cOtilde(nk)$}{Õ(nk)} Time}\label{sec:nk}

In this section, we prove \cref{thm:stalgmainnk}. As a first simple step, we split the
text \(T\) into a set \(\S\) of overlapping fragments, such that it suffices to compute, for each
\(S \in \S\), the positions of \(S\) in an interval \(\fragment{0}{\Oh(k)}\) where
a $(k,w)$-error occurrence of \(P\) starts.

\begin{remark}
    \dglabel{7.9.30-1}(Standard Trick: Can reduce to \(\Oh(n/k)\) instances \((P,
    S)\) where \(S = T\fragmentco{i}{i + m + 2k}\) with \(i \equiv_k 0\))
    For a pattern \(P\) of length \(m\), a text \(T\) of length \(n\), and a threshold
    \(k > 1\) write \(\S\) for a set of $\cO(n/k)$ fragments $S_i \coloneqq T\fragmentco{i}{\max\{i + m + 2k,n\}}$ such
    that $i \equiv 0 \pmod k$.

    We have \(\OccW_k(P, T) = \bigcup_{S_i \in \S} i + \OccW_k(P, S_i)\), as
    each fragment $U$ of $T$ with $\edw{P}{U}\leq k$
    has a length in $\fragment{m-k}{m+k}$ and hence is contained in at least one and at most
    four fragments \(S \in \S\).
\end{remark}

\begin{figure}[t!]
    \centering
    \begin{subfigure}[t]{\linewidth}
        \centering
        \includegraphics[page=1,width=.9\linewidth]{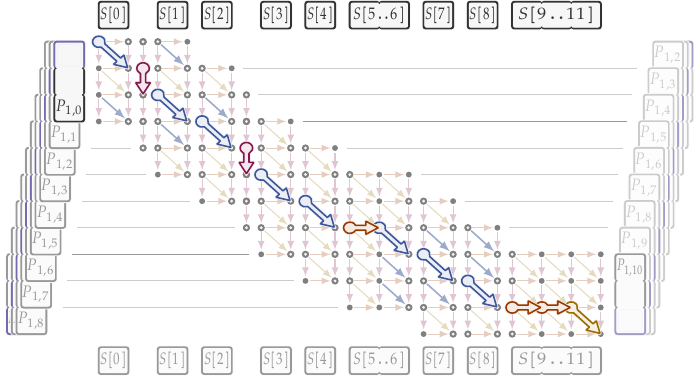}
        \caption{For each \(i \in \fragment{0}{|P|}\), we depict the slice \(\GPSd_{i,i}\) as a
            subgraph of \(\oAGw(P,S)\), the corresponding fragments of
            \(P\) and \(S\), as well as the corresponding part of the alignment \(\A\).
            For each slice, we highlight in gray its portal vertices.
            Observe that portal vertices of consecutive slices correspond to the same
            vertices in \(\oAGw(P,S)\).}
            \label{mjhlpavqep}
    \end{subfigure}\vskip1ex
    \begin{subfigure}[t]{.48\linewidth}
        \centering
        \includegraphics[page=1,height=.8\linewidth,width=\linewidth]{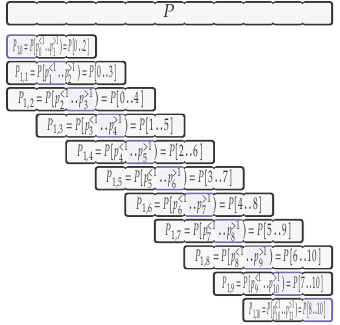}
        \caption{The $d$-slices of $P$ in detail. For each slice \(P_{d,i}\), we highlight the character
        \(P\position{i}\).}
    \end{subfigure}\hfill%
    \begin{subfigure}[t]{.47\linewidth}
        \centering
        \includegraphics[page=5,width=\linewidth]{figs/g03}
        \caption{The graph \(\GPSd\), which is obtained by identifying the portal vertices
            of consecutive slices of \cref{mjhlpavqep}. Again, we highlight portal
            vertices in gray.
        }\label{fig:gpsd}
    \end{subfigure}%
    \vskip-2ex%
    \caption{The $d$-slices \(P_{d,i}\) of $P$, the $d$-slices \(\GPSd_{i,i}\) of
        \(\oAGw(P,S)\), and the graph \(\GPSd\) with respect to alignment $\A: P
        \onto S$ from \cref{oiuoqannlf,byutjsdglx}.
        In order to conserve space but still have a meaningful example,
        we depict slices for \(d = 1\), even though alignment alignment $\A$ is of cost \(6\).
        For the depicted example this does not cause any issues.
        In general, the alignment needs to be
        fully contained in the slices---which holds when the alignment is of
        cost at most \(d\).\vskip-3ex}
\label{fig:7.10.18-1}
\end{figure}

In a similarly standard fashion, we may filter out any \(S \in \S\) that is
too far from \(P\) under the unweighted edit distance to contain any $(k,w)$-error occurrences.
For the remaining $S \in \S$, we obtain a cheap unweighted alignment \(P \onto S\)
as a useful byproduct.

\begin{lemma}
    \dglabel{lem:PtoF}[fct:LValignment,thm:pilis](After preprocessing a string \(T\) of
    length \(n\) and a pattern \(P\) of length \(m \le n + k\) for \(\Oh(n)\) time,
    can compute for any fragment \(S\) of \(T\) of length at most \((m+2k)\)
    an unweighted alignment \(\A : P \onto S\) of cost at
    most \(d = 4k\) or determine that \(\OccW_k(P,S) = \emptyset\) in time \(\Oh(d^2)\) for a
    single \(S\))
    Consider a threshold \(k > 0\), a text \(T\) of length \(n\), and a pattern \(P\) of length \(m \le n + k\).
    \begin{itemize}
        \item For any fragment \(S\) of \(T\) of length at most \(m+2k\),
            if $\ed(P,S) >  4k$ then $\OccW_k(P,S) = \emptyset$.
        \item After preprocessing \(T\) for \(\Oh(n)\) time, given any fragment \(S\) of
            \(T\) of length at most \(m+2k\), we can in $\cO(k^2)$ time check whether $\ed(P,S)\leq 4k$, and,
            if so, compute the breakpoint representation of an optimal unweighted alignment $\A: P
            \onto S$.
    \end{itemize}
\end{lemma}
\begin{proof}
    We prove the contraposition of the first item.
    Fix a fragment \(S\) of \(T\) and suppose that \(\OccW_k(P,S) \neq \emptyset\);
    that is, suppose that there are positions $\ell$ and $r$ and an alignment $\B:P \onto
    S\fragmentco{\ell}{r}$ of weighted cost at most $k$.
    We argue that \(\B\) may be transformed into an unweighted alignment \(\A : P \onto S\) of cost at most
    \(4k\): consider the (unweighted) alignment \(\A : P \onto S\) that is obtained as an extension
    from \(\B\) by adding insertions for $S\fragmentco{0}{\ell}$ and $S\fragmentco{r}{|S|}$ appropriately.

    First, from $\ed(P,S\fragmentco{\ell}{r}) \leq \edw{P}{S\fragmentco{\ell}{r}} \le k$
    we conclude that \(\B\) contributes a cost of at most \(k\) to \(\A\).
    Next, \(\edw{P}{S\fragmentco{\ell}{r}} \le k\) additionally implies \(|S\fragmentco{\ell}{r}|\geq
    m-k\); which in turn implies \(|S|-|S\fragmentco{\ell}{r}| \leq (m + 2k) - (m - k) =
    3k\).
    Hence, the extra insertions incur a total cost of at most \(3k\).
    In total, \(\A\) has thus a cost of at most \(4k\), completing the proof.

    The second item follows directly from the combination of
    \cref{fct:LValignment,thm:pilis}.
\end{proof}

\begin{figure}[t!]
    \begin{subfigure}[t]{\linewidth}
        \centering
        \includegraphics[page=1,width=.9\linewidth]{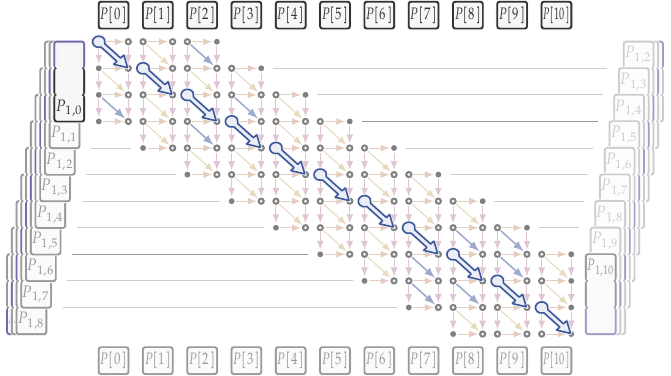}
        \caption{For each \(i \in \fragment{0}{|P|}\), we depict the slice \(\GPd_{i,i}\) as a
            subgraph of \(\oAGw(P,P)\), the corresponding fragments of
            \(P\), as well as the corresponding part of the alignment \(\A\);
            compare \cref{mjhlpavqep}.
            For each slice, we highlight in gray its portal vertices.
            Observe that portal vertices of consecutive slices correspond to the same
            vertices in \(\oAGw(P,S)\).}
            \label{fyleajysdq}
    \end{subfigure}\vskip1ex
    \begin{subfigure}[t]{.46\linewidth}
        \centering
        \includegraphics[page=5,width=.85\linewidth]{figs/g03b}
        \caption{The graph \(\GPSd\), which is obtained by identifying the portal vertices
            of consecutive slices of \cref{fyleajysdq}; compare \cref{fig:gpsd} Again, we highlight portal
            vertices in gray.}
    \end{subfigure}\hfill%
    \begin{subfigure}[t]{.51\linewidth}
        \centering
        \includegraphics[page=2,scale=1.15]{figs/g04}
        \caption{The subgraph \(\GPSd_{5,6}\) in detail.
            As the alignment $\A: P \onto S$ matches \(P\fragment{5}{6} = S\fragment{3}{4}\),
        we have \(\GPSd_{5,6} = \GPd_{5,6}\).}\label{evatvwbtfb}
    \end{subfigure}%
    \caption{The $d$-slices \(\GPd_{i,i}\) of
        \(\oAGw(P,P)\), and the graph \(\GPd\) with respect to the trivial identity
        alignment; compare \cref{fig:7.10.18-1}.
        We also depict a detailed view of the subgraph \(\GPSd_{5,6} = \GPd_{5,6}\).
    }\label{fig:7.10.18-2}
\end{figure}

In a next step, we intend to use the cheap unweighted alignments \(P \onto S\), and in
particular that any such alignment matches large parts of \(P\) with \(S\) without any edits. As many of
said matched parts are shared across the different \(S \in \S\), we intend to save on
computation time by essentially precomputing boundary-to-boundary shortest-path distances
for the parts where fragments of \(P\) are matched to fragments of \(P\).
For a formal exposition, we start with naming parts of the corresponding alignment graphs.
Consult \cref{fig:7.10.18-1,fig:7.10.18-2} for a visualization of an example.

\begin{definition}[The \(i\)-th \(d\)-slice \(P_{d,i}\) of a string \(P\), the graphs
    \(\GPSd_{\ell,r}\) (and \(\GPd_{\ell,r}\)) as subgraphs of \(\oAGw(P,S)\) (and
    \(\oAGw(P,P)\)) and the corresponding distance
    matrices \(\DPSd_{\ell,r}\) (and \(\DPd_{\ell,r}\))]
    \dglabel{7.9.30-3}\hfill
    Consider a string \(P\) of length \(m\), a threshold \(d > 0\), and an alignment \(\A
    : P \onto S\) of cost at most \(d\).
    \begin{itemize}
        \item For \(i \in \fragmentco{0}{m}\), the \emph{\(i\)-th \(d\)-slice of \(P\)} is the string
            \(P_{d,i} \coloneqq P\fragmentco{p^{<d}_i}{p^{>d}_{i+1}}\), where
            \(p^{<d}_{i} \coloneqq \max(0,i-2d)\) and \(p^{>d}_{i + 1} \coloneqq
            {\min(m,i+ 1 + 2d)}\).
            Further, we write
            \(s_i \coloneqq \min\{ s  \mid (i, s) \in \A \}\) for \(i \in \fragmentco{0}{m}\) and \(s_m \coloneqq |S|\).

        \item For \(i \in \fragmentco{0}{m}\), the \emph{\(i\)-th \(d\)-slice of
            \(\oAGw(P,S)\)}
            is the subgraph \(\GPSd_{i,i} \coloneqq \oAGw(P_{d,i},\A(P\position{i}))\) (induced by
            $\fragment{p^{<d}_i}{p^{>d}_{i+1}}\times \fragment{s_i}{s_{i+1}}$).

            For \(0 \le \ell < r < m\), we write \(\GPSd_{\ell,r} \coloneqq \bigcup_{\ell
            \le j \le r} \GPSd_{j,j}\) and we write \(\GPSd \coloneqq \GPSd_{0,m}\).

            We also write $\VPSd_i \coloneqq \fragment{p^{<d}_i}{p^{>d}_{i}}\times \{s_i\}$ for the
            \(i\)-th set of portal vertices of \(\GPSd\); we index the vertices top-to-bottom,
            that is, \( \VPSd_i = \{ \VPSd_i\position{0}, \dots,
            \VPSd_i\position{p^{>d}_{i} - p^{<d}_i}\}\) with \(\VPSd_i\position{j}
            \coloneqq (p^{<d}_i + j,s_i)\).

        \item For \(0 \le \ell \le r \le m\), we write \(\DPSd_{\ell,r}\) for the matrix
            of distances from \(\VPSd_{\ell}\) to \(\VPSd_{r}\), that is,
            \[\DPSd_{\ell,r}\position{a,b} \coloneqq \dist_{\GPSd} ((p^{<d}_\ell - a,
            s_\ell), (p^{>d}_r - b, s_r))\] for $a \in
            \fragment{0}{p^{>d}_\ell-p^{<d}_\ell}$
            and $b \in \fragment{0}{p^{>d}_r-p^{<d}_r}$.
    \end{itemize}
    For the special case of \(S = P\) and \(\A = \mathrm{id}\), we omit the
    superscript \(S\) and write \(\GPd\), \(\VPd\), and \(\DPd\).
\end{definition}

Next, we formally define the data structure problem that enables us to efficiently compute
$\OccW_k(P, S)$ for each $S \in \S$.

\begin{problem}[DMVO]{Distance Matrix Vector Oracle, DMVO{\tt($P$, $d$, $w$)}}
    \label{prob:dmvo}
    \PInput{{\tt DMVO-Init($P$, $d$, $w$)}: Given a pattern \(P\) of length \(m\), a
        positive integer \(d\), and
        oracle access to a normalized weight function $w:\sqEsigma \to \intvl{0}{W}$,
        preprocess \(P\) and \(d\).}

    \PItem{Queries}{{\tt DMVO-Query($i$, $j$, $v$)}: For integers $0\le i \le j \le m$
        and a vector $v$ of dimension $(p^{>d}_j-p^{<d}_j+1)$,
        return the vector $\DPd_{i,j}\oplus v$.
    }
\end{problem}

\begin{remark}
    One might wonder why a {\tt DMVO-Query} computes a matrix-vector product instead of
    allowing for direct random access to $\DPd_{i,j}$.
    Crucially, this allows us to exploit that $\DPd_{i,j}\oplus v$ may be computed via
    $\DPd_{i,j}\oplus v = \DPd_{i,c}\oplus (\DPd_{c,j}\oplus v)$ for \(c \in
    \fragment{i}{j}\)---which ultimately allows us to save a logarithmic factor compared to an approach that
    first computes (single elements of) \(\DPd_{i,c} \oplus \DPd_{c,j}\). Indeed,
    given \(\Oh(\log m)\)-time random access to \(\DPd_{i,c}\) and \(\DPd_{c,j}\),
    one may obtain \(\Oh(\log^2 m)\)-time random access to \(\DPd_{i,c} \oplus
    \DPd_{c,j}\) by rewriting each such random access as a multiplication of a \(1
    \times d\) and a \(d \times 1\) (Monge) matrix and employing \cref{thm:smawk} (that is, the SMAWK algorithm).
\end{remark}

\begin{lemma}
    \dglabel{clm:preprocess}(\(\Oh(d\log m)\)-time
    computation of any \(\DPd_{\ell,r} \oplus v\) ({\tt DMVO-Query}) after global
    $\Oh(md\log^2 m)$-time preprocessing ({\tt DMVO-Init}))
    We can implement \ref{prob:dmvo} such that {\tt DMVO-Init} takes time \(\Oh(md\log^2
    m)\) and a single {\tt DMVO-Query} takes time \(\Oh(d \log m)\).
\end{lemma}
\begin{proof}
    {\tt DMVO-Init}. For each \(s \in \fragment{0}{\lceil \log m\rceil}\), we
    partition \(\fragmentco{0}{m}\) into subintervals of length \(2^s\)\!,%
    \footnote{For each \(s\), the last subinterval might be shorter.}
    and
    for each such subinterval \(\fragmentco{\ell}{r}\),
    we use Klein's algorithm (\cref{lem:MSSP}) on $\GPd_{\ell,r}$
    to compute a data structure with
    \(\Oh(\log m)\)-time random access to \(\DPd_{\ell,r}\)\!,
    and---crucially---\(\DPd_{\ell,i}\) and \(\DPd_{i,r}\) for every \(i\in
    \fragment{\ell}{r}\).

    {\tt DMVO-Query}. Suppose that we are given integers $0\le i \le j \le m$ and a vector
    \(v\).
    We first compute the largest integer $s$ such that $\fragmentco{i}{j}$ contains an
    integer $x \coloneqq c \cdot 2^s$ (for an integer \(c \ge 0\)).
    Then, we consider two consecutive length-\(2^s\) intervals
    \(\fragmentco{x - 2^s}{x}\) and \(\fragmentco{x}{x + 2^s}\)
    for which
    we have precomputed (random access to) the corresponding distance matrices \(\DPd_{x-2^s,x}\) and
    \(\DPd_{x, x+2^s}\)---and, crucially, \(\DPd_{i,x}\) and \(\DPd_{x,j}\).
    Using \cref{thm:smawk} (the SMAWK algorithm) twice, we compute and return
    $\DPd_{i,x}\oplus (\DPd_{x,j}\oplus v)$.

    {\bf Correctness.}
    First observe that every path from $\VPd_i$ to $\VPd_j$ in
    $\GPd$ crosses $\VPd_x$ for each \(x \in \fragmentoo{i}{j}\). Thus, we know that
    $\DPd_{i,j} = \DPd_{i,x}\oplus \DPd_{x,j}$.
    Further, by \cref{fct:monge}, $\DPd_{i,x}$ and $\DPd_{x,j}$ are Monge.

    Even though we compute distances in graphs $\GPd_{\ell,r}$ for different pairs $(\ell,r)$, they are also distances in~$\GPd$;
    this is implied by \cref{lem:oagw} (path monotonicity in $\oAGw(P,P)$) and the
    fact that $\VPd_\ell$ and $\VPd_r$ separate $\GPd_{\ell,r}$ from the remainder of $\GPd$.
    In particular, such a preprocessing of $\GPd_{\ell,r}$ yields $\Oh(\log m)$-time
    random access to $\DPd_{\ell,i}$ and $\DPd_{i,r}$ for every $i\in \fragmentoo{\ell}{r}$.

    Recall that for an interval \(\fragmentco{i}{j}\)  we choose \(s\)  as the maximal
    integer for which an integer \(c\) exists  with   \(i \le x = c \cdot 2^s < j\).
    Observe that this choice implies that \(\fragmentco{i}{j}\) contains neither \((c-1)
    \cdot 2^s\) nor \((c+1) \cdot 2^s\), that is, we have
    \(x - 2^s = (c-1) \cdot 2^s < i\) and \(j \le (c+1) \cdot 2^s = x + 2^s\).
    Hence, our precomputation of \(\DPd_{x-2^s,x}\) and
    \(\DPd_{x, x+2^s}\) indeed includes random access to \(\DPd_{i,x}\) and \(\DPd_{x,j}\).

    {\bf Running time.}
    For the running time of {\tt DMVO-Init},
    observe that for fixed \(\ell\) and \(r\), applying \cref{lem:MSSP} (Klein's algorithm)
    to $\GPd_{\ell,r}$ costs
    $\Oh(d\cdot (r-\ell+1)\cdot \log m)$ time; which yields $\Oh(md\cdot \log m)$
    total time for all intervals of a fixed length.
    Across the \(\Oh(\log m)\) different interval lengths, the preprocessing thus takes
    time $\Oh(md\log^2 m)$ in total.

    The running time of {\tt DMVO-Query} is as claimed as it is dominated by a constant number of calls
    to \cref{thm:smawk}
    for $\Oh(d)\times \Oh(d)$ Monge matrices to which we have $\Oh(\log m)$-time random access.
\end{proof}

Next, we establish how to use \ref{prob:dmvo} to obtain a fast algorithm to compute
\((k,w)\)-error occurrences of \(P\) in a string \(S \in \S\).

\begin{algorithm}[t]
    \SetKwBlock{Begin}{}{end}
    \SetKwFunction{DMVOqeury}{DMVO-Query}
    \boosted{$\D$ $\gets$ {\tt DMVO($P$, $d$, $w$)}, $\A : P \onto S$ $\gets$
        \(\{((0,\varepsilon),(0,\varepsilon)),\dots, ((m,\varepsilon),(|S|,\varepsilon)) \}\), $d$, $k$, $w$}\Begin{
        Using a left-to-right pass over the breakpoints in \(\A\), compute
        \Begin{
            the source vertices \(V^\top \coloneqq \{ v^\top_i  \mid i \in
            \fragment{0}{s_{2d + 1}} \}\) with \(v^\top_i \coloneqq (i, 0)\)\;
            the left-center boundary portal-set \( \VPSd_{2d + 1}\)\;
            the set of important portal positions \(I \coloneqq \{ i_0 = 2d+1, \dots,
            i_{q} = m - 2d - 1\}\),
            where a portal position \(i_j\) is
            \emph{marked} if it is to the right of a pure range, that is, if
            \(\A(P\fragmentco{i_{j-1}}{i_j}) = P\fragmentco{i_{j-1}}{i_j}\)\;
            the center-right boundary portal-set \( \VPSd_{m - 2d - 1} \) \;
            and the target vertices \(V^\bot \coloneqq \{ v^\bot_i  \mid i \in
            \fragment{0}{|S| - s_{m-2d - 1}} \}\) with \(v^\bot_i \coloneqq (i + s_{m-2d-1},
            m)\).
        }

        Use Klein's algorithm (\cref{lem:MSSP})
        on \(\GPSd_{m - 2d - 1, m}\)
        to obtain \(D_{m - 2d - 1, \bot}\),
        where \(
            D_{m - 2d - 1, \bot}\position{i, j}  \coloneqq \dist_{\GPSd}( \VPSd_{m - 2d -
            1}\position{i}, v^\bot_j)
        \)\;
        Use SMAWK (\cref{thm:smawk}) for $v$ $\gets$ $D_{m - 2d - 1, \bot} \oplus \vec{0}$\;

        \ForEach{$j \gets q, \dots, 1$}{
            \leIf{portal position \(i_j\) is marked}{
                $v$ $\gets$ \DMVOqeury{$i_{j-1}$, $i_j$, $v$}\;
            }{
                Use Klein's algorithm
                on \(\GPSd_{i_{j-1}, i_j}\)
                to obtain \(\DPSd_{i_{j-1}, i_j}\); then SMAWK for
                $v$ $\gets$ $\DPSd_{i_{j-1}, i_j} \oplus v$
            }
        }

        Use Klein's algo.\
        on \(\GPSd_{0, 2d + 1}\)
        to obtain \(D_{\top, 2d + 1}\),
        where \(
            D_{\top, 2d + 1}\position{i, j}  \coloneqq \dist_{\GPSd}( v^\top_i, \VPSd_{2d
                + 1}\position{j})
        \); then SMAWK for $v$ $\gets$ $D_{\top, 2d + 1} \oplus v$\;

        \Return{$ \{ i  \mid v\position{i} \le k \} $}\;
    }
    \caption{Using \ref{prob:dmvo} for a sped-up computation of \((k,w)\)-error
        occurrences of a string \(P\) in a string \(S\) that is given via a breakpoint
        representation of an alignment of cost at most \(d\). Observe that without using
        \ref{prob:dmvo}, the algorithm would amount to computing the top-to-bottom shortest-path
        distances in \(\GPSd\) using Klein's algorithm (\cref{lem:MSSP}). We achieve a speed-up
        by using precomputed shortest-path distances for the subgraphs where \(\A\) matches
        fragments exactly.}\label{alg:boostedPM}
\end{algorithm}
\begin{lemma}
    \dglabel{7.9.30-2}[obs:band](Given alignment \(\A : P \onto S\)
    of cost at most \(d\) and a \(k \in\fragment{0}{d}\), can compute \(\OccW_k(P,
    S)\) using \(\Oh(d)\) calls to {\tt DMVO-Query} plus \(\Ohtilde(d^2)\) extra time)
    Suppose that we are given an instance of a \ref{prob:dmvo} data structure for a
    positive integer threshold \(d\),
    a pattern \(P\) of length \(m > 4d + 1\),
    and a normalized weight function $w:\sqEsigma \to \intvl{0}{W}$.

    Given the breakpoint representation of an alignment $\A : P\onto S$ of unweighted cost
    at most $d$ (implicitly defining a string $S\in \Sigma^*$) and a threshold $k\in
    \fragment{0}{d}$, we can compute $\OccW_{k}(P,S)$ using \(\Oh(d)\) calls to {\tt
    DMVO-Query} and \(\Oh(d^2 \log m)\) extra time.
\end{lemma}
\begin{proof}
    \textbf{Notations.}
    We say that a position (or slice) \(i\) is \emph{left} if $p^{<d}_i = 0$ (or \(i \le 2d\)) and
    we say that a position (or slice) \(i\)  is \emph{right} if $p^{>d}_{i+1} = m$ (or \(i
    \ge m - 2d - 1\)).
    A slice is \emph{center} if it is neither left nor right.
    Our assumption \(m > 4d + 1\) ensures that each position (and slice) is at most one of
    left, right, and center.

    We say that a (center) position \(i\) with \(\A(P\position{i}) = P\position{i}\) is
    \emph{pure}.
    We group maximal stretches of pure (center) positions into \emph{pure stretches}.
    We say that a position where a pure stretch starts is \emph{important}, and we say that a position
    where a pure stretch ends is \emph{important and marked}.
    We write \(I \coloneqq \{ i_0 = 2d+1, \dots,
    i_{q} = m - 2d - 1\}\)
    for the set of important portal positions, where we artificially add \(i_0\) and
    \(i_q\) in case they are not present.

    Next, we name important vertices of \(\GPSd\).
    The leftmost center portal-set $\VPSd_{i_0}$ is the \emph{left-center boundary}; the
    rightmost center portal-set $\VPSd_{i_q}$ is
    the \emph{center-right boundary}.
    The \emph{source} vertices are the topmost vertices of left slices,
    that is, $\VPSd_\top\coloneqq \{0\}\times
    \fragment{s_0}{s_{2d + 1}}$;
    the \emph{target} vertices are the bottommost vertices of right
    slices, that is, $\VPSd_\bot\coloneqq \{m\}\times
    \fragment{s_{m-2d -1}}{s_m}$.
    Observe that
    $\VPSd_\bot$ and $\VPSd_\top$ are the only vertices in $\GPSd$ with the first coordinate equal to
    $m$ and $0$, respectively.

    Consult \cref{fig:nk1} for a visualization of an example.

    \begin{figure}[t!]
        \centering
        \begin{subfigure}[t]{\linewidth}
            \centering
            \includegraphics[page=2,width=.95\linewidth]{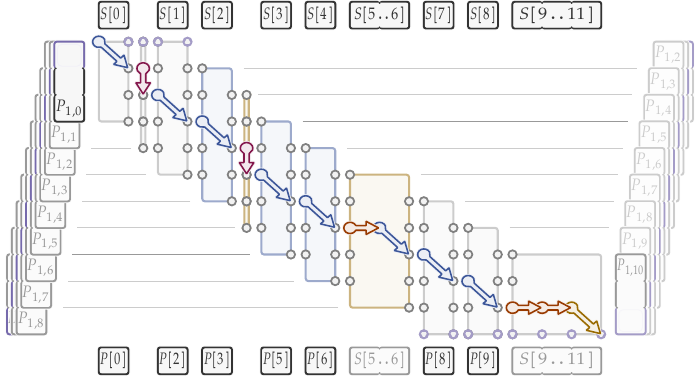}
            \caption{
            For each \(i \in \fragment{0}{|P|}\), we depict the slice \(\GPSd_{i,i}\) in a
            stylized manner, the corresponding fragments of
            \(P\) and \(S\), as well as the corresponding part of the alignment \(\A\);
            compare \cref{mjhlpavqep}.
            For each slice, we highlight in gray its portal vertices; we also highlight
            source vertices and target vertices in light purple.
            Left and right slices are colored gray, the (center) pure slices are colored blue, while the remaining slices are colored yellow.
            }\label{tjcwmetzrj}
        \end{subfigure}\vskip1ex
        \begin{subfigure}[t]{.62\linewidth}
            \centering
            \includegraphics[page=4,width=.9\linewidth]{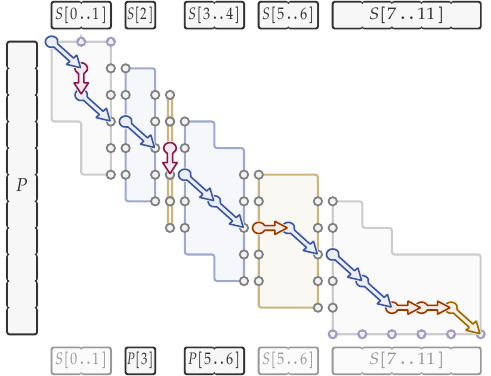}
            \caption{The graphs \(\GPSd_{i_{q},m}\) and \(\GPSd_{0,i_{0}}\) (colored
                gray), as well as for each maximal interval $\fragment{i_{j-1}}{i_j}$ of
                center positions that are pure (colored blue) and not pure (colored
                yellow) the graph \(\GPSd_{i_{j-1},i_j}\) as a
                union of slices from \cref{tjcwmetzrj}.
                We again highlight portal vertices, as well as source and target vertices.}
        \end{subfigure}\hfill%
        \begin{subfigure}[t]{.35\linewidth}
            \centering
            \includegraphics[page=3,scale=1.4]{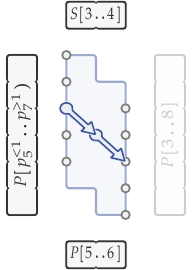}
            \caption{The subgraph \(\GPSd_{5,6}\) in a stylized manner; compare
                \cref{evatvwbtfb}.\vphantom{\(\GPSd_{5,6}\)}
                As the alignment $\A: P \onto S$ matches \(P\fragment{5}{6} = S\fragment{3}{4}\),
                we have \(\GPSd_{5,6} = \GPd_{5,6}\).}
        \end{subfigure}
        \caption{Stylized variants of the $d$-slices \(\GPSd_{i,i}\) of \(\oAGw(P,S)\) and
            the subgraphs considered by the algorithm underlying \cref{7.9.30-2}.
            The source vertices in the top left and the target vertices
            in the bottom right are marked in light purple.
            Consecutive slices with the same color are merged.
            The portals are exactly those vertices that belong to subgraphs with different
            colors.}
        \label{fig:nk1}
    \end{figure}

    \textbf{Algorithm description.}
    We first use a single pass over the breakpoint description of \(\A\) to identify
    the source vertices, the left-center boundary, the boundaries of (center) pure
    stretches (and the corresponding portal-sets $\VPSd_{i_j}$ for $j \in
    \fragmentoo{0}{q}$), the center-right boundary, and the target vertices.
    Additionally, using (the breakpoint representation of) $\A$, we explicitly construct
    \(\GPSd_{i_{q},m}\), \(\GPSd_{0,i_{0}}\), as well as \(\GPSd_{i_{j-1},i_j}\) for each
    maximal interval $\fragment{i_{j-1}}{i_j}$ of center positions that are not pure in a
    straightforward manner in time linear in the total size of these graphs.

    Using Klein's algorithm (\cref{lem:MSSP}), we compute (random
    access to) the matrix of distances from any center-right boundary
    vertex to any target vertex in \(\GPSd_{i_{q},m}\); we then min-plus multiply the obtained matrix with a
    suitably-sized zero vector using \cref{thm:smawk} (SMAWK algorithm) to obtain a vector \(v\).

    Next, we process the center positions in $I$ from right to left. For a pure stretch
    \((i_{j-1},i_j)\), that is, when $i_j$ is marked, we use a
    {\tt DMVO-Query} to compute the min-plus product \(\DPd_{i_{j-1},i_j} \oplus v\).
    For a non-pure stretch \((i_{j-1},i_j)\),
    that is, when $i_j$ is not marked, we use Klein's algorithm (\cref{lem:MSSP})
    on \(\GPSd_{i_{j-1},i_j}\) to obtain (random access to) \(\DPSd_{i_{j-1},i_j}\) and use
    \cref{thm:smawk} to compute \(\DPSd_{i_{j-1},i_j} \oplus v\).

    Finally, again using Klein's algorithm (\cref{lem:MSSP}), we compute (random
    access to) the matrix of distances from any source vertex to any left-center boundary
    vertex in \(\GPSd_{0,i_0}\); we then min-plus multiply the obtained matrix with \(v\)
    using \cref{thm:smawk} (SMAWK algorithm) and return any position of the resulting
    vector that has value at most \(k\).

    Consult \cref{alg:boostedPM} for a detailed description as pseudo-code.

    \textbf{Correctness.}
    We argue about the correctness in two steps. First, we show that each shortest
    $(0,\ell)$-to-$(m,r)$ path of length at most \(k\) in \(\GPSd\) indeed corresponds to an
    alignment of \(P\) onto \(S\fragmentco{\ell}{r}\) of the same cost.
    We then argue that portal-sets act as separators of \(\GPSd\) and thus allow us to split the
    distance computations and reuse the precomputed distances via \ref{prob:dmvo} whenever the
    corresponding subgraph of \(\GPSd\) is isomorphic to a subgraph of \(\GPd\) (encapsulated in
    \ref{prob:dmvo}).

    \begin{claim}
        \label{clm:restrict}
        For each fragment $S\fragmentco{\ell}{r}$ of $S$ and each $k\in\intvl{0}{d}$,
        the following two statements are equivalent:
        \begin{enumerate}
            \item\label{it:restrict:a} $\edw{P}{S\fragmentco{\ell}{r}}\le k$; \; and
            \item\label{it:restrict:b} $(0,\ell)$ is a source, $(m,r)$ is a target, and
                $\dist_{\GPSd}((0,\ell), (m,r))\le k$.
        \end{enumerate}
    \end{claim}
    \begin{claimproof}
        The $\eqref{it:restrict:a} \Leftarrow \eqref{it:restrict:b}$ implication follows
        directly from \cref{lem:oagw} (distance preservation between $\AGW(P,S)$ and
        $\oAGw(P,S)$) and the fact that $\GPSd$ is a subgraph of $\oAGw(P,S)$.

        For the converse implication, suppose that there is an alignment $\cO:P\onto
        S\fragmentco{\ell}{r}$ of weighted cost at most $k\le d$.
        It suffices to prove that $\cO$ (interpreted as a path from $(0,\ell)$ to $(m,r)$
        in $\oAGw(P,S)$) is contained in $\GPSd$.
        By \cref{obs:band}, every vertex $(x,y)\in \cO$ satisfies $y\in \fragment{x-d}{x+|S|-m+d}$.
        Let us focus on a single edge $(x,y)\to (x',y')$ of $\cO$ with
        $x'\in\{x,x+1\}$ and $y'\in \{y,y+1\}$.
        Let us fix an index $i\in \fragmentco{0}{m}$ such that $\fragment{y}{y'}\subseteq
        \fragment{s_i}{s_{i+1}}$.

        Since the path corresponding to $\A$ intersects at most $d$ diagonals, we have $(i-s_i)-(m-|S|)\le d$.
        Additionally, $s_i \le y \le x+|S|-m+d$.
        Combining these two inequalities, we get $i \le x+2d$, that is, $x \ge i-2d$.
        Since $x\ge 0$, we conclude that $x\ge \max(0,i-2d)=p^{<d}_i$.

        Similarly, we have $(s_{i+1}-(i+1))-(0-0)\le d$ and $s_{i+1}\ge y' \ge x'-d$.
        Combining these two inequalities, we get $x'-d-(i+1)\le d$, that is, $x' \le
        (i+1)+2d$.
        Since $x\le m$, we conclude that $x \le \min(m,i+1+2d)=p^{>d}_{i+1}$.
        Thus, the edge $(x,y)\to (x',y')$ is contained in $\oAGw(P'_i,S_i)$ and hence in
        \(\GPSd\).
    \end{claimproof}

    Next, observe that for pure positions, the
    corresponding \(d\)-slices of \(\oAGw(P, S)\) and \(\oAGw(P, P)\) are
    indeed isomorphic, that is, we have \( \GPSd_{i,s_i} = \GPd_{i,i}\)---naturally, this
    observation extends to pure stretches.
    Hence, our calls to {\tt DMVO-Query} for a pure stretch \((i_{j-1},i_j)\) compute in fact
    \(\DPSd_{i_{j-1},i_j} \oplus v = \DPd_{i_{j-1},i_j} \oplus v \).

    Finally, observe that each portal-set $\VPSd_{i_j}$ is indeed a separator of \(\GPSd\)---which allows us to
    split shortest path computations along this set and then combine the results using a
    min-plus product.

    In total, we indeed compute, for each source vertex, its minimum shortest-path distance to a
    target vertex.

    {\bf Running time.} First observe that the subgraphs
    corresponding to left and right positions (of which there are \(\Oh(d)\) in total)
    each have a size of \(\Oh(d^2)\).
    Next, as the alignment \(\A\) makes at most \(d\) edits, the total size of subgraphs
    between pure stretches is again \(\Oh(d^2)\) (as they span \(\Oh(d))\) positions of $S$).
    Similarly, we may bound the number of pure stretches by \(\Oh(d)\)---which yields the
    number of calls to {\tt DMVO-Query}.
    Finally, we perform \(\Oh(d)\) min-plus matrix-vector multiplications using SMAWK
    (\cref{thm:smawk}).

    In total, our calls to Klein's algorithm (\cref{lem:MSSP}) take time \(\Oh(d^2 \log
    m)\); all performed min-plus multiplications take the same time in total.
\end{proof}

Finally, we combine the various intermediate results of this section to obtain
our first main result.
\stalgmainnk
\begin{proof}
    We first proceed as in \cref{7.9.30-1} to obtain a set \(\S\) of \(\Oh(n/k)\) strings,
    each of length at most \(m + 2k\).
    It suffices to compute $\OccW_k(P,S)$ for all $S \in \S$ and merge the results in $\cO(n)$
    time via bucket sort.

    Next, we preprocess \(T\) for \cref{lem:PtoF} and
    we initialize \cref{clm:preprocess} on \(P\), \(d
    \coloneqq 4k\), and \(w\).

    For each \(S \in \S\), we proceed as follows.
    First, if \(m \le 4d + 1\), we use the classical dynamic programming algorithm to
    compute \(\OccW_k(P,S)\).
    Otherwise, we use \cref{lem:PtoF} to obtain a weighted alignment \(\A_S :
    P \onto S\) of cost at most \(d \coloneqq 4k\) or conclude that $\OccW_k(P,S) = \emptyset$.
    Next, we use \cref{7.9.30-2} on \(\A_S\), with the initialized \ref{prob:dmvo} data
    structure, \(d \coloneqq 4k\) and \(k\).

    We briefly argue about the running time.
    Our global preprocessing takes time \(\Oh(km \log^2 m)\).
    For a fixed $S \in \S$, \cref{lem:PtoF,7.9.30-2,clm:preprocess} combined run in
    time \(\Oh(k^2 \log m)\); in the case when we run the classic DP algorithm, we have
    \(|S| \le m + 2k <  m + d = \Oh(d)\) and thus the DP algorithm takes time
    $\Oh((|P|+1)\cdot (|S|+1))=\Oh(d^2) = \Oh(k^2)$.
    Over all $\cO(n/k)$ elements of $\S$, this takes $\cO((n/k)(k^2 + k^2 \log m)) =
    \cO(nk\log m)$ time. In total, we obtain the claimed running time.
\end{proof}

\section{Toolbox for Monge Matrices and Ferns}
\dglabel{sec:monge_toolbox}[7-16-4,7-16-5,def:bd,def:bdm,wfwfocnbjn,def:fern,7-16-6,7-16-7,def:fernprop,qvhuraqrwd](Definitions and generally useful facts and results for
Monge matrices and ferns)

	In this section, we recall useful results on the representation, the structure, and the multiplication of Monge matrices.
	Further, we introduce the notion of \emph{fern matrices}.
	A fern matrix represents some pairwise distances from a set of leftmost nodes in the top boundary of an alignment graph to a set of rightmost nodes in the bottom boundary of this alignment graph.
	Specifically, a fern matrix and the corresponding actual boundary-to-boundary distance matrix must agree on all values that are at most~$k$, that is, the two matrices are $k$-equivalent.
	Fern matrices are crucial for the following reason.
	While we may not be able to efficiently compute the actual boundary-to-boundary distance matrix, we show that we can efficiently compute a ($k$-equivalent) fern matrix that is Monge.

\subsection{Toolbox for Monge Matrices}

Explicitly storing Monge matrices is in general too expensive for our purposes.
Fortunately, their structure often allows us to represent them using much less space.

\begin{definition}
    \dglabel{7-16-4}(The density matrix \(\dens{A}\) and the core \(\core{A}\) of size
    \(\delta(A) \coloneqq |\core{A}|\)) of a matrix \(A\); the condensed representation of \(A\))
    For a matrix $A$ of size $p \times q$, the \emph{density matrix} of $A$, denoted
    $\dens{A}$, is a matrix of size $(p - 1) \times (q - 1)$ such that
    \[
        \dens{A}\position{i, j} = A\position{i, j + 1} + A\position{i + 1, j} - A\position{i,
        j} - A\position{i + 1, j + 1}
    \] for $i \in \fragmentco{0}{p-1}$ and $j \in
    \fragmentco{0}{q-1}$.

    The \emph{core} of a matrix $A$ is the set
    \[
        \core{A} \coloneqq \{(i, j, \dens{A}\position{i,
        j}) \mid i \in \fragmentco{0}{p-1}, j \in \fragmentco{0}{q-1}, \dens{A}\position{i, j}
        \neq 0 \}.\]
    We denote the \emph{size of the core} of $A$ by $\delta(A) \coloneqq |\core{A}|$.

    For a matrix $A$, the \emph{condensed representation of $A$} comprises in the values
    in the topmost row and the leftmost column of $A$ and the core of $A$.
\end{definition}

\begin{fact}
    \dglabel{remark:submatrix_smallcore}(The core of any contiguous submatrix of a matrix
    $A$ is at most $\delta(A)$)
    The core of any contiguous submatrix of a matrix $A$ is at most $\delta(A)$.
\end{fact}

Monge matrices with small core admit small representations that allow fast access to their entries.

\begin{definition}[{\cite[Lemma 4.4]{gk24}}]
    \dglabel{lm:build_mtx_ds}(The Core-based Matrix Oracle $\mds(A)$ of a
    matrix $A$,~\cite[Lemma 4.4]{gk24})
    The \emph{Core-based Matrix Oracle data structure} $\mds(A)$ of a matrix $A$ of size
    $p \times q$ with $N \coloneqq \max(p,q)$ provides the following interface.
    \begin{description}
        \item[Random access:] given $i \in \fragmentco{0}{p}$ and $j \in
            \fragmentco{0}{q}$, in time $\Oh(\log N)$ returns $A\position{i, j}$.
        \item[Full core listing:] in time $\Oh(1 + \delta(A))$ returns $\core{A}$.\qedhere
    \end{description}
\end{definition}

\begin{lemma}[{\cite[Lemma 4.4]{gk24}}]
    \dglabel{lm:build_mtx_dscomp}(Given a matrix \(A\) of size \(p \times q\), we can construct a core-based matrix
    oracle for \(A\) in time \(\Oh(pq + \delta(A) \log(p + q))\), {\cite[Lemma 4.4]{gk24}})
    Given a matrix \(A\) of size \(p \times q\), we can construct a core-based matrix
    oracle for \(A\) in time \(\Oh(pq + \delta(A) \log(p + q))\).
\end{lemma}
\begin{proof}
    Compared to \cite[Lemma 4.4]{gk24}, we pay \(\Oh(pq)\) to first compute the core of
    \(A\) which (together with the first row and column of \(A\)) use as the input to the
    original algorithm.
\end{proof}

Given the core-based matrix oracle data structure of two matrices, one can efficiently
obtain the core-based matrix oracle data structure of their min-plus product.

\begin{lemmaq}[{\cite[Lemma 4.6]{gk24}}]
    \dglabel{lm:matrix_mult_w}(Given $\mds(A)$ and $\mds(B)$ for Monge \(A\) and \(B\),
    can compute $\mds(A \oplus B)$ in time $\Oh((N + \delta(A) + \delta(B)) \log N)$,
    where \(N\) is the maximum dimension of \(A\) and \(B\),~\cite[Lemma 4.6]{gk24})
    There is an algorithm that, given $\mds(A)$ and $\mds(B)$ for Monge matrices
    $A \in \mathbb{Z}_{\ge 0}^{p \times q}$ and $B \in \mathbb{Z}_{\ge 0}^{q \times r}$
    with $N \coloneqq \max(p,q,r)$, builds
    $\mds(A \oplus B)$ in time $\Oh((N + \delta(A) + \delta(B)) \log N)$.
\end{lemmaq}

The following operation enables us to efficiently retrieve representations of contiguous
submatrices of Monge matrices with small core.

\begin{lemmaq}[{\cite[Lemma 4.7]{gk24}}]
    \dglabel{lm:ds_tools}(Given $\mds(A)$ for some Monge matrix $A$, can compute $\mds(C)$
    for any contiguous submatrix $C$ of~$A$ in time  $\Oh((N + \delta(A)) \log N)$ where
    \(N\) is the maximum dimension of \(A\))
    There is an algorithm that, given $\mds(A)$ for a Monge matrix $A$ of maximum dimension
    $N$, builds $\mds(C)$ in time $\Oh((N + \delta(A)) \log N)$ for any contiguous
    submatrix $C$ of~$A$.
\end{lemmaq}

For our purposes, in the interest of efficiency, instead of exactly computing a desired
distance matrix $D$, it suffices to compute a different matrix $D'$ that agrees with $D$
on all values that do not exceed some threshold $k$.

\begin{definition}
    \dglabel{7-16-5}(\(k\)-equivalence of integers, $a\meq{k} b$)
    For a real threshold $k\in \mathbb{R}$, we say that real numbers $a,b\in \mathbb{R}$
    are \emph{$k$-equivalent}, denoted $a\meq{k} b$, if $a=b$ or $k\le \min\{a,b\}$.

    Two matrices $A,B\in \mathbb{R}^{p\times q}$ are $k$-equivalent, denoted $A \meq{k}
    B$, if all their entries are $k$-equivalent, that is, $A\position{i,j}\meq{k}
    B\position{i,j}$ holds for each $i\in \fragmentco{0}{p}$ and $j\in \fragmentco{0}{q}$.
\end{definition}
\begin{remark}
    Equivalently, two numbers \(a\) and \(b\) are \(k\)-equivalent if they are equal
    whenever at least one of them is at most \(k\).

    The $\meq{k}$ relation is a \emph{congruence} with respect to the $(\min,+)$ algebra on
    $\mathbb{R}_{\ge 0}$: it is an equivalence relation such that, for every $a,a',b,b'\in
    \mathbb{R}_{\ge 0}$, if $a\meq{k}b$ and $a'\meq{k}b'$, then $a+a'\meq{k}b+b'$ and
    $\min\{a,a'\}\meq{k} \min\{b,b'\}$.
    Thus, $\meq{k}$ is also a congruence with respect to the $(\min,+)$-product.
\end{remark}

Consider a scenario where we have to compute a matrix that is $k$-equivalent to the
product of a (long) size-$\ell$ sequence of integer Monge matrices of dimensions at most
$N$ with cores at most $t$.
In this scenario, we can use the following lemma to ensure that the core of all computed
matrices does not exceed $\cOtilde(\sqrt{k})$ and that the output can be produced in time
$\cOtilde(\ell \cdot N \cdot (t+\sqrt{k}))$.

\begin{lemmaq}[{\cite[Lemma 4.14]{gk24}}]
    \dglabel{lm:matrix_mult_k}(Given \(k\) and $\mds(A)$ and $\mds(B)$ for Monge matrices \(A\) and
    \(B\), can compute $\mds(C')$ of a matrix \(C'\) that is \(k\) equivalent to \(A
    \oplus B\) and $\delta(C') \le \Oh(N \sqrt k)$,
    in time $\Oh((N + \delta(A) + \delta(B)) \log N)$, where \(N\) is the
    maximum dimension of \(A\) and \(B\),~\cite[Lemma 4.14]{gk24})
    There is an algorithm that, given $\mds(A)$ and $\mds(B)$ for Monge matrices
    $A \in \mathbb{Z}_{\ge 0}^{p \times q}$ and $B \in \mathbb{Z}_{\ge 0}^{q \times r}$ with $N
    \coloneqq \max
    (p,q,r)$ and a $k \in \mathbb{Z}_{> 0}$, in time $\Oh((N + \delta(A) + \delta(B)) \log N)$
    builds $\mds(C')$, where $C' \in \mathbb{Z}_{\ge 0}^{p \times r}$ is a Monge matrix
    that satisfies $\delta(C') \le \Oh(N \sqrt k)$ and $C' \meq{k} A \oplus B$.
\end{lemmaq}

In the regime of small integer edit weights, we often operate on bounded-difference
matrices.

\begin{definition}[{\cite[Definition 4.8]{gk24}}]
    \dglabel{def:bd}(The \(\beta\)-bounded difference property of an integer matrix)
    Let $\beta\in \Zp$ be a positive integer and let $A\in \mathbb{Z}^{p\times q}$ be an
    integer matrix.
    We say that $A$ is $\beta$-bounded-difference if
    \begin{itemize}
        \item $|A\position{i,j}-A\position{i,j-1}|\le \beta$ holds for each $i\in
            \fragmentco{0}{p}$ and $j\in \fragmentco{1}{q}$, and
        \item $|A\position{i,j}-A\position{i-1,j}|\le \beta$ holds for each $i\in
            \fragmentco{1}{p}$ and $j\in \fragmentco{0}{q}$.\qedhere
    \end{itemize}
\end{definition}

\begin{fact}
    \dglabel{remark:submatrix_bd}(Contiguous submatrices of integer-valued
    $\beta$-bounded-difference matrices are integer-valued and $\beta$-bounded-difference)
    Any contiguous submatrix of an integer-valued $\beta$-bounded-difference matrix is also an
    integer-valued $\beta$-bounded-difference matrix.
\end{fact}

Crucially, integer-valued bounded-difference Monge matrices have small cores.

\begin{lemmaq}[{\cite[Lemma 4.9]{gk24}}]
    \dglabel{lem:bdtocore}($\beta$-bounded-difference implies $\delta(A)\le
    2\beta(\min\{p,q\}-1)$,~\cite[Lemma 4.9]{gk24})
    If a Monge matrix $A\in \mathbb{Z}^{p\times q}$ is $\beta$-bounded-difference, then
    $\delta(A)\le 2\beta(\min\{p,q\}-1)$.
\end{lemmaq}

Moreover, the bounded-difference property is preserved under the min-plus product.

\begin{lemmaq}[{\cite[Lemma 4.10]{gk24}}]
    \dglabel{lem:bdproduct}(The $(\min,+)$-product of \(\beta\)-bounded difference
    matrices is \(\beta\)-bounded difference,~\cite[Lemma 4.10]{gk24})
    Let $A\in \mathbb{Z}^{p\times q}$ and $B\in \mathbb{Z}^{q\times r}$ be
    $\beta$-bounded-difference matrices.
    Then, $C\coloneqq A\oplus B$ is also a $\beta$-bounded-difference matrix.
\end{lemmaq}

\subsection{Ferns and Distance Matrices}

\begin{definition}[{\cite[Definitions 5.5 and 5.6]{gk24}}]
    \dglabel{def:bdm}
    For two strings $X,Y \in \Sigma^*$, the \emph{boundary distance matrix} $D^w_{X,Y}$
    stores the distances from every \emph{input} vertex on the top-left boundary to every
    \emph{output} vertex on the bottom-right boundary of the alignment graph $\AGw(X,Y)$.
    Formally, consider the sequence of points \((\lambda_i)_{i\in \fragment{0}{|X|+|Y|}}\) with
    \[
        \lambda_i \coloneqq \begin{cases}
            (|X|-i,0) & \text{for } i\in \fragment{0}{|X|},\\
            (0, i-|X|) & \text{for }i\in \fragment{|X|}{|X|+|Y|},\\
      \end{cases}.
    \]
    and the sequence of points \((\rho_i)_{i\in \fragment{0}{|X|+|Y|}}\) with
    \[
        \rho_i \coloneqq \begin{cases}
            (|X|,i) & \text{for }i\in \fragment{0}{|Y|},\\
            (|X|+|Y|-i,|Y|) & \text{for }i\in \fragment{|Y|}{|X|+|Y|}.
        \end{cases}
    \]
    Then, $D^w_{X,Y}$ is an $(|X|+|Y|+1) \times (|X|+|Y|+1)$ distance matrix with
    \[
        D^w_{X,Y}\position{a,b} \coloneqq
        \dist_{\AGw(X,Y)}(\lambda_a,\rho_b).
        \qedhere
    \]
\end{definition}

The following observation allows us to utilize the toolbox described in
\cref{sec:monge_toolbox}.

\begin{fact}[{\cite[Observation 5.7]{gk24}}] \dglabel{lm:bounded_core}
    For two strings $X,Y \in \Sigma^*$ and a weight function $w : \sqEsigma \to
    \fragment{0}{W}$,
    the matrix $D^w_{X,Y}$ is $(W + 1)$-bounded-difference and hence $\delta(D^w_{X,Y}) =
    \cO(W \cdot (|X|+ |Y|))$.
\end{fact}

Recall from the technical overview that a \emph{puzzle piece} is a pair of strings; to
keep in line with our intuition that the said strings are substrings of a pattern \(P\)
and a text \(T\), we typically denote puzzle pieces as \((\dP, \dT)\).

\begin{definition}
    \dglabel{tprplazkfz}
    For a puzzle piece \((\dP, \dT)\), write $G\coloneqq \oAGw(\dP, \dT)$
    for the corresponding augmented alignment
    graph and \(\dist_G\) for the shortest path distances in \(G\).
    For \((\dP, \dT)\), and integer $\Delta, \nabla \in \fragment{0}{|\dT|}$,
    the \emph{$(\Delta\to\nabla)$-restricted distance matrix} of \((\dP, \dT)\) with respect
    to $w$ is a $(\Delta+1)\times (\nabla+1)$ matrix $\dmat{\dP}{\dT}{w}{\Delta\to\nabla}$
    where, for each $i\in \fragment{0}{\Delta}$ and \(j\in\fragment{0}{\nabla}\), we set
    \[
        \dmat{\dP}{\dT}{w}{\Delta\to\nabla}\position{i,j} \coloneqq \dist_G((0,i),
        (|\dP|,|\dT|-\nabla+j)).
    \]
    If \(\Delta=\nabla\), we simplify our notation and write \(\Delta\) instead of
    \(\Delta\to\nabla\).
\end{definition}

\begin{remark}
    \dglabel{wfwfocnbjn}[tprplazkfz,fct:monge]
    The \((\Delta\to\nabla)\)-restricted distance matrix of \cref{tprplazkfz} stores
    selected costs
    to traverse the alignment graph from the top boundary to the bottom boundary.

    In particular, if $i \le |\dT|-\nabla+j$, then $\dmat{\dP}{\dT}{w}{\Delta\to\nabla}\position{i,j} =
    \edw{\dP}{\dT\fragmentco{i}{|\dT|-\nabla+j}}$.
    Consequently, $\dmat{\dP}{\dT}{w}{\Delta\to\nabla}$ encodes the edit distances from
    $\dP$ to all fragments of
    $\dT$ that start within the $\Delta+1$ leftmost positions of $\dT$ and end within the
    $\nabla+1$ rightmost positions of $\dT$.
    As long as $k \le |\dP|-|\dT|+\min\{\Delta, \nabla\}$, this includes all fragments of length at least
    $|\dP|-k$, and hence all $(k,w)$-error occurrences of $\dP$ in $\dT$.

    From \cref{fct:monge}, we also learn that $\dmat{\dP}{\dT}{w}{\Delta\to\nabla}$ is indeed a Monge
    matrix, as it is a submatrix of the boundary distance matrix.
\end{remark}

Unfortunately, we are not able to directly compute the matrix \(\dmat{\dP}{\dT}{w}{\Delta}\)
in an efficient manner. Nevertheless, for our purposes, it suffices to compute a suitable
approximation---which we can do efficiently.
We formalize the required concepts next.

\begin{definition}\dglabel{def:fern}
    For a puzzle piece \((\dP, \dT)\), integer $\Delta, \nabla  \in
    \fragment{0}{|\dT|}$, and a threshold $k\in \mathbb{R}_{\ge 0}$,
    a \emph{$(\Delta \to \nabla, k)$-fern}%
    \footnote{In the unweighted setting, a similar structure called \emph{seaweed}
        fulfills a similar purpose when paired with its corresponding \emph{seaweed product}.
        Hence, we chose the name \emph{fern}, as common knowledge tells us that (liquid)
        seaweed (extract) is a good fertilizer to cultivate many species of ferns.}
    of strings $\dP$ and $\dT$ with respect to \(w\)
    is a matrix that is $k$-equivalent to $\dmat{\dP}{\dT}{w}{\Delta\to\nabla}$.
    We write \(\fmats{\dP}{\dT}{w}{\Delta\to\nabla}{k}\) for the set of all
    \((\Delta\to\nabla,k)\)-ferns of \(\dP\) and \(\dT\) with respect to \(w\).

    Whenever \(\Delta = \nabla\), we simplify our notation and write \(\Delta\) instead of
    \(\Delta\to\nabla\).
\end{definition}
\begin{remark}
    Ferns are not uniquely defined;
    that is,
    typically we have \(|\fmats{\dP}{\dT}{w}{\Delta\to\nabla}{k}| > 1\).
    This is intentional to allow for
    efficient algorithms that compute and multiply ferns.
\end{remark}

\begin{lemma}
    \dglabel{lem:trim}[remark:submatrix_smallcore,remark:submatrix_bd]
    Consider a puzzle piece \((\dP, \dT)\), integers $\Delta', \Delta, \nabla', \nabla \in
    \fragment{0}{|\dT|}$ with $\Delta' \leq \Delta$ and $\nabla' \leq \nabla$, and
    thresholds $k', k\in \mathbb{R}_{\ge 0}$ with $k'\leq k$.
    For any $F \in \fmats{\dP}{\dT}{w}{\Delta\to\nabla}{k}$,
    we have
    \[
        F' \coloneqq F\position{\fragment{0}{\Delta'},\fragment{\nabla-\nabla'}{\nabla}} \in
          \fmats{\dP}{\dT}{w}{\Delta'\to\nabla'}{k'}.
    \]

    Further the core of $F'$ is at most $\delta(F)$ and, if \(F\) is integer-valued and
    \(\beta\)-bounded difference, then so is \(F'\).
\end{lemma}
\begin{proof}
    For all $i \in \fragment{0}{\Delta'}$ and $j \in \fragment{0}{\nabla'}$, we have
    \begin{align*}
        F\position{i, \nabla - \nabla' + j} & \meq{k} \dist_G((0,i),
        (|\dP|,|\dT|-\nabla + \nabla - \nabla' + j))\\
                                            & = \dist_G((0,i),
                                            (|\dP|,|\dT|- \nabla' + j)).
    \end{align*}
    This implies that $F'$ is $k$-equivalent (and hence $k'$-equivalent) to $\dmat{\dP}{\dT}{w}{\Delta'\to\nabla'}$.

    The claimed upper bounds on the core and the difference of $F'$ follow instantly from
    \cref{remark:submatrix_smallcore} and \cref{remark:submatrix_bd}, respectively.
\end{proof}

For a puzzle piece \((\dot{P},\dot{T})\), we define its \emph{shift} as $\shift(\dot{P},\dot{T})
\coloneqq |\dot{T}|-|\dot{P}|$ and its \emph{torsion} as $\tor(\dot{P},\dot{T})
\coloneqq |\shift(\dot{P},\dot{T})|$.
The torsion of a sequence of puzzle pieces is the sum of the torsions of the individual
pieces.

We ultimately intend to fit together multiple puzzle pieces into longer strings, so-called \(\Delta\)-puzzles.
We formalize related concepts next.

\begin{definition}
    \dglabel{7-16-6}
    For an integer $\Delta>0$, we say that strings $S, S'\in \Sigma^{\ge \Delta}$ fit
    (with respect to \(\Delta\))
    if they have an overlap of size $\Delta$, that is,
    if
    \[
        S\fragmentco{|S|-\Delta}{|S|}=S'\fragmentco{0}{\Delta}.
    \]
    In this case, we denote their \emph{$\Delta$-concatenation} as
    \[
         S\odot_{\Delta} S' \coloneqq  S\odot S'\fragmentco{\Delta}{|S'|} = S\fragmentco{0}{|S|-\Delta} \odot S'.
    \]

    Two puzzle pieces \((\dot{P},\dot{T})\) and \((\ddot{P},\ddot{T})\)
    \emph{fit} (with respect to \(\Delta\)),
    denoted by \((\dot{P},\dot{T}) \mapsto_{\Delta} (\ddot{P},\ddot{T})\),
    if \[
        \dot{T}\fragmentco{|\dot{T}|-\Delta}{|\dot{T}|} = \ddot{T}\fragmentco{0}{\Delta}
        \quad\text{and}\quad
        \dot{P}\fragmentco{|\dot{P}|-\Delta}{|\dot{P}|} = \ddot{P}\fragmentco{0}{\Delta}.
    \]
    For two fitting puzzle pieces \((\dot{P},\dot{T}) \mapsto_{\Delta} (\ddot{P},\ddot{T})\),
    we define the \emph{stitched piece} as
    \((\dot{P}\odot_{\Delta} \ddot{P},\dot{T}\odot_{\Delta}\ddot{T})\).
\end{definition}
\begin{remark}
    For a sequence of puzzle pieces \(\P \coloneqq (\dot{P}_1, \dot{T}_1), \dots,
    (\dot{P}_z, \dot{T}_z)\) where consecutive pieces \(\Delta\)-fit, we also say that the
    sequences \(\dot{P}_1, \dots, \dot{P}_z\) and \(\dot{T}_1, \dots, \dot{T}_z\) form
    \(\Delta\)-puzzles; we write
    \(P \coloneqq \val_\Delta(\dot{P}_1, \dots, \dot{P}_z)\) and
    \(T \coloneqq \val_\Delta(\dot{T}_1, \dots, \dot{T}_z)\) for the resulting stitched
    strings.
    For simplicity, we also say that the sequence \(\P\) has the value \((P, T)\).
\end{remark}

For a puzzle piece \((\dot{P},\dot{T})\) and a fixed \(\Delta\), the first
\(\Delta\) and the last \(\Delta\) characters of either string may end up overlapping
corresponding parts of a neighboring puzzle piece.
Hence, when dealing with the alignment graphs of pieces (and when computing shortest
paths therein), we have to ensure that we avoid duplicate vertices along the paths
that we consider: if we were to naively compute top-to-bottom distances for each
piece, we would be unable to combine such distances for consecutive pieces.
This motivates the following definition.

\begin{definition}
    \dglabel{7-16-7}
    For a puzzle piece $(\dP,\dT)$ and an integer \(\Delta\), besides
    $\dmat{\dP}{\dT}{w}{\Delta}$, we define
    \begin{itemize}
        \item the \emph{internal} distance matrix $\dmati{\dP}{\dT}{w}{\Delta} \coloneqq
            \dmat{\dP\fragmentco{\floor{\Delta/2}}{|\dP|-\ceil{\Delta/2}}}{\dT}{w}{\Delta}$;
        \item the \emph{leading} distance matrix $\dmatl{\dP}{\dT}{w}{\Delta} \coloneqq
            \dmat{\dP\fragmentco{0}{|\dP|-\ceil{\Delta/2}}}{\dT}{w}{\Delta}$;
        \item the \emph{trailing} distance matrix $\dmatt{\dP}{\dT}{w}{\Delta} \coloneqq
            \dmat{\dP\fragmentco{\floor{\Delta/2}}{|\dP|}}{\dT}{w}{\Delta}$.
    \end{itemize}
    For a threshold $k\in \mathbb{R}_{\ge 0}$,
    \begin{itemize}
        \item an internal \((\Delta,k)\)-fern is a matrix \(k\)-equivalent to
            \(\dmati{\dP}{\dT}{w}{\Delta}\); we use \(\fmatsi{\dP}{\dT}{w}{\Delta}{k}\)
            for the set of all such ferns.
        \item a leading \((\Delta,k)\)-fern is a matrix \(k\)-equivalent to
            \(\dmatl{\dP}{\dT}{w}{\Delta}\); we use \(\fmatsl{\dP}{\dT}{w}{\Delta}{k}\)
            for the set of all such ferns.
        \item a trailing \((\Delta,k)\)-fern is a matrix \(k\)-equivalent to
            \(\dmatt{\dP}{\dT}{w}{\Delta}\); we use \(\fmatst{\dP}{\dT}{w}{\Delta}{k}\)
            for the set of all such ferns.
            \qedhere
    \end{itemize}
\end{definition}
\begin{remark}
    For our purposes, it suffices to define internal, leading, and trailing distance
    matrices only for \(\Delta = \nabla\), allowing for a slightly simplified definition
    compared to \cref{tprplazkfz,def:fern}.
\end{remark}

We conclude with discussing the main property of ferns: they decompose under the
\((\min,+)\)-product.
We argue via the corresponding distance matrices.

\begin{theorem}
    \dglabel{cor:dproduct}[obs:band]
    Consider a sequence of puzzle pieces \(\P \coloneqq (\dot{P}_0, \dot{T}_0), \dots,
    (\dot{P}_{z-1}, \dot{T}_{z-1})\) that forms a pair of \(\Delta\)-puzzles of size $z \ge 2$ with values
    \(P \coloneqq \dP_0\odot_\Delta \cdots \odot_\Delta \dP_{z-1}\) and
    \(T \coloneqq \dT_0 \odot_\Delta \cdots \odot_\Delta \dT_{z-1}\).

    For every integer $k \le \floor{\Delta/2} - \tor(\P)$, we have
    \[
        \dmat{P}{T}{w}{\Delta} \meq{k} \dmatl{\dP_0}{\dT_0}{w}{\Delta}\oplus
        \Big(\bigoplus_{s=1}^{z-2} \dmati{\dP_s}{\dT_s}{w}{\Delta}\Big) \oplus
        \dmatt{\dP_{z-1}}{\dT_{z-1}}{w}{\Delta}.
    \]
\end{theorem}
\begin{proof}
    It is instructive to unfold the definitions of the leading, internal, and trailing
    distance matrices; for notational simplicity, we do so by defining a ``cut'' version
    of the sequence \((\dP_i)\) that removes the overlaps of consecutive strings.
    Formally, we set
    \begin{align*}
        P_0 &\coloneqq
        \dP_0\fragmentco{0}{|\dP_0|-\ceil{\Delta/2}},\\
        P_s &\coloneqq
        \dP_s\fragmentco{\floor{\Delta/2}}{|\dP_s|-\ceil{\Delta/2}} \quad\text{for every
        \(s \in \fragmentoo{0}{z-1}\), and}\\
        P_{z-1} &\coloneqq
        \dP_{z-1}\fragmentco{\floor{\Delta/2}}{|\dP_{z-1}|}.
    \end{align*}

    Now, we start by analyzing how the torsion of \(\P\) evolves in the span of the
    sequence.
    To that end, define sequences $(p_s)_{s=0}^z$ with $p_s \coloneqq |P_0|+\cdots
    +|P_{s-1}|=|P_0\odot \cdots \odot P_{s-1}|$
    and $(t_s)_{s=0}^z$ with $t_s \coloneqq |\dT_0|+\cdots +|\dT_{s-1}|-s\cdot \Delta$, which
    equals $|\dT_0 \odot_\Delta \cdots \odot_\Delta \dT_{s-1}|-\Delta$ for $s\in
    \fragmentoc{0}{z}$.

    \begin{claim}
        \label{xtagmnrjyl}
        For each $s \in \fragmentoo{0}{z}$, we have
        \[
            k \le \min\{p_s-t_s, (|P|-p_s)-(|T|-t_s)+\Delta\}.
        \]
    \end{claim}
    \begin{claimproof}
        Observe that, indeed, we have
        \[
            P= \dP_0\odot_\Delta \cdots \odot_\Delta \dP_{z-1}=P_0\odot \cdots \odot P_{z-1}.
        \]
        In particular, for every $s \in \fragmentoo{0}{z}$,
        we have
        \[
            p_s = |\dP_0\odot\cdots \odot \dP_{s-1}|-\ceil{\Delta/2}
            \quad \text{and} \quad
            |P| - p_s
            = |\dP_s\odot\cdots \odot \dP_{z-1}|-\floor{\Delta/2}.
        \]
        Consequently, we obtain
        \begin{align*}
            p_s - t_s + \Delta
            &= |\dP_0\odot\cdots \odot \dP_{s-1}|-|\dT_0\odot_\Delta \cdots \odot_\Delta
            \dT_{s-1}|+\floor{\Delta/2} \\
            &= \floor{\Delta/2} - \sum_{t=0}^{s-1}\left(|\dT_t|-|\dP_t|\right) \ge
            \floor{\Delta/2} - \tor(\P) \ge k
        \end{align*}
        and
        \begin{align*}
             (|P|-p_s)-(|T|-t_s)
            + \Delta
            &= |\dP_s\odot\cdots \odot \dP_{z-1}| - |\dT_s\odot_\Delta \cdots \odot_\Delta
            \dT_{z-1}|+\ceil{\Delta/2} \\
            &= \ceil{\Delta/2} - \sum_{t=s}^{z-1}\left(|\dT_t|-|\dP_t|\right) \ge
            \floor{\Delta/2} - \tor(\P) \ge k;
        \end{align*}
        which completes the proof of the claim.
    \end{claimproof}

    Next, consider the alignment graph $G \coloneqq \oAGw(P,T)$ and observe that, for each
    $s\in \fragmentco{0}{z}$, the subgraph induced by $\fragment{p_s}{p_{s+1}}\times
    \fragment{t_s}{t_{s+1}+\Delta}$ is isomorphic to $G_s \coloneqq \oAGw(P_s,\dT_s)$.

    \begin{claim}
        \label{mstxosxeuz}
        For any sequence $(i_s)_{s=0}^z$ with $i_s\in \fragment{0}{\Delta}$ for $s\in
        \fragment{0}{z}$,
        we have
        $\dmat{P}{T}{w}{\Delta}\position{i_0,i_z} \le \sum_{s=0}^{z-1}
        \dmat{P_s}{\dT_s}{w}{\Delta}\position{i_s,i_{s+1}}$.
    \end{claim}
    \begin{claimproof}
        By triangle inequality, we have
        \begin{align*}
            \dmat{P}{T}{w}{\Delta}\position{i_0,i_z}
            &= \dist_{G}((0,i_0),(|P|,|T|-\Delta+i_z))
            = \dist_G((p_0,t_0+i_0),(p_z,t_z+i_z)) \\
            &\le \sum_{s=0}^{z-1} \dist_G((p_s,t_s+i_s),(p_{s+1},t_{s+1}+i_{s+1})).
        \end{align*}
        Moreover, the aforementioned embedding of $G_s$ into $G$ yields
        \begin{align*}
            &\dist_G((p_s,t_s+i_s),(p_{s+1},t_{s+1}+i_{s+1}))\\
            &\quad\le\dist_{G_s}((0,i_s),(p_{s+1}-p_s,t_{s+1}-t_s+i_{s+1}))
            = \dmat{P_s}{\dT_s}{w}{\Delta}\position{i_s,i_{s+1}}.
        \end{align*}
        Combining the inequalities yields the claim.
    \end{claimproof}

    As the sequence \((i_s)\) is arbitrary, \cref{mstxosxeuz} implies
    \[
        \dmat{P}{T}{w}{\Delta}\position{i_0,i_z} \le
        \Big(\bigoplus_{s=0}^{z-1}
        \dmat{P_s}{\dT_s}{w}{\Delta}\Big)\position{i_0,i_z};
    \] that is, the natural statement that forcing a path to contain certain vertices may
    only increase its length.
    We conclude with a proof that we may indeed describe also a \emph{shortest} path of
    length at most \(k\) in a similar manner.

    \begin{claim}
        For any \(i_0, i_z \in \fragment{0}{\Delta}\),
        if
        $\dmat{P}{T}{w}{\Delta}\position{i_0,i_z} \le k$, then we also have
        \[
            \Big(\bigoplus_{s=0}^{z-1}
            \dmat{P_s}{\dT_s}{w}{\Delta}\Big)\position{i_0,i_z} \le
            \dmat{P}{T}{w}{\Delta}\position{i_0,i_z}.
        \]
    \end{claim}
    \begin{claimproof}
        Consider a shortest path $\A$ from $(0,i_0)=(p_0,t_0+i_0)$ to
        $(|P|,|T|-\Delta+i_z)=(p_z,t_z+i_z)$ and assume that its length does not exceed
        $k$.
        For each $s\in \fragmentoo{0}{z}$, let $(p_s,\talt_s)\in \A$ be an arbitrary
        vertex on~$\A$ whose first coordinate is equal to $p_s$.

        By \cref{obs:band}, we have $\talt_s \in \fragment{p_s-k}{p_s+|T|-|P|+k}$.
        In particular, we have
        \begin{align*}
            \talt_s &\ge p_s - k \ge p_s - (p_s-t_s) = t_s\quad \text{and}\\
            \talt_s &\le p_s+|T|-|P|+k \le p_s+|T|-|P|+
            (|P|-p_s)-(|T|-t_s)+\Delta=t_s+\Delta.
        \end{align*}
        Hence, \cref{xtagmnrjyl} yields
        $\talt_s = t_s + i_s$ for some $i_s \in \fragment{0}{\Delta}$.
        By \cref{lem:oagw}, we have
        \begin{align*}
            &\dist_G((p_s,t_s+i_s),(p_{s+1},t_{s+1}+i_{s+1}))\\
            &\quad = \dist_{G_s}((0,i_s),(p_{s+1}-p_s,t_{s+1}-t_s+i_{s+1}))\\
            &\quad =\dmat{P_s}{\dT_s}{w}{\Delta}\position{i_s,i_{s+1}}.
        \end{align*}
        Thus, $\dmat{P}{T}{w}{\Delta}\position{i_0,i_z}\le \sum_{s=0}^{z-1}
        \dmat{P_s}{\dT_s}{w}{\Delta}\position{i_s,i_{s+1}}$ holds for some
        $i_1,\ldots,i_{z-1}\in \fragment{0}{\Delta}$, which in turn implies the claim.
    \end{claimproof}

    In total, we obtain
    $\dmat{P}{T}{w}{\Delta}\meq{k} \bigoplus_{s=0}^{z-1} \dmat{P_s}{\dT_s}{w}{\Delta}$
    and thus the claim.
\end{proof}

\begin{remark}
    Since $\meq{k}$ is a congruence relation, in the context of \cref{cor:dproduct}, a
    $(\Delta,k)$-fern of $(P,T)$ can be obtained as the $(\min,+)$-product of a leading
    $(\Delta,k)$-fern of $(\dP_0,\dT_0)$, internal $(\Delta,k)$-ferns of $(\dP_s,\dT_s)$ for
    $s\in \fragmentco{1}{z-1}$, and a trailing $(\Delta,k)$-fern of $(\dP_{z-1},\dT_{z-1}$).
\end{remark}

\begin{remark}
    We may intuitively understand the proof of \cref{cor:dproduct} as follows: a shortest
    path of length at most \(k\) may use at most \(k\) horizontal or vertical edges of the
    alignment graph. Thus, all shortest path that we care about are contained in a
    diagonal band that allows for a ``safety margin'' of at least \(k\) to the left and
    to the right.

    Now, when cutting the original fern into small ferns, we just have to ensure that
    the inputs and outputs of every small fern fully contains the safety margin---for this,
    we need the bound on the torsion of the sequence of puzzle pieces.
\end{remark}

\subsection{Algorithms for Multiplying Matrices}

\begin{definition}
    \dglabel{def:fernprop}
    We say that $P$ is a \emph{matrix property} if the following holds for every dimension $q\in \Zp$:
    There is a family $\Prp{q}\subseteq \Rp^{q\times q}$ of $q\times q$ matrices that \emph{satisfy $P$},
    and this family is compatible with the  $(\min,+)$-product,
    that is, for every two matrices $A,B\in \Prp{q}$, we have \( A \oplus B\in \Prp{q}\).

    We say that $P$ admits \emph{\(T_M^P(q)\)-time matrix multiplication} if there is an algorithm
    that, given two matrices $A,B\in \Prp{q}$, in time \(T_M^P(q)\)
    computes $A \oplus B$.

    We say that $P$ admits \emph{\(T_V^P(q)\)-time vector multiplication} if there is an algorithm
    that, given a matrix \(A\in \Prp{q}\) and a vector \(v\in \Rp^{q\times 1}\),
    in time \(T_V^P(q)\) computes a vector $A \oplus v \in \Rp^{q\times 1}$.
\end{definition}
\begin{remark}
    In this work, we typically study the matrix properties of being Monge or bounded
    difference; they are matrix
    properties due to \cite[Theorem~2]{Tis15} and \cref{lem:bdproduct}.
\end{remark}

\begin{lemma}
    \dglabel{uogiuvnbph}[fct:monge-product-is-monge]
    The Monge property is a matrix property
    that admits matrix multiplication in $\Oh(q^2)$ time and vector multiplication in $\Oh(q)$ time.
\end{lemma}
\begin{proof}
    We use \cite{SMAWK} (see \cref{thm:smawk}); the $(\min,+)$-product of Monge
    matrices is Monge again by \cite[Theorem~2]{Tis15}.
\end{proof}
\begin{remark}
    \dglabel{qvhuraqrwd}[lem:small_ed,lem:MSSPverify]
    Observe that we may use \cite{MSSP}  (see \cref{lem:MSSP}) to compute Monge ferns for
    a pair of strings (see \cref{lem:small_ed} which follows \cref{lem:MSSPverify}).
    Unfortunately, this is fast enough only in very specific scenarios.
    Hence, we later devise more sophisticated algorithms to efficiently construct Monge ferns.
\end{remark}

\begin{lemma}
    \dglabel{azjqbhozka}[lem:bdtocore,lm:matrix_mult_k,fct:monge-product-is-monge]
    For every threshold $\beta \in \Zp$, the \(\beta\)-bounded difference (integer) Monge matrices
    define a matrix property. Represented as core-based  matrix oracles, they admit
    \(\Oh(\beta q \log q)\)-time matrix multiplication and
    \(\Oh(q \log q)\)-time vector multiplication.
\end{lemma}
\begin{proof}
    By \cref{lem:bdtocore}, our input matrices have a core of \(\Oh(\beta q)\); hence,
    \cref{lm:matrix_mult_k} allows us to multiply them in time \(\Oh(\beta q \log
    q)\).
    Finally, the resulting matrix is again \(\beta\)-bounded difference by
    \cref{lem:bdproduct} and Monge by \cite[Theorem~2]{Tis15}.

    For the vector multiplication, we proceed as for general weights; however, we
    afterward list the resulting vector explicitly using \(\Oh(q)\) queries to the
    core-based matrix oracle.
\end{proof}

Finally, it is instructive to restate (a simplified) \cite[Theorem~5.10]{gk24} in our new terminology.

\begin{lemma}[{\cite[Theorem~5.10, simplified]{gk24}}]
    \dglabel{lem:etdyn}[lm:ds_tools,lm:build_mtx_ds,lem:bdtocore]
    There is a data structure that supports all of the following operations.
    \begin{description}
        \item[Initialization.]
            Given non-empty strings \(X, Y \in \Sigma^+\) of total length \(n\) and oracle access to an integer weight
            function \(w : \sqEsigma \to \fragment{0}{W}\), initialize the data structure
            in time \(\Oh(Wn^2 \log n)\).
        \item[Fern retrieval.]
            Given integer \(\Delta\) and \(k\) return a \((W+1)\)-bounded difference
            Monge \((\Delta,k)\)-fern (represented as a core-based matrix oracle of size
            and core \(\Oh(W\Delta)\)) in time \(\Oh(W\Delta \log n)\).
        \item[Character edits.]
            Apply a single character edit to either \(X\) or \(Y\) in time \(\Oh(Wn \log
            n)\).
    \end{description}
\end{lemma}
\begin{proof}
    Compared to \cite{gk24}, we output a fern instead of a Core-based matrix oracle of the
    whole distance matrix.
    To that end, we use \cref{lm:ds_tools} to extract the core corresponding to the fern
    (submatrix) of the distance matrix that is returned by the original data structure.

    Now, in particular, we argue about the core of the output fern; the running time guarantees
    then follow from \cref{lm:build_mtx_ds}.
    First, as a contiguous submatrix of a \((W+1)\)-bounded difference Monge fern, the output fern has the
    same properties.
    Next, by \cref{lem:bdtocore}, the fern thus has a core of at most \(\Oh(\Delta W)\);
    completing the proof.
\end{proof}

\section{Efficient Solutions for Growing Ferns, Verifying Occurrences, and Listing Fragments}\label{sec:verify}

As mentioned in the introduction, the Landau--Vishkin algorithm for pattern matching
under unweighted edit distance~\cite{LandauV89} can be viewed as a repeated application of
a subroutine
that, given an $\cO(k)$-size interval $I$ of positions in $T$, computes $I \cap
\OccE_k(P,T)$.
A solution for this problem, called \verify, that takes $\cO(k^2)$ time in the \pillar
model can be found in a work~\cite{ColeH98} of Cole and Hariharan, who used it as a
subroutine in their $\cO(n+k^4\cdot n/m)$ algorithm for \PMED.
Unfortunately, under weighted edit distance, one cannot hope to obtain a solution with the
same complexity: an algorithm for \verify under weighted edit distance taking $\cO(k^2)$
time in the \pillar model would imply an $\cO(n+k^2)$-time algorithm for deciding whether
the weighted edit distance of two strings of length at most $n$ is at most $k$, thus refuting the
lower bound of
Cassis, Kociumaka, and Wellnitz~\cite{ckw23} for this problem---see \cref{thm:lb_wed}.
In this section, we show that \verify can be solved in $\cOtilde(k^3)$ time in the \pillar
model in the case of weighted edit distance.
Further, if the weights are integers upper bounded by some threshold $W$, we show that
\verify can be solved in $\cOtilde(\min (k^{2.5}, W\cdot k^2))$ time in the \pillar model.

Our techniques also lead to an efficient algorithm for the problem of computing all
fragments~$U$ of $T$ with $\edw{P}{U} \leq k$, as well as $\edw{P}{U}$.
This problem is formally defined below.

\begin{problem}[ListAllOccs]{ListAllOccs($P$, $T$, $k$, $w$)}
    \label{keyxtssyqb}

    \PInput{A text $T$ of length $n$, a pattern $P$ of length $m$, an integer
        threshold $k >0$, and oracle access to a
        normalized weight function
    \(w:\sqEsigma \to \intvl{0}{W}\).}
    \POutput{The set $\{(i,j, \edw{P}{T\fragmentco{i}{j}}) : 0 \le i \le j \le n \text{
    and } \edw{P}{T\fragmentco{i}{j}} \le k \}$.}
\end{problem}

On the way to our solutions for the \verify problem and the \allOccs problem, we derive an
efficient procedure for constructing ferns of $P$ and $T$ (recall \cref{def:fern}), which
plays a critical role in \cref{sec:dpm}.

\begin{problem}[GrowFern]{GrowFern($P$, $T$, $k$, $w$)}
    \label{prob:grow_fern}
    \PInput{A text $T$ of length $n$, a pattern $P$ of length $m$,
    an integer threshold $k > 0$, and oracle access to a normalized weight function
    \(w:\sqEsigma \to \intvl{0}{W}\).}
    \POutput{A $(\min(n,k),k)$-fern of $P$ and $T$ with respect to $w$.}
\end{problem}

Instances of the \growFern problem where the involved strings are of length \(\Oh(k)\) can
be solved efficiently using an MSSP data structure \cref{lem:MSSPverify}.

\begin{lemma}
    \dglabel{gucbenbjuj}[lem:MSSPverify](Any instance of \growFern with \(|P|+|T|=\Oh(k)\)
    can be solved in $\Oh(k^2 \log k)$)
    Any instance of \growFern with \(|P|+|T|=\Oh(k)\) can be solved in
    time $\Oh(k^2 \log k)$.
\end{lemma}
\begin{proof}
    We first build the data structure of \cref{lem:MSSPverify} for $\oAGw(P,T)$ in
    $\cO(k^2 \log k)$ time.
    Then,
    we compute $\edw{P}{T\fragmentco{i}{j}}$ for all $(i,j) \in \fragment{0}{\min(n,k)}
    \times \fragment{n-\min(n,k)}{n}$ using said data structure in $\Oh(k^2 \log k)$
    time---we issue
    queries asking for the distance from $(0,i)$ to $(m,n-j)$ for every $(i,j)\in
    \fragment{0}{\min(n,k)}^2$.
\end{proof}

\subsection{Growing Ferns: the Case of Small Edit and Self-edit Distance}

Let us recall the definition of \emph{self-edit distance}, a notion introduced
in~\cite{ckw23}.
We say that an alignment $\A : X \onto X$ is a \emph{self-alignment}
if $\A$ does not align any character $X\position{x}$ to itself.
We define the \emph{self-edit distance} of $X$ as $\selfed(X) \coloneqq \min_\A
\ed_{\A}(X, X)$,
where the minimization ranges over self-alignments $\A : X \onto X$. In words,
$\selfed(X)$ is the minimum (unweighted) cost of a self-alignment.
We can interpret a self-alignment as a
$(0, 0) \leadsto (|X|, |X|)$ path in the alignment graph $\AG(X, X)$
that does not contain any edge of the main diagonal.
Self-edit distance enjoys the following useful properties.

\begin{lemmaq}[Properties of $\selfed{}$, {\cite[Lemma 4.2]{ckw23}} and {\cite[Lemma
    6.2]{gk24}}]
    \dglabel{fct:selfed-properties}
    Let $X \in \Sigma^*$ denote a string. Then, all of the following hold:
    \begin{description}
        \item[Monotonicity.\!\!\!\!] For any $\ell' \le \ell \le r \le r'\in
            \fragment{0}{|X|}$, we have
            \[
                \selfed(X\fragmentco{\ell}{r}) \leq \selfed(X\fragmentco{\ell'}{r'}).
            \]
        \item[Sub-additivity.\!\!\!\!] For any $\mu \in \fragment{0}{|X|}$, we have
            \[
                \selfed(X) \leq \selfed(X\fragmentco{0}{\mu}) +
                \selfed(X\fragmentco{\mu}{|X|}).
            \]
        \item[Triangle inequality.\!\!\!\!] For any $Y\in \Sigma^*$, we have
            \(
            \selfed(Y) \le \selfed(X)+2\ed(X,Y).
            \)
        \item[Continuity.\!\!\!\!] For any $i\in \fragmentoo{0}{|X|}$, we have
            \(
            \selfed(X\fragment{0}{i})-\selfed(X\fragmentco{0}{i})  \in \{0,1\}.
            \)
            \qedhere
    \end{description}
\end{lemmaq}

The following lemma gives a structural characterization of a pair of strings that are at
small edit distance and have small self-edit distance.

\begin{lemmaq}[{\cite[Lemma 7.1]{gk24}; see also~\cite[Lemma 4.6 and Claim 4.11]{ckw23}}]
    \dglabel{lem:decomp}
    There is a \pillar algorithm that, given strings $X,Y$ and a positive integer $k$
    satisfying $\selfed(X)\le k \le |X|$ and $\ed(X,Y)\le k$, in $\Oh(k^2)$ time builds a
    decomposition $X=\bigodot_{i=0}^{z-1} X_i$ and a sequence of fragments
    $(Y_i)_{i=0}^{z-1}$ of $Y$ that satisfies all of the following.
    \begin{itemize}
        \item Each phrase $X_i=X\fragmentco{x_i}{x_{i+1}}$ is of length $x_{i+1}-x_i\in
            \fragmentco{k}{2k}$ and each fragment $Y_i = Y\fragmentco{y_i}{y'_{i+1}}$
            satisfies $y_i=\max(0, x_i-k)$ and $y'_{i+1}=\min(x_{i+1}+3k,|Y|)$.
        \item There is a set $F\subseteq \fragmentco{0}{z}$ of size $|F|=\Oh(k)$ such that
            $X\fragmentco{x_i}{x_{i+1}}=X\fragmentco{x_{i-1}}{x_i}$ and
            $Y\fragmentco{y_i}{y'_{i+1}}=Y\fragmentco{y_{i-1}}{y'_i}$ hold for each $i\in
            \fragmentco{0}{z}\setminus F$ (in particular, $0\in F$).
    \end{itemize}
    The algorithm returns the set $F$ and, for each $i\in F$, the endpoints of
    $X_i=X\fragmentco{x_i}{x_{i+1}}$.\footnote{This determines the whole decomposition
        because $(x_i)_{i=\ell}^r$ is an arithmetic progression for every
    $\fragmentoc{\ell}{r}\subseteq \fragmentoc{0}{z}\setminus F$.}
\end{lemmaq}

In the case when the weights are small integers, we utilize the following lemma from
\cite{gk24} that allows us to efficiently compute, for each $i$, a representation of
$D^w_{X_i,Y_i}$.\footnote{\cite[Lemma 7.2]{gk24} returns an intricate data structure that
is outlined in \cite[Theorem 5.10]{gk24}. In \cref{lem:addtoboxds}, instead of describing
this data structure in detail, we specify as output only the part of said data structure
that is useful to us.}

\begin{lemmaq}[{\cite[see Lemma 7.2]{gk24}}]
    \dglabel{lem:addtoboxds}
    Consider a weight function $w : \sqEsigma \to \fragment{0}{W}$, a positive
    integer~$k$, and two strings $X, Y \in \Sigma^{\le n}$ such that $\selfed(X) \le k \le
    |X|$ and $\ed(X,Y)\le k$.
    Moreover, suppose that $X=\bigodot_{i=0}^{m-1}X_i$,  $(Y_i)_{i=0}^{m-1}$, and
    $F\subseteq\fragmentco{0}{m}$ satisfy the conditions in the statement of
    \cref{lem:decomp}.

    There is an algorithm that, given oracle access to $w$, the integer $k$, the strings
    $X,Y$, the set~$F$, and the endpoints of $X_i$ for each $i\in F$, takes $\Oh(W\cdot
    k^2\log^2n)$ time plus $\Oh(k^2)$ \pillar operations and returns a reference to
    $\mds(D^w_{X_i,Y_i})$ for each $i \in F$.

    A variant of this algorithm that takes $\Oh(k^{2.5}\log^2 n)$ time plus $\Oh(k^2)$
    \modelname{} operations returns, for each $i \in F$, a reference to $\mds(M_i)$, where
    $M_i$ is a Monge matrix~$M_i$ such that $M_i \meq{k} D^w_{X_i,Y_i}$ and
    $\delta(M_i)=\cO(k^{1.5})$.
\end{lemmaq}

In the next lemma, we show an efficient solution for \allOccs in the case when $\selfed(P)
\le k$ and $\edu{P}{T} \le k$.
The proof of said lemma closely follows the proof of \cite[Lemma 7.3]{gk24}.

\begin{lemma}[{\tt\apmsmallsed($P$, $T$, $k$, $w$)}, {\cite[adapted from Lemma
    7.3]{gk24}}]
    \dglabel"{lem:small_sed}[gucbenbjuj,lem:decomp,lm:build_mtx_ds,fct:monge,lem:addtoboxds,lm:ds_tools,lm:bounded_core,lm:matrix_mult_w,lem:bdproduct,lm:matrix_mult_k]
    Any instance of the \growFern problem where $\selfed(P) \le k$ and $\edu{P}{T} \le k$ can be solved in
    \begin{enumerate}
        \item $\Oh(k^3 \log (nk))$ time in the \pillar model, with the returned fern being
            Monge;
        \item $\Oh(W\cdot k^2 \log^2 (nk))$ time in the \pillar model if all the weights
            are integers, with the returned fern being Monge and
            $(W+1)$-bounded-difference;
        \item $\cO(k^{2.5} \log^2 (nk))$ time in the \pillar model if all the weights are
            integers, with the returned fern~$\Phi$ being Monge and satisfying
            $\delta(\Phi) = \cO(k^{1.5})$.
    \end{enumerate}
\end{lemma}
\begin{proof}
    First, observe that $|n - m| \leq \edu{P}{T} \leq k$.

    If $m < 20k$, we have $n < 19k$ and it hence suffices to use \cref{gucbenbjuj}.

    We henceforth assume that $m\ge 20 k$, which implies that $n \geq 19k$.
    We first run the algorithm of \cref{lem:decomp} for $P$, $T$, and~$k$, arriving at a
    decomposition $P=\bigodot_{i=0}^{z - 1}P_i$
    and a sequence of fragments $(T_i)_{i=0}^{z-1}$, such that
    $P_i=P\fragmentco{p_i}{p_{i+1}}$ and $T_i=T\fragmentco{t_i}{t'_{i+1}}$ for
    $t_i=\max\{p_i-k,0\}$ and $t'_{i+1}=\min\{p_{i+1}+3k,n\}$ for all $i\in
    \fragmentco{0}{z}$.
    The decomposition is represented using a set $F$ of size $\Oh(k)$ such that
    $P_i=P_{i-1}$ and $T_i=T_{i-1}$ holds for each $i\in \fragmentco{0}{z}\setminus F$ and
    the endpoints of $P_i$ for $i\in F$.
    To easily handle corner cases, we set $t'_0\coloneqq t'_1$ and $t_z \coloneqq t_{z-1}$
    (so that the sequences $(t_i)_{i=0}^z$ and $(t'_i)_{i=0}^z$ remain monotone), and we
    assume that $\fragmentco{0}{z}\setminus F\subseteq \fragment{2}{z-5}$; if the
    original set $F$ does not satisfy the latter condition, we insert to it the missing $\Oh(1)$
    elements.

    For each $i\in \fragmentco{0}{z}$, consider a subgraph $G_i$ of $\oAGw(P,T)$ induced
    by $\fragment{p_i}{p_{i+1}}\times \fragment{t_i}{t'_{i+1}}$.
    Denote by $G$ the union
    of all subgraphs $G_i$.
    For each $i \in \fragment{0}{z}$, denote $V_i \coloneqq \{p_i\}\times
    \fragment{t_i}{t'_i}$.
    Note that $V_0 \supseteq \{0\}\times \fragment{0}{2k}$.
    Additionally, $V_z \supseteq \{m\}\times
    \fragment{n-k}{n}$ since $p_z - p_{z-1} \geq k$ implies $p_{z-1} \leq m-k$, $t_{z-1}=
    p_{z-1}-k \leq m-2k$, and $n \geq m-k$.
    Further, $V_i = V(G_i) \cap V(G_{i - 1})$ for $i \in
    \fragmentoo{0}{z}$.
    Moreover, for integers $0\le i \le  j \le z$, let $D_{i, j}$ denote the matrix of
    pairwise distances from $V_i$ to $V_j$ in \(G\), where rows of $D_{i, j}$ represent
    vertices of $V_i$ in the increasing order of the second coordinate, and columns of
    $D_{i, j}$ represent vertices of $V_j$ in the increasing order of the second
    coordinate.
    Note that $D_{i,j}$ is a Monge matrix by \cref{fct:monge}.

    \begin{claim}\label{claim:dii}
        \begin{enumerate}
            \item We can compute $D_{i,i+1}$ for all $i \in F$ in $\cO(k^3 \log k)$ time
                in total in the \pillar model.\label{dii:gen}
            \item If all the weights are integers, we can compute $\mds(D_{i,i+1})$, for
                all $i \in F$, in $\Oh(W \cdot k^2 \log^2 n)$ time in the \pillar
                model;\label{dii:intw}
            \item If all the weights are integers, we can compute $\mds(Z_{i,i+1})$, for
                all $i \in F$, where $Z_{i,i+1}$ is a Monge matrix such that $Z_{i,i+1}
                \meq{k} D_{i,i+1}$ and $\delta(Z_{i,i+1}) = \cO(k^{1.5})$, in $\Oh(k^{2.5}
                \log^2 n)$ time in the \pillar model.\label{dii:intsqrtk}
        \end{enumerate}
    \end{claim}
    \begin{claimproof}
        By \cref{lem:oagw} (path monotonicity), the shortest paths (in $G$) between
        vertices of
        $G_i$ stay within $G_i$.
        Hence, $D_{i, i + 1}$ can be computed in $\Oh(k^2 \log k)$
        time using an MSSP data structure (\cref{lem:MSSP}) for each $i\in F$, for a total
        time of $\Oh(k^3 \log k)$.

        In the case of integer weights, we can instead use either of the following
        approaches.
        \begin{itemize}
            \item We can use \cref{lem:addtoboxds} to construct $\mds(D^w_{P_i,T_i})$ for
                all $i\in F$ using $\Oh(W \cdot k^2 \log^2 n)$ total time in the \pillar
                model.
                Then, for each $i \in F$, we can extract $\mds(D_{i,i+1})$ in $\Oh(W \cdot
                k \log k)$, time using \cref{lm:ds_tools}, for a total of $\cO(W \cdot k^2
                \log k)$ time, noting that the core of each considered matrix is of size
                $\Oh(W\cdot k)$ due to \cref{lm:bounded_core}.
            \item We can use \cref{lem:addtoboxds} to construct, for each $i\in F$,
                $\mds(M_i)$ for a Monge matrix $M_i$ such that $M_i \meq{k} D^w_{P_i,T_i}$
                and $\delta(M_i) = \cO(k^{1.5})$ in $\cO(k^{2.5} \log^2 n)$ total time in
                the \pillar model.
                Then, for each $i \in F$, we can extract $\mds(Z_{i,i+1})$ in $\Oh(k^{1.5}
                \log k)$, for a total of $\cO(k^{2.5} \log k)$ time, using
                \cref{lm:ds_tools}.\claimqedhere
        \end{itemize}
    \end{claimproof}

    As shown in the proof of {\cite[Claim 7.4]{gk24}}, our definition of $F$ ensures that
    $D_{i, i + 1} = D_{i-1,i}$ for each $i \in \fragmentco{0}{z}\setminus F$. We rely on
    the following fact to prove \cref{claim:d0z}.

    \begin{claimq}[{\cite[Claim 7.4]{gk24}}] \label{clm:D-formula}
        Let $0=j_0 < j_1 < \ldots < j_{|F| - 1} < j_{|F|} = z$ be the elements of $F\cup
        \{z\}$.
        We have
        \[D_{0, z} = \bigoplus_{i=0}^{|F|-1} D_{j_i, j_i + 1}^{\oplus (j_{i + 1} -
        j_i)}. \claimqedhere\]
    \end{claimq}

    \begin{claim}\label{claim:d0z}
        \begin{enumerate}
            \item We can compute $D_{0,z}$ in $\Oh(k^3 \log n)$ time in the \pillar model.
            \item If all weights are integers, $D_{0,z}$ is $(W+1)$-bounded difference and
                we can compute it in time $\cO(W\cdot k^2 \log^2 n)$ time in the \pillar
                model.
            \item If all weights are integers, in $\cO(k^{2.5} \log^2 n)$ time in the
                \pillar model, we can compute a Monge matrix $Z_{0,z}$ such that $Z_{0,z}
                \meq{k} D_{0,z}$
                and $\delta(Z_{0,z}) = \cO(k^{1.5})$.
        \end{enumerate}
    \end{claim}
    \begin{claimproof}
        In the case of arbitrary weights, we first employ \cref{claim:dii}(\ref{dii:gen})
        to compute, for each $i \in \fragmentco{0}{|F|}$,  the matrix $D_{j_i, j_i + 1}$,
        in $\cO(k^3 \log k)$ time in the \pillar model in total.
        Given $D_{a, b}$ and $D_{b, c}$ for any integers $0\le a \le b \le c \le z$, we can
        calculate $D_{a, c}$ using \cref{thm:smawk} in time $\Oh(k^2)$.
        Therefore, using binary exponentiation, we can calculate $M_i$ where $M_i \coloneqq
        D_{j_i, j_i + 1}^{\oplus (j_{i + 1} - j_i)}$ in time $\Oh(k^2 \log n)$.
        Doing this for all $i \in \fragmentco{0}{|F|}$ requires $\Oh(k^3 \log n)$ time.
        We then calculate $D_{0, z}$ using the fact that $D_{0, z} = M_0 \oplus M_1 \oplus
        \cdots \oplus M_{|F| - 1}$ in time $\Oh(|F| \cdot k^2) = \Oh(k^3)$.

        In the case of integer weights, we present two algorithms.
        \begin{itemize}
            \item \emph{An $\cO(W\cdot k^2 \log^2 n)$-\pillar{}-time algorithm.} We employ
                \cref{claim:dii}(\ref{dii:intw}) to compute $\mds(D_{j_i, j_i + 1})$ for
                all $i \in \fragmentco{0}{|F|}$, in $\cO(W \cdot k^2 \log^2 n)$ total time
                in the \pillar model.
                We then use \cref{lm:matrix_mult_w} instead of \cref{thm:smawk} and binary
                exponentiation to compute $\mds(D_{0,z})$.
                This requires $\cO(W \cdot k^2 \log^2 n)$ time in total as each matrix
                $D_{j_i, j_i + 1}$ is $(W+1)$-bounded-difference due to
                \cref{lm:bounded_core}, each computed matrix is $(W+1)$-bounded-difference
                due to \cref{lem:bdproduct}, and thus the core of each considered matrix
                is $\cO(W \cdot k)$ due to \cref{lem:bdtocore}.
                In the end, we extract the $(W+1)$-bounded-difference matrix $D_{0,z}$
                from $\mds(D_{0,z})$ in $\cO(k^2 \log k)$ time (see
                \cref{lm:build_mtx_ds}).
            \item \emph{An $\cO(k^{2.5} \log^2 n)$-\pillar{}-time algorithm.} We employ
                \cref{claim:dii}(\ref{dii:intsqrtk}) to compute, for each $i \in
                \fragmentco{0}{|F|}$, $\mds(Z_{j_i, j_i + 1})$ for a Monge matrix $Z_{j_i,
                j_i + 1}$ that is $k$-equivalent to $D_{j_i, j_i + 1}$ and its core has
                size $\cO(k^{1.5})$ in total time $\cO(k^{2.5} \log^2 n)$ in the \pillar
                model.
                We then use \cref{lm:matrix_mult_k} instead of \cref{thm:smawk} and binary
                exponentiation to compute a matrix $Z_{0,z}$ that is $k$-equivalent to
                $D_{0,z}$.
                \cref{lm:matrix_mult_k} guarantees that the size of the core of each
                computed matrix is $\cO(k^{1.5})$ and hence the total running time for all
                $\cO(k \log n)$ min-plus multiplications is $\cO(k^{2.5} \log^2 n)$.
                In the end, we again extract the matrix from its $\mds$ representation in
                $\cO(k^2 \log k)$ time, noting that its core is of size $\cO(k^{1.5})$.
                \claimqedhere
        \end{itemize}
    \end{claimproof}

    \begin{claim}[See also \cref{clm:restrict}]\label{clm:ver:restrict}
        For every fragment $T\fragmentco{\ell}{r}$ of $T$ and threshold $d\in [0,k]$, the
        following are equivalent:
        \begin{enumerate}[(a)]
            \item\label{it:ver:restrict:a} $\edw{P}{T\fragmentco{\ell}{r}}\le d$;
            \item\label{it:ver:restrict:b} $\ell\le t'_0$, $r\ge t_z$, and $\dist_{G}((0,\ell),(m,r))\le d$.
        \end{enumerate}
    \end{claim}
    \begin{claimproof}
        The $\ref{it:ver:restrict:a} \Leftarrow \ref{it:ver:restrict:b}$ implication
        follows directly from \cref{lem:oagw} (distance preservation between $\AGW(P,T)$
        and $\oAGw(P,T)$) and the fact that $G$ is a subgraph of $\oAGw(P,T)$.

        For the converse implication, suppose that there is an alignment $\cO:P\onto
        T\fragmentco{\ell}{r}$ of weighted cost at most $d\le k$.
        It suffices to prove that $\cO$ (interpreted as a path from $(0,\ell)$ to $(m,r)$
        in $\oAGw(P,T)$) is contained in $G$.
        For this, let us focus on a single edge $(p,t)\to (p',t')$ of $\cO$ with
        $p'\in\{p,p+1\}$ and $t'\in \{t,t+1\}$.
        Let us fix an index $i\in \fragmentco{0}{z}$ such that $\fragment{p}{p'}\subseteq
        \fragment{p_i}{p_{i+1}}$.
        Note that $(t'-p')-(r-m)\le k$ implies $t' \le p'+k+r-m\le p'+k+n-m\le p'+2k\le
        p_{i+1}+2k$.
        Since $t' \le n$, we conclude that $t'\le t'_{i+1}$.
        Further, note that $(p-t)-(0-\ell)\le k$ implies $t \ge p-k-\ell \ge p-k \ge p_{i}-k$.
        Since $t \ge 0$, we conclude that $t \ge t_i$.
        Thus, since $t,t' \in \fragment{t_i}{t'_{i+1}}$, the edge $(p,t)\to (p',t')$ is
        contained in $G_i$ and hence in $G$.
    \end{claimproof}

    The above claim guarantees that the matrix returned by each of the algorithms
    encapsulated in \cref{claim:d0z} is an $\cO(k) \times \cO(k)$ Monge matrix whose
    top-right contiguous $(k+1) \times (k+1)$ submatrix $\Phi$ is in
    $\fmats{P}{T}{w}{k}{k}$ due to the fact that $V_0 \supseteq \{0\}\times
    \fragment{0}{k}$ and $V_z \supseteq \{m\}\times
    \fragment{n-k}{n}$.
    We extract $\Phi$ in $\cO(k^2 \log k)$ time (see \cref{lm:build_mtx_ds}) and return
    it.
    In the case when $m \geq 20k$, the time complexity of the algorithm is dominated by
    the running time of the procedure encapsulated in \cref{claim:d0z}.
    This concludes the proof of the lemma.
\end{proof}

\subsection{Growing Ferns: The General Case}

Before proving the main lemma of this section, let us state some useful lemmas from \cite{ckw23,gk24}.

\begin{lemmaq}[{\cite[Lemma 4.5]{ckw23}}]
    \dglabel{lem:selfed}
    There is an $\Oh(k^2)$-time \pillar algorithm that, given a string $X\in \Sigma^n$
    and an integer $k\in \mathbb{Z}_{\ge 0}$ determines whether $\selfed(X)\le k$ and, if so,
    retrieves (the breakpoint representation of) an optimal self-alignment $\A : X \onto X$.
\end{lemmaq}

\begin{lemmaq}[{\cite[part of Lemma 4.9]{ckw23}}]
    \dglabel{lem:alg-periodic}
    There is an algorithm that, given two strings $X,Y\in \Sigma^{\le n}$ and integers
    $1\le d \le k\le n$ such that $\selfed(X)\le k$ and $\big||X|-|Y|\big|\le
    2d$, as well as (oracle access to) a normalized weight function
    $w : \sqEsigma \to \intvl{0}{W}$,
    computes $\edwk{d}{X}{Y}$ as well as the breakpoint representation of the
    underlying alignment (if this value is not $\infty$) in $\Oh(k^2 d\log n)$ time in the
    \pillar model.
\end{lemmaq}

\begin{lemmaq}[{\cite[Lemma 4.4]{ckw23}}]
    \dglabel{prop:disjoint-alignments-bound-sed}
    For two strings $X, Y \in \Sigma^*$ and positions
    $i \le j, i' \le j' \in \fragment{0}{|Y|}$,
    write $\A\in \Als(X,Y\fragmentco{i}{j})$ and $\A'\in \Als(X,Y\fragmentco{i'}{j'})$.
    If the alignments $\A$ and $\A'$, interpreted as paths in $\AG(X, Y)$, do not contain any common diagonal edge,
    then \[
        \selfed(X)\le
        |i-i'|+\ed_\A(X,Y\fragmentco{i}{j})
        +\ed_{\A'}(X,Y\fragmentco{i'}{j'})+|j-j'|.\tag*{\qedhere}
    \]
\end{lemmaq}

\begin{lemmaq}[{\cite[Lemma 7.6]{gk24}}]\dglabel{lem:int_wed}
    There is a \pillar algorithm that, given two strings $X, Y \in \Sigma^{\le n}$ as well
    as oracle access to a weight function $w : \sqEsigma \to \fragment{0}{W}$, finds
    $\edw{X}{Y}$ as well as the breakpoint representation of a $w$-optimal alignment $\A :
    X \onto Y$.
    The running time of the algorithm is $\Oh(\min\{k^2 \cdot W \log^2 n, k^{2.5} \log^3
    n\})$ time plus $\Oh(k^2 \log \min \{n, W + 1\})$ \pillar operations, where $k =
    \edw{X}{Y}$.
\end{lemmaq}

\begin{lemma}[{\tt\apmsmalled($P$, $T$, $k$, $w$)}]
    \dglabel"{lem:small_ed}[fct:LValignment,lem:trim,lem:selfed,lem:small_sed,fct:selfed-properties,obs:band,lem:alg-periodic,lem:int_wed,prop:disjoint-alignments-bound-sed,lem:wed,lm:matrix_mult_w,lm:matrix_mult_k]
    Any instance of the \growFern problem can be solved in
    \begin{enumerate}
        \item $\Oh(k^3 \log (nk) \log n)$ time in the \pillar model with the returned fern
            being Monge;
        \item $\Oh(W \cdot k^2 \log^2 (nk))$ time in the \pillar model if all weights are
            integers, with the returned fern being Monge and $(W+1)$-bounded-difference;
        \item $\Oh(k^{2.5} \log^3 (nk))$ time in the \pillar model if all weights are
            integers, with the returned fern~$\Phi$ being Monge and satisfying
            $\delta(\Phi) = \cO(k^{1.5})$.
    \end{enumerate}
\end{lemma}
\begin{proof}
    We first compute $\edu{P}{T}$ in $\cO(k^2)$ time in the \pillar model using
    \cref{fct:LValignment}.

    If $\edu{P}{T} > 3k$, then we can output a trivial fern all of whose entries are set
    to $k+1$. This follows from the fact that if there exist $i \leq k$ and $j \geq n-k$
    such that the $(0,i)$-to-$(m,j)$ shortest path in $\oAGw(P,T)$ has weight at most $k$,
    then
    $\edu{P}{T} \leq |T\fragmentco{0}{i}| +
    \edw{P}{T\fragmentco{i}{j}} + |T\fragmentco{j}{n}| \leq 3k$.

    If $\edu{P}{T} \in \fragmentoc{k}{3k}$, we simply replace $k$ with $\edu{P}{T}$, thus
    reducing to the case when $k\geq \edu{P}{T}$. In the end, we simply trim the returned
    fern in $\cO(k^2)$ time according to \cref{lem:trim}.

    We can thus henceforth assume that $\edu{P}{T} \le k$
    and distinguish between two cases depending on whether $\selfed(P) \leq 57k$; this can
    be determined in $\cO(k^2)$ time in the \pillar model using \cref{lem:selfed}.

    \emph{Case I: $\selfed(P) \leq 57k$.}
    We call \apmsmallsed{\tt($P$, $T$, $57k$, $w$)} from \cref{lem:small_sed}, which,
    within the stated bounds, returns a Monge fern $\Phi \in
    \fmats{P}{T}{w}{\min(n,57k)}{57k}$.
    It suffices to trim this fern in $\cO(k^2)$ time according to \cref{lem:trim} thus
    obtaining a Monge $(\min(n,k),k)$-fern of $P$ and~$T$ with respect to $w$ whose
    difference bounds and core size are inherited from the corresponding constraints on
    $\Phi$.

    \emph{Case II: $\selfed(P) > 57k$.}
    The monotonicity and continuity properties of self-edit distance
    (\cref{fct:selfed-properties})
    imply that, using $\cO(\log m)$ invocations of \cref{lem:selfed} in a binary search
    fashion, we
    can compute the following fragments in $\cO(k^2 \log m)$ time in the \pillar model:
    \begin{itemize}
        \item a prefix $P_p$ of $P$ with self-edit distance $28k$;
        \item a prefix $P_p'$ of $P$ with self-edit distance $14k$;
        \item a suffix $P_s$ of $P$ with self-edit distance $28k$;
        \item a suffix $P_s'$ of $P$ with self-edit distance $14k$.
    \end{itemize}
    Due to the monotonicity, sub-additivity, and triangle inequality properties of
    self-edit distance (\cref{fct:selfed-properties}), the fragments $P_p$ and $P_s$ do
    not overlap, that is, we have $P = P_p P_\mu P_s$ for some string $P_\mu$ of self-edit
    distance at least $57k-28k-28k= k$.

    We then intend to compute $\min_{0\le u \le v \le
    |P_p|+n-m+k}\edwk{k}{P_p}{T\fragmentco{u}{v}}$ along with a witness fragment
    $T\fragmentco{p_1}{p_3}$.
    Note that $|P_p|+n-m+k \leq |P_p|+2k$ and hence, if this minimum value is at most $k$,
    we have $p_1 \in \fragment{0}{3k}$ and $p_3 \in
    \fragmentco{|P_p|+n-m-2k}{|P_p|+n-m+k}$.
    It thus suffices to call \apmsmallsed{\tt($P_p$, $T\fragmentco{0}{|P_p|+n-m+k}$,
    $28k$, $w$)} using \cref{lem:small_sed} and iterate over its entries.
    If the computed value is $\infty$, we output a trivial fern all of whose entries are
    set to $k+1$.
    Otherwise, that is, if the computed value is not $\infty$, we then compute the
    breakpoint representation of a witness alignment~$\A_p$, and we set $T_p \coloneqq
    T\fragmentco{p_1}{p_3} = \A_p(P_p)$.
    To see that this is correct, suppose that there exists an alignment $\A: P \onto
    T\fragmentco{\ell}{r}$ of weighted cost at most $k$.
    Let $T\fragmentco{\ell}{v}\coloneqq \A(P_p)$.
    By \cref{obs:band}, we have $v \le |P_p| + n-m+k$, so we have
    $\edw{P_p}{T\fragmentco{\ell}{v}}\le k$ for $0\le i \le v \le |P_p|+n-m+k$.
    It remains to explain how we compute the breakpoint representation of $\A_p$.
    In the case of general weights, we do so using \cref{lem:alg-periodic} in $\cO(k^3
    \log m)$ time in the \pillar model.
    In the case of integer weights, we employ \cref{lem:int_wed} which requires
    $\Oh(\min\{k^2 \cdot W \log^2 n, k^{2.5} \log^3 n\})$ time plus $\Oh(k^2 \log \min
    \{n, W + 1\})$ \pillar operations.

    Symmetrically, we compute
    $\min_{m-k-|P_s| \le u \le v \le n}\edwk{k}{P_s}{T\fragmentco{u}{v}}$.
    As above, if this value is $\infty$, we output a trivial fern all of whose entries are
    set to $k+1$.
    Otherwise, that is, if this value is not $\infty$, we also obtain the breakpoint
    representation of a witness alignment $\A_s$, and we set
    $T_s \coloneqq T\fragmentco{s_1}{s_3} \coloneqq \A_s(P_s)$.
    The time complexity and the correctness of this procedure follow by arguments
    symmetric to the one made in the previous paragraph.

    We henceforth assume that we have alignments $\A_p$ and $\A_s$ at hand and set $T_p'
    \coloneqq T\fragmentco{p_1}{p_2} \coloneqq \A_p(P_p')$ and $T_s' \coloneqq
    T\fragmentco{s_1}{s_2} \coloneqq \A_s(P_s')$.

    \begin{claim}\label{claim:cross}
        For any two positions $\beta,\gamma\in \fragment{0}{n}$ such that
        $\edw{P}{T\fragmentco{\beta}{\gamma}}\le k$,
        \begin{multline*}
            \edw{P}{T\fragmentco{\beta}{\gamma}} = \edw{P\fragmentco{0}{|P'_p|}}{T\fragmentco{\beta}{p_2}} \\ +
            \edw{P\fragmentco{|P'_p|}{m-|P'_s|}}{T\fragmentco{p_2}{s_2}} \\ +
        \edw{P\fragmentco{m-|P'_s|}{m}}{T\fragmentco{s_2}{\gamma}}.\end{multline*}
    \end{claim}
    \begin{claimproof}
        We say that a vertex $(r_1,c_1)$ dominates a vertex $(r_2,c_2)$ if $r_1 \geq r_2$ and $c_1 \geq c_2$.

        Let us consider an optimal alignment $\A: P \onto T\fragmentco{\beta}{\gamma}$.
        Our goal is to show that, when we interpret alignments $\A$, $\A_p$, and $\A_s$ as
        paths in $\AGW(P,T)$, each of the following holds.
        \begin{enumerate}
            \item $\A$ and $\A_p$ share a vertex that is dominated by
                $(|P_p'|,p_2)$.\label{it:cross1}
            \item $\A$ and $\A_p$ share a vertex that dominates $(|P_p'|,p_2)$ and is
                dominated by $(|P_p|,p_3)$.\label{it:cross2}
            \item $\A$ and $\A_s$ share a vertex that dominates $(m-|P_s|,s_1)$ and is
                dominated by $(m-|P_s'|,s_2)$.\label{it:cross3}
            \item $\A$ and $\A_s$ share a vertex that dominates
                $(m-|P_s'|,s_2)$.\label{it:cross4}
        \end{enumerate}
        The setting is illustrated in \cref{fig:crossing_alignments}.
        We prove the existence of the two vertices specified in
        \cref{it:cross1,it:cross2}---the existence of the remaining ones follows by
        symmetry.

        \begin{figure}[t!]
            \centering
            \scalebox{.95}{\begin{tikzpicture}[xscale=.5, yscale=-.5, black!90]
    \draw[rounded corners=1pt, thin,fill=black!3!white] (1,0) rectangle (8, 8);
    \draw[rounded corners=1pt, thin,fill=black!13!white] (1,0) rectangle (5, 4);
    \draw[rounded corners=1pt, thin,fill=black!3!white] (15.5,13.5) rectangle (22, 20);
    \draw[rounded corners=1pt, thin,fill=black!13!white] (15.5,13.5) rectangle (18.5,16);

    \draw[line width = 1pt, rounded corners = 1.5pt, black!80!white,
    fill=black!80!white!4!white] (0, -3.2) rectangle node{$T$} (23, -4.4);
    \draw[inner sep=-1pt, line width = 1pt, rounded corners = 1.5pt, black!80!white,
    fill=black!80!white!4!white] (1, -1.8) rectangle node{$T_p$} (8, -3);
    \draw[inner sep=-1pt, line width = 1pt, rounded corners = 1.5pt, black!80!white,
    fill=black!80!white!4!white] (1, -0.3) rectangle node{\smash{$T_p'$}} (5, -1.6);
    \draw[inner sep=-1pt, line width = 1pt, rounded corners = 1.5pt, black!80!white,
    fill=black!80!white!4!white] (15.5, -1.8) rectangle node{$T_s$} (22, -3);
    \draw[inner sep=-1pt, line width = 1pt, rounded corners = 1.5pt, black!80!white,
    fill=black!80!white!4!white] (18.5, -0.3) rectangle node{\smash{$T_s'$}} (22, -1.6);
    \draw[very thin] (22,-0) -- (22,20);

    \draw[line width = 1pt, rounded corners = 1.5pt, black!80!white,
    fill=black!80!white!4!white] (-3.2, 0) rectangle node{$P$} (-4.4, 20);
    \draw[line width = 1pt, rounded corners = 1.5pt, black!80!white,
    fill=black!80!white!4!white] (-1.8, 0) rectangle node{$P_p$} (-3, 8);
    \draw[very thin] (-1.6,8) -- (8,8) -- (8,-1.6);
    \draw[line width = 1pt, rounded corners = 1.5pt, black!80!white,
    fill=black!80!white!4!white] (-0.3, 0) rectangle node{$P_p'$} (-1.6, 4);
    \draw[very thin] (-0,4) -- (5,4) -- (5,-0);
    \draw[line width = 1pt, rounded corners = 1.5pt, black!80!white,
    fill=black!80!white!4!white] (-1.8, 13.5) rectangle node{$P_s$} (-3, 20);
    \draw[very thin] (-1.6,13.5) -- (15.5,13.5) -- (15.5,-1.6);
    \draw[line width = 1pt, rounded corners = 1.5pt, black!80!white,
    fill=black!80!white!4!white] (-0.3, 16) rectangle node{$P_s'$} (-1.6, 20);
    \draw[very thin] (-0,16) -- (18.5,16) -- (18.5,-0);

    \draw[rounded corners=1pt, double distance=1.25pt, white, double=black!90] (0,0) rectangle (23, 20);

    \draw [white,double=color1, double distance=1.25pt, very thick] plot [smooth, tension=.75] coordinates { (1,0) (2,2.5) (5,4) (6,7) (8,8)};
    \node[circle, inner sep=.05cm, label=left:$\color{c1}{\A_p}$] at (2.5,3) {};

    \draw [white,double=color5, double distance=1.25pt, very thick] plot [smooth, tension=.75] coordinates { (2,0) (3.75,4.25) (6.5,6) (8,10) (15.5,12) (18,17) (22,19) (23,20)};
    \node[circle, inner sep=.05cm, label=above:$\color{c5}{\A}$] at (11.5,10.5) {};

    \draw [white,double=color1, double distance=1.25pt, very thick] plot [smooth, tension=.75] coordinates { (15.5,13.5) (16,15) (18.5,16) (21,18) (22,20)};
    \node[circle, inner sep=.05cm, label=right:$\color{c1}{\A_s}$] at (19.75,16.5) {};

    \node[fill = white, circle, inner sep=.075cm] at (1,0) {};
    \node[fill = c1!50!color1,circle, inner sep=.05cm, label={above left}:$\color{c1}{\pi_1}$] at (1,0) {};
    \node[fill = white, circle, inner sep=.075cm] at (5,4) {};
    \node[fill = c1!50!color1,circle, inner sep=.05cm, label={above right}:$\color{c1}{\pi_2}$] at (5,4) {};
    \node[fill = white, circle, inner sep=.075cm] at (8,8) {};
    \node[fill = c1!50!color1,circle, inner sep=.05cm, label=right:$\color{c1}{\pi_3}$] at (8,8) {};

    \node[fill = white, circle, inner sep=.075cm] at (15.5,13.5) {};
    \node[fill = c1!50!color1,circle, inner sep=.05cm, label={above left}:$\color{c1}{\sigma_1}$] at (15.5,13.5) {};
    \node[fill = white, circle, inner sep=.075cm] at (18.5,16) {};
    \node[fill = c1!50!color1,circle, inner sep=.05cm, label={above right}:$\color{c1}{\sigma_2}$] at (18.5,16) {};
    \node[fill = white, circle, inner sep=.075cm] at (22,20) {};
    \node[fill = c1!50!color1,circle, inner sep=.05cm, label=below:$\color{c1}{\sigma_3}$] at (22,20) {};

    \node[fill = white, circle, inner sep=.075cm] at (2,0) {};
    \node[fill = c5!50!color5, circle, inner sep=.05cm, label={below right}:$\color{c5}{(0,\beta)}$] at (2,0) {};
    \node[fill = white, circle, inner sep=.075cm] at (23,20) {};
    \node[fill = c5!50!color5, circle, inner sep=.05cm, label=right:$\color{c5}{(m,\gamma)}$] at (23,20) {};

    \node[fill = white, circle, inner sep=.075cm] at (3.1,3.1) {};
    \node[fill = black, circle, inner sep=.05cm, label={above right}:$u_1$] at (3.1,3.1) {};
    \node[fill = white, circle, inner sep=.075cm] at (5.5,5.2) {};
    \node[fill = black, circle, inner sep=.05cm, label={above right}:$u_2$] at (5.5,5.2) {};
    \node[fill = white, circle, inner sep=.075cm] at (17.25,15.5) {};
    \node[fill = black, circle, inner sep=.05cm, label={above right}:$u_3$] at (17.25,15.5) {};
    \node[fill = white, circle, inner sep=.075cm] at (21.55,18.75) {};
    \node[fill = black, circle, inner sep=.05cm, label={below left}:$u_4$] at (21.55,18.75) {};

\end{tikzpicture}}
            \caption{An illustration of the setting in the proof of \cref{claim:cross}
                with $\pi_1=(0,p_1)$, $\pi_2=(|P_p'|,p_2)$, $\pi_3=(|P_p|,p_3)$, $\sigma_1
                = (m-|P_s|,s_1)$, $\sigma_2 = (m-|P_s'|,s_2)$, and $\sigma_3 = (m,s_3)$.
                Vertices $u_1$, $u_2$, $u_3$, and $u_4$ satisfy the requirements of
                \cref{it:cross1}, \cref{it:cross2}, \cref{it:cross3}, and
                \cref{it:cross4}, respectively; observe that there are vertices that
                satisfy the requirements of \cref{it:cross2}.
                An alignment that contains $(|P_p'|,p_2)$ and $(m-|P_s'|,s_2)$ and hence
                witnesses the equality at the statement of the lemma can be obtained as
                follows:
                follow $\A$ from $(0,\beta)$ to $u_1$,
                $\A_p$ from $u_1$ to $u_2$,
                $\A$ from $u_2$ to $u_3$,
                $\A_s$ from $u_3$ to $u_4$, and
                $\A$ from $u_4$ to $(m,\gamma)$.
            }\label{fig:crossing_alignments}
        \end{figure}
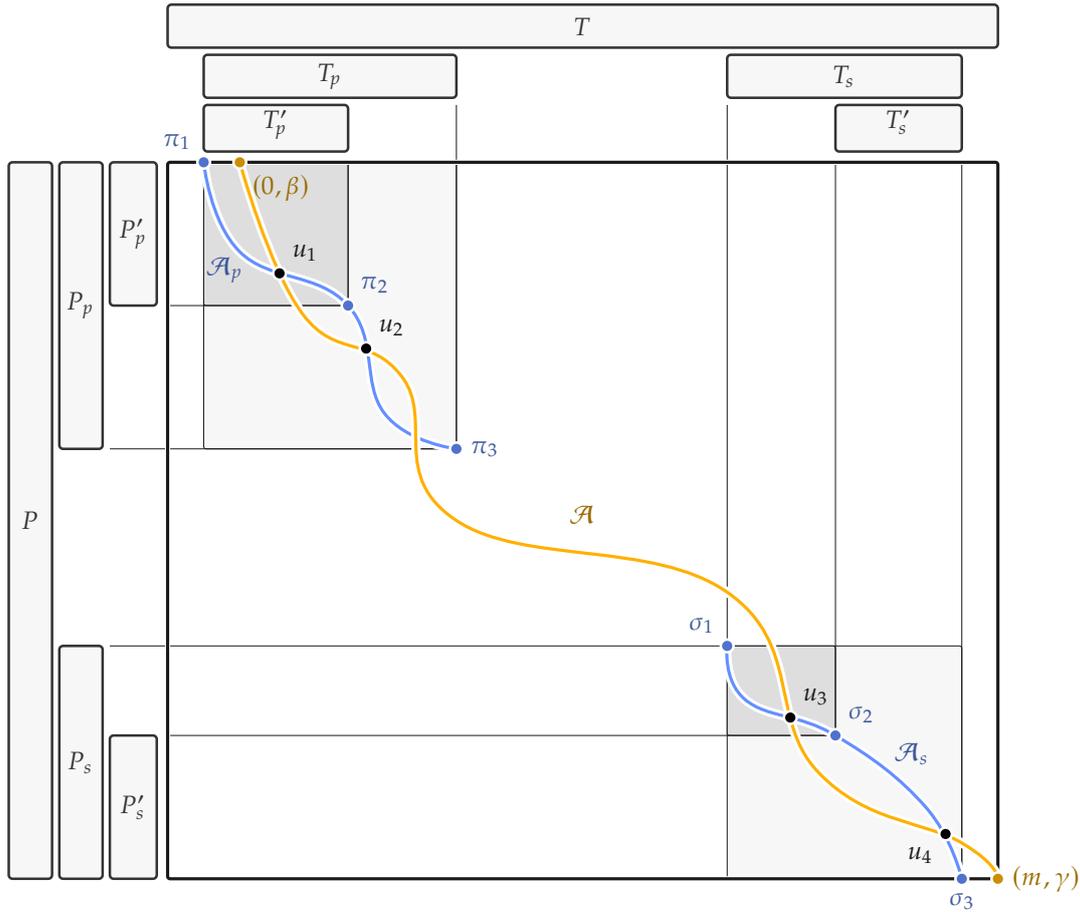

        \emph{Proof of \cref{it:cross1}:} Toward a contradiction, suppose that $\A$ and
        $\A_p$ do not share any vertex dominated by $(|P_p'|,p_2)$.
        Then, in particular, $\A$ and $\A_p$ do not share any diagonal edge with an
        endpoint dominated by $(|P_p'|,p_2)$.
        The condition of \cref{prop:disjoint-alignments-bound-sed} is thus satisfied and
        we have that
        \[
            \selfed(P_p')\le
            |\beta-p_1|+
            \ed_\A(P_p',\A(P_p'))
            +\ed_{\A_p}(P_p',T_p')+
            |\beta + |\A(P_p')| - p_2 | \le 3k + k + k + 5k = 10k,
        \]
        since $p_1,\beta \in \fragment{0}{3k}$ and the lengths of $\A(P_p')$
        and $T_p'$ differ by at most $k$.
        This implies that $\selfed(P_p') \leq 10k<14k$, which contradicts the definition of $P'_p$.

        \emph{Proof of \cref{it:cross2}:} Toward a contradiction, suppose that $\A$ and
        $\A_p$ do not share any vertex that dominates $(|P_p'|,p_2)$ and is dominated by
        $(|P_p|,p_3)$.
        Then, in particular, these two alignments do not share any diagonal edge with such
        an endpoint, and we can again use \cref{prop:disjoint-alignments-bound-sed}.
        Let us denote $P\fragmentco{|P_p'|}{|P_p|}$ by~$P_p''$.
        We have
        \begin{multline*}
            \selfed(P_p'')\le
            |\beta + |\A(P_p')| - p_2|+
            \ed_\A(P_p'',\A(P_p''))
            +\ed_{\A_p}(P_p'',\A_p(P_p''))+
            |\beta + |\A(P_p)| - p_3 | \\
            \le 5k + k + k + 5k = 12k,
        \end{multline*}
        similarly to before.
        By the sub-additivity and triangle inequality properties of self-edit distance
        from \cref{fct:selfed-properties}, we then have that $\selfed(P_p) \leq
        \selfed(P_p') + \selfed(P_p'') \le 14k+12k<28k$,
        which contradicts the definition of $P_p$.

        We are now ready to conclude the proof of the claim. We tweak alignment $\A$ by
        making it behave like alignment $\A_p$ between a pair of vertices $u_1$ and $u_2$
        that satisfy the conditions specified in \cref{it:cross1} and \cref{it:cross2},
        respectively, and like alignment $\A_s$ between a pair of vertices $u_3$ and $u_4$
        that satisfy the conditions specified in \cref{it:cross3} and \cref{it:cross4}.
        Since global alignments are also locally optimal, the cost of the obtained
        alignment $\A'$ remains unchanged.
        As $\A'$ is an optimal alignment and contains $(|P_p'|,p_2)$ and $(m-|P_s'|,s_2)$,
        the statement of the claim follows.
    \end{claimproof}

    We are now ready to compute the output.
    We first call \apmsmallsed{\tt($P_p'$, $T\fragmentco{0}{p_2}$, $14k$, $w$)} using the
    appropriate variant of \cref{lem:small_sed}, noting
    that its conditions are satisfied:
    $\selfed(P'_p)=14k$ and $\edu{P'_p}{T\fragmentco{0}{p_2}}\le p_1+k\le 4k$.
    Let us denote the returned fern by~$\Phi_p$.
    We also call \apmsmallsed{\tt($P_s'$, $T\fragmentco{s_2}{n}$, $14k$, $w$)},
    obtaining a fern $\Phi_s$.
    We also compute $\psi \coloneqq
    \edw{P\fragmentco{|P'_p|}{|P|-|P'_s|}}{T\fragmentco{p_2}{s_2}}$ in
    $\cO(k^3 \log^2 (nk))$ time in the \pillar model using \cref{lem:wed} in the case of
    general weights and
    in $\Oh(\min\{k^2 \cdot W \log^2 n, k^{2.5} \log^3 n\})$ time plus $\Oh(k^2 \log \min
    \{n, W + 1\})$ \pillar operations using \cref{lem:int_wed} in the case of integer
    weights.
    Finally, we compute a Monge matrix that is $k$-equivalent to the min-plus product of
    the last column $\Phi_p\position{\fragment{0}{k}, 14k}$ of $\Phi_P$ with the first row
    $\Phi_s\position{0, \fragment{13k}{14k}}$ of $\Phi_s$ and increment all entries of
    their min-plus product by $\psi$, thus obtaining a matrix $\Psi$.
    Due to \cref{claim:cross}, $\Psi \in \fmats{P}{T}{w}{k}{k}$.
    The time required for computing $\Psi$ (using one of
    \cref{thm:smawk,lm:matrix_mult_w,lm:matrix_mult_k} depending on the case) is dominated
    by the time required by the calls to the appropriate variant of \apmsmallsed.
\end{proof}

\subsection{Efficient Solutions for the \allOccs and \verify Problems}

\begin{corollary}[{\tt\listalloccs($P$, $T$, $k$, $w$)}]
    \dglabel"{cor:fragments}[lem:MSSPverify,fct:LValignment,lem:small_sed]
    \allOccs can be solved in $\cO(\min (\max(k,n-m)\cdot k^2 \log^2 (mk) , n (m+k) \log
    (nm)) )$ time in the \pillar model.

    In the case of integer weights, \allOccs can also be solved in $\cO(\max(k,n-m) \cdot
    k \log^2 (mk) \cdot \min(\sqrt{k} \log (mk),W))$ time in the \pillar model.
\end{corollary}
\begin{proof}
    If $n< m-k$, then $P$ has no $(k,w)$-error occurrences in $T$, so we henceforth assume $n\ge m-k$.

    First, we observe that \cref{lem:MSSPverify} directly yields an algorithm with running
    time $\cO(n(m+k) \log (nm))$.
    It suffices to preprocess the graph $\oAGw(P,T)$ according to \cref{lem:MSSPverify} in
    $\Oh(nm\log(nm))$ time and ask $\Oh(nk)$ queries in $\Oh(\log(nm))$ time each:
    For each $i\in \fragment{0}{n}$, we ask $\Oh(k)$ queries to retrieve
    $\edw{P}{T\fragmentco{i}{j}}$ for all $j\in \fragment{i}{n}\cap
    \fragment{i+m-k}{i+m+k}$.

    Let us now describe efficient algorithms for the case when $n\le m+2k$.
    Observe that if there are positions $i,j \in \fragmentco{0}{n}$ such that
    $\edw{P}{T\fragmentco{i}{j}}\leq k$, we have $\ed(P,T) \leq |T\fragmentco{0}{i}| +
    \edw{P}{T\fragmentco{i}{j}} + |T\fragmentco{j}{n}| \leq 4k$.
    We can check whether
    $\ed(P,T) \leq 4k$ in $\cO(k^2)$ time in the \pillar model using
    \cref{fct:LValignment}, and, if this is not the case return the empty set.
    In the case when $\ed(P,T) \leq 4k$, we call \apmsmalled($P$, $T$, $4k$, $w$) from
    \cref{lem:small_sed} thus obtaining a fern $\Phi \in \fmats{P}{T}{w}{\Delta}{4k}$ for
    $\Delta = \min(n,4k)$.
    In order to compute the desired set of occurrences, it then suffices to iterate over
    the entries of $\Phi$ and, for each entry $\Phi\position{\ell, r} \le k$ such that
    $\ell  \le n-\Delta+r$,
    report the fragment $T\fragmentco{\ell}{n - \Delta + r}$ along with the value $d$.
    This post-processing takes $\Oh(k^2)$ time.

    Next, we present efficient algorithms for the case when $n\ge m+2k$.
    We consider fragments $T_t\coloneqq T\fragmentco{tk}{\min\{m+(t+2)k,n\}}$ for all
    $t\in \fragment{0}{\floor{(n-m+k)/k}}$.
    Each of them is of length at most $m+2k$, so we can apply the procedure from the
    previous case to list all $(k,w)$-error occurrences of $P$ in $T_t$, along with the
    underlying costs.
    For each $(k,w)$-error occurrence $T_t\fragmentco{i'}{j'}$, if $i'\le k$, we report a
    $(k,w)$-error occurrence $T\fragmentco{i' + tk}{j' + tk}$ of the same cost.

    Every $(k,w)$-error occurrence $T_t\fragmentco{i'}{j'}$ corresponds to a $(k,w)$-error
    occurrence $T\fragmentco{i'+tk}{j'+tk}$ of the same cost, and thus the reported
    occurrences are correct.
    Moreover, $i'\in \fragmentco{0}{k}$ implies $i'+tk\in \fragmentco{tk}{(t+1)k}$, so the
    reported occurrences are distinct.
    It remains to prove that all $(k,w)$-error occurrences are indeed reported.
    For this, consider a $(k,w)$-error occurrence $T\fragmentco{i}{j}$.
    The condition $j-i\ge m-k$ implies $i \le n-m+k$. Hence, $t\coloneqq \floor{i/k}\in
    \fragment{0}{\floor{(n-m+k)/k}}$.
    Furthermore, the condition $j-i \le m+k$ implies $j \le i+m+k < (t+1)k + m+k =
    m+(t+2)k$.
    Thus, $T\fragmentco{i}{j}=T_t\fragmentco{i'}{j'}$ for $i'\coloneqq i-tk\in
    \fragmentco{0}{k}$ and $j'\coloneqq j-tk$.
    Consequently, our subroutine for $T_t$ reports $T_t\fragmentco{i'}{j'}$ as a
    $(k,w)$-occurrence of~$P$, and thus we also report $T\fragmentco{i}{j}$.

    The running time is dominated by listing the $(k,w)$-error occurrences of $P$ in the
    fragments~$T_t$.
    Due to the assumption $n\ge m+2k$, the number of such fragments can be written as
    $1+\floor{(n-m+k)/k}\le (n-m+2k)/k\le 2(n-m)/k$.
    For each of them, the procedure takes $\Oh(k^3 \log^2(mk))$ time in the \pillar model
    in the general case and $\Oh(k^2 \log^2 (mk) \cdot \min(\sqrt{k}\log(mk),W))$ time in
    the \pillar model in the case of integer weights. The stated bounds follow.
\end{proof}

\begin{corollary}[{\tt\algverify($P$, $T$, $k$, $I$, $w$)}]
    \dglabel{lem:verify}[cor:fragments]
    The \verify problem can be solved in the \pillar model in time $\cO(k^2(k + |I|)
    \log^2 (mk))$.

    In the case of integer weights, The \verify problem can be solved in the \pillar model
    in time $\cO((k+|I|) \cdot k \log^2 (mk) \cdot \min(\sqrt{k} \log (mk),W))$.
\end{corollary}
\begin{proof}
    We assume $I\subseteq \fragment{0}{n}$ without loss of generality (otherwise, we set
    $I\coloneqq I\cap \fragment{0}{n}$).
    We define $T'\coloneqq T\fragmentco{\min I}{\min\{n,\max I + m+k\}}$ and use
    \cref{cor:fragments} to list all $(k,w)$-error occurrences of $P$ in $T'$
    along with their costs.
    While processing these $(k,w)$-error occurrences $T'\fragmentco{i'}{j'}$, we skip
    those with $i'\ge |I|$, and group the remaining ones by the starting position $i'$.
    For each starting position $i'\in \fragmentco{0}{|I|}$ with at least one $(k,w)$-error
    occurrence, we report a pair $(i,d)$, where $i\coloneqq i'+\min I$ and $d$ is the
    minimum cost among the $(k,w)$-error occurrences $T'\fragmentco{i'}{j'}$.

    Clearly, whenever we report a pair $(i,d)$, the underlying $(k,w)$-error occurrence
    $T'\fragmentco{i'}{j'}$ corresponds to a $(k,w)$-error occurrence $T\fragmentco{i}{j}$
    of cost $d$, where $i\coloneqq i'+\min I$ and $j\coloneqq j'+\min I$.
    Conversely, if $i\in I$ and $T\fragmentco{i}{j}$ is a $(k,w)$-error occurrence, then
    $j-i \le m+k$ implies $j\le i+m+k \le \max I + m+k$, and thus $T'\fragmentco{i'}{j'}$,
    where $j'\coloneqq j-\min I$, is a $(k,w)$-error occurrence.
    It is listed by the subroutine of \cref{cor:fragments}, so our algorithm reports a
    pair $(i,d)$ for $d\le \edw{P}{T\fragmentco{i}{j}}=\edw{P}{T'\fragmentco{i'}{j'}}$.

    As far as the running time is concerned, observe that $n'\coloneqq |T'|\le |I|+m+k$,
    and thus the algorithm of \cref{cor:fragments} takes $\Oh(\max(k,|I|+k)\cdot
    k^2\log(mk))=\Oh((k+|I|)k^2\log^2(mk))$ time.
    In the case of integer weights, the algorithm of \cref{cor:fragments} takes
    $\cO((k+|I|) \cdot k \log^2 (mk) \cdot \min(\sqrt{k} \log (mk),W))$ time.
\end{proof}

\section{An \texorpdfstring{\boldmath $\cOtilde(n+k^4 \cdot n/m)$}{Õ(n+k⁴·n/m)}-time Algorithm}\label{sec:reduction}

In line with several works on pattern matching, we focus on a bounded-ratio version
of \PMWED, that is, we assume that the length of the text is smaller than $\threehalfs m +
k$

\begin{remark}[The standard trick to reduce the text-length to be roughly the pattern-length]
    \dglabel{rem:std}
    Given a text of arbitrary length $n$, we can reduce the problem to $\floor{2n/m}$
    instances of the bounded-ratio version of the problem
    using the so-called \emph{standard trick}~\cite{Abrahamson}.
    That is,
    we consider the overlapping fragments
    \[T_i \coloneqq  T\fragmentco{\floor{i\cdot {m}/2}} {\min\{n, \floor{(i+3)\cdot {m}/2} + k
    - 1\}},\]
    for \(i \in \fragmentco{0}{\floor{2n/m}-1}\),
    and compute the
    $(k,w)$-error occurrences of $P$ in each of $T_0, \dots, T_{\floor{2n/m}-1}$.
    We then merge the obtained partial results.
    Thus, an algorithm that solves the bounded-ratio version of
    \PMWED in time $f(m,k,W)$ in the \pillar model (where $W$ is an upper-bound on the weights),
    directly yields an algorithm that solves an arbitrary instance of
    \PMWED in time $\cO(n/m)\cdot f(m,k,W)$ in the \pillar model.
\end{remark}

In \cref{sec:red_to_SM}, we show that we can focus on a variant of \PMWED where the
pattern is almost periodic,
that is, the pattern is at small edit distance from a string with small period.
In particular, we show how an algorithm from~\cite{ckw20,ckw22} for
\PMED can be used to reduce \PMWED to several instances of (weighted)
\verify (which we can solve using \cref{lem:verify}) and an instance of the following problem.%
\footnote{The unweighted analogue of \SM from \cite{ckw22} is called
\NPM, is formally stated below, and has an extra parameter $d$ which we have set to $16k$ in
this work.}

\SMproblem

We formalize our reduction as follows.

\begin{restatable*}[Reduction of \PMWED to \SM and \verify]{lemma}{redtoSMV}
    \dglabel{lem:red_to_SM}[prp:EIalg,cor:breaks,cor:rep_reg,lem:approx_per]
    An instance of \PMWED with $n < \threehalfs m + k$ can be reduced in
    $\cO(k^{3.5}\sqrt{\log m \log k})$ time in the \pillar model to one of the following:
    \begin{itemize}
        \item an instance \SM{\tt($P$, $T'$, $k$, $Q$, $\A_{P}$, $\A_{T'}$, $w$)},
            where $T'$ is a fragment of $T$;
        \item $\cO(k)$ instances \verify with the same $P$, $T$, $k$, and $w$, each with
            an interval of size $\cO(k)$.\qedhere
    \end{itemize}
\end{restatable*}

Our goal is to then solve \SM
by reducing it to an instance of (a variant of) the so-called
dynamic puzzle matching problem, in a similar fashion
to~\cite{ckw22}.
To this end, let us recall the notion of a puzzle from~\cite{ckw22}.

\begin{definition}[A \(\Delta\)-puzzle of strings and the value of a \(\Delta\)-puzzle, {\cite[Definition 1.4]{ckw22}}]
    \dglabel{def:puzzle}
    For a $\Delta \in \Zz$, we say that $z\ge 2$ strings $S_1,\ldots,S_z$ form a
    \emph{$\Delta$-puzzle} if
    \begin{itemize}
        \item $|S_i|\ge \Delta$ for each $i\in \fragment{1}{z}$, and
        \item $S_i\fragmentco{|S_i|-\Delta}{|S_i|} = S_{i+1}\fragmentco{0}{\Delta}$ for
            each $i\in \fragmentco{1}{z}$.
    \end{itemize}
    The \emph{value} of the puzzle is  $\val_{\Delta}(S_1,\ldots,S_z)\coloneqq S_1
    \odot S_2\fragmentco{\Delta}{|S_2|}\odot S_3\fragmentco{\Delta}{|S_3|}\cdots
    S_z\fragmentco{\Delta}{|S_z|}$.
\end{definition}

Below, we provide a weighted variant of the dynamic puzzle matching problem
from~\cite{ckw22}.
While we use a more intricate variant of this problem,
this simpler version allows us to
introduce useful notions in its (simpler) context and
help the reader develop intuition.

\begin{problem}[Dynamic\-Puzzle\-Matching]{Dynamic\-Puzzle\-Matching{\tt($k$, $\Delta$,
    $w$, $\Sb$, $\Sm$, $\Sf$)}}
    \label{kaqmrblnrj}

    \PInput{Positive integers $k$ and $\Delta$, oracle access to a normalized weight function
        \(w : \sqEsigma \to \intvl{0}{W}\),
        as well as string families $\Sb$, $\Sm$, and $\Sf$ of \emph{leading}, \emph{internal},
    and \emph{trailing} pieces, respectively.}

    \PObject{A sequence $\I=(U_1,V_1)(U_2,V_2)\cdots(U_z,V_z)$ of
        ordered pairs of strings (a \emph{DPM-sequence}), that additionally satisfies
        the following two conditions at initialization time and after each update.
        \begin{enumerate}
            \item $U_1,V_1 \in \Sb$, $U_z, V_z  \in \Sf$,
                and, for all $i\in \fragmentoo{1}{z}$, $U_i,V_i \in \Sm$,\label{it:orig}
            \item the \torn \(\tor(\I) \coloneqq \sum_{i=1}^z \big||U_i|-|V_i|\big|\)
                satisfies \(\tor(\I) \le \Delta/2 - k\).
                \label{it:unique}
    \end{enumerate}}

    \PInit{{\tt DPM-Init($\I'$)}: Initialize $\I$ as $\I'$.}

    \PUpdate{\begin{description}[left=0em..0em,itemindent=-.5em]
            \item {\tt DPM-Delete(\(i\))}: Delete the $i$-th pair of strings.
            \item {\tt DPM-Insert($(U',V')$, \(i\))}:
                Insert the pair of strings $(U',V')$ after the {$i$-th} pair of strings.
            \item {\tt DPM-Substitute($(U',V')$, \(i\)):}
                Substitute the $i$-th pair of strings with the pair of strings $(U',V')$.
    \end{description}}

    \PQuery{{\tt DPM-Query}: Return
        $\OccW_k(\I) \coloneqq
        \OccW_k(\val_{\Delta}(U_1,\ldots, U_z), \val_{\Delta}(V_1,\ldots, V_z))$
    under a promise that $U_1, \ldots, U_z$ and  $V_1, \ldots, V_z$ are $\Delta$-puzzles.}
\end{problem}

\begin{remark}
    A rather straightforward adaptation of a reduction from~\cite{ckw22} would yield an
    instance of \DPM
    with $\cO(k^3)$ update operations.
    While each of these operations can be handled in $\cOtilde(k)$ time in the unweighted
    case,
    in the weighted case we only know how to perform each of them in $\cO(k^2)$ time.
    This way, we could obtain an algorithm for \SM and, in turn, for \PMWED,
    running in time $\cOtilde(k^5)$ in the \pillar model.
\end{remark}

The overhead for processing the introduced DPM update operations in the weighted case
essentially boils down to the fact that the distance matrices we consider are not
unit-Monge (in particular, constant-bounded-difference) as they are in the unweighted
case:
while the $(\min,+)$-multiplication of $k \times k$ \emph{unit-Monge matrices} can be
performed in $\cOtilde(k)$ time~\cite{Tis15},
the $(\min,+)$-multiplication of arbitrary $k \times k$ Monge matrices requires
$\cOtilde(k^2)$ time (see~\cref{thm:smawk}).
We address the challenge posed by the absence of the unit-Monge property---and
thereby devise an algorithm that runs in $\cOtilde(k^4)$ time in the \pillar model---as
follows.
\begin{itemize}
    \item We refine \DPM to \BCDPM,
        allowing for several extra operations.
    \item We reduce the \SM problem to \BCDPM.
    \item While the total number of DPM operations in our instance of \BCDPM
        is still $\cO(k^3)$, we perform them in total time $\cOtilde(k^4)$ by
        employing our efficient implementation of \BCDPM from \cref{sec:dpm} (see
        \cref{xxdsyqwfxu}) which reduces them to $\cO(k^2)$ matrix-matrix multiplications
        and
        $\cO(k^3)$ (cheaper) matrix-vector multiplications.
\end{itemize}
We obtain the following solution for the \SM problem.

\begin{restatable*}{lemma}{solveSM}
    \dglabel"{lem:solve_SM}[lem:smallbeta,fct:rj,fct:puzzleP,fct:puzzleT,fact:druns,lem:swaps,lem:specialbound-simpl,xxdsyqwfxu]
(\SM is in time $\cO(k^4 \log^2(mk))$ in the \pillar model)
    \SM can be solved in $\cO(k^4 \log^2(mk))$ time in the \pillar model, with the output
    set $\OccW_k(P,T)$ represented as $\Oh(k^4)$ disjoint arithmetic progressions.
\end{restatable*}

Further, in \cref{sec:marking}, we obtain the following result for integer metric weight
functions.

\begin{restatable*}{lemma}{solveintSM}
    \dglabel"{lem:solve_int_SM}[lem:solve_SM,fct:rj,sus:heavy,lem:swaps,cor:klocked,lem:heavy-alg,def:fjgj,lem:correctlight,lem:heavy-total,lem:fewreds,lem:swaps,xxdsyqwfxu,lem:heavy-alg,lem:heavy-bound](\SM is in time $\Ohtilde(k^{3.5} W^4)$ in the \pillar model for metric integer weight functions)
    \SM can be solved in $\Ohtilde(k^{3.5} W^4)$ time in the \pillar model if $w$ is an
    integer metric weight function, with the output set $\OccW_k(P,T)$ represented as a
    collection of disjoint arithmetic progressions.
\end{restatable*}

All in all,
a combination of \cref{lem:red_to_SM} with \cref{lem:solve_SM,lem:solve_int_SM} and our
efficient algorithms for the \verify problem (\cref{lem:verify}) yields \cref{thm:main}.

\pillark
\begin{proof}
    Consider an instance of \PMWED with $n \leq \threehalfs m + k$.
    Using \cref{lem:red_to_SM},
    we reduce this instance, in $\cO(k^{3.5}\sqrt{\log m \log k})$ time in the \pillar model to
    \begin{itemize}
        \item an instance \SM{\tt($P$, $T'$, $k$ $Q$, $\A_{P}$, $\A_{T'}$, $w$)},
            where $T'$ is a fragment of~$T$,
        \item $\cO(k)$ instances of \verify with the same $P$, $T$, $k$, and $w$, each with an interval of size $\cO(k)$.
    \end{itemize}

    In the case of arbitrary weights, the obtained instance of \SM can be solved in
    $\cO(k^4 \log^2 (mk))$ time in the \pillar model using \cref{lem:solve_SM}, while the
    obtained instances of \verify can be solved in $\cO(k^4 \log^2 (mk))$ time in total in
    the \pillar model due to \cref{lem:verify}.
    \PMWED can thus be solved in $\cO(k^4 \log^2 (mk))$ time in the \pillar model in the
    case of arbitrary weights.

    If $w$ is an integer metric weight function, the obtained instance of \SM can be
    solved in $\Oh(k^{3.5} W \log^2 (mk))$ time in the \pillar model using
    \cref{lem:solve_int_SM}, while the obtained instances of \verify can be solved in
    $\cO(k^{3.5} W \log^2 (mk))$ time in total in the \pillar model due to
    \cref{lem:verify}.
    \PMWED can thus be solved in $\cO(k^{3.5}W \log^2 (mk))$ time in the \pillar model if
    $w$ is an integer metric weight function.
\end{proof}

\subsection{Reduction to Almost Periodic Patterns}\label{sec:red_to_SM}

As a first step in solving the bounded-ratio version of
\PMWED,
we analyze the pattern according to \cite[Lemma 6.4]{ckw20}.

\begin{lemma}[{\tt Analyze($P$, $k$)},~{\cite[Lemma 6.4]{ckw20}}]
    \dglabel{prp:EIalg}({\tt Analyze($P$, $k$)}, $\Oh(k^2)$ time \pillar model algorithm for structurally
    analyzing $P$ to obtain either breaks, repetitive regions, or to conclude that \(P\) is
    almost periodic~{\cite[Lemma 6.4]{ckw20}})
    Let $P$ denote a string of~length $m$ and let $k \le m$ denote a positive integer.
    Then, there is an algorithm that computes one of~the following.
    \begin{enumerate}
        \item A set of $2k$ disjoint fragments $B_1,\ldots, B_{2k}$ in $P$, called \emph{breaks},
            each having period $\per(B_i)> m/\alphav k$ and length $|B_i| = \lfloor
            m/\betav k\rfloor$.
        \item Disjoint fragments $H_1,\ldots, H_{r}$ in $P$, called \emph{repetitive regions},
            of~total length $\sum_{i=1}^r |H_i| \ge \deltavN/\deltavD \cdot m$ such
            that each region~$H_i$ satisfies
            $|H_i| \ge m/\betav k$ and is constructed along with a primitive approximate
            period $Q_i$
            such that $|Q_i| \le m/\alphav k$ and $\edl{H_i}{Q_i} = \ceil{\betav k/m\cdot
            |H_i|}$.\footnote{Observe that we have renamed the repetitive regions to
            \(H_{\star}\), as we use \(R_{\star}\) later with a different meaning.}
        \item A primitive approximate period $Q$ of~$P$
            with $|Q|\le m/\alphav k$ and $\edl{P}{Q} < \betav k$.\label{item:alm-per}
    \end{enumerate}
    \noindent The algorithm runs in $\Oh(k^2)$ time in the \pillar
    model.\lipicsEnd
\end{lemma}

If the analysis of the pattern using \cref{prp:EIalg}
yields breaks we employ the following lemma.

\begin{fact}[$\Oh(k^3)$-time \pillar model algorithm for computing candidate positions in the case of breaks, {\cite[(proofs of) Lemmas 5.21 and 6.12]{ckw20}}]\dglabel{lm:impEdA}
    Suppose that we are given a threshold $k$,
    a pattern~$P$ of~length $m$,
    disjoint breaks $B_1,\dots,B_{2k}$ in $P$, each satisfying $\per(B_i) \ge m / \alphav
    k$,
    and a text $T$ of~length $n < \threehalfs m + k$.
    Then, we can compute a set of $\cO(k)$ intervals, each of length $k$, whose union is a
    superset of $\OccE_k(P, T)$ in $\Oh(k^3)$
    time in the \pillar model.
\end{fact}

\begin{corollary}[$\cO(k^3)$-time \pillar model reduction from \PMWED to \verify in the case of breaks]
    \dglabel{cor:breaks}[lm:impEdA]
    Suppose that we are given a threshold $k$,
    a pattern $P$ of~length $m$,
    disjoint breaks $B_1,\dots,B_{2k}$ in~$P$, each satisfying $\per(B_i) \ge m / \alphav k$,
    a text $T$ of~length $n < \threehalfs m + k$,
    and oracle access to a normalized weight function $w : \sqEsigma \to \intvl{0}{W}$.

    Then, computing $\OccW_k(P,T)$ can be reduced in $\cO(k^3)$ time in the \pillar model
    to solving $\cO(k)$ instances of \verify with the same $P$, $T$, $k$, and $w$, each
    with an interval of size $\cO(k)$.
\end{corollary}
\begin{proof}
    We first apply \cref{lm:impEdA} to compute $\cO(k)$ $k$-length intervals $I_1, \ldots
    , I_t$ whose union is a
    superset of $\OccE_k(P, T)$ in $\Oh(k^3)$ time in the \pillar model.
    Then, as $\OccW_k(P,T) \subseteq \OccE_k(P,T)$ (see also~\cref{fact:simple}), it suffices
    to solve $\cO(k)$ instances of \verify with the same $P$, $T$, $k$, and $w$, each with
    an interval of size $\cO(k)$, and merge the results.
\end{proof}

Let us now consider the case when  the analysis of the pattern using \cref{prp:EIalg}
returns disjoint repetitive regions.
To complete the reduction in this case, we need the following result for the \NPM problem
from \cite[Section 1.3]{ckw22} that is formally defined below and is an unweighted
analogue of the \SM problem.

\NPMproblem

\begin{fact}[{\cite[Lemma 1.3]{ckw22}}]
    \dglabel^{fastnpm}(\NPM is in time $\cOtilde(d^{3.5})$ in the \pillar model)
    We can solve the \NPM problem in $\cO(d^{3.5} \sqrt{\log n \log d})$ time
    in the \modelname model.
\end{fact}

The following result is shown in \cite[Sections 2.3 and 3.1]{ckw22}; here, we provide a
sketch of its proof.

\begin{lemma}[$\Ohtilde(k^{3.5})$-time \pillar model algorithm for computing candidate positions in the case of repetitive regions, {\cite[Sections 2.3 and 3.1]{ckw22}}]
    \dglabel{lm:impEdB}[fastnpm]
    Suppose that we are given a threshold~$k$,
    a pattern~$P$ of length $m$,
    and a text $T$ of~length $n < \threehalfs m + k$.
    Suppose that we are also given disjoint repetitive regions $H_1,\ldots, H_{r}$ in $P$
    of~total length at least $\sum_{i=1}^r |H_i| \ge {}^{\deltavN}\!/\!{}_{\deltavD}\, m$
    such that each region $H_i$ satisfies $|H_i| \ge m/\betav k$ and has a
    primitive approximate period~$Q_i$
    with $|Q_i| \le m/\alphav k$ and $\ed(H_i,Q_i^*) = \ceil{\betav k/m\cdot |H_i|}$.

    Then, we can compute a set of $\cO(k)$ intervals, each of length $k$, whose union is a
    superset of $\OccE_k(P, T)$ in $\Oh(k^{3.5} \sqrt{\log m \log k})$ time in the \pillar
    model.
\end{lemma}
\begin{proof}[Proof sketch]
    For each repetitive region $H_i$, set $k_i \coloneqq \floor{\betavh \cdot k/m \cdot |H_i|}$
    and $d_i \coloneqq \ceil{\betav\cdot k/m \cdot |H_i|}$.
    We can compute $\OccE_{k_i}(H_i,T)$ by making several calls to the
    procedure encapsulated \cref{fastnpm}; consult the
    discussion following \cite[Fact 2.14]{ckw22} as well as \cite[Section 3.1]{ckw22} for
    a rigorous description of these instances.
    The total time taken by all these calls is $\Oh(k^{3.5} \sqrt{\log m \log k})$ due to \cref{fastnpm}.

    Next, using the sets $\OccE_{k_i}(H_i,T)$ we can identify in $\cO(k^2 \log \log
    k)$ time
    $\cO(k)$ length-$k$ intervals whose union is a superset of $\OccE_{k}(P,T)$ using a marking scheme.
\end{proof}

\begin{corollary}[$\Ohtilde(k^{3.5})$-time \pillar model reduction from \PMWED to \verify
    in the case of repetitive regions]
    \dglabel{cor:rep_reg}[lm:impEdB]
    Suppose that we are given a threshold $k$,
    a pattern $P$ of~length $m$,
    a text $T$ of~length $n < \threehalfs m + k$,
    and oracle access to a normalized weight function $w :\sqEsigma \to \intvl{0}{W}$.
    Suppose that we are also given disjoint repetitive regions $H_1,\ldots, H_{r}$ in $P$
    of~total length at least $\sum_{i=1}^r |H_i| \ge {}^{\deltavN}\!/\!{}_{\deltavD}\, m$
    such that each region $H_i$ satisfies $|H_i| \ge m/\betav k$ and has a
    primitive approximate period~$Q_i$
    with $|Q_i| \le m/\alphav k$ and $\ed(H_i,Q_i^*) = \ceil{\betav k/m\cdot |H_i|}$.

    Then, computing $\OccW_k(P,T)$ can be reduced in $\Oh(k^{3.5} \sqrt{\log m \log k})$
    time in the \pillar model to solving $\cO(k)$ instances of \verify with the same $P$,
    $T$, $k$, and $w$, each with an interval of size $\cO(k)$.
\end{corollary}
\begin{proof}
    We apply \cref{lm:impEdB} to compute $\cO(k)$ $k$-length intervals $I_1, \ldots , I_t$ whose union is a
    superset of $\OccE_k(P, T)$ in $\Oh(k^{3.5} \sqrt{\log m \log k})$ time in the \pillar model.
    Then, as $\OccW_k(P,T) \subseteq \OccE_k(P,T)$ (see also~\cref{fact:simple}),
    it suffices to solve $\cO(k)$ instances of \verify with the same $P$, $T$, $k$, and
    $w$, each with an interval of size $\cO(k)$, and merge the results.
\end{proof}

Finally, we consider the case when the analysis of the pattern using~\cref{prp:EIalg}
returns a primitive approximate period $Q$ of~$P$ with $|Q|\le m/\alphav k$ and
$\edl{P}{Q} < \betav k$.
We show that, in this case, \PMWED
can be reduced, in $\cO(k^2)$ time in the \pillar model, to an instance \SM($P$, $T'$,
$k$, $Q$, $\A_P$, $\A_T$, $w$),
where $T'$ is a fragment of $T$.
We use the following two facts.
The first one is a simplified version of \cite[Lemma 6.8]{ckw20}, while the second one is
a simplified version of \cite[Lemma 6.5]{ckw20}

\begin{fact}[\texttt{FindRelevantFragment($P$, $T$, $k$, $d$, $Q$)},
    Trimming $T$ to an approximately periodic fragment with all approximate occurrences,~{\cite[compare Lemma
    6.8]{ckw20}}]
    \dglabel{lem:erelevant}($\Oh(d^2)$ \pillar algorithm for computing approximately
    periodic fragment of $T$ containing all approximate occurrences,~{\cite[compare Lemma 6.8]{ckw20}})
    Let $P$ denote a pattern of length~$m$, let $T$ denote a text of~length $n$,
    and let $0 \le k\le m$ denote a threshold such that $n<\threehalfs m+k$.
    Further, let $d\ge 2k$ denote a positive integer and let
    $Q$ denote a primitive string that satisfies $|Q|\le m/8d$ and $\edl{P}{Q}\le d$.

    Then, there is an algorithm that computes a fragment $T'$ of $T$ such that
    $\edl{T'}{Q}\le 3d$ and $|\OccE_k(P,T)|=|\OccE_k(P,T')|$.
    The algorithm runs in $\Oh(d^2)$ time in the \pillar model.
\end{fact}

\begin{fact}[\texttt{FindAWitness($k$, $Q$, $S$)}, $\Oh(k^2)$-time \pillar algorithm that can be used to rotate $Q$ such that $\edl{P}{Q} = \edp{P}{Q}$,~{\cite[compare Lemma
    6.5]{ckw20}}]\dglabel{lem:witness}
    Let $k$ denote a positive integer,
    let~$S$ denote a string,
    and let $Q$ denote a primitive string that satisfies $|S|\ge (2k+1)|Q|$.

    Then, we can compute a \emph{witness} $x\in \Zz$
    such that $\edp{S}{\rot^{-x}(Q)}=\edl{S}{Q}\le k$,
    or report that $\edl{S}{Q}>k$.
    The algorithm takes $\Oh(k^2)$ time in the \pillar model.
\end{fact}

\begin{lemma}[$\cO(k^2)$-time \pillar model reduction from \PMWED to \SM in the case of approximate periodicity]
    \dglabel{lem:approx_per}[lem:erelevant,lem:witness,fct:alignment]
    Suppose that we are given a threshold $k$,
    a pattern $P$ of length~$m$,
    a primitive approximate period $Q$ of~$P$
            with $|Q|\le m/\alphav k$ and $\edl{P}{Q} < \betav k$,
    a text $T$ of~length $n < \threehalfs m + k$,
    and oracle access to a normalized weight function $w : \sqEsigma \to \intvl{0}{W}$.

    Then, computing $\OccW_k(P,T)$ can be reduced in $\cO(k^2)$ time in the \pillar model to solving an instance \SM{\tt($P$, $T'$, $k$, $Q$, $\A_{P}$, $\A_{T'}$, $w$)},
            where $T'$ is a fragment of $T$.
\end{lemma}
\begin{proof}
    We first use \cref{lem:erelevant} with $d=16k$ in order to replace~$T$ by its fragment
    $T'$
    which contains all $k$-error occurrences of $P$ (and hence all $(k,w)$-error
    occurrences by \cref{fact:simple}) and is at small edit distance from a fragment of
    $\Q$;
    it is readily verified that all conditions of \cref{lem:erelevant} are satisfied in
    this call.
    We can assume that the obtained text is of length at least $m-k$ as otherwise
    $\OccW_k(P,T) = \emptyset$.

    We then use \cref{lem:witness} to rotate $Q$ so that it satisfies $\edl{P}{Q} =
    \edp{P}{Q}$ as follows (similarly to~\cite{ckw22}).
    We make a call {\tt FindAWitness($16k$, $Q$, $P$)}
    to derive $x\in \Zz$ such that $\edp{P}{\rot^{-x}(Q)}=\edl{P}{Q}$ in $\cO(k^2)$ time in the \pillar model,
    noting that $|P| \ge 128k|Q| \ge (2\cdot 16k + 1)|Q|$.
    Then, we make a call \alignment{$P$, $Q$, $x$} to the function of
    \cref{fct:alignment},
    which yields an optimal (unweighted) alignment $\A : P \onto
    Q^\infty\fragmentco{x}{y}$ of cost $\edl{P}{Q}$
    in $\cO(k^2)$ time in the \pillar model.
    Next, we extract a fragment $Q'$ of $P$
    that $\A$ matches without any edits against $Q^\infty\fragmentco{x'}{x'+|Q|}$
    for some $x'\equiv x \bmod {|Q|}$.
    Finally, we replace $Q$ with~$Q'$, and set $q \coloneqq |Q|$.

    In order to complete the reduction, it suffices to obtain optimal (unweighted)
    alignments from \(P\) and~\(T\) to fragments of \(\Q\),
    since we have already shown that all the assumptions of \SM are satisfied.
    A call \alignment{$P$, $Q$, $x$} returns an optimal alignment $\A_P: P \onto
    Q^\infty\fragmentco{0}{y_P}$ of cost $d_P = \edl{P}{Q}$
    (for some $y_P\in \Zz$) in $\cO(k^2)$ time in the \pillar model.
    Then, we perform a call {\tt FindAWitness($48k$, $Q$, $T$)}
    to compute $x\in \Zz$ such that $\edp{P}{\rot^{-x}(Q)}=\edl{P}{Q}$,
    noting that $|T| \ge m - k \geq 128k|Q| - k \ge (2\cdot 48k +1)|Q|$.
    This call takes $\cO(k^2)$ time in the \pillar model.
    We set $x_T$ to be the unique integer in $\fragmentco{0}{q}$ that is equivalent to $x$
    modulo $q$.
    Then, a call \alignment{$T$, $Q$, $x_T$}
    returns an optimal alignment $\A_T : T \onto Q^\infty{\fragmentco{x_T}{y_T}}$ of cost
    $d_T = \edl{T}{Q}$
    (for some $y_T\in \Zz$) in $\Oh(k^2)$ time in the \pillar model.
\end{proof}

We are now ready to complete the reduction.

\redtoSMV
\begin{proof}
    We analyze the pattern according to \cref{prp:EIalg}.
    Depending on whether this analysis returns breaks, repetitive regions, or an
    approximate period, we employ \cref{cor:breaks}, \cref{cor:rep_reg}, or
    \cref{lem:approx_per}, respectively.
\end{proof}

\subsection{Puzzle Pieces for $P$ and $T$}\label{subsec:puzzles}

For the rest of \cref{sec:reduction}, we fix an instance \SM{\tt($P$, $T$, $k$, $Q$, $\A_P$, $\A_T$, $w$)}
and set {$\kappa \coloneqq k+d_P+d_T$}
and {$\tau \coloneqq q\ceil{{\kappa}/{2q}}$}.

After providing intuition, we encapsulate the main ideas of \cite[Section 3.2]{ckw22} in \cref{fct:rj}.
For the sake of an intuitive explanation, let us assume here that $q \gg k$.
Due to the primitivity of $Q$, the fact that $d_P$ and $d_T$ are $\cO(k)$, and the large
number of $Q$'s exact occurrences in both $P$ and any fragment of~$T$ of length roughly
$m$, each $k$-error occurrence of the pattern much start $\cO(k)$ positions away from a
position of $T$ that $\A_T$ aligns with a position of $\Q$ equivalent to $0 \pmod q$.
This allows us to create $\cO(m/q)$ substrings $R_j$ of $T$, each of length $m+\cO(k)$,
that contain all $k$-error occurrences of $P$ in $T$.
This intuition is formalized below (for arbitrary values of $q$ and $k$).

\begin{lemmaq}[$\cO(d_T+1)$-time algorithm for computing a collection of short substrings $R_j$ of $T$ that capture all $k$-edit occurrences and there are only a few distinct ones, {\cite[Section 3.2]{ckw22}}]\dglabel{fct:rj}
    Consider an instance of the \SM problem.
    Let $\kappa \coloneqq k+d_P+d_T$, $\tau \coloneqq q\ceil{\kappa/2q}$, and
    \[J \coloneqq \fragment{\ceil{{(x_T-\kappa)}/{\tau}}}{\floor{{(y_T-y_P+\kappa)}/{\tau}}}
    =\{j\in \mathbb{Z} \mid x_T \le j\tau+\kappa  \text{ and } j\tau + y_P -\kappa \le y_T\}.\]
    There exists a family $(R_j)_{j\in J}$ of substrings of $T$, where
    $R_j = T\fragmentco{r_j}{r'_j}$ for all $j\in J$, such that:
    \begin{itemize}
        \item $|R_j| \leq m + 3\kappa - k$ for all $j \in J$, and
        \item $\OccE_k(P, T)=\bigcup_{j\in J} \big(\OccE_k(P, R_j) + r_j \big)\subseteq
            \bigcup_{j\in \Z} \fragment{jq-x_t-\kappa-d_T}{jq-x_T+\kappa+d_T}$.
    \end{itemize}

    Given the alignment $\A_T$, in \(\Oh(d_T + 1)\) time, we can construct the sequences
    $(r_j)_{j\in J}$ and $(r'_j)_{j\in J}$, represented as concatenations of $\Oh(d_T+1)$
    arithmetic progressions with difference $\tau$.%
    \footnote{Parameters $\kappa$ and $\tau$ are defined in the first sentence of
    \cite[Section 3.2]{ckw22}, $J$ and the $R_j$s are defined in \cite[Definition
    3.6]{ckw22}, the bound on the lengths of the $R_j$s is from \cite[Remark 3.9]{ckw22},
    while $\OccE_k(P, T)=\bigcup_{j\in J} \big(\OccE_k(P, R_j) + r_j \big)$ is from
    \cite[Corollary 3.11]{ckw22}. The efficient computation of sequences $(r_j)_{j\in J}$
    and
    $(r'_j)_{j\in J}$ is from \cite[Lemma 3.12]{ckw22}.}
\end{lemmaq}

Finally, let us present a first algorithm for \SM, which is particularly useful for small
sets
\(J\), that is, when \(q \gg k\).

\begin{lemma}[Efficient algorithm for the \SM problem for the case $q \gg k$]
    \dglabel{fct:verifyRj}[fct:rj,lem:verify]
    For each $j\in J$, we can compute the set $\OccW_k(P,R_j)$
    in $\cO(k^3 \log^2 (mk))$ time in the \pillar model.
    In particular, we can solve the \SM problem in $\cO(k^3|J| \log^2 (mk))$ time in the
    \pillar model.

    In the case of integer weights, for each $j\in J$, we can compute the set
    $\OccW_k(P,R_j)$
    in $\cO(k^{2} W \log^2 (mk))$ time in the \pillar model and hence solve the \SM
    problem in $\cO(k^{2} W |J| \log^2 (mk))$ time in the \pillar model.

    In either case, we explicitly return the output set $\OccW_k(P,T)$, which is of size
    $\cO(k|J|)$.
\end{lemma}
\begin{proof}
    A combination of \cref{fct:rj} with \cref{fact:simple} yields that
    $\OccW_k(P, T)=\bigcup_{j\in J} \big(\OccW_k(P, R_j) + r_j \big)$.
    Now, recall from \cref{fct:rj} that, for each $j \in J$, we have $|R_j| \leq m+3\kappa-k$.
    Thus, for each $j \in J$, we can compute $\OccW_k(P,R_j)$ using
    \algverify{} from \cref{lem:verify}.

    For general weights, each call to \algverify{} takes
    $\cO(k^3 \log^2 (mk))$ time in the \pillar model
    and hence, by merging the partial results, we obtain a \pillar algorithm with running
    time $\cO(k^3|J| \log^2 (mk))$.

    In the case of integer weights, each call to \algverify{} takes
    $\cO(k^2 W \log^2 (mk))$ time in the \pillar model
    and hence, by merging the partial results, we obtain a \pillar algorithm with running
    time $\cO(k^{2} W |J| \log^2 (mk))$.
\end{proof}

We next recall the definitions (from \cite{ckw22}) of suitable puzzle pieces for the strings \(P\) and~\(T\)
toward a reduction from \SM to \DPM,
as well as several results about their properties and efficient computation.

We define partitions of \(P\) and \(T\) into
\emph{tiles} from which appropriate puzzle pieces can be then constructed.

\begin{definition}[$\tau$-tile partition of a string with respect to a primitive string, {\cite[Definition 4.1]{ckw22}}]
    \label{def:can_part}
    Consider a string $S$, a primitive string $Q$ of length $q$, an integer $\tau\in \Zp$
    divisible by $q$, and an alignment $\A_S : S \onto Q^\infty\fragmentco{x_S}{y_S}$, where
    $x_S\in \fragmentco{0}{q}$ and $y_S\ge x_S$.

    Partition \(\Q\) into blocks of length \(\tau\) and number them starting from \(1\).
    For the \(j\)-th block \(Q_j \coloneqq \Q\fragmentco{\max\{x_S, (j-1)\tau\}}{\min\{y_S,
    j\tau\}}\), we define the \emph{\(j\)-th tile of \(S\) (with respect to \(\A_S\))} as\[
        S\fragmentco{s_{i-1}}{s_i} \coloneqq \A_S^{-1}(Q_j).
    \]
    Further, we define the \emph{$\tau$-tile partition of $S$ with respect to $\A_S$}
    as the partition \[
    S = \bigodot_{i=1}^{\beta_S} S\fragmentco{s_{i-1}}{s_i},\]
    where \(\beta_S \coloneqq \ceil{{y_S}/{\tau}}\) is the index of the last non-empty tile of \(S\).
\end{definition}

\begin{fact}[Upper bound on the number of substrings constructed by \cref{fct:rj} with
    respect to the number of tiles in $P$, {\cite[proof of Lemma
    4.5]{ckw22}}]\label{fact:Jbetap}
    We have $|J|=\Oh(\beta_P)$.
\end{fact}

A combination of \cref{fact:Jbetap} and \cref{fct:verifyRj},
yields that, if $\beta_P  =\ceil{{y_P}/{\tau}}$ is very small
(that is if the tiles are very long),
then we can already efficiently solve the \SM problem.
\begin{lemmaq}[Efficient solution for the \SM problem parameterized by the number of tiles
    in $P$]\dglabel{lem:smallbeta}[fct:verifyRj,fact:Jbetap]
    We can solve \SM in time $\cO(k^3 \beta_P \log^2 (mk))$ in the \pillar model
    returning explicitly the output set $\OccW_k(P,T)$, which is of size $\cO(k \beta_P)$.

    In the case of integer weights, we can solve \SM in time $\cO(k^2 W \beta_P \log^2
    (mk))$ in the \pillar model
    returning explicitly the output set $\OccW_k(P,T)$, which is of size $\cO(k \beta_P)$.
\end{lemmaq}
In particular, \cref{lem:smallbeta} allows us to assume (without loss of generality) that
\(\beta_P \ge 20\) as, otherwise, \cref{lem:smallbeta} implies \cref{lem:solve_SM}.
We make this assumption in the remainder of \cref{subsec:puzzles} and in
\cref{subsec:redDPM}.

Now, let $P = \bigodot_{i=1}^{\beta_P} P\fragmentco{p_{i-1}}{p_i}$
denote the $\tau$-tile partition of $P$ with respect to $\A_P$,
and let $T=\bigodot_{i=1}^{\beta_T} T\fragmentco{t_{i-1}}{t_i}$ denote
the $\tau$-tile partition of $T$ with respect to $\A_T$.

Let us set $\Delta \coloneqq  6\kappa$ and $z \coloneqq \beta_P - 17$.
By (essentially) extending the tiles from the \(\tau\)-tile partition of \(P\) by an additional
\(\Delta\) characters, we obtain a \(\Delta\)-puzzle with value \(P\).

\begin{lemmaq}[$\Delta$-puzzle for $P$ from its tile partition, {\cite[Lemma 4.6]{ckw22}}]\dglabel{fct:puzzleP}
    The following sequence $P_1,\ldots,P_z$ forms a $\Delta$-puzzle with value $P$:
    \begin{itemize}
        \item $P_1 \coloneqq P\fragmentco{p_0}{p_2+\Delta}$;
        \item $P_{i} \coloneqq P\fragmentco{p_{i}}{p_{i+1}+\Delta}$ for $i\in \fragmentoo{1}{z}$;
        \item $P_z \coloneqq P\fragmentco{p_{z}}{|P|}$.\qedhere
    \end{itemize}
\end{lemmaq}

Similarly, we obtain \(\Delta\)-puzzles for \(T\).
In particular, we use the $\tau$-tile partition of $T$ in order to represent the fragments
$R_j$ (from \cref{fct:rj}) as $\Delta$-puzzles.

\begin{lemmaq}[$\Delta$-puzzle for $T$ from its tile partition, {\cite[Lemma 4.7]{ckw22}}]
    \dglabel{fct:puzzleT}
    For each $j\in J$, the following sequence $T_{j,1},\ldots,T_{j,z}$ forms a
    $\Delta$-puzzle with value $R_j$:
    \begin{itemize}
        \item $T_{j,1} \coloneqq T\fragmentco{r_j}{t_{j+2}+\Delta}$;
        \item $T_{j,i} \coloneqq T\fragmentco{t_{j+i}}{t_{j+i+1}+\Delta}$ for $i\in \fragmentoo{1}{z}$;
        \item $T_{j,z} \coloneqq T\fragmentco{t_{j+z}}{r'_j}$.\qedhere
    \end{itemize}
\end{lemmaq}

Taken together, we obtain the families of puzzle pieces that we use in the remainder of this
work.

\begin{definition}[{\cite[Definition 4.8]{ckw22}}]
    \dglabel^{def:puzzleset}(Families of leading, internal, and trailing puzzle pieces, {\cite[Definition 4.8]{ckw22}})
    \begin{itemize}
        \item We write \(\Sb \coloneqq \{P_1\} \cup \{T_{j,1} \mid j\in J\}\) for the
            family of
            \emph{leading} puzzle pieces.
        \item We write \(\Sm \coloneqq \{P_{i} : i\in \fragmentoo{1}{z}\} \cup \{T_{i} : i \in
            \fragmentoo{\min J +1}{\max J +z} \}\) for the family of
            \emph{internal} puzzle pieces.
        \item We write \(\Sf \coloneqq \{P_z\} \cup \{T_{j,z} : j\in J\}\) for the family of
            \emph{trailing} puzzle pieces.\qedhere
    \end{itemize}
\end{definition}

\begin{remark}[{\cite[Remark 4.9]{ckw22}}]
    \dglabel^{rem:notconv}(Notational convention for puzzle pieces,~\cite[Remark 4.9]{ckw22})
    For convenience, for $i\in \fragmentoo{\min J + 1}{\max J +z}$, we write
    $T_i\coloneqq T\fragmentco{t_i}{t_{i+1}+\Delta}$.
    Observe that by construction, we have \(T_i = T_{j,i'}\) for all \(j + i' = i\) with
    \(i' \in \fragmentoo{1}{z}\); that is,
    overlapping parts of different \(R_j\)'s share their internal pieces.
    Observe further that this is an essential property for our approach to work: when
    moving from \(R_j\) to \(R_{j + 1}\), we exploit that we need to only shift the pieces
    \(T_i\), and not recompute them altogether.
\end{remark}

Finally, the following lemma establishes that the pairs of pieces \(P_i\) and \(T_{j,i}\)
from \cref{def:puzzleset}
are roughly of the same length on average.

\begin{lemmaq}[{\cite[Lemma 4.10]{ckw22}}]
    \dglabel{fact:lengths_diff}[def:puzzleset](Upper bound on the cumulative difference of length between
    puzzle pieces comprising $R_j$ and $P$,~{\cite[Lemma 4.10]{ckw22}})
    For each $j\in J$, we have $\sum_{i=1}^z \big| |T_{j,i}|-|P_i|\big| \le 3\kappa-k$.
\end{lemmaq}

\paragraph*{Special Puzzle Pieces}

In order for puzzle pieces to be useful to us, we need to be able to efficiently
compute the families \(\Sb\), \(\Sm\), and \(\Sf\) specified in \cref{def:puzzleset}.
As it turns out, when considering instances of \SM, most puzzle pieces are substrings of
\(\Q\)---this means that it suffices to locate and compute the \emph{special} pieces, that is, the pieces that
are different from some specific substrings of \(\Q\).
We start with a formal definition of special puzzle pieces.

\begin{definition}[Special puzzle pieces, {\cite[Definition 4.11]{ckw22}}]
    \dglabel^{def:specialpiece}
    We say that an internal piece is \emph{special} if and only if it is different from
    \(\Q\fragmentco{0}{\tau+\Delta}\). We write \(\sp{P}\) for the set of \emph{special
    internal pieces} of $P$ and \(\sp{T}\) for the set of \emph{special internal
    pieces} of $T$; that is, we set
    \begin{align*}
        \sp{P} &\coloneqq \{P_i : i \in \fragmentoo{1}{z}
        \text{ and } P_i \neq Q^\infty\fragmentco{0}{\tau+\Delta}\} \quad\text{and}\\
            \sp{T} &\coloneqq \{T_i :i \in \fragmentoo{1 + \min J}{z + \max J}
            \text{ and } T_i \neq Q^\infty\fragmentco{0}{\tau+\Delta}\}.
    \end{align*}
    Further, we say that a leading piece \(T_{j,1}\) is \emph{special} if and only if it is different from
    \(Q^\infty\fragmentco{-\kappa}{2\tau+\Delta}\) and that a trailing piece \(T_{j,z}\)
    is \emph{special} if and only if it is different from \(Q^\infty\fragmentco{z\tau}{y_P+\kappa}\).
    Similar to before, we write \(\spb{T}\) for the set of special leading and trailing
    pieces of \(T\).
\end{definition}

Indeed, there are only very few special pieces.
We compute them in $\cO(k)$ time in the \pillar model using the following lemma.

\begin{lemma}[{\cite[Lemma 1.6]{ckw22}}]
    \dglabel{lem:specialbound-simpl}(The median edit ditance of families of leading,
    internal, and trailing puzzle pieces is small, there are $\cO(k)$ special puzzle
    pieces that can be computed in $\cO(k)$ time,~{\cite[Lemma 1.6]{ckw22}})
    The median edit distance of each of the families \(\Sb\), \(\Sm\), and $\Sf$ is \(\Oh(k)\).
    Further, each of the multisets \(\sp{P}\),
    \(\sp{T}\), and \(\spb{T}\) is of size \(\Oh(k)\) and can be computed
    in \(\Oh(k)\) time in the \pillar model.
    \lipicsEnd
\end{lemma}

\subsection{Reduction to (Bleach-Commit) Dynamic Puzzle Matching}\label{subsec:redDPM}

Let us sketch a warm-up reduction from \SM to \DPM. In an instance of \DPM, we iterate
over sequences $\I_j$ of puzzle pieces,
such that the value of the $\Delta$-puzzle consisting of first components is $P$ and that
of the second components is~$R_j$,
for each $j\in J$. Then, it would suffice to ask a {\tt DPM-Query} for each~$\I_j$ and
merge the results.
A formal definition of the $\I_j$s is provided below.

\begin{definition}[Sequence of pairs of puzzle pieces for $T_j$ and $P$, {\cite[Definition 4.21]{ckw22}}]
    \label{def:seq-i}
    For each $j \in J$, set \(
    \I_j \coloneqq (P_1, T_{j,1}) (P_2, T_{j,2}) \cdots (P_z, T_{j,z}).
    \)
\end{definition}
Note that, due to \cref{fct:puzzleP,fct:puzzleT}, we have
\(\val_{\Delta}(P_1,\ldots,P_z)=P\)
and \(\val_{\Delta}(T_{j,1},\ldots,T_{j,z})=R_j\).

In order to achieve a more efficient solution, we distinguish so-called \emph{canonical
pairs} in a DPM-sequence~$\I$,
with respect to our fixed instance of the \SM problem; recall that
\(\Delta \coloneqq 6(k + d_P + d_T)\) and \(\tau \coloneqq q\ceil{\kappa/2q}\).

\begin{definition}[Canonical pairs of pieces, {\cite[Definition 4.23]{ckw22}}]
    \label{def:plain}
    We call a pair of pieces \emph{canonical} if it equals \[
        \mathcal{Q}\coloneqq (\Q\fragmentco0{\tau+\Delta},\Q\fragmentco0{\tau+\Delta}).
        \tag*{\qedhere}
    \]
\end{definition}
A crucial combinatorial property exploited in~\cite{ckw22}, is the fact that, for unweighted edit distance,
one can ``trim'' long runs of canonical pairs without affecting the output of a {\tt DPM-Query}.
Here, we extend this observation to weighted edit distance.
To this end, let us recall the definition of the \(\Comp\) operator for DPM-sequences.
We call a maximal contiguous subsequence of internal pairs that satisfy some property
an \emph{internal run} of pairs that satisfy said property.

\begin{definition}[Trimmed sequences of pairs of pieces according to plain pairs and a threshold, {\cite[Definition 4.27]{ckw22}}]
    Consider a DPM-sequence \(\I\),
    where the elements of some subset \(\rpl(\I)\) of the internal canonical DPM-pairs are
    labeled as \emph{plain}.
    (We call all other pieces \emph{non-plain}.)

    For a positive integer \(\comp\), we write \(\Comps(\I, \rpl(\I), \comp)\) for
    the DPM-sequence obtained from \(\I\) by removing exactly one plain DPM-pair from any internal run
    of plain DPM-pairs of length at least \(\comp + 1\);
    if \(\I\) does not contain any such run, we set \(\Comps(\I, \rpl(\I),
    \comp) \coloneqq \I\).

    Further, we write \(\Comp(\I, \rpl(\I), \comp) \coloneqq \Comps^{\star}(\I,
    \rpl(\I), \comp)\) for an iterated application of \(\Comps\) until the DPM-sequence
    remains unchanged. (That is, in \(\Comp(\I, \rpl(\I), \comp)\) every
    internal run of plain DPM-pairs is of length at most \(\comp\).)
\end{definition}

For technical reasons, we may not want to include all canonical pairs
    \(\mathcal{Q}\) in the set \(\rpl(\I)\).

\begin{remark}[{\cite[Remark 4.28]{ckw22}}]
    \dglabel^{rem:plaincan}(Red pieces and plain pairs, trimmed sequences of pairs with respect to red pieces)
     Write \(\rred{P} \supseteq \sp{P}\) for
    a set of \emph{red} pieces of \(P\) and write \(\rred{T} \supseteq \sp{T}\) for a set
    of red pieces of \(T\). Now, for a \(j \in J\), we set \[
        \rpl(\I_j) \coloneqq \{ (P_i, T_{j,i}) \mid i \in \fragmentoo{1}{z} \text{ and }
            P_i \notin \rred{P} \text{ and } T_{j,i} \notin \rred{T}.
        \] Abusing notation, we write \(\Comp(\I_j, \rred{P}, \rred{T}, \comp) \coloneqq
    \Comp(\I_j, \rpl(\I_j), \comp)\).
\end{remark}

The following lemma, whose proof is similar to the proof of \cite[Lemma 4.32]{ckw22},
implies that we can cap the exponent of internal runs of canonical pieces to $k+2$ without altering
the set of $(k,w)$-error occurrences.
The proof of \cite[Lemma 4.32]{ckw22} relied on the greedy nature of
optimal unweighted alignments to show an analogue of the inclusion
$\OccW_k(\I) \subseteq \OccW_k(\I^+)$ shown in \cref{lem:redundantk}.
In the presence of weights, we have to take a different approach;
this is the main difference between the proof of the following lemma and
the proof of \cite[Lemma 4.32]{ckw22}.

\begin{lemma}
    \dglabel{lem:redundantk}(If the torsion is small, we can trim/expand sequences of at
    least $k+2$ plain pairs without changing the set of $k$-edit occurrences)
    Fix a string \(\hat{Q}\), integers \(k \geq 0\) and \(\Delta > 0\), and a normalized
    weight function $w: \sqEsigma \to \intvl{0}{W}$.
    Further, consider a DPM-sequence \(\I=(U_1, V_1)(U_2, V_2)\cdots(U_z, V_z)\) whose pieces form \(\Delta\)-puzzles and
    that has a \torn of \(\tor(\I) \le \Delta/2 - k\), and a set \(\rpl(\I)\) of
    DPM-pairs labeled as plain, where \(\rpl(\I) \subseteq \{ (U_i, V_i) \in \I \mid
        i\in\fragmentoo1z \text{ and } U_i =
    V_i = \hat{Q}\} \).

    Suppose that $\I$ contains an internal run
    \(\X \coloneqq (U_{x},V_{x})\cdots(U_{x + k'}, V_{x + k'})\)
    such that all DPM-pairs in $\X$ are in \(\rpl(\I)\) and $k' \geq k+1$.

    Let $\I^-$ be the DPM-sequence obtained by replacing $\X$ with $\X^- \coloneqq
    (U_{x},V_{x})\cdots(U_{x + k'-1}, V_{x + k'-1})$
    and $\I^+$ be the DPM-sequence obtained by replacing $\X$ with $\X^+=\X \odot
    (\hat{Q},\hat{Q})$.
    We have $\OccW_k(\I) \subseteq \OccW_k(\I^-)$ and $\OccW_k(\I) \subseteq
    \OccW_k(\I^+)$.
\end{lemma}
\begin{proof}
    The following claim states that the deletion of an internal piece from a $\Delta$-puzzle
    yields a $\Delta$-puzzle.

    \begin{claim}[{\cite[see Claim 4.33]{ckw22}}]
        \label{clm:cutout}
        Given a \(\Delta\)-puzzle \(\mathcal{X} \coloneqq X_1,\dots,X_i, X_{i + 1},\dots,X_z\)
        with \(X_i = X_{i + 1}\) and value \(X \coloneqq \val_{\Delta}(\mathcal{X})\),
        the sequence \[
            \mathcal{X}' \coloneqq
            X_1,\dots,X_{i-1},X_{i+1},\dots,X_z
            = X_1,\dots,X_{i},X_{i+2},\dots,X_z
            \] forms a \(\Delta\)-puzzle with value \[
            \val_{\Delta}(\mathcal{X}') =
            X\fragmentco{0}{\xi_{i} + \floor{\Delta/2}} X\fragmentco{\xi_{i} + |X_{i}| -
            \ceil{\Delta/2}}{|X|} =
            X\fragmentco{0}{\xi_i + p} X\fragmentco{\xi_{i+1} + p}{|X|},
        \] where
        \(\xi_i \coloneqq \sum_{t=1}^{i-1} |X_t| - \Delta\) and \(p \in \fragment{0}{|X_{i}|
        - \Delta}\).\claimqedhere
    \end{claim}

    Let \(U \coloneqq \val_{\Delta}(U_1,\dots, U_z)\) and
    \(V \coloneqq \val_{\Delta}(V_1,\dots, V_z)\).
    For each \(j \in \fragmentoo{1}{z}\), write $U\fragmentco{u_j}{u'_j} = U_j$ and
    $V\fragmentco{v_j}{v'_j} = V_j$;
    that is $u_j = \sum_{t=1}^{j-1}(|U_t|-\Delta)$ and
    $v_j = \sum_{t=1}^{j-1}(|V_t|-\Delta)$.
    With \cref{clm:cutout} in mind,
    we also write
    \(\hat{U}_j \coloneqq U\fragmentco{\hat{u}_j}{\hat{u}'_j}\), where
    \begin{align*}
        \hat{u}_j &\coloneqq u_j + \floor{\Delta/2}
        = \sum_{t=1}^{j-1}(|U_t|-\Delta) + \floor{\Delta/2} \qquad\text{and}\\
        \hat{u}'_j &\coloneqq \hat{u}_j + |U_j| - \Delta
        = \sum_{t=1}^{j-1}(|U_t|-\Delta) + |U_j|
        - \ceil{\Delta/2}
        = u'_j - \ceil{\Delta/2}.
    \end{align*}
    For convenience, we define \(\hat{u}'_1\) and \(\hat{u}_z\) analogously.
    Observe that we have \(\hat{u}'_j = \hat{u}_{j+1}\),
    that is, we obtain a partition
    \(U = U\fragmentco{0}{\hat{u}'_1} \hat{U}_2 \cdots \hat{U}_{z-1}
    U\fragmentco{\hat{u}_z}{|U|}\).

    Consider an alignment $\A : U \onto V\fragmentco{a}{b}$ that satisfies
    $\edwa{\A}{U}{V\fragmentco{a}{b}}=\edwk{k}{U}{V\fragmentco{a}{b}}$.

    \begin{claim}[{\cite[see Claim 4.34]{ckw22}}]\label{clm:cutout2}
        For any \(j \in \fragmentoo1z\), the fragment $\hat{V}_j
        \coloneqq V\fragmentco{\hat{v}_j}{\hat{v}'_j}
        \coloneqq \A(\hat{U}_j)$ is contained in $V_j$.
    \end{claim}
    \begin{claimproof}
        Observe that, by construction, we have
        \begin{align*}
            v_j &= u_j + \sum_{t=1}^{j-1} (|V_t|-|U_t|)\\
                &\le u_j + \sum_{t=1}^z \big||U_t|-|V_t|\big|\\
                &= u_j + \tor(\I)\\
                &\le u_j + \floor{\Delta/2} - k\\
                &\le a+\hat{u}_j - k \le \hat{v}_j.
        \end{align*}
        Symmetrically, we obtain
        \begin{align*}
            \hat{v}'_j \le b+\hat{u}'_j-|U| +k
            &\le |V|-|U|+u'_j-\ceil{\Delta/2}+k\\
            &\le |V|-|U|+u'_j-\tor(\I) = |V|-|U|+u'_j-\sum_{t=1}^z \big||U_t|-|V_t|\big|)\\
            &\le |V|-|U|+u'_j - \sum_{t=j+1}^z (|V_t|-|U_t|)= v'_j,
        \end{align*}
        thus completing the proof.
    \end{claimproof}

    As the alignment \(\A\) can make at most \(k\) edits, at least two of the first components of the
    (at least $k+2$) DPM-pairs of~\(\X\) get aligned without any edits.
    We may thus assume without loss of generality that \(\hat{V}_{i} = \hat{U}_{i}\) for some $i \in \fragmentco{x}{x+k'}$.
    \cref{clm:cutout2} ensures that \(\hat{V}_i\) is contained
    in \(V_i\).

    We first prove that $\OccW_k(\I) \subseteq \OccW_k(\I^-)$.
    Let us remove the pair $(U_i,V_i)$ from $\I$, thus obtaining $\I^-$.
    Let \(U^- \coloneqq \val_{\Delta}(U_1, \dots, U_{i-1}, U_{i+1}, \dots, U_z)\) and
    \(V^- \coloneqq \val_{\Delta}(V_1,\dots, V_{i-1}, V_{i+1}, \dots, V_z)\).
    Since $\val_\Delta(V_x, \dots, V_{x+k'})$ has a period $|U_i|-\Delta$,
    and $\hat{V}_i = \hat{U}_i$ is of length $|U_i|-\Delta$,
    we conclude that
    \begin{align*}
        \edw{U}{V\fragmentco{a}{b}}
        &=\edw{(U\fragmentco{0}{\hat{u}_i}}{V\fragmentco{a}{\hat{v}_i}}\\
        &\qquad+ \edw{U\fragmentco{\hat{u}_i+|U_i|-\Delta}{|U|}}{V\fragmentco{\hat{v}_i+|U_i|-\Delta}{b}}\\
        &= \edw{U^-\fragmentco{0}{\hat{u}_i}}{V^-\fragmentco{a}{\hat{v}_i}}\\
        &\qquad+ \edw{U^-\fragmentco{\hat{u}_i}{|U^-|}}{V^-\fragmentco{\hat{v}_i}{b-|U_i|+\Delta}}\\
        &\ge \edw{U^-}{V^-\fragmentco{a}{b-|U_i|+\Delta}}.
    \end{align*}
    We can thus obtain an alignment \(\A^- : U^- \onto
    V^-\fragmentco{a}{b-|U_i|+\Delta}\) of weighted cost at most \(k\), completing the
    proof of $\OccW_k(\I) \subseteq \OccW_k(\I^-)$.

    We now prove $\OccW_k(\I) \subseteq \OccW_k(\I^+)$.
    Let us replace $\X$ by $\X^+$, thus obtaining $\I^+$, by inserting $(\hat{Q},\hat{Q})$ just after
    \((U_i, V_i)\).
    Set
    \[ U^+ \coloneqq \val_{\Delta}(U_1, \dots, U_{i}, \hat{Q}, U_{i+1}, \dots, U_z)
        \quad\text{and}\quad
        V^+ \coloneqq \val_{\Delta}(V_1, \dots, V_{i}, \hat{Q}, V_{i+1}, \dots, V_z).
    \]
    Similarly to before, due to the periodicity of $\val_\Delta(V_x, \dots, V_{x+k'})$, we have
    \begin{align*}
        \edw{U}{V\fragmentco{a}{b}}
        &=\edw{U\fragmentco{0}{\hat{u}'_i}}{V\fragmentco{a}{\hat{v}'_i}}
        + \edw{U\fragmentco{\hat{u}'_i}{|U|}}{V\fragmentco{\hat{v}'_i}{b}}\\
        &=\edw{U\fragmentco{0}{\hat{u}'_i}}{V\fragmentco{a}{\hat{v}'_i}} \\
        &\quad+ \edw{\hat{Q}\fragmentco{0}{|\hat{Q}|-\Delta}}{\hat{Q}\fragmentco{0}{|\hat{Q}|-\Delta}}\\
        &\quad+ \edw{U\fragmentco{\hat{u}'_i}{|U|}}{V\fragmentco{\hat{v}'_i}{b}}\\
        &=\edw{U^+\fragmentco{0}{\hat{u}'_i}}{V^+\fragmentco{a}{\hat{v}'_i}}\\
        &\quad+
        \edw{U^+\fragmentco{\hat{u}'_i}{\hat{u}'_i+|\hat{Q}|-\Delta}}{V^+\fragmentco{\hat{v}'_i}{\hat{v}'_i+|\hat{Q}|-\Delta}}
        \\
        &\quad+
        \edw{U^+\fragmentco{\hat{u}'_i+|\hat{Q}|-\Delta}{|U^+|}}{V^+\fragmentco{\hat{v}'_i+|\hat{Q}|-\Delta}{b+|\hat{Q}|-\Delta}}\\
        &\ge \edw{U^+}{V^+\fragmentco{a}{b+|\hat{Q}| - \Delta}}.
    \end{align*}
    We can thus obtain an alignment \(\A^+ : U^+ \onto
    V^+\fragmentco{a}{b+|\hat{Q}|-\Delta}\) of weighted cost at most \(k\), completing the
    proof of $\OccW_k(\I) \subseteq \OccW_k(\I^+)$.
    \lipicsEnd
\end{proof}

Noting that $\Delta = 6\kappa$
and hence $\tor(\I_j)\leq 3\kappa-k = \Delta/2 - k$, a combination of
\cref{fact:lengths_diff,lem:redundantk}
yields the following result, which is an analogue of
\cite[Lemma 4.30]{ckw22} for weighted edit distance.

\begin{corollary}
    \dglabel{fact:druns}[fact:lengths_diff,lem:redundantk]
    For any \(j \in J\), we have
    \[
        \OccW_k(P, R_j) = \OccW_k(\I_j)
        = \OccW_k( \Comp(\I_j, \sp{P}, \sp{T}, k + 2) ).
        \tag*{\lipicsEnd}
    \]
\end{corollary}

\begin{remark}
    The $k+2$ in \cref{fact:druns} can be optimized to $k+1$,
    which is a more intuitive parameter;
    we opted to not implement this optimization as the proof of an
    analogously refined variant of
    \cref{lem:redundantk} would be lengthier that that of \cref{lem:redundantk}.
\end{remark}

We next argue that, even with \cref{fact:druns} at hand, we cannot
iterate over all distinct DPM-sequences
$\Comp(\I_j, \sp{P}, \sp{T}, k + 2)$, for $j \in J$, with $o(k^3)$ update operations in the worst case.
Let us fix a pair $(P_x,T_y) \in \sp{P} \times \sp{T}$ and consider the maintenance of a DPM-sequence $\I$
that iterates over all specified sequences.
Suppose that at some point, the length-$(k+4)$ sequence
\[(P_x,Q^\infty\fragmentco{0}{\tau+\Delta}) \mathcal{Q}^{k+2} (Q^\infty\fragmentco{0}{\tau+\Delta},T_y)\]
is a contiguous subsequence of $\I$.
As we, intuitively, shift $P$ over $T$, the exponent of the internal run of canonical
pairs enclosed by the DPM-pairs of strings containing $P_x$ and $T_y$
might change $\Omega(k)$ times.
It is actually not hard to construct an instance of \SM where $\sp{P} \times \sp{T}$ is of
size $\Omega(k^2)$
and each pair in $\sp{P} \times \sp{T}$ is ``responsible'' for $\Omega(k)$ updates as in
the considered example for $P_x$ and $T_y$.
In order to circumvent this issue, we refine \DPM by introducing some further update
operations that admit faster implementations
than those in the original problem (while offloading some computation to the query
procedure).
A refined reduction allows us to only use $\cO(k^2)$ \emph{expensive} update operations
and $\cO(k^3)$ \emph{cheap} update operations.

\begin{definition}[Blank pieces for representing runs of canonical pairs]
    For a positive integer $j$, a \emph{blank piece of order} $j$ is a compact
    representation of a sequence of $j$ copies of $\mathcal{Q}$ as $(\mathcal{Q},j)$ .
\end{definition}

\begin{definition}
    \dglabel^{def:unpack}(Operation \unpack for replacing blank pieces with canonical pairs)
    Operation \unpack takes as input a sequence $\J$ of ordered pairs of strings and blank
    pieces and replaces each blank piece of the form $(\mathcal{Q},j)$ with $j$ copies of
    $\mathcal{Q}$.
    The output is denoted by $\unpack(\J)$.

    We say that a sequence $\J$ is a \blanked-representation of a sequence $\I$ of ordered
    pairs of strings if $\I=\unpack(\J)$.
\end{definition}

\begin{problem}{Bleach\-CommitDPM}
    \label{pupolgyhot}

    \PObject{A sequence $\Y$ of puzzle pieces and blank pieces (a
    \emph{BCDPM-sequence}).}

    \PInit{{\tt BCDPM-Init($\Y'$, $\alpha$, $k$, $\Delta$, $\Sb$, $\Sm$, $\Sf$, $\mathcal{Q}$, $w$)}:
        Given
        \begin{itemize}
            \item a BCDPM-sequence $\Y'$,
            \item positive integers $k$, $\Delta = \cO(k)$,
                and $\alpha = \cO(k)$,
            \item string families $\Sb$, $\Sm$, and $\Sf$ of \emph{leading}, \emph{internal}, and
                \emph{trailing} pieces, respectively,
                that satisfy \(\ed(\Sb) + \ed(\Sm) + \ed(\Sf) = \Oh(k)\),
            \item a canonical pair of strings $\mathcal{Q}$,
            \item and oracle access to a normalized weight function
            \(w : \sqEsigma \to \intvl{0}{W}\),
        \end{itemize}
        initialize $\Y$ as $\Y'$.}

    \PInvariant{At initialization time and after each update,
        all blank pieces in $\Y$ are of order at most $\alpha$
        and the DPM-sequence \[
            \I \coloneqq \unpack(\Y) =(U_1,V_1)(U_2,V_2)\cdots(U_z,V_z)
        \] satisfies
        $U_1,V_1 \in \Sb$, $U_z, V_z  \in \Sf$,
        and, for all $i\in \fragmentoo{1}{z}$, $U_i,V_i \in \Sm$.
    }

    \PUpdate{\begin{description}[left=0em..0em,itemindent=-.5em]
            \item {\tt BCDPM-Delete(\(i\))}:
                Delete the $i$-th element of $\Y$.
            \item {\tt BCDPM-Insert($(U',V')$, \(i\))}:
                Insert the pair $(U',V')$ of strings after the $i$-th element of $\Y$.
            \item {\tt BCDPM-Substitute($(U',V')$, \(i\)):}
                Substitute the $i$-th element of $\Y$, which is a pair of strings, with the pair of strings $(U',V')$.
            \item {\tt BCDPM-Commit(\(i\), $\zeta$):}
                Replace the $i$-th element of $\Y$, which is a blank piece, with $\zeta$ copies of $\mathcal{Q}$.
            \item {\tt BCDPM-Bleach(\(i\), \(\zeta\)):}
                Replace the $\zeta$ pieces, that follow the $i$-th element of~$\Y$ and are
                all equal to $\mathcal{Q}$, with a single blank piece
                $(\mathcal{Q},\zeta)$.
            \item {\tt BCDPM-Set(\(i\), $\zeta$):}
                Substitute the $i$-th element of $\Y$, which is a blank piece, with blank piece $(\mathcal{Q},\zeta)$.
    \end{description}}

    \PQuery{{\tt BCDPM-Query}: Return
        \[\OccW_k(\Y) \coloneqq
        \OccW_k(\val_{\Delta}(U_1,\ldots, U_z), \val_{\Delta}(V_1,\ldots, V_z))\]
        under a promise that $U_1, \ldots, U_z$ and  $V_1, \ldots, V_z$ are
        $\Delta$-puzzles and that \(\tor(\I) \le \Delta/2 - k\).}
\end{problem}

The following theorem is analogous to \cite[Lemma 4.36]{ckw22}.
The authors of~\cite{ckw22} showed that one can iterate over \(\I'_j \coloneqq
\Comp(\I_j, \rred{P}, \rred{T}, \comp)\),
for all $j \in J\setminus \{\min J\}$ satisfying
$\I'_{j-1} \neq \I'_j$, in an instance of \DPM using a small number of updates.
They showed that if we skip over some \(\I'_j\) that is equal to \(\I'_{j + 1}\) can be handled by
extending, by one element, some maintained arithmetic progressions (with difference $\tau$)
of $(k,w)$-error occurrences.
The crucial novelties in the proof of the following theorem are the definition of
BCDPM-sequences $\Y'_j$ that satisfy $\unpack(\Y'_j) = \I'_j$, the replacement of all DPM
update operations apart from $\cO(k^2)$ of them with (cheap) updates of the type {\tt
BCDPM-Set} and (the proof of) \cref{claim:few_blank}.
To achieve the latter, we crucially observe that most operations either extend or shrink
an internal run of plain pairs: we maintain such runs using blank pieces, and spend extra
time during the query to treat them.

\begin{theorem}[\protect\swaps{$T$, $P$, $k$, $w$, $Q$, $\A_P$, $\A_T$, $\rred{P}$,
    $\rred{T}$ $\comp$}, Reduction from the \SM problem to the \BCDPM problem]
    \dglabel{lem:swaps}[def:puzzleset,lem:specialbound-simpl,fct:rj,fact:lengths_diff]
    Suppose that we are given
    an instance \SM{\tt($P$, $T$, $k$, $Q$, $\A_P$, $\A_T$, $w$)} with
    \[\beta_P \coloneqq \ceil{{y_P}/q\ceil{{(k+d_P+d_T)}/{2q}}} \geq 20,\]
    and a positive integer \(\comp = \cO(k)\)
    as well as sets of puzzle
    pieces \(\rred{P} \supseteq \sp{P}\) and \(\rred{T} \supseteq \sp{T}\) of size at most $\redsize$ each.

    In time \(\Oh(\redsize^2 \comp (k+\log n))\) in the \pillar
    model, we can reduce the computation of a representation of \[
        \mathcal{G} \coloneqq \bigcup_{j \in J}
        (r_j+\OccW_k(\Comp(\I_j, \rred{P}, \rred{T}, \comp)))
    \]
    that consists in \(\Oh(k \redsize^2 \comp)\) arithmetic
    progressions with difference \(\tau = \Oh(\max\{q,k\})\)
    to an instance of the \BCDPM problem
    with the maintained BCDPM-sequence initialized with a call
    {\tt BCDPM-Init($\Y'$, $\comp$, $k$, $\Delta$, $\Sb$, $\Sm$, $\Sf$, $\mathcal{Q}$,
    $w$)}, where $\Y'$ is an BCDPM-sequence of length \(\Oh( \redsize \comp )\),
    under
    \begin{itemize}
        \item \(\Oh(\redsize^2)\) calls to each of {\tt BCDPM-Delete}, {\tt
            BCDPM-Insert}, {\tt BCDPM-Substitute}, {\tt BCDPM-Commit}, and {\tt BCDPM-Bleach}, and
        \item \(\Oh(\redsize^2\comp) \) calls to each of {\tt BCDPM-Set} and {\tt BCDPM-Query}, such that
            the total number of blank pieces over all calls to {\tt BCDPM-Query} is $\cO(\redsize^2 \comp)$.
    \end{itemize}
\end{theorem}
\begin{proof}
    We set $\mathcal{S}_\beta$, $\mathcal{S}_\mu$, $\mathcal{S}_{\varphi}$
    according to \cref{def:puzzleset},
    choose $\Delta \coloneqq 6\kappa = 6(k + d_P +
    d_T)=\cO(k)$,
    and set $\mathcal{Q}\coloneqq (\Q\fragmentco0{\tau+\Delta},\Q\fragmentco0{\tau+\Delta})$.
    We compute sets $\sp{P}$ and $\sp{T}$ in $\cO(k)$ time in the \pillar model using \cref{lem:specialbound-simpl}.
    Note that \cref{lem:specialbound-simpl} ensures that \(\ed(\Sb) + \ed(\Sm) + \ed(\Sf) = \Oh(k)\).

    We initialize the \BCDPM data structure with a call to the operation
    {\tt BCDPM-Init($\Y'_{\min J}$, $\comp$, $k$, $\Delta$, $\Sb$, $\Sm$, $\Sf$, $\mathcal{Q}$, $w$)}.
    We denote the maintained BCDPM-sequence by \(\Y\).
    Then, we wish to transform \(\Y\) to be equal to \(\Y'_j\), for each $j \in J$, so that we can
    call a {\tt BCDPM-Query} to compute $\OccW_k(\Y'_j)$ and construct arithmetic
    progressions of $(k,w)$-occurrences, akin to~\cite{ckw22}.
    In the next claim, we generate sequences $\update_j$ of \BCDPM updates operations
    (BCDPM-updates)
    that allow us to do exactly that and satisfy the conditions of the statement of the
    theorem.

    \begin{claim}\label{claim:kaupdates}
        In $\cO(\redsize^2\comp \log n)$ time in total, we can compute,
        for all $j \in J\setminus \{\min J\}$ with
        $\Y'_{j-1} \neq \Y'_j$,
        a sequence $\update_j$ of \BCDPM update operations that transforms
        $\Y'_{j-1}$
        to
        $\Y'_j$,
        such that the required invariants are maintained at all times.
        Further, these sequences of updates are returned in increasing order with respect to $j$
        and contain in total
        \begin{itemize}
            \item \(\Oh(\redsize^2)\) calls to each of {\tt BCDPM-Delete}, {\tt
                BCDPM-Insert}, {\tt BCDPM-Substitute}, {\tt BCDPM-Commit}, and
                {\tt BCDPM-Bleach}, and
            \item \(\Oh(\redsize^2\comp)\) calls to {\tt BCDPM-Set}.
        \end{itemize}
    \end{claim}
    \begin{claimproof}
        In order to develop some intuition, let us examine how the length of the run
        of plain pairs that succeeds the
        head of $\I$ changes in the course of the algorithm.
        This length may decrease as a DPM pair~$A$ containing some $T_t \in \rred{T}$
        approaches the head of $\I$ while we slide $P$ on $T$.
        When such a DPM-pair disappears in the process of transforming
        $\I_{t-1}$ to $\I_{t}$,
        the length of the considered run might increase depending on where the leftmost
        pair with a red piece in $\I_t$ is;
        roughly speaking, the run of plain pairs after $A$ in $\I_{t-1}$ becomes the run of
        plain pairs that succeeds the head of $\I_t$.

        We say that a run of plain pairs is \emph{enclosed} by the two DPM-pairs
        that are adjacent to it.
        Suppose that in each of $\I_{j-1}$ and $\I_j$ there is a run of plain pairs that is
        enclosed by a pair that contains a fragment $P_i \in \rred{P}\cup \{P_1, P_z\}$
        and a pair that contains a fragment $T_t \in \rred{T}$.
        The lengths of these two runs differ by at most one; details are provided below.
        Observe that any run of plain pairs in $\I_{j-1}$ that is not enclosed by a pair
        that contains a fragment $P_i \in \rred{P}\cup \{P_1, P_z\}$ and a pair that contains
        a fragment $T_t \in \rred{T}$
        either remains intact in $\I_j$ or is shifted to the left by one position.
        In particular, each change to the length of a run of plain pairs as we transform $\I'_{j-1}$ to~$\I'_j$ can
        be attributed to a single pair of fragments $P_i \in \rred{P}\cup \{P_1, P_z\}$
        and $T_t \in \rred{T}$
        getting closer to (or farther from) each other.

        Let us call an internal run of plain pairs of a DPM-sequence \emph{active} if it
        is of size smaller than~$\comp$ and either of the following holds:
        \begin{itemize}
            \item it is enclosed by the leading pair and a pair containing a piece of
                $\rred{T}$,
            \item it is enclosed by the trailing pair and a pair containing a piece of
                $\rred{T}$,
            \item it is enclosed by distinct DPM-pairs such that one contains a piece of
                $\rred{P}$ and no piece from $\rred{T}$
                and the other contains a piece of $\rred{T}$ and no piece from $\rred{P}$.
        \end{itemize}
        Internal runs of plain pairs that are not active are called \emph{inactive}.
        For a DPM-sequence~$\S$ with some canonical pairs marked as plain, let us denote
        by $\blact(\S)$ the \blanked-representation of $\S$,
        obtained by replacing each active internal run of $r$ plain pairs
        by a blank piece $(\mathcal{Q},r)$.
        For each \(j\in J\),
        let $\Y'_j \coloneqq \blact(\I'_j)$.
        Note that $\Y'_{j-1} \neq \Y'_j$ if and only if $\I'_{j-1} \neq \I'_j$ for each $j \in J\setminus \{\min J\}$.

        For convenience, apart from \(\Y\), we also explicitly maintain the sequence \(\I
        = \unpack(\Y)\), initialized as \(\I'_{\min J}\), that we transform to be equal to
        each \(\I'_{j}\).
        In particular, we first compute updates for $\I$ and then transform them to
        updates for $\Y$, ensuring that, at all times, $\Y \coloneqq \blact(\I)$.
        We store $\Y$ and $\I$ as doubly-linked lists $Y$ and $I$, respectively, and
        maintain bidirectional pointers between pairs in $Y$ and $I$ that contain a common
        red piece, labelled with said red piece(s).
        (We do not explicitly mention when such pointers need to be updated.)

        We compute sequences of updates for $I$. Each of them is of one of the following
        forms:
        \begin{itemize}
            \item $\textsf{list-sub}(\textsf{pointer}, (P_i,T_{j,i}))$:
                given a handle $\textsf{pointer}$ to an element of $I$,
                substitute $(P_i,T_{j,i})$ for this element;
            \item $\textsf{list-ins}(\textsf{pointer}, (P_i,T_{j,i}))$:
                given a handle $\textsf{pointer}$ to an element of $I$,
                insert $(P_i,T_{j,i})$ before this element;
            \item $\textsf{list-del}(\textsf{pointer})$:
                given a handle $\textsf{pointer}$ to an element of $I$,
                delete its successor in $I$.
        \end{itemize}
        Further, each update is of the following types, and is annotated with its type:
        \begin{itemize}
            \item Type 1: A substitution of the head/tail of the sequence or an update
                that only involve pairs that contain a red piece.
            \item Type 2: An insertion or a deletion of a plain pair that results in the
                creation or the destruction of a run of plain pairs.
            \item Type 3: An insertion or a deletion of a plain pair that changes whether
                a run of plain pairs is active.
            \item Type 4: An insertions or a deletion of a plain pair that
                shrinks or expand an active internal run that remains active after the update.
        \end{itemize}
        In total, our sequences will contain $\cO(\redsize^2)$ updates of Types 1-3 and
        $\cO(\redsize^2 \comp)$ updates of Type 4.
        We show how to maintain $\Y$ subject to each update on $I$, such that $\Y
        \coloneqq \blact(\I)$ by transforming, in $\cO(\log \redsize)$ time, each update
        for $I$ to
        \begin{itemize}
            \item $\cO(1)$ BCDPM-updates for $\Y$ if the update is of Type 1, 2, or 3, or
            \item to a single {\tt BCDPM-Set} if the update is of Type 4.
        \end{itemize}
        We maintain $Y$ analogously.

        In order to facilitate the efficient transformation of updates for $I$ to updates
        for $\Y$, we use a balanced binary search tree to support the following operation
        in $\cO(\log \redsize)$ time at the cost of an $\cO(\log \redsize)$-time additive
        overhead for each update operation on $Y$:
        given a handle to an element of $Y$, return its rank, that is,
        the number of elements that precede it in~$Y$.
        Additionally, in $\cO(\redsize \log \redsize)$ time, we insert the elements of $\rred{P} \cup \{P_1,
        P_z\}$ in a doubly linked list $L_P$ and the elements of
        \(\rred{T}\) in a doubly linked list $L_T$,
        sorted with respect to their starting positions.
        At all times, we store a bidirectional pointer between each element of each
        of $L_P$ and $L_T$ and the DPM-pair that contains it in $I$.
        (We do not explicitly mention when such pointers need to be updated.)
        For both \(P\) and \(T\), we write \(\textsf{pointer}(P_i)\) (or
        \(\textsf{pointer}(T_i)\))
        for the pointer from the element of $L_P$ (or \(L_T\)) corresponding to $P_i$ (or
        \(T_i\)) to the element of $I$ that contains $P_i$ (or \(T_i\)).

        Throughout the course of the algorithm, to ensure the correct and efficient
        maintenance of our pointers, we only update pairs that contain a piece in
        $\rred{P}$ in each of $I$ and $Y$ using substitution operations.

        We now explain how to map each computed update for $I$ to a
        BCDPM-update for~$\Y$.
        Observe that using our pointers and rank structure, given a handle to an element
        of $I$ that contains a red piece, we can compute in $\cO(\log \rho)$ time the rank
        of the pair that contains the same red piece in $Y$ and, hence, in $\Y$.
        We transform an update for $I$ to an update for $Y$ and $\Y$ as follows:
        \begin{itemize}
            \item If the update is of Type 1, we simply change the index of the update
                based on the computed rank.
            \item If the update is of Type 2, we consider two cases:
                \begin{itemize}
                    \item if the update results in the creation of a run, we issue a {\tt
                        BCDPM-Insert} update, followed by a {\tt BCDPM-Bleach} update;
                    \item if the update results in the destruction of a run, we issue a
                        {\tt BCDPM-Delete} update.
                \end{itemize}
            \item For a Type 4 update, we also consider two cases:
                \begin{itemize}
                    \item if the update is of the form $\textsf{list-del}$ and changes the
                        status of an inactive run to active, we issue a {\tt BCDPM-Bleach}
                        update followed by a {\tt BCDPM-Set} update;
                    \item if the update is of the form $\textsf{list-ins}$ and changes the
                        status of an active run to inactive, we issue a {\tt BCDPM-Commit}
                        update.
                \end{itemize}
            \item If the update is of Type 5, we transform it to a {\tt BCDPM-Set} update.
        \end{itemize}

        We now explain how to compute the desired updates for $I$, which we insert in a
        global min-heap with keys in
        $\fragmentoc{\min J}{\max J}$.
        When we are done generating updates, we pop them from the heap, appending
        each update with key $j$ to the sequence $\update_j$.
        This process indeed returns sequences
        $\textsf{update}_j$ in increasing order with respect to~$j$.

        We first issue updates for the head and the tail of the sequence.\footnote{Some
        of the issued updates might be redundant or duplicate, but this is not a problem.}
        To this end, we first compute the union $\textsf{HeadTail}$ of the sets
        \[
            \{ j \in J : T_{j,1} \neq Q^\infty\fragmentco{-\kappa}{2\tau+\Delta} \}
            \quad\text{and}\quad
            \{ j \in J : T_{j,z} \neq Q^\infty\fragmentco{z\tau}{y_P+\kappa} \},
        \]
        in $\cO(k)$ time using \cref{lem:specialbound-simpl}.
        For each $j \in \textsf{HeadTail} \setminus \min{J}$, for each $e \in \{1,z\}$,
        we issue an update
        $\textsf{list-sub}(\textsf{pointer}(P_e), (P_e,T_{j,e}))$
        with key $j$.
        Further, for each $j \in \textsf{HeadTail} \setminus \max{J}$, for each $e \in \{1,z\}$,
        we issue update $\textsf{list-sub}(\textsf{pointer}(P_e), (P_e,T_{j+1,e}))$
        with key $j+1$.
        The number of such issued updates is $\cO(k)=\cO(\redsize)$ and their all of Type
        1.

        Let us now issue all updates that involve some red piece.
        For each pair of fragments $P_i \in \rred{P}$ and $T_t \in \rred{T}$ we do the following.
        \begin{itemize}
            \item If $t-i \in \fragmentoc{\min J}{\max J}$, we issue updates:
                \begin{itemize}
                    \item $\textsf{list-sub}(\textsf{pointer}(P_i), (P_i,T_t))$
                        with key $t-i$,  substituting $(P_i, T_{t-1})$ with $(P_i, T_t)$
                        and
                    \item if $i+1<z$ and $P_{i+1}\not\in \rred{P}$,
                        $\textsf{list-del}(\textsf{pointer}(P_i))$
                        with key $t-i$, deleting the pair  $(P_{i+1}, T_t)$.
                \end{itemize}
            \item If $t-i+1 \in \fragmentoc{\min J}{\max J}$, we issue updates:
                \begin{itemize}
                    \item if $i>2$ and $P_{i-1} \not\in \rred{P}$,
                        $\textsf{list-ins}(\textsf{pointer}(P_i), (P_{i-1},T_t))$ with key
                        $t-i+1$, inserting $(P_{i-1}, T_t)$ just before the pair
                        containing $P_i$, and
                    \item if $T_{t+1} \not \in \rred{T}$,  $\textsf{list-sub}(\textsf{pointer}(P_i), (P_i,T_{t+1}))$
                        with key $t-i+1$, replacing $(P_i, T_{t})$ with $(P_i, T_{t+1})$.
                \end{itemize}
        \end{itemize}
        Over all pairs of red pieces, we thus issue $\cO(\redsize^2)$ such updates in
        total, and each of these updates only involves pairs that contain a red piece. We
        thus annotate these updates as Type 1.

        Next, we compute updates that enable the efficient maintenance of the (trimmed)
        run of plain pairs that is enclosed by a pair
        containing some $P_i \in \rred{P}\cup\{P_1, P_z\}$
        and a pair containing some $T_t \in \rred{T}$, when such a run exists.

        First, we consider the case when the pair containing $P_i$ precedes the pair
        containing $T_t=T_{j-1,t-(j-1)}$ in some $\I'_{j-1}$
        and these two pairs enclose a non-empty run of plain pairs.
        In this case, the run would either shrink in $\I'_j$ or retain its length $\comp$.
        Let the successor of $P_i$ in $L_P$ be $P_v$.
        Further, if $T_t$ is not the first element in $L_T$ let its predecessor in $L_T$
        be $T_u = T_{j-1,u-(j-1)}$; otherwise, set $u\coloneqq -\infty$.
        The following conditions must be satisfied:
        \begin{enumerate}[(i)]
            \item $t - (j - 1) \leq v$, which is equivalent to $j > t - v$, so that $P_v$
                is not in a pair strictly between the two pairs and hence none of the
                elements of $\rred{P}$ is,
            \item $u - (j - 1) \leq i $, which is equivalent to $ j > u - i$, so that
                there is no non-plain pair containing some element of $\rred{T}$ between
                the two pairs, and
            \item $t - (j - 1) - i > 1 $, which is equivalent to $ j \leq t - i - 1$, so
                that the precedence condition is satisfied and the two pairs are not
                adjacent.
        \end{enumerate}
        Further, the run should shrink in $\I'_j$ only if the two pairs are close enough,
        that is, if
        $t - (j - 1) - i -1 \leq \comp$, which is equivalent to $j > t - i - \comp - 1$.
        Hence, for each \[
            j \in \fragmentoc{\max\{ \min J, t - v, u - i, t - i - \comp - 1  \}}{\min \{
        \max J, t - i - 1 \}},\] we issue an update
        $\textsf{list-del}(\textsf{pointer}(P_i))$ with key $j$, deleting a plain pair
        from the run succeeding the pair that contains $P_i$
        with key $j$.
        Observe that only the last update may result in a destruction of the run, while
        only the first update may change its status by turning it from inactive to active.
        We can compute whether this is the case for these two updates in $\cO(\log
        \redsize)$ time, and hence annotate all updates appropriately (each of them is of
        Type 2, 3, or 4).
        Observe that the total number of issued updates is at most $(t - i - 1)-(t - i -
        \comp - 1) = \comp$ and at most two of them are not of Type 4.

        We now consider the complementary case when in each of $\I'_{j-1}$ and $\I'_j$,
        the pair containing $P_i$ succeeds the pair containing $T_t$
        and these two pairs either enclose a non-empty run of plain pairs or are
        adjacent.\footnote{Observe that the pairs can only be adjacent in $\I'_{j-1}$.}
        In this case, the run would either be expanded by one plain pair in $\I'_j$ or
        retain its length $\comp$.
        Let the predecessor of $P_i$ in $L_P$ be $P_y$.
        Further, if $T_t$ is not the last element in $L_T$ let its successor in $L_T$ be
        $T_x = T_{j-1,x-(j-1)}$; otherwise, set $x\coloneqq \infty$.
        The following conditions must be satisfied:
        \begin{enumerate}[(i)]
            \item $t - (j - 1) > y $, which is equivalent to $ j \leq t - y$, so that
                $P_x$ is strictly to the left of both pairs in $\I'_{j-1}$, as if this is
                not the case then
                the condition on $\I'_j$ would not be satisfied,
            \item $x - (j - 1) > i $, which is equivalent to $ j \leq x - i$, so that
                $T_y$ (if it exists) is strictly to the right of both pairs in
                $\I'_{j-1}$, as if this is not the case then
                the condition on $\I'_j$ would not be satisfied,
            \item $t - (j - 1) < i $, which is equivalent to $ j > t - i + 1$, so that the
                precedence condition is satisfied.
        \end{enumerate}
        Further, the run should be expanded in $\I'_j$ only if the two pairs are close enough,
        that is, if
        $i - (t - (j - 1)) -1 < \comp $, which is equivalent to $ j \leq \comp + t - i +
        1$.
        Hence, for each \[
            j \in \fragmentoc{\max\{ \min J, t - i + 1 \}}{\min \{ \max J, t - y, x - i,
        \comp + t - i + 1 \}},\] we issue an update
        $\textsf{list-ins}(\textsf{pointer}(P_i),
        (Q^\infty\fragmentco{0}{\tau+\Delta},Q^\infty\fragmentco{0}{\tau+\Delta}))$, inserting a
        plain pair in the (potentially empty) run preceding the pair that contains $P_i$
        with key $j$.
        Observe that only the first update may result in the creation of a run, while only
        the first update may change its status from active to inactive.
        We can compute whether this is the case for these two updates in $\cO(\log
        \redsize)$ time, and hence annotate all updates appropriately (each of them is of
        Type 2, 3, or 4).
        Observe that the total number of such issued updates is at most $(\comp + t - i +
        1)-(t - i + 1) = \comp$ and at most two of them are not of Type 4.

        We can compute all updates for a given pair (in the intermediate interface)
        in $\cO(\comp)$ time using lists $L_P$ and $L_T$.

        The invariants of \BCDPM are clearly satisfied at initialization.
        By construction, throughout the course of the algorithm each blank piece in $\Y$
        is of order at most $\comp$.
        Additionally the second invariant is satisfied at all times, since all the updates
        involving the head or tail of $\Y$
        only involve pairs in $\mathcal{S}_\beta \times \mathcal{S}_\beta$ and
        $\mathcal{S}_{\varphi} \times \mathcal{S}_{\varphi}$, respectively,
        while all remaining updates involve pairs in $\mathcal{S}_\mu \times
        \mathcal{S}_\mu$.
    \end{claimproof}

    We are now ready to complete the reduction to \BCDPM; see~\cref{alg:arprogr} for a
    pseudocode implementation.

    Using \cref{claim:kaupdates}, we initialize the
    set \[
        \mathcal{U}=\{\textsf{update}_j : j\in \fragmentoc{\min J}{\max J} \text{ and }
        \Y'_{j-1} \neq \Y'_j\}.
    \]
    To simplify the construction of certain arithmetic progressions later in the proof,
    we insert to $\mathcal{U}$ the following single-element (void) sequences of updates.
    For each $j \in \fragmentoc{\min J}{\max J}$ such that $r_j - r_{j-1}\neq \tau$ and $\Y'_{j-1} = \Y'_j$,
    we insert to $\mathcal{U}$ the sequence $\textsf{update}_j$ that consists of a single
    element: {\tt BCDPM-Substitute(\((P_1,T_{j,1})\),\(1\))}, that is,
    the substitution of the first pair of the maintained sequence with $(P_1,T_{j,1})$.
    By \cref{fct:rj}, we only generate $\cO(d_T+1)$ new updates.

    Now, we initialize the sequence $\Y$ with $\Y'_{\min J}$.
    We consider the sequences of updates $\textsf{update}_j$ in~$\mathcal{U}$
    in increasing order with respect to~$j$
    and apply them to the maintained sequence~$\Y$.
    Prior to performing the sequence of updates $\textsf{update}_j$ we do the following.
    Suppose that the maintained sequence~$\Y$ corresponds to $\Y'_{t}$, that is,
    either $\textsf{update}_j$ is the first element of $\mathcal{U}$ and $t=\min J$
    or the previously applied sequence of updates was $\textsf{update}_{t}$.
    First, we compute $\OccW_k(\Y'_t)$ using a \texttt{BCDPM-Query}, noting that
    \(\tor(\I'_t) \le \tor(\I_t) \le \Delta/2 - k\) is satisfied at all times due to the
    fact that $\I'_j$ is a subsequence of~$\I_j$ and~\cref{fact:lengths_diff}.
    Now, observe that we have $\Y'_t = \cdots = \Y'_{j-1}$
    and hence $\OccW_k(\Y'_t) = \cdots = \OccW_k(\Y'_{j-1})$.
    Further, for each
    $i \in \fragmentoo{t}{j}$, we have $r_i-r_{i-1} = \tau$ and hence
    $r_i + \OccW_k(\Y'_i) = r_t + \OccW_k(\Y'_t) + (i-t)\tau$.
    We can thus efficiently return
    $\bigcup_{i = t}^{j-1} (r_i+\OccW_k(\Y'_t))$
    as the union of $\cO(|\OccW_k(\Y'_i)|)$ arithmetic progressions
    \[
        \bigcup_{a \in \OccW_k(\Y'_t)} \{r_t+a + i \cdot \tau : i \in \fragmentco{0}{j-t}\}.
    \]
    (Observe that in the case when $t=j-1$, each of the constructed arithmetic
    progressions consists of a single element.)
    Finally, we apply the sequence of updates $\textsf{update}_j$.
    The case when there are no more updates to be processed is treated analogously (with
    $j$ set to $\max J+1$).
    The upper bound on the number of returned arithmetic progressions is embedded in the
    analysis of the time complexity of the algorithm.

    \begin{algorithm}[t!]
        \SetKwBlock{Begin}{}{end}
        \SetKwFunction{DPMquery}{BCDPM-Query}
        \swaps{$T$, $P$, $k$, $w$, $Q$, $\A_P$, $\A_T$, $\rred{P}$, $\rred{T}$, $\comp$}\Begin{
            $\mathcal{G} \gets \emptyset$\;
            Construct set $\mathcal{U}$ of sequences $\textsf{update}_j$ specified in
            \cref{claim:kaupdates}\tcp*{$\cO(\redsize^2 \comp \log n)$ time}
            \ForEach{$j \in \fragmentoc{\min J}{\max J}$ such that $r_j - r_{j-1}\neq
                \tau$ and $\Y'_{j-1} = \Y'_j$}{\label{line:void1}
                $\textsf{update}_j \gets \{${\tt BCDPM-Substitute(\((P_1,T_{j,1})\),\(1\))}\(\}\)\;
                $\mathcal{U} \gets \mathcal{U} \cup \{\textsf{update}_j\}$\;\label{line:void2}
            }
            {\tt BCDPM-Init($\Y'_{\min J}$, $\comp$, $k$, $\Delta$, $\Sb$, $\Sm$, $\Sf$,
            $\mathcal{Q}$, $w$)}\tcp*{$\cO(\redsize \comp + \redsize \log k)$ time}
            $t \gets \min J$\;
            \ForEach{$j \in \{ i \in \fragmentoc{\min J}{\max J + 1} : \textsf{update}_i
                \neq \emptyset \text{ or } i = \max J + 1\}$ in
                incr.~order}{\label{line:for1}
                $\OccW_k(\Y'_j) \gets \DPMquery$\;\label{line:dpmq}
                \ForEach{$a \in \OccW_k(\Y'_j)$}{
                    $\mathcal{G} \gets \mathcal{G} \cup \{r_{t} + a + i \cdot \tau : i \in
                    \fragmentco{0}{j-t}\}$\;\label{line:arprogr}
                }
                \If {$j \leq \max J$}{
                    Apply the sequence of updates $\textsf{update}_{j}$ on $\Y$\;
                    $t \gets j$\;\label{line:for2}
                }
            }
            \Return{$\mathcal{G}$}\;
        }
        \caption{Reduction of computing $\mathcal{G}$ to an instance of \BCDPM{}.}\label{alg:arprogr}
    \end{algorithm}

    We now proceed to analyze the time required by \cref{alg:arprogr}.

    First, recall that the elements of $\mathcal{U}$ as returned by \cref{claim:kaupdates}
    in $\cO(k^2 \comp \log n)$ time are sorted in increasing order with respect to $j$ and
    a representation of $(r_j)_{j\in J}$
    as the concatenation of $\Oh(d_T+1)$ arithmetic progressions with difference $\tau$
    can be computed in $\Oh(d_T+1)$ time in the \pillar model due to \cref{fct:rj}.
    Thus, Lines~\ref{line:void1}--\ref{line:void2} can be implemented in $\cO(k^2 \comp)$
    time in a parallel left-to-right scan of $\mathcal{U}$ and the aforementioned
    representation of $(r_j)_{j\in J}$,
    maintaining the sortedness of $\mathcal{U}$.

    Observe that each $\I'_{j}$, for $j \in J$ is of length $\cO(k \comp)$ as it consists
    of a head,
    a tail, $\cO(\redsize)$ non-plain pairs and $\cO(\redsize)$ internal runs of plain
    pairs, each
    of length at most $\comp$.
    Sequence $\Y'_{\min J}$ can be thus computed in $\cO(\redsize \comp + \redsize \log
    \redsize)$
    time in the \pillar model: the head and the tail are computed in $\cO(k)$ time in the
    \pillar model due to
    \cref{lem:specialbound-simpl}, while the pairs consisting of internal pieces can be
    computed
    by processing the elements of $\rred{P}$ and $\rred{T}$ in a left-to-right manner,
    after sorting them
    with respect to their starting positions in $\cO(\redsize \log \redsize)$ time.
    The BCDPM sequence in the call to {\tt BCDPM-Init} is thus a sequence of length
    $\cO(\redsize \comp)$.

    The for-loop of Lines~\ref{line:for1}--\ref{line:for2} is executed $\cO(\redsize^2
    \comp)$ times
    as this is an upper bound on the number of sequences of updates that are processed and
    hence on the calls to {\tt BCDPM-Query}.
    For each $j\in J$, due to \cref{fact:lengths_diff}, we have
    that the difference of the two $\Delta$-puzzles represented by the first and second
    components of the elements
    of each $\I'_j$ is at most $3\kappa-k$ and hence $|\OccW_k(\Y'_j)|=\cO(k)$.
    Thus, each query in \cref{line:dpmq} returns a set of size $\cO(k)$.
    Hence, apart from the time required for \BCDPM updates and queries, each iteration of
    the for-loop takes $\cO(k)$ time,
    for a total of $\cO(k \redsize^2 \comp)$ time.
    It readily follows that the algorithm outputs $\cO(k \redsize^2 \comp)$ arithmetic
    progressions.

    In total, in $\cO(k \redsize^2 \comp + \redsize^2\comp \log n)$ time in the \pillar
    model,
    the problem in scope reduces to an instance of \BCDPM
    with $\mathcal{S}_\beta+\mathcal{S}_\mu+\mathcal{S}_{\varphi}=\cO(k)$ (due to
    \cref{lem:specialbound-simpl}), $\Delta = 6\kappa = \Oh(k)$,
    $\Y$ initialized as a sequence of length $\cO(\redsize \comp)$, and
    updates satisfying the condition in the lemma's statement due to
    \cref{claim:kaupdates}.

    It only remains to bound the total number of blank pieces over all calls to {\tt
    DPM-Query}.
    \begin{claim}\label{claim:few_blank}
        The total number of blank pieces over all calls to {\tt DPM-Query} is
        $\cO(\redsize^2 \comp)$.
    \end{claim}
    \begin{claimproof}
        In a sequence $\Y'_j$, a blank piece is adjacent either to the leading element of
        $\Y'_j$, or to the trailing element of $\Y'_j$, or
        to a pair of strings containing a piece from $\rred{P}$ and a pair of strings
        containing a piece from $\rred{T}$.

        We have already established that the number of calls to {\tt BCDPM-Query} are
        $\cO(\redsize^2 \comp)$.
        Thus, we get a trivial bound of $\cO(\redsize^2 \comp)$ for the number of
        blank pieces adjacent to the
        leading/trailing element of the BCDPM-sequences involved in all calls to {\tt
        BCDPM-Query}.
        It thus remains to bound, over all calls to {\tt BCDPM-Query}, the number of blank pieces
        adjacent to two internal elements of the BCDPM-sequence
        such that one contains a piece from $\rred{P}$ and the other contains a piece from $\rred{T}$.
        For each of the $\cO(\redsize^2)$ pairs $(P_x,T_y) \in \rred{P} \times \rred{T}$,
        a pair containing $P_x$ and a pair containing $T_y$ can be adjacent at the same
        time to an internal run of canonical pieces
        of size smaller than $\comp$ in $\I'_j$ for $\cO(\comp)$ values of $j$.
        As only internal runs of $\I'_j$ of size smaller than $\comp$ are represented by blank pieces, the bound follows.
    \end{claimproof}
    This concludes the proof of the theorem.
\end{proof}

\subsection{Wrap-up: Solution for Almost Periodic Patterns}

In order to complete the description of our efficient algorithm for \SM,
we use the data structure encapsulated in the following theorem, which is proved in
\cref{sec:dpm}.

\begin{restatable*}[Efficient implementation for \BCDPM]{theorem}{rdpmds}
    \dglabel{xxdsyqwfxu}[lemma2,lemma1,11-3-1,lem:trim,lem:small_ed,uogiuvnbph,azjqbhozka,cor:dproduct]
    There is an implementation for \BCDPM that satisfies all of the following.
    \begin{itemize}
        \item {\tt BCDPM-Init($\Y'_{\min J}$, $\comp$, $k$, $\Delta$, $\Sb$, $\Sm$, $\Sf$,
            $\mathcal{Q}$, $w$)} can be performed in time
            \begin{itemize}
                \item \(\Oh(k^4 \log^2 (k \lambda) + |\Y'| k^2 )\) for general weights and
                    in time
                \item \(\Oh(k^3 W \log^2 (k \lambda) + |\Y'| k W \log k )\) for integer weights,
            \end{itemize}
            where $\lambda$ is the length of the longest string in $\Sb \cup \Sm \cup \Sf$;
        \item any call to {\tt BCDPM-Set} can be performed in constant time;
        \item any call to {\tt BCDPM-Query} can be performed in \(\Oh(k \rho)\) time for
            general weights and in  \(\Oh(k \rho \log \Delta)\) time for integer weights,
            where \(\rho\) is the number of blank pieces at the time of the query;
        \item any call to any other operation can be performed in time
            \begin{itemize}
                \item \(\Oh(k^2 \log(k + |\Y|))\) for general weights and in time
                \item  \(\Oh(kW \log^2(k + |\Y|))\) for integer weights,
            \end{itemize}
            where $\Y$ is the
            BCDPM sequence to which the operation is applied.\qedhere
    \end{itemize}
\end{restatable*}

We are now ready to prove \cref{lem:solve_SM}.

\solveSM
\begin{proof}
    Let us set
    {$\tau \coloneqq q\ceil{{k+d_P+d_T}/{2q}}$},
    and $\Delta \coloneqq 6\kappa$.

    If $\beta_P  \coloneqq \ceil{{y_P}/{\tau}}<20$, we are done, as we can solve the \SM
    problem in $\cO(k^3 \beta_P \log^2 (mk))$ time in the \pillar model, returning an
    $\cO(k)$-size explicit representation of $\OccW_k(P,T)$, due to \cref{lem:smallbeta}.
    We henceforth suppose that $\beta_P > 20$, in which case all the statements of
    \cref{subsec:puzzles,subsec:redDPM} hold.

    \cref{fct:rj} states that $\OccE_k(P, T) = \bigcup_{j\in J} \big(\OccE_k(P, R_j) + r_j
    \big)$.
    Due to \cref{fact:simple}, we also have $\OccW_k(P, T) \subseteq \bigcup_{j\in J}
    \big(\OccW_k(P, R_j) + r_j \big)$.
    Then, \cref{fct:puzzleP,fct:puzzleT} state that $\val_{\Delta}(P_1,\ldots,P_z) = P$
    and $\val_{\Delta}(T_{j,1},\ldots,T_{j,z}) = R_j$, respectively.
    Further, due to \cref{fact:druns}, we have that,
    for all \(j \in J\), $\OccW_k(P, R_j) = \OccW_k(\I_j) = \OccW_k( \Comp(\I_j, \sp{P},
    \sp{T}, k + 2)$.
    It thus suffices to compute \[
        \mathcal{G} \coloneqq \bigcup_{j \in J}
        (r_j+\OccW_k(\Comp(\I_j, \sp{P}, \sp{T}, k+2))).
    \]

    To this end, we call procedure \swaps{$T$, $P$, $k$, $w$, $Q$, $\A_P$, $\A_T$,
        $\sp{P}$, $\sp{T}$,
    $k+2$} from \cref{lem:swaps}, which computes $\mathcal{G}$ represented as \(\Oh(k^4)\)
    arithmetic progressions with difference \(\tau\) since we have $\sp{P}+\sp{T}=\cO(k)$
    from \cref{lem:specialbound-simpl}.
    In what follows, we use the implementation of \BCDPM from \cref{xxdsyqwfxu} for
    arbitrary weight functions.
    Let us analyze the time complexity of this call.
    The reduction takes time \( \Oh( k^3 (k + \log n)) \) in the \pillar model.
    We next analyze the time required by our calls to BCDPM operations:
    \begin{itemize}
        \item We call
            {\tt BCDPM-Init($\Y'$, $k+2$, $k$, $\Delta$, $\Sb$, $\Sm$, $\Sf$, $\mathcal{Q}$, $w$)}
            with $\Y'$ being an BCDPM sequence of size $\cO(k^2)$.
            This requires $\cO(k^4 \log^2 (mk) + k^2 \cdot k^2) = \cO(k^4 \log^2 (mk))$
            time, as the lengths of all strings in $\Sb \cup \Sm \cup \Sf$ are $\cO(m)$.
        \item We make \(\Oh(k^2)\) calls to each of {\tt BCDPM-Delete}, {\tt
            BCDPM-Insert}, {\tt BCDPM-Substitute}, {\tt BCDPM-Commit}, and {\tt
            BCDPM-Bleach}.
            These calls are processed in total time $\cO(k^2 \cdot k^2 \log k) = \cO(k^4
            \log k)$ time, since the maintained sequence $\Y$ is of length polynomial in
            $k$
            throughout the course of the algorithm.
        \item We make \(\Oh(k^3)\) calls to {\tt BCDPM-Set} and {\tt BCDPM-Query}, with
            the total number of blank pieces over all calls to {\tt BCDPM-Query} being
            $\cO(k^3)$.
            Let us denote the number number of calls to {\tt BCDPM-Query} by $\nu$, and
            the number of blank pieces in the $i$-th call to the {\tt BCDPM-Query} by
            $b_i$.
            The time required for the \(\Oh(k^3)\) calls to {\tt BCDPM-Set} and {\tt
            BCDPM-Query} is then $\cO(k\cdot \big(k^3 + \sum\nolimits_{i \in \fragment{1}{\nu}}
            b_i\big)\big) = \cO(k^4)$.
    \end{itemize}
    In total, our call to \swaps thus takes $\cO(k^4 \log^2 (mk))$ time in the \pillar model.
\end{proof}

\section{An \texorpdfstring{\boldmath $\cOtilde(n+k^{3.5}\cdot W^4 \cdot n/m)$}{Õ(n+k³˙⁵·W⁴·n/m)}-time Algorithm for Metric Integer Weights}\label{sec:marking}

In this section, we incorporate the results of \cref{sec:verify,sec:dpm} into the approach
of \cite{ckw22} in order to prove the following result, which is the crucial step on the
way to proving \cref{thm:main} for integer metric weights.
\solveintSM*
Most statements in this section have direct unweighted counterparts in \cite{ckw20,ckw22}.

\subsection{Toolbox}

\begin{lemma}[\texttt{FindAWitness($k$, $Q$, $S$, $w$)}, {\cite[see Lemma
    6.5]{ckw20}}]
    \dglabel^{lm:witness}[cor:fragments,lem:int_wed]
    Let $k$ denote a positive integer,
    let $S$ denote a string,
    let $Q$ denote a primitive string that satisfies $|S|\ge (2k+1)|Q|$ or $|Q|\le 3k+1$,
    and let $w:\Esigma^2 \to \fragment{0}{W}$ be a weight function.

    Then, we can compute a \emph{witness} $Q^\infty\fragmentco{x}{y}$
    such that $\edw{S}{Q^\infty\fragmentco{x}{y}}=\edwl{S}{Q}\le k$,
    or report that $\edwl{S}{Q}>k$.
    In the former case, the algorithm returns the breakpoint representation of a
    $w$-optimal alignment $\A : S \onto Q^\infty\fragmentco{x}{y}$.
    The algorithm takes $\Ohtilde(k^2\cdot W)$ time in the \modelname model.
\end{lemma}
\begin{proof}
    For a set $A\subseteq \mathbb{Z}$ and an integer $p>0$, we define
    $A\bmod p := \{a \bmod p \mid a\in A\}$.
    We first compute a (short) interval $J$ such that $\OccE_k(S,Q^\infty)\bmod |Q| \subseteq J
    \bmod |Q|$.
    If $|Q|\le 3k+1$, then we simply set $J=\fragmentco{0}{|Q|}$.
    Otherwise, we decompose the prefix $S\fragmentco{0}{(2k+1)|Q|}$ into $2k+1$ fragments
    $S_i$ of length~$|Q|$.
    For each fragment $S_i$, we construct $\cycEqOp{S_i}{Q}=\{j \in \Z : Q = \rot^{j}(S_i)\}$.
    Next, we define an auxiliary
    set $I$ as the union of~intervals $\fragment{p}{p+k}$
    such that for at least $k + 1$ fragments $S_i$ of~$S$,
    we have $\fragment{p}{p+k}\cap \cycEqOp{S_i}{Q}\ne \emptyset$.
    Finally, we set $J\subseteq \fragmentco{0}{2|Q|}$ to be a shortest interval satisfying
    $I\bmod |Q|\subseteq J\bmod |Q|$.
    Here, $J\bmod |Q|$ can be interpreted as a shortest cyclic interval (modulo $|Q|$)
    containing $I\bmod |Q|$.

    Having computed the set $J$, we use \cref{cor:fragments} to determine
    $\edwk{k}{S}{Q^\infty\fragmentco{x}{y}}$
    for all positions $x\in J$ and $y\in J + \fragment{|S|-k}{|S|+k}$.
    If all the values are $\infty$, we report that $\edwl{S}{Q}>k$.
    Otherwise, we pick a fragment $Q^\infty\fragmentco{x}{y}$ minimizing
    $\edwk{k}{S}{Q^\infty\fragmentco{x}{y}}$,
    and we find an $w$-optimal alignment $\A : S \onto Q^\infty\fragmentco{x}{y}$ using
    \cref{lem:int_wed}.

    The correctness is based on the aforementioned characterization of~$J$:
    \begin{claim}\label{cl:wt0}
        The interval $J$ satisfies
        $\OccW_k(S,Q^\infty)\bmod |Q| \subseteq J \bmod |Q|$.
    \end{claim}
    \begin{claimproof}
        The claim trivially holds if $|Q|\le 3k+1$, so we assume that $|Q|> 3k+1$.
        For every $i \in \fragment{0}{2k}$, define $S_i \coloneqq S\fragmentco{i\,|Q|}{(i+1)\,|Q|}$.
        Consider an optimum alignment between $S$ and its $(k,w)$-error occurrence
        $Q^\infty\fragmentco{x}{y}$,
        and let $Q_i=Q^\infty\fragmentco{x_i}{x_{i+1}}$ denote the fragment aligned to $S_i$.
        Consider the multi-set $R:= \bigcup_i \cycEqOp{S_i}{Q}$.
        Next, consider the values $\delta_i := x_i-i|Q|$ for $i \in \fragment{0}{2k}$.
        We have $\delta_0=x$, and $\delta_{i+1}=\delta_i+|Q_i|-|S_i|$ for $i>0$.
        Since \[
            \sum_{i=0}^{2k}\big||Q_i|-|S_i|\big|\le \sum_{i=0}^{2k}\edw{Q_i}{S_i}\le k,
        \]
        all values $\delta_i$ belong to an interval of~the form $\fragment{p}{p+k}$ for some integer $p$.
        Moreover, note that $Q_i=S_i$ holds for at least $k+1$ fragments $Q_i$;
        these fragments satisfy $Q=\rot^{\delta_i}(S_i)$ and thus contribute $\delta_i$ to $R$.
        We conclude that there is an interval $\fragment{p}{p+k}$ containing $x$ and at least
        $k+1$ elements of~$R$.
        Consequently, we have $x\in I$. By definition of~$J$, this yields $x\bmod |Q| \in
        J \bmod |Q|$.
    \end{claimproof}

    Now, let $Q^\infty\fragmentco{x}{y}$ denote a witness that satisfies
    $\edwl{S}{Q}=\edw{S}{Q^\infty\fragmentco{x}{y}}$.
    By \cref{cl:wt0}, there is a matching fragment
    $Q^\infty\fragmentco{x'}{y'}=Q^\infty\fragmentco{x}{y}$
    starting at $x'\in J$. Thus, we may assume without loss of~generality that $x\in J$.
    As we verify all possible starting positions in $J$ using \cref{cor:fragments}, we
    correctly compute the starting position $x$  and the ending position $y$ of the
    minimum-cost occurrence of $S$ in $Q^\infty$.

    As for the running time, we prove the following characterization of~$J$:
    \begin{claim}\label{cl:wt1}
        The interval $J$ satisfies $|J|\le 3k+1$.
    \end{claim}
    \begin{claimproof}
        The claim trivially holds if $|Q|\le 3k+1$, so we assume that $|Q|>3k+1$.
        Recall that the multiset~$R:= \bigcup_{i=0}^{2k} \cycEqOp{S_i}{Q}$.
        As the string $Q$ is primitive, $R$ is the union of~at most $2k+1$ infinite
        arithmetic progressions with difference $|Q|$.
        In particular, if $\fragment{p}{p+k}$ and $\fragment{p'}{p'+k}$
        contain at least $k+1$ elements of~$R$ each, then $(\fragment{p}{p+k}\bmod |Q|)
        \cap (\fragment{p'}{p'+k} \bmod |Q|) \ne \emptyset$, and thus $\fragment{p'}{p'+k}
        \bmod |Q| \subseteq \fragment{p-k}{p+2k}\bmod |Q|$.
        Since $J$ is the union of~such intervals $\fragment{p'}{p'+k}$,
        we have $J \bmod |Q| \subseteq \fragment{p-k}{p+2k}\bmod |Q|$.
        By definition of~$I$, we conclude that $|I| \le |\fragment{p-k}{p+2k}|=3k+1$.
    \end{claimproof}

    Now, observe that computing the multiset $R$ (represented as the union of~infinite
    arithmetic progressions modulo $|Q|$) takes $\Oh(k)$ time in the \modelname
    model; computing the sets $I$ and $J$ can be done in $\Oh(k \log\log k)$ time
    by sorting $R$ (restricted to $\fragmentco{0}{|Q|}$) and a subsequent cyclic scan over~$R$.
    Further, by \cref{cl:wt1}, we call \cref{cor:fragments} for a substring of $Q^\infty$ of length $|S|+\Oh(k)$.
    This costs $\Ohtilde(k^2 W)$ time, just like the application of \cref{lem:int_wed}.
\end{proof}

\subsection{Locked Fragments and their Properties}

For strings $S$ and $Q$ and a weight function $w$, we write $\edwp{S}{Q} := \min
\{\edw{S}{Q^\infty\fragmentco{0}{j}} \mid j \in \mathbb{Z}_{\ge 0}\}$  to denote the minimum edit
distance between a string $S$ and any prefix of~a string $Q^\infty$.
Further, we write $\edwl{S}{Q}\coloneqq \min\{\edw{S}{Q^\infty\fragmentco{i}{j}} \mid i, j
\in \mathbb{Z}_{\ge 0}, i \le j\}$
to denote the minimum edit distance between $S$ and any substring of~$Q^\infty$,
and we set $\edws{S}{Q} \coloneqq \min\{\edw{S}{Q^\infty\fragmentco{i}{j|Q|}} \mid i, j
\in \mathbb{Z}_{\ge 0}, i \le j|Q|\}$.

\begin{definition}[{\cite[see Definition~5.5]{ckw20}}]
    Let $S$ denote a string, let $Q$ denote a primitive string, and let $w$ be a  weight function.
    We say that a fragment $L$ of~$S$ is \emph{locked} (with respect to $Q$ and $w$)
    if at least one of~the following holds:
    \begin{itemize}
        \item For some integer $\alpha$, we have $\edwl{L}{Q}=\edw{L}{Q^\alpha}$.
        \item The fragment $L$ is a suffix of~$S$ and $\edwl{L}{Q}=\edwp{L}{Q}$.
        \item The fragment $L$ is a prefix of~$S$ and $\edwl{L}{Q}=\edws{L}{Q}$.
        \item We have $L = S$.
        \qedhere
    \end{itemize}
\end{definition}

\begin{lemma}[{\cite[see Lemma~5.6]{ckw20}}]
    \dglabel{lem:locked}
    Let $S$ denote a string, let $Q$ denote a primitive string, and let $w$ be an integer
    weight function.
    There exist disjoint locked fragments $L_1,\ldots,L_{\ell}$ of $S$
    with $\edwl{L_i}{Q} > 0$ such that \[
        \edwl{S}{Q}=\sum_{i=1}^{\ell} \edwl{L_i}{Q}\quad
        \text{and}\quad\sum_{i=1}^{\ell}|L_i| \le (5|Q|+1)\edwl{S}{Q}.\]
\end{lemma}
\begin{proof}
    Let us choose integers $x\le y$ so that $\edwl{S}{Q}=\edw{S}{Q^\infty\fragmentco{x}{y}}$.
    Without loss of~generality, we may assume that $x\in \fragmentco{0}{|Q|}$.
    If $y \le |Q|$, then $|S|\le |Q|+\edwl{S}{Q}$; thus setting the whole string~$S$
    as the only locked fragment satisfies the claimed conditions.
    Hence, we may assume that $y > |Q|$.

    An arbitrary $w$-optimum alignment of~$S$ and $Q^\infty\fragmentco{x}{y}$ yields
    a partition  $S=S_0^{(0)}\cdots S_{\!s^{(0)}}^{(0)}$
    with $s^{(0)}=\floor{(y-1)/|Q|}$
    such that  $\edw{S}{Q^\infty\fragmentco{x}{y}} = \sum_{i=0}^{s^{(0)}} \Delta^{(0)}_i$,
    where
    \[\Delta^{(0)}_i = \begin{cases}
        \edw{S^{(0)}_0}{Q\fragmentco{x}{|Q|}} &\text{if }i=0,\\
        \edw{S^{(0)}_i}{Q} &\text{if }0 < i < s^{(0)},\\
        \edw{S^{(0)}_{s^{(0)}}}{Q\fragmentco{0}{y-s^{(0)}|Q|}} & \text{if }i = s^{(0)}.
    \end{cases}\]

    We start with this partition and then coarsen it by exhaustively applying
    the merging rules specified below, where each rule is applied only if the previous rules cannot
    be applied.
    In each case, we re-index the unchanged
    fragments $S^{(t)}_i$ to obtain a new partition $S = S^{(t+1)}_0\cdots
    S^{(t+1)}_{\!s^{(t+1)}}$ and re-index the corresponding values $\Delta^{(t)}_i$ accordingly.
    We say that a fragment $S^{(t)}_i$ is \emph{interesting}
    if $i=0$, $i=s^{(t)}$, $S^{(t)}_i\ne Q$, or $\Delta_i^{(t)}>0$.
    \begin{enumerate}
        \item\label{it:type1} If subsequent fragments $S^{(t)}_i$ and $S^{(t)}_{i+1}$ are
            both interesting, then merge $S^{(t)}_i$ and $S^{(t)}_{i+1}$, obtaining
            $S^{(t+1)}_i := S^{(t)}_i S^{(t)}_{i+1}$ and $\Delta^{(t+1)}_i :=
            \Delta^{(t)}_i + \Delta^{(t)}_{i+1}$.
        \item\label{it:type2} If $0 < i < s^{(t)}$ and $\Delta^{(t)}_i>0$,
            then merge the subsequent fragments $S^{(t)}_{i-1}=Q$, $S^{(t)}_i$, and
            $S^{(t)}_{i+1}=Q$, obtaining $S^{(t+1)}_{i-1} := S^{(t)}_{i-1} S^{(t)}_i
            S^{(t)}_{i+1}$, and set $\Delta^{(t+1)}_{i-1}:=\Delta^{(t)}_{i}-1$.
        \item\label{it:type3} If $0 < i = s^{(t)}$ and $\Delta^{(t)}_{i}>0$,
            then merge the subsequent fragments $S^{(t)}_{i-1}=Q$ and
            $S^{(t)}_{i}$, obtaining $S^{(t+1)}_{i-1} := S^{(t)}_{i-1}S^{(t)}_{i}$, and
            set $\Delta^{(t+1)}_{i-1} := \Delta^{(t)}_{i}-1$.
        \item\label{it:type4} If $0 = i < s^{(t)}$ and $\Delta^{(t)}_i>0$,
            then merge the subsequent fragments $S^{(t)}_i$ and $S^{(t)}_{i+1}=Q$,
            obtaining $S^{(t+1)}_i := S^{(t)}_i S^{(t)}_{i+1}$,
            and set $\Delta^{(t+1)}_{i}:=\Delta^{(t)}_{i}-1$.
    \end{enumerate}

    Let $S=S_0\cdots S_{s}$ denote the obtained final partition.
    We select as locked fragments all the fragments~$S_i$ with $\edwl{S_i}{Q}> 0$.
    Below,  we show that this selection satisfies the desired properties.
    We start by proving that we indeed picked locked fragments.
    \begin{claim}\label{clm:locked}
        Each fragment $S_i^{(t)}$ of~each partition $S=S_0^{(t)}\cdots S_{s^{(t)}}^{(t)}$
        satisfies at least one of~the following:
        \begin{itemize}
            \item $\edw{S^{(t)}_i}{Q^\alpha}\le \edwl{S^{(t)}_i}{Q}+\Delta^{(t)}_i$ for some
                integer $\alpha$;
            \item $i=s^{(t)}$ and $\edwp{S^{(t)}_i}{Q}\le \edwl{S^{(t)}_i}{Q}+\Delta^{(t)}_i$;
            \item $i=0$ and $\edws{S^{(t)}_i}{Q}\le \edwl{S^{(t)}_i}{Q}+\Delta^{(t)}_i$;
            \item $i=0=s^{(t)}$.
        \end{itemize}
    \end{claim}
    \begin{claimproof}
        We proceed by induction on $t$. The base case follows from the definition of~the
        values $\Delta^{(0)}_i$.

        As for the inductive step, we assume that the claim holds for all fragments
        $S^{(t)}_i$
        and we prove that it holds for all fragments $S^{(t+1)}_i$. We consider several
        cases based on the merge rule applied.
        \begin{enumerate}
            \item For a type-\ref{it:type1} merge of~interesting fragments $S^{(t)}_i$ and
                $S^{(t)}_{i+1}$ into $S^{(t+1)}_i$, it suffices to prove that
                $S^{(t+1)}_i$ satisfies the claim.
                \begin{itemize}
                    \item If $0 < i < s^{(t+1)}$, then $\edw{S^{(t)}_i}{Q^\alpha}\le
                        \edwl{S^{(t)}_i}{Q}+\Delta^{(t)}_i$ and
                        $\edw{S^{(t)}_{i+1}}{Q^{\alpha'}}\le
                        \edwl{S^{(t)}_{i+1}}{Q}+\Delta^{(t)}_{i+1}$ hold by the inductive
                        assumption for some integers $\alpha,\alpha'$. Consequently,
                        \begin{multline*}\edw{S^{(t+1)}_i}{Q^{\alpha+\alpha'}}=\edw{S^{(t)}_i
                            S^{(t)}_{i+1}}{Q^{\alpha} Q^{\alpha'}}\le
                            \edw{S^{(t)}_i}{Q^\alpha}+\edw{S^{(t)}_{i+1}}{Q^{\alpha'}}\\
                            \le
                            \edwl{S^{(t)}_i}{Q}+\Delta^{(t)}_i+\edwl{S^{(t)}_{i+1}}{Q}+\Delta^{(t)}_{i+1}\le
                            \edwl{S^{(t+1)}_i}{Q}+\Delta^{(t+1)}_i.
                        \end{multline*}
                    \item If $0 < i = s^{(t+1)}$, then $\edw{S^{(t)}_i}{Q^\alpha}\le
                        \edwl{S^{(t)}_i}{Q}+\Delta^{(t)}_i$ and $\edwp{S^{(t)}_{i+1}}{Q}\le
                        \edwl{S^{(t)}_{i+1}}{Q}+\Delta^{(t)}_{i+1}$ hold by the inductive
                        assumption for some integer $\alpha$. Consequently,
                        \begin{multline*}\edwp{S^{(t+1)}_i}{Q}=\edwp{S^{(t)}_i S^{(t)}_{i+1}}{Q}
                            \le \edw{S^{(t)}_i}{Q^\alpha}+\edwp{S^{(t)}_{i+1}}{Q}\\
                            \le
                            \edwl{S^{(t)}_i}{Q}+\Delta^{(t)}_i+\edwl{S^{(t)}_{i+1}}{Q}+\Delta^{(t)}_{i+1}\le
                            \edwl{S^{(t+1)}_i}{Q}+\Delta^{(t+1)}_i.
                        \end{multline*}
                    \item The analysis of~the case that $0 = i < s^{(t+1)}$ is symmetric to
                        that of~the above case---this can be seen by reversing all strings in
                        scope.
                    \item If $0 = i = s^{(t+1)}$, then the claim holds trivially.
                \end{itemize}
            \item For a type-\ref{it:type2} merge of~$S^{(t)}_{i-1}$, $S^{(t)}_i$, and
                $S^{(t)}_{i+1}$ into $S^{(t+1)}_{i-1}$, it suffices to prove that
                $S^{(t+1)}_{i-1}$ satisfies the claim.
                \begin{itemize} \item If $\edwl{S^{(t+1)}_{i-1}}{Q}\ge \edwl{S^{(t)}_i}{Q}+1$,
                    we observe that  $\edw{S^{(t)}_i}{Q^\alpha}\le
                    \edwl{S^{(t)}_i}{Q}+\Delta^{(t)}_i$ holds by the inductive assumption for
                    some integer $\alpha$. Consequently,
                    \begin{multline*}
                        \edw{S^{(t+1)}_{i-1}}{Q^{\alpha+2}}=\edw{QS^{(t)}_iQ}{QQ^{\alpha}Q}\le
                        \edw{S^{(t)}_i}{Q^\alpha}\\
                        \le \edwl{S^{(t)}_i}{Q}+\Delta^{(t)}_i \le
                        \edwl{S^{(t+1)}_{i-1}}{Q} - 1 +
                        \Delta^{(t)}_i=\edwl{S^{(t+1)}_{i-1}}{Q} + \Delta^{(t+1)}_{i-1}.
                    \end{multline*}
                    \item If $\edwl{S^{(t+1)}_{i-1}}{Q} < \edwl{S^{(t)}_i}{Q}+1$, then let $x' \le
                        y'$ denote integers that satisfy
                        $\edwl{S^{(t+1)}_{i-1}}{Q}=\edw{S^{(t+1)}_{i-1}}{Q^\infty\fragmentco{x'}{y'}}$.
                        This also yields integers $x'',y''$ with
                        $x'\le x'' \le y'' \le y'$ such that\begin{multline*}
                            \edw{S^{(t+1)}_{i-1}}{Q^\infty\fragmentco{x'}{y'}} \\=
                            \edw{Q}{Q^\infty\fragmentco{x'}{x''}}
                            +\edw{S^{(t)}_i}{Q^\infty\fragmentco{x''}{y''}}
                            +\edw{Q}{Q^\infty\fragmentco{y''}{y'}}.
                        \end{multline*}
                        Due to\begin{multline*}
                            \edwl{S^{(t)}_i}{Q}
                            \le \edw{S^{(t)}_i}{Q^\infty\fragmentco{x''}{y''}}\\
                            \le \edw{S^{(t+1)}_{i-1}}{Q^\infty\fragmentco{x'}{y'}}
                            = \edwl{S^{(t+1)}_{i-1}}{Q} < \edwl{S^{(t)}_i}{Q}+1,
                        \end{multline*}
                        we have
                        $\edw{Q}{Q^\infty\fragmentco{x'}{x''}}+\edw{Q}{Q^\infty\fragmentco{y''}{y'}}<1$.
                        As the string $Q$ is primitive, this means that $x',x'',y'',y'$
                        are all multiples
                        of~$|Q|$.
                        Consequently,
                        \begin{multline*}
                            \edw{S^{(t+1)}_{i-1}}{Q^{(y'-x')/|Q|}}
                            =\edw{S^{(t+1)}_{i-1}}{Q^\infty\fragmentco{x'}{y'}} \\=
                            \edwl{S^{(t+1)}_{i-1}}{Q} \le \edwl{S^{(t+1)}_{i-1}}{Q} +
                            \Delta^{(t-1)}_{i-1}.
                        \end{multline*}
                \end{itemize}
            \item For a type-\ref{it:type3} merge of~$S^{(t)}_{i-1}$ and $S^{(t)}_i$ into
                $S^{(t+1)}_{i-1}$, it suffices to prove that $S^{(t+1)}_{i-1}$ satisfies
                the claim.
                \begin{itemize}
                    \item If $\edwl{S^{(t+1)}_{i-1}}{Q}\ge\edwl{S^{(t)}_i}{Q}+1$, we
                        observe that  $\edwp{S^{(t)}_{i}}{Q}\le
                        \edwl{S^{(t)}_{i}}{Q}+\Delta^{(t)}_{i}$ holds by the inductive
                        assumption. Consequently,
                        \begin{multline*}
                            \edwp{S^{(t+1)}_{i-1}}{Q}=\edwp{QS^{(t)}_i}{Q} \le
                            \edwp{S^{(t)}_i}{Q}\\
                            \le \edwl{S^{(t)}_{i}}{Q}+\Delta^{(t)}_{i} \le
                            \edwl{S^{(t+1)}_{i-1}}{Q} - 1 +
                            \Delta^{(t)}_i=\edwl{S^{(t+1)}_{i-1}}{Q} +
                            \Delta^{(t+1)}_{i-1}.
                        \end{multline*}
                    \item If $\edwl{S^{(t+1)}_{i-1}}{Q}<\edwl{S^{(t)}_i}{Q}+1$, then let
                        $x' \le
                        y'$ denote integers that satisfy
                        $\edwl{S^{(t+1)}_{i-1}}{Q}=\edw{S^{(t+1)}_{i-1}}{Q^\infty\fragmentco{x'}{y'}}$.
                        This also yields an integer $x''$ with
                        $x'\le x'' \le y'$ such that\[
                            \edw{S^{(t+1)}_{i-1}}{Q^\infty\fragmentco{x'}{y'}} =
                            \edw{Q}{Q^\infty\fragmentco{x'}{x''}}
                            +\edw{S^{(t)}_i}{Q^\infty\fragmentco{x''}{y'}}.
                        \]
                        Due to\begin{multline*}
                            \edwl{S^{(t)}_i}{Q}
                            \le \edw{S^{(t)}_i}{Q^\infty\fragmentco{x''}{y'}}
                            \le \edw{S^{(t+1)}_{i-1}}{Q^\infty\fragmentco{x'}{y'}}\\
                            = \edwl{S^{(t+1)}_{i-1}}{Q} < \edwl{S^{(t)}_i}{Q}+1,
                        \end{multline*}
                        we have
                        $\edw{Q}{Q^\infty\fragmentco{x'}{x''}} < 1$.
                        As the string $Q$ is primitive, this means that $x',x''$ are both
                        multiples
                        of~$|Q|$.
                        Consequently,
                        \begin{multline*}
                            \edwp{S^{(t+1)}_{i-1}}{Q}
                            =\edw{S^{(t+1)}_{i-1}}{Q^\infty\fragmentco{x'}{y'}}=
                            \edwl{S^{(t+1)}_{i-1}}{Q}\\ \le \edwl{S^{(t+1)}_{i-1}}{Q} +
                            \Delta^{(t-1)}_{i-1}.
                        \end{multline*}
                \end{itemize}
            \item The analysis of~type-\ref{it:type4} merges is symmetrical to that of
                type-\ref{it:type3} merges---this can be seen by reversing all strings in
                scope.
        \end{enumerate}

        This completes the proof~of~the inductive step.
    \end{claimproof}

    Observe that if no merge rule can be applied to a partition $S=S_{0}^{(t)}\cdots
    S_{s^{(t)}}^{(t)}$,
    then $s^{(t)}=0$ or $\Delta_{0}^{(t)}= \cdots = \Delta_{s^{(t)}}^{(t)}=0$.
    Consequently, \cref{clm:locked} implies that all fragments $S_i$ in the final
    partition $S=S_0\cdots S_s$
    are locked.

    \begin{claim}\label{clm:locked_short}
        For each partition $S=S_0^{(t)}\cdots S_{\!s^{(t)}}^{(t)}$,
        the total length $\lambda^{(t)}$ of~interesting fragments satisfies\[
            \lambda^{(t)} + 2|Q|\sum_{i=0}^{s^{(t)}} \Delta_i^{(t)} \le
            (5|Q|+1)\edwl{S}{Q}.
        \]
    \end{claim}
    \begin{claimproof}
        We proceed by induction on $t$.
        In the base case of~$t=0$, each interesting fragment other than $S_{0}^{(0)}$ and
        $S_{s^{(0)}}^{(0)}$
        satisfies $\Delta_i^{(0)}>0$. Hence, the number of~interesting fragments is at
        most $2+\sum_{i=0}^{s^{(0)}} \Delta_i^{(0)}= 2+\edwl{S}{Q}$.
        Moreover, the length of~each fragment $S^{(0)}_i$ does not exceed
        $|Q|+\Delta_i^{(0)}$.
        Consequently,
        \[
            \lambda^{(0)} + 2|Q|\sum_{i=0}^{s^{(0)}} \Delta_i^{(0)} \le
            (2+\edwl{S}{Q})|Q|+(2|Q|+1)\sum_{i=0}^{s^{(0)}}\Delta_i^{(0)}\le(5|Q|+1)\,\edwl{S}{Q}.
        \]This completes the proof~in the base case.

        As for the inductive step, it suffices to prove that
        $\lambda^{(t+1)}+2|Q|\sum_{i=0}^{s^{(t+1)}} \Delta_i^{(t+1)} \le
        \lambda^{(t)}+2|Q|\sum_{i=0}^{s^{(t)}} \Delta_i^{(t)}$:
        \begin{itemize}
            \item For a type-\ref{it:type1} merge
                (where we merge two interesting fragments),
                we have
                \[\lambda^{(t+1)} +2|Q|\sum_{i=0}^{s^{(t+1)}} \Delta_i^{(t+1)}
                =  \lambda^{(t)} + 2|Q|\sum_{i=0}^{s^{(t)}} \Delta_i^{(t)}.\]
            \item For a type-\ref{it:type2}, type-\ref{it:type3}, or type-\ref{it:type4}
                merge
                (where we merge a fragment with its one or two non-interesting neighbors),
                we have
                \[\lambda^{(t+1)} +2|Q|\sum_{i=0}^{s^{(t+1)}} \Delta_i^{(t+1)}
                \le \lambda^{(t)} + 2|Q|+2|Q|\sum_{i=0}^{s^{(t+1)}} \Delta_i^{(t+1)}
                =\lambda^{(t)} + 2|Q|\sum_{i=0}^{s^{(t)}} \Delta_i^{(t)}.\]
        \end{itemize}
        Overall, we obtain the claimed bound.
    \end{claimproof}
    We conclude that the total length of~interesting fragments $S_i$ does not exceed
    $(5|Q|+1)\edwl{S}{Q}$.

    \begin{claim}\label{clm:locked_whole}
        We have $\edwl{S}{Q}=\sum_{i=0}^{s} \edwl{S_i}{Q}$.
    \end{claim}
    \begin{claimproof}
        The claim is immediate if $s = 0$; hence, assume that $s \ge 1$.
        Observe that the inequality $\sum_{i=0}^{s} \edwl{S_i}{Q}\le \edwl{S}{Q}$
        easily follows from disjointness of~fragments $S_i$;
        thus, we focus on proving $\edwl{S}{Q}\le \sum_{i=0}^{s} \edwl{S_i}{Q}$.

        For $0\le i \le s$, let $Q_i$ denote a substring of~$Q^\infty$ that satisfies
        $\edwl{S_i}{Q}=\edw{S_i}{Q_i}$.
        Since each $S_i$ is locked (by \cref{clm:locked}),
        we may assume that for $0 < i < s$ the substring $Q_i$ is a power of~$Q$,
        the substring $Q_s$ is a prefix of~a power of~$Q$,
        and the substring $Q_0$ is a suffix of~a power of~$Q$.
        Consequently, $Q_0\cdots Q_s$ is a substring of~$Q^\infty$,
        and we have\[
            \edwl{S}{Q} \le \edw{S_0\cdots S_s}{Q_0\cdots Q_s}
                       \le \sum_{i=0}^s \edw{S_i}{Q_i} = \sum_{i=0}^s \edwl{S_i}{Q},
                   \] thus completing the proof.
    \end{claimproof}

    The locked fragments created satisfy $\edwl{S}{Q}=\sum_{i=1}^\ell \edwl{L_i}{Q}$
    due to \cref{clm:locked_whole}.
    Moreover, since $\edwl{S_i}{Q}>0$ holds only for interesting fragments,
    \cref{clm:locked_short}
    yields $\sum_{i=1}^{\ell} |L_i| \le (5|Q|+1)\,\edwl{S}{Q}$, completing the proof.
\end{proof}

\begin{definition}[{\cite[see Definition~5.10]{ckw20}}]\label{def:klocked}
    Let $S$ denote a string, let $Q$ denote a primitive string,
    let $h\ge 0$ denote an integer, and let $w$ be a weight function.
    We say that a prefix $L$ of~$S$ is \emph{$h$-locked} (with respect to~$Q$)
    if at least one of~the following holds:
    \begin{itemize}
        \item For every $p\in \fragmentco{0}{|Q|}$, if $\edwp{L}{\rot^p(Q)}\le h$, then
            $\edwp{L}{\rot^p(Q)}=\edw{L}{Q^\infty\fragmentco{-p}{j|Q|}}$ for some \(j\in\Z\).
        \item We have $L=S$.
            \qedhere
    \end{itemize}
\end{definition}

\begin{lemma}[{\cite[see Lemma~5.11]{ckw20}}]
    \dglabel{lem:klocked_comb}[lem:locked]
    Let $S$ denote a string, let $Q$ denote a primitive string, let $h\ge 0$ be an
    integer, and $w$ be an integer weight function.
    There are disjoint locked fragments $L_1,\ldots,L_{\ell} \preceq S$,
    such that $L_1$ is a $k$-locked prefix of~$S$, $L_{\ell}$ is a suffix of~$S$,
    $\edwl{L_i}{Q} > 0$ for $1 < i < \ell$, \[
        \edwl{S}{Q}=\sum_{i=1}^{\ell} \edwl{L_i}{Q},\quad
        \text{and}\quad\sum_{i=1}^{\ell}|L_i| \le (5|Q|+1)\edwl{S}{Q}+2(h+1)|Q|.\]
\end{lemma}
\begin{proof}
    We proceed as in the proof~of~\cref{lem:locked} except that $\Delta_{0}^{(0)}$
    is artificially increased by $h+1$, the prefix $S_0$ in the final partition
    is included as $L_1$ among the locked fragments even if $\edwl{S_0}{Q}=0$,
    and the suffix $S_s$ is included as $L_{\ell}$ among the locked fragments even if $\edwl{S_s}{Q}=0$.

    It is easy to see that \cref{clm:locked,clm:locked_whole} remain satisfied,
    whereas the upper bound in \cref{clm:locked_short} is increased by $2(h+1)|Q|$.
    We only need to prove that $S_0$ is a $h$-locked prefix of~$S$.
    For this, we prove the following claim using induction.

    \begin{claim}\label{clm:klocked}
        For each partition $S=S_0^{(t)}\cdots S_{s^{(t)}}^{(t)}$, at least one of~the
        following holds.
        \begin{itemize}
                \item For every $p\in \fragmentco{0}{|Q|}$, if
                    $\edwp{S_0^{(t)}}{\rot^p(Q)}\le h-\Delta_0^{(t)}$, then
                    $\edw{S_0^{(t)}}{Q^\infty\fragmentco{|Q|-p}{j|Q|}} \le
                    \edwp{S_0^{(t)}}{\rot^p(Q)}+\Delta_0^{(t)}$ holds for some integer
                    $j$.
                \item We have $S_0^{(t)}=S$.
        \end{itemize}
    \end{claim}
    \begin{claimproof}
        We proceed by induction on $t$. In the base case of~$t=0$, the claim holds
        trivially since
        $\edwp{S_0^{(0)}}{\rot^p(Q)} \ge 0 > h-\Delta_0^{(0)}$ holds for every $p$ due to
        $\Delta_0^{(0)}\ge h+1$.

        As for the induction step, we assume that the claim holds for $S_0^{(t)}$ and we
        prove that it holds
        for $S_0^{(t+1)}$. The claim holds trivially if the merge rule applied did not
        affect $S_0^{(t)}$.
        Given that  $S_0^{(t)}$ is interesting by definition, the merges that might affect
        $S_0^{(t)}$ are of~type~\ref{it:type1} (if $S_1^{(t)}$ is interesting)
        or~\ref{it:type4} (otherwise).

        \begin{enumerate}
            \item Consider a type-\ref{it:type1} merge of~$S_0^{(t)}$ and $S_1^{(t)}$. If
                $s^{(t)}=1$, then $S_0^{(t+1)}=S$ satisfies the claim trivially.
                Hence, we may assume that $1 < s^{(t)}$ so that \cref{clm:locked} yields
                $\edw{S_1^{(t)}}{Q^{\alpha}}\le \edwl{S_1^{(t)}}{Q} + \Delta_{1}^{(t)}$ for some integer $\alpha$.
                Let us fix $p\in \fragmentco{0}{|Q|}$ with $\edwp{S_0^{(t+1)}}{\rot^p(Q)}\le h-\Delta_0^{(t+1)}$.
                Due to $\Delta_0^{(t+1)}\ge \Delta_0^{(t)}$, this yields
                $\edwp{S_0^{(t)}}{\rot^p(Q)}\le h-\Delta_0^{(t)}$, so the inductive assumption
                implies $\edw{S_0^{(t)}}{Q^\infty\fragmentco{|Q|-p}{j|Q|}} \le
                \edwp{S_0^{(t)}}{\rot^p(Q)}+\Delta_0^{(t)}$ for some integer $j$.
                Consequently,
                \begin{multline*}
                    \edw{S_0^{(t+1)}}{Q^\infty\fragmentco{|Q|-p}{(j+\alpha)|Q|}}
                    = \edw{S_0^{(t)}S_1^{(t)}}{Q^\infty\fragmentco{|Q|-p}{j|Q|}Q^{\alpha}} \\
                    \le \edw{S_0^{(t)}}{Q^\infty\fragmentco{|Q|-p}{j|Q|}} + \edw{S_1^{(t)}}{Q^{\alpha}}\\
                    \le \edwp{S_0^{(t)}}{\rot^p(Q)}+\Delta_0^{(t)}+\edwl{S_1^{(t)}}{Q} + \Delta_{1}^{(t)}
                    \le \edwp{S_0^{(t+1)}}{\rot^p(Q)} + \Delta_0^{(t+1)}.
                \end{multline*}
            \item Consider a type-\ref{it:type4} merge of~$S_0^{(t)}$ and $S_1^{(t)}$.
                Let us fix $p\in \fragmentco{0}{|Q|}$ with $\edwp{S_0^{(t+1)}}{\rot^p(Q)}\le h-\Delta_0^{(t+1)}$.
                \begin{itemize}
                    \item If $\edwp{S_0^{(t+1)}}{\rot^p(Q)}\ge\edwp{S_0^{(t)}}{\rot^p(Q)}+1$,
                        then \[\edwp{S_0^{(t)}}{\rot^p(Q)} \le\edwp{S_0^{(t+1)}}{\rot^p(Q)} -1 \le h-\Delta_0^{(t+1)}-1=
                        h-\Delta_0^{(t)}.\]
                        Hence, the inductive assumption implies
                        $\edw{S_0^{(t)}}{Q^\infty\fragmentco{|Q|-p}{j|Q|}} \le
                        \edwp{S_0^{(t)}}{\rot^p(Q)}+\Delta_0^{(t)}$ for some integer $j$.
                        Consequently,
                        \begin{multline*}
                            \edw{S_0^{(t+1)}}{Q^\infty\fragmentco{|Q|-p}{(j+1)|Q|}}
                            = \edw{S_0^{(t)}Q}{Q^\infty\fragmentco{|Q|-p}{j|Q|}Q}\\
                            \le \edw{S_0^{(t)}}{Q^\infty\fragmentco{|Q|-p}{j|Q|}}
                            \le \edwp{S_0^{(t)}}{\rot^p(Q)}+\Delta_0^{(t)}\\
                            \le \edwp{S_0^{(t+1)}}{\rot^p(Q)}-1 + \Delta_0^{(t)}
                            = \edwp{S_0^{(t+1)}}{\rot^p(Q)} + \Delta_0^{(t+1)}.
                        \end{multline*}
                    \item If $\edwp{S_0^{(t+1)}}{\rot^p(Q)}<\edwp{S_0^{(t)}}{\rot^p(Q)}+1$,
                        then let $y'$ denote an arbitrary integer that satisfies
                        $\edwp{S_0^{(t+1)}}{\rot^p(Q)}=\edw{S_0^{(t+1)}}{Q^\infty\fragmentco{|Q|-p}{y'}}$.
                        This also yields an integer $y''$ with
                        $|Q|-p\le y'' \le y'$ such that\[
                            \edw{S^{(t+1)}_{0}}{Q^\infty\fragmentco{|Q|-p}{y'}} =
                            \edw{S^{(t)}_0}{Q^\infty\fragmentco{|Q|-p}{y''}} +
                            \edw{Q}{Q^\infty\fragmentco{y''}{y'}}.
                        \]
                        Due to\begin{multline*}
                            \edwp{S_0^{(t)}}{\rot^p(Q)}
                            \le  \edw{S^{(t)}_0}{Q^\infty\fragmentco{|Q|-p}{y''}}
                            \le \edw{S^{(t+1)}_{0}}{Q^\infty\fragmentco{|Q|-p}{y'}}\\
                            = \edwp{S_0^{(t+1)}}{\rot^p(Q)}<\edwp{S_0^{(t)}}{\rot^p(Q)}+1,
                        \end{multline*}
                        we have
                        $\edw{Q}{Q^\infty\fragmentco{y''}{y''}}<1$.
                        As the string $Q$ is primitive, this means that $y'',y'$ are both multiples
                        of~$|Q|$.
                        Consequently,
                        \[\edw{S^{(t+1)}_{0}}{Q^\infty\fragmentco{|Q|-p}{y'}} =
                        \edwp{S_0^{(t+1)}}{\rot^p(Q)} \le \edwp{S_0^{(t+1)}}{\rot^p(Q)} +
                    \Delta^{(t+1)}_{0}.\]
                \end{itemize}
        \end{enumerate}
        This completes the proof~of~the inductive step.
    \end{claimproof}
    Given that the final partition $S=S_{0}^{(t)}\cdots S_{s^{(t)}}^{(t)}$
    satisfies $s^{(t)}=0$ or $\Delta_{0}^{(t)}= 0$, we conclude that $S_0$ is indeed $h$-locked.
\end{proof}

\begin{lemma}[{\tt Locked($S$, $Q$, $d$, $h$,$w$)}, {see \cite[Lemma 6.9]{ckw20}}]
    \dglabel{lem:klocked}[lem:klocked_comb,lem:witness,lem:locked]
    Let $S$ denote a string, let $Q$ denote a primitive string,  let $w: \Esigma^2 \to
    \fragment{0}{W}$ be an integer weight function,
    let $d$ denote a positive integer such that $\edwl{S}{Q}\le d$
    and $|S| \ge (2d+1)|Q|$, and let $h\in \Zz$.

    Then, there is an algorithm that computes
    disjoint locked fragments $L_1,\ldots,L_{\ell} \preceq S$
    such that
        \begin{itemize}
        \item $S=L_1 \cdot \bigodot_{i=1}^{\ell-1} (Q^{\alpha_i} L_{i+1})$ for
            positive integers $\alpha_1,\ldots,\alpha_{\ell-1}$;
        \item $L_1$ is an $h$-locked prefix of $S$ and $L_{\ell}$ is a suffix of~$S$;
        \item
            \(\edwl{S}{Q}=\sum_{i=1}^{\ell} \edwl{L_i}{Q}\) and  $\edwl{L_i}{Q} > 0$ for $i\in \fragmentoo{1}{\ell}$; and
        \item \(\displaystyle
                \sum_{i=1}^\ell |L_i| \le (5|Q|+1)\edwl{S}{Q} + 2(h+1)|Q|.
            \)
    \end{itemize}
    The algorithm takes $\Ohtilde(d^2\cdot W+h)$ time in the \pillar model.
\end{lemma}
\begin{proof}
    We implement the process from the proof of \cref{lem:klocked_comb}.

    First, we use \cref{lem:witness} to construct a $w$-optimal alignment $\A : S \onto
    Q^\infty\fragmentco{x}{y}$ between $S$ and a substring of~$Q^\infty$.
    Then, based on the alignment $\A$, we construct a decomposition $S=S_0^{(0)}\cdots
    S^{(0)}_{s^{(0)}}$
    such that $S^{(0)}_i$ is aligned against \[
        Q^{(0)}_i :=
        Q^\infty\fragmentco{\max(x,|Q|(\lceil{x/|Q|\rceil}+i-1))}{\min(|Q|(\lceil{x/|Q|\rceil}+i),y)}
    \] by alignment $\A$, and a sequence $\Delta^{(0)}_i$ such that we have
     $\Delta^{(0)}_i=\edw{S^{(0)}_i}{Q^{(0)}_i}$ for $i>0$ and
        $\Delta^{(0)}_0=\edw{S^{(0)}_0}{Q^{(0)}_0}+h+1.$
    Since this sequence might be long, we only generate \emph{interesting} fragments~$S^{(0)}_i$
    and store them, along with the values $\Delta^{(0)}_i$, in a queue $F$ in
    left-to-right order.
    (Recall that $S^{(t)}_i$ is interesting if $i=0$, $i=s^{(t)}$, $S_i^{(t)}\ne Q$, or
    $\Delta_i^{(t)}>0$.)

    The process of constructing the interesting fragments $S^{(0)}_i$ is somewhat tedious.
    We maintain a fragment $Q^\infty\fragmentco{\ell_Q}{r_Q}$, interpreted as $Q^{(0)}_i$
    for increasing values of $i$, a fragment $S\fragmentco{\ell_S}{r_S}$, interpreted as a
    candidate for $S^{(0)}_i$,
    and an integer $\Delta$, interpreted as $\Delta^{(0)}_i$.
    They are initialized to $Q^\infty\fragmentco{x}{|Q|\lceil{x/|Q|\rceil}}$,
    $S\fragmentco{0}{Q|\lceil{x/|Q|\rceil}-x}$, and $k+1$, respectively.

    Next, we process pairs $(s,q)$ corresponding to subsequent errors in the alignment $A$.
    The interpretation of the $j$-th pair $(s,q)$ is that $S\fragmentco{0}{s}$ is aligned
    with $Q^\infty\fragmentco{x}{x+q}$ with $j$ errors so that the $j$-th error is a
    substitution of $S\position{s}$ into $Q^\infty\position{x+q}$, and insertion of
    $Q^\infty\position{x+q}$, or a deletion of $S\position{s}$.

    The first step of processing $(s,q)$ is only performed if
    $Q^\infty\fragmentco{x}{x+q}$ is not (yet)
    contained in $Q^\infty\fragmentco{\ell_Q}{r_Q}$.
    If this is not the case, then we push $S\fragmentco{\ell_S}{r_S}$ with budget $\Delta$
    to the queue $F$ of interesting fragments, and we update the maintained data:
    The fragment $Q^\infty\fragmentco{\ell_Q}{r_Q}$ is set to be the fragment of~$Q^\infty$
    matching $Q$ and containing $Q^\infty\position{x+q}$;
    between the previous and the current value of $Q^\infty\fragmentco{\ell_Q}{r_Q}$,
    there are zero or more copies of $Q$ aligned in $A$ without error. Hence, we skip the
    same number of copies of $Q$ in $S$
    (these are the uninteresting fragments $S^{(0)}_i$)
    and set $S\fragmentco{\ell_S}{r_S}$ to be the subsequent fragment of length $|Q|$.
    Finally, the budget $\Delta$ is reset to $0$.

    In the second step, we update $r_S$ according to the type of the currently processed error:
    We increment~$r_S$ in case of deletion of $S\position{s}$ and
    we decrement $r_S$  in case of insertion of
    $Q^\infty\position{x+q}$. This way, we guarantee that
    $|S\fragmentoo{s}{r_S}|=|Q^\infty\fragmentoo{x+q}{r_Q}|$,
    and that $A$ aligns $S\fragmentco{\ell_S}{r_S}$ with
    $Q^\infty\fragmentco{\ell_Q}{r_Q}$ provided
    that we have already processed all errors involving $Q^\infty\fragmentco{\ell_Q}{r_Q}$.
    Additionally, we increase $\Delta$ to acknowledge the currently processed error between
    $S\fragmentco{\ell_S}{r_S}$ and $Q^\infty\fragmentco{\ell_Q}{r_Q}$.

    In a similar way, we process $(s,q)=(|S|,y)$, interpreting it as extra substitution.
    This time, however, we do not increase $\Delta$ (because this is a not a real error).
    Finally, we push $S\fragmentco{\ell_S}{|S|}=S^{(0)}_{s^{(0)}}$ with budget~$\Delta$ to
    the queue $F$.

    In the second phase of the algorithm, we transform the decomposition
    $S=S_0^{(0)}\cdots S^{(0)}_{s^{(0)}}$ and the sequence $\Delta_0^{(0)}\cdots
    \Delta^{(0)}_{s^{(0)}}$
    using the four types of merge operations described in the proof of \cref{lem:locked}.

    We maintain an invariant that a stack $L$ contains already processed interesting fragments,
    all with budget equal to $0$, in left-to-right order (so that the top of $L$
    represents the rightmost one), while $F$ contains fragments that have not been
    processed yet (and may have positive budgets) also in the left-to-right order (so that
    the front of $F$ represents the leftmost one).
    Additionally, the currently processed fragment $S\fragmentco{\ell}{r}$ is guaranteed
    to be to the right of all fragments in $L$ and to the left of all fragments in~$F$.
    The fragments in $L$, the fragment $S\fragmentco{\ell}{r}$, and the fragments in $F$
    form the sequence of all interesting fragments in the current decomposition
    $S=S_0^{(t)}\cdots S^{(t)}_{s^{(t)}}$.

    In each iteration of the main loop, we pop the front fragment $S\fragmentco{\ell}{r}$
    with budget $\Delta$ from the queue~$F$ and exhaustively perform merge operations
    involving it:
    We first try applying a type-\ref{it:type1} merge with the fragment to the left (which
    must be on the top of $L$).
    If this is not possible, we type applying a type-\ref{it:type1} merge with the
    fragment to the right (which must be on the front of $F$).
    If also this is not possible, then $S\fragmentco{\ell}{r}$ is surrounded by
    uninteresting fragments.
    In this case, we perform a type-\ref{it:type2}, type-\ref{it:type3}, or \ref{it:type4}
    merge provided that $\Delta > 0$. Otherwise, we push $S\fragmentco{\ell}{r}$ to $L$
    and proceed to the next iteration.

    Finally, the algorithm returns the sequence of (locked) fragments represented in the stack $L$.

    The correctness of the algorithm follows from \cref{lem:klocked}; no deep insight is needed
    to prove that our implementation indeed follows the procedure described in the proof
    of \cref{lem:locked} and extended in the proof of \cref{lem:klocked}.

       As the alignment $\A$ is of size $|\A|\le d$, initially, the are
    $\Oh(d)$ interesting fragments with total budget $\Oh(d+h)$.
    Each subsequent iteration decreases
    the number of interesting locked fragments or their total budget, so there are
    at most $\Oh(d + k)$ iterations in total. Overall the algorithm runs in $\Oh(d+k)$ time
    in the \modelname model, on top of the cost of constructing $\A$, which is $\Ohtilde(d^2 W)$ due to \cref{lem:witness}.
\end{proof}
\begin{corollary}
    \dglabel{cor:klocked}[lem:verify,lem:klocked]
    Consider an instance of \SM, where $w$ is an integer metric weight function.
    The values \(d^w_P \coloneqq \edwl{P}{Q}\) and \(d^w_T \coloneqq \edwl{T}{Q}\) satisfy $d_w^P+d_w^T = \Oh(kW)$
    and can be computed in $\Ohtilde(k^2 W^3)$ time.
    Moreover, for any parameter $\lpref\in \Zp$, the sets \(\mathcal{L}^P \coloneqq
    \locked(P,Q,d^w_P,\lpref,w)\) and \(\mathcal{L}^T \coloneqq \locked(T,Q,d^w_T,0,w)\)
    can be constructed in $\Ohtilde(k^2 W^3 + \lpref)$ time.
\end{corollary}
\begin{proof}
    Since the weights are upper-bounded by $w$,
    we have $d^w_P = \edwl{P}{Q} \le W\cdot \edwl{P}{Q} = W\cdot d_P = \Oh(kW)$ and
    analogously for $d^w_T$.
    The values $d^w_P$ and $d^w_T$ can be computed in $\Ohtilde(k^2W^3)$ time using
    \cref{lem:verify}.
    The cost of using \cref{lem:klocked} is $\Ohtilde(k^2W^3+\lpref)$.
\end{proof}

\subsection{Analyzing the Text Using Locked Fragments}

In what follows, we adapt the marking scheme of \cite{ckw22} to \SM.

\begin{definition}[{\cite[see Definition~5.5]{ckw22}}]
    \dglabel{def:marking}
    For a weight function $w:\Esigma^2\to \fragment{0}{W}$, a text \(T\), a pattern \(P\),
    a primitive string \(Q\) with \(\edwl{P}{Q} =
    d^w_P\) and \(\edwl{T}{Q} = d^w_T\), and corresponding sets of locked
    fragments \(\mathcal{L}^P \coloneqq \locked(P,Q,d^w_P,\lpref,w)\) and \(\mathcal{L}^T
    \coloneqq \locked(T,Q,d^w_T,0,w)\),
    write \(\markf\) for the function that
    maps an integer~$v$ to a (weighted) number of
    locked fragments in \(\mathcal{L}^P\) that
    (almost) overlap locked fragments in \(\mathcal{L}^T\) when aligning \(P\) to
    position $v$.

    Formally, we first define a function
    \(\markf: \Z \times \Zz \times \mathcal{L}^P\times  \mathcal{L}^T \to \Zz \) by
    \begin{multline*}
        \markf(v,\ktotm, L^P = P\fragmentco{\ell_P}{r_P}, L^T = T\fragmentco{\ell_T}{r_T}
        )\\
         \coloneqq \begin{cases}
            \edwl{L^T}{Q},&\text{if } \ell_P=0 \text{ and }  v\in \fragmentoo{\ell_T -
            \ktotm-r_P}{r_T + \ktotm-\ell_P}\text{;}\\
            \min\{ \edwl{L^T}{Q}, \edwl{L^P}{Q} \}, &\text{if } \ell_P \neq 0 \text{ and }
            v\in \fragmentoo{\ell_T - \ktotm-r_P}{r_T + \ktotm-\ell_P} \text{;}\\
        0,& \text{otherwise.}
            \end{cases}
        \end{multline*}
     Now, set
        \[\markf(v,\ktotm,\mathcal{L}^P,\mathcal{L}^T)
        \coloneqq \sum_{L^P \in \mathcal{L}^P} \sum_{L^T \in \mathcal{L}^T}
        \markf(v,\ktotm, L^P, L^T ).\tag*{\qedhere}\]
\end{definition}

When $\ktotm$, $\mathcal{L}^P$, and $\mathcal{L}^T$ are clear from context we may just
write
$\markf(v)$ for $\markf(v,\ktotm,\mathcal{L}^P,\mathcal{L}^T)$.

We continue with a set of useful observations about our marking scheme.
Fix an instance of \SM and sets of locked
fragments \(\mathcal{L}^P \coloneqq \locked(P,Q,d^w_P,\lpref,w)\) and \(\mathcal{L}^T \coloneqq \locked(T,Q,d^w_T,0,w)\),
with $\lpref = 2kW$.
Let $\mu \coloneqq \max(d^w_P,d^w_T,\lpref,1) = \cO(k W)$.

\begin{lemma}[{\cite[see Lemma~5.6]{ckw22}}]
    \dglabel{lem:mark-pos}
    For every $\ktotm\in \Zz$, we have
    \[\sum_{v\in \Z} \markf(v,\ktotm,\mathcal{L}^P,\mathcal{L}^T)
    = \cO(W^2 \cdot k^2\cdot (\ktotm + |Q|)).\]
\end{lemma}
\begin{proof}
    Fix a locked fragment
    \(L^P = P\fragmentco{\ell_P}{r_P}\in \mathcal{L}^P\) and
    a locked fragment
    \(L^T = P\fragmentco{\ell_T}{r_T}\in \mathcal{L}^T\).
    If $\ell_P\ne 0$, then the definition of \(\markf\) yields
    \begin{align*}
        \sum_{v\in \Z}
        \markf(v,\ktotm, L^P, L^T)
            & \le \sum_{v\in \fragmentoo{\ell_T-\ktotm -r_P}{r_T+\ktotm-\ell_P}}
            \min\{ \edwl{L^T}{Q}, \edwl{L^P}{Q} \} \\
            & \le \min\{ \edwl{L^T}{Q}, \edwl{L^P}{Q} \} (|L^T| + |L^P| + 2\ktotm)\\
            & \le \edwl{L^T}{Q}(|L^P| + 2\ktotm) + \edwl{L^P}{Q}|L^T|.
    \end{align*}
    Similarly, if $\ell_P =0$, then the definition of \(\markf\) yields
    \begin{align*}
        \sum_{v\in \Z}
        \markf(v,\ktotm, L^P, L^T)
            & \le \sum_{v\in \fragmentoo{\ell_T-\ktotm -r_P}{r_T+\ktotm-\ell_P}}
            \edwl{L^T}{Q} \\
            & \le \edwl{L^T}{Q}(|L^T| + |L^P| + 2\ktotm)\\
            & \le \edwl{L^T}{Q}(|L^P| + 2\ktotm) + d^w_T|L^T|.
    \end{align*}
    Overall, we have
    \begin{align*}
        \sum_{v\in \Z} \markf(v,\ktotm,\mathcal{L}^P,\mathcal{L}^T)& \le \sum_{L^P \in \mathcal{L}^P} \sum_{L^T \in \mathcal{L}^T} (\edwl{L^T}{Q}(|L^P| + 2\ktotm)+\edwl{L^P}{Q}|L^T|) + d^w_T\sum_{L^T \in \mathcal{L}^T}|L^T|\\
                                                                   & \le d^w_T\sum_{L^P \in \mathcal{L}^P}(|L^P|+2\ktotm) + (d^w_P+d^w_T)\sum_{L^T \in \mathcal{L}^T}|L^T|\\
                                                                   & \le d^w_T(d^w_P+2)2\ktotm + d^w_T\sum_{L^P \in \mathcal{L}^P}|L^P| + (d^w_P+d^w_T)\sum_{L^T \in \mathcal{L}^T}|L^T|\\
                                                                   & = \Oh(\mu^2 \ktotm + d^w_T(d^w_T+\lpref)|Q|+(d^w_P+d^w_T)d^w_T|Q|) \\
                                                                   & = \Oh(\mu^2 (\ktotm+|Q|))\\
                                                                   & = \Oh(k^2 W^2 (\ktotm + |Q|)).
    \end{align*}
\end{proof}

\begin{definition}\label{def:restrset}
    For a set $\mathcal{U}$ of fragments of a string $S$ and an interval $I\subseteq \mathbb{Z}$,
    we write
    $\res{\mathcal{U}}{I} \coloneqq \{S\fragmentco{x}{y} \in \mathcal{U} : \fragmentco{x}{y} \cap I \neq \emptyset\}$.
\end{definition}

\begin{definition}[{\cite[Definition~5.8]{ckw22}}]\dglabel{def:light}
    For any fixed thresholds $\ktotm\in \Zz$ and $\tradeoff\in \Zp$, we say that a position \(v\in \fragment{0}{n-m+k}\) is
    \emph{light} if the following conditions are simultaneously satisfied:
    \begin{itemize}
        \item $\markf(v, \ktotm,\mathcal{L}^P,\mathcal{L}^T) < \tradeoff$,
        \item $\res{\mathcal{L}^T}{\fragmentco{v-\ktotm}{v+\ktotm}}=\emptyset$,
        \item $\res{\mathcal{L}^T}{\fragmentco{v+m-\ktotm}{v+m+\ktotm}}=\emptyset$, and
        \item $\res{\mathcal{L}^T}{\fragmentco{v+|L_1^P|-\ktotm}{v+|L_1^P|+\ktotm}}=\emptyset$.
    \end{itemize}
    Otherwise, the position $v\in \fragment{0}{n-m+k}$ is called \emph{heavy}.
    We denote the sets of heavy and light positions by $\Hv$ and $\light$, respectively.
\end{definition}

\SetKwFunction{heavy}{Heavy}
\begin{lemma}[$\protect\heavy(P,T,k,w,Q,\LP,\LT,\ktotm,\tradeoff)$, {\cite[see Lemma~5.9]{ckw22}}]
    \dglabel{lem:heavy-alg}[def:marking,def:light,lem:witness]
    Consider an instance of \SM with an integer metric weight function,
    families $\LP = \locked(P,Q,d^w_P,\lpref,w)$ and
    $\LT = \locked(T,Q,d^w_T,0,w)$, as well as thresholds $\ktot\in \Zz$ and $\tradeoff\in
    \Zp$,

    The set $\Hv$ of heavy positions, represented as the union of $\Oh(k^2 W^2)$ disjoint
    integer ranges,
    can be computed in $\Ohtilde(k^2 W^3)$ time.
\end{lemma}
\begin{proof}
    We implement the marking process according to \cref{def:marking}, assigning
    $\tradeoff$ extra marks to all positions made heavy due to the last three items in
    \cref{def:light}.
    Formally, we produce the following $\Oh(\mu^2)=\Oh(k^2 W^2)$ weighted intervals:
    \begin{itemize}
        \item $\fragmentoo{\ell_T-\ktotm-r_P}{r_T+\ktotm-\ell_P}$ of weight $\min\{
            \edwl{L^T}{Q}, \edwl{L^P}{Q} \}$ for each locked fragment
            $L^P=P\fragmentco{\ell_P}{r_P}\in \LP$ with $\ell_P\ne 0$ and
            $L_T=T\fragmentco{\ell_T}{r_T}\in \LT$,
        \item $\fragmentoo{\ell_T-\ktotm-r_P}{r_T+\ktotm-\ell_P}$ of weight $\edwl{L^T}{Q}$
            for each locked fragment $L^P=P\fragmentco{\ell_P}{r_P}\in \LP$ with
            $\ell_P=0$ and $L_T=T\fragmentco{\ell_T}{r_T}\in \LT$,
        \item $\fragmentoo{\ell_T-\ktotm}{r_T+\ktotm}$ of weight $\tradeoff$ for each
            locked fragment $L_T=T\fragmentco{\ell_T}{r_T}\in \LT$,
        \item $\fragmentoo{\ell_T-\ktotm-m}{r_T+\ktotm-m}$ of weight $\tradeoff$ for each
            locked fragment $L_T=T\fragmentco{\ell_T}{r_T}\in \LT$,
        \item $\fragmentoo{\ell_T-\ktotm-|L_1^P|}{r_T+\ktotm-|L_1^P|}$ of weight
            $\tradeoff$ for each locked fragment $L_T=T\fragmentco{\ell_T}{r_T}\in \LT$.
    \end{itemize}
    The heavy positions are exactly those positions in $\fragment{0}{n-m+k}$ that are
    contained in intervals of total weight at least $\tradeoff$;
    they can be computed using a sweep-line procedure with events corresponding to
    interval endpoints.
    The number of events is $\Oh(k^2W^2)$, so the output consists of $\Oh(k^2W^2)$ disjoint
    ranges.
    In terms of the running time, the bottleneck is computing the locked costs; each such
    cost $c$ can be computed in $\Ohtilde(c^2 W)$ time using \cref{lem:witness}, and this
    sums up to $\Ohtilde(k^2 W^3)$ because $\sum c = \Oh(\mu)= \Oh(kW)$.
\end{proof}

\begin{lemma}[{\cite[see Lemma~5.10]{ckw22}}]
    \dglabel{lem:heavy-bound}[lem:mark-pos,def:light,lem:heavy-alg]
    The total number of heavy positions does not exceed
    \[|\Hv|
               =\cO((kW/\tradeoff + 1)\cdot kW \cdot (\ktotm+|Q|)).\]
    Moreover, for any integer $b\in \Zp$,
    \[\left|\bigcup_{v\in \Hv}\fragment{v-b}{v+b} \right|
        = \cO((kW/\tradeoff + 1)\cdot kW \cdot (\ktotm+b+|Q|)).\]
\end{lemma}
\begin{proof}
    By \cref{lem:mark-pos}, the number of positions violating the first condition of
    \cref{def:light}
    does not exceed $\cO((kW)^2\cdot (\ktotm+|Q|)/\tradeoff)$.
    As for the remaining conditions, each locked fragment $L^T\in \mathcal{L}^T$ may yield
    at most $3(|L^T|+2\ktotm)$
    heavy positions, for a total of:
    \[\sum_{L^T\in \mathcal{L}^T} 3(|L^T|+2\ktotm)
    	= \Oh((|Q|+\ktotm)d^w_T) = \Oh(kW(|Q|+\ktotm)).\]

    As for the second claim, observe that if $v$ is heavy, then all positions in
    $\fragment{v-b}{v+b} \cap \fragment{0}{n-m+k}$ would
    be heavy if we increased the threshold $\ktotm$ to $\ktotm+b$.
    This is because the left endpoint of all intervals considered in the proof of
    \cref{lem:heavy-alg} contains a $-\ktotm$ term whereas the right endpoint contains a
    $+\ktotm$ term (and there is no other dependency on $\ktotm$).
    Further, $\left|\bigcup_{v\in \Hv}\fragment{v-b}{v+b} \setminus \fragment{0}{n-m+k}
    \right| \leq 2b$, since $\Hv \subseteq \fragment{0}{n-m+k}$.
\end{proof}

Further, consider a partition of the positions of $T$ in $\fragment{0}{n-m+k}$ into heavy
and light using \cref{lem:heavy-alg}.
First, we show how to compute $(k,w)$-error occurrences that start at heavy positions;
this is a straightforward application of \cref{fct:rj,fct:verifyRj}.
Then, we reduce the problem of computing $(k,w)$-error occurrences that start at light
positions
to an instance $\swaps(T, P, k, w, Q, \A_T, \A_P, \rred{P}, \rred{T}, \alpha)$, where
$\rred{P}$ and $\rred{T}$
are both of size $\cO(kW)$ and contain all the internal pieces that (almost) overlap
locked fragments,
and $\eta \in \fragment{1}{kW}$.

\subsection{Computing Occurrences Starting at Heavy Positions}
\label{sus:heavy}

\begin{lemma}[\texttt{HeavyMatches($P$, $T$, $k$, $w$, $Q$, $\A_P$, $\A_T$, $\Hv$)},
    {\cite[see Lemma 6.1]{ckw22}}]
    \dglabel"{lem:heavy-total}[lem:heavy-alg,fct:rj,lem:verify,lem:heavy-bound]
    Given an instance of the \SM problem
    and the set $\Hv$ of heavy positions constructed using \cref{lem:heavy-alg},
    $\OccW_k(P,T)\cap \Hv$ can be computed in $\Ohtilde((kW/{\tradeoff}+1)k^3W^3)$ time in
    the \pillar model.
\end{lemma}
\begin{proof}
    Recall that the set $\Hv$ is represented as the union of $\Oh(k^2W^2)$ disjoint heavy
    ranges
    (listed in the left-to-right order).

    Our algorithm starts with an application \cref{fct:rj} to construct the sequence
    $(r_j)_{j\in J}$, represented as a concatenation of $\Oh(k)$ arithmetic progressions.
    Our first goal is to enumerate elements of the set $J' \coloneqq \{j\in J :
    \fragmentco{r_j}{r_{j+1}}\cap \Hv \ne \emptyset\}$.
    For this, we simultaneously traverse the sequence $(r_j)_{j\in J}$ along with the
    heavy ranges constituting $\Hv$.
    For each heavy range $H \subseteq \Hv$, we list all $j\in J$ such that
    $\fragmentco{r_j}{r_{j+1}}\cap H \ne \emptyset$
    (only the smallest such $j$ might have already been listed for an earlier heavy
    range).
    In the second phase, we compute $O\coloneqq\bigcup_{j\in J'} (r_j+\OccW_k(P,R_j))\cap
    \fragmentco{r_j}{r_{j+1}}$
    using \cref{lem:verify}, making sure that the positions are listed in the
    left-to-right order.
    Finally, we simultaneously scan~$O$ and the heavy ranges constituting $\Hv$, reporting
    all positions $v\in O\cap \Hv$.

    As for correctness, observe that $R_j=T\fragmentco{r_j}{r'_j}$,
    so $O \subseteq \OccW_k(P,T)$ and, due to the final filtering step, we output a subset of $\OccW_k(P,T)\cap \Hv$.
    To prove the converse inclusion, consider a position $v\in \OccW_k(P,T)\cap \Hv$.
    A combination of \cref{fct:rj} with \cref{fact:simple} yields that there exists $j\in
    J$ such that $v\in (r_j+\OccW_k(P,R_j))\cap \fragmentco{r_j}{r_{j+1}}$
    and, by the definition of $J'$, we also have $j\in J'$. Consequently, $v$ is indeed reported.

    As for the complexity analysis, let us first compute the running time in terms of $|J'|$.
    Constructing the sequence $(r_j)_{j\in J}$ costs $\Oh(k)$ time.
    The set $J'$ can be computed in $\Oh(|J'|+k+k^2W^2)=\Oh(|J'|+k^2W^2)$ time.
    The applications of \cref{lem:verify} cost $\Ohtilde(k^2 W)$ time each (in the \pillar
    model),
    for a total of $\Ohtilde(|J'|\cdot k^2W^2)$ time in the \pillar model.
    A (crude) upper bound on the output size is $\cO(|J'|k^2W^2)$, so the final filtering
    step works in $\Oh((1+|J'|)k^2W^2)$ time.
    Overall, the running time in the \pillar model is $\Oh((1+|J'|)k^2W^2)$.

    It remains to bound $|J'|$. Observe that each interval $\fragmentco{r_j}{r_{j+1}}$
    (possibly except for the last one with $j=\max J$) has length at most $\tau+d_T$.
    Moreover, the total length of any $s$ intervals $\fragmentco{r_j}{r_{j+1}}$ with $0<r_j < r_{j+1} < n$
    is at least $s\tau-d_T$.
    On top of that, there might be one non-empty interval with $r_j=0$ and one non-empty interval with $r_{j+1}=n$.
    We conclude that
    \[ |J'| \le 2 + \frac{\sum_{j\in J'\setminus \{\max J\}} |\fragmentco{r_j}{r_{j+1}}|+d_T}{\tau}
    \le 2+\frac{\left|\bigcup_{v\in \Hv} \fragment{v-\tau-d_T}{v+\tau+d_T}\right|+d_T}{\tau}.\]
    Since $\tau = \Theta(\max(|Q|,k))$ and $d_T,\ktotm=\Oh(kW)$,
    the bound of \cref{lem:heavy-bound} yields that $|J'|=\Oh(({kW}/{\tradeoff}+1)kW^2)$.
    Hence, the total running time is $\Ohtilde((kW/{\tradeoff}+1)k^3W^3)$ just as claimed.
\end{proof}

\subsection{Computing Occurrences Starting at Light Positions}
\label{sus:light}

\paragraph*{Combinatorial Insights}

The following fact follows directly from~\cref{def:light}.
\begin{fact}[{\cite[see Fact 6.2]{ckw22}}]\dglabel{fact:restrlckd}
    For any light position $v$ in $T$ and for any $x\in
    \fragmentoo{v+m-\ktotm}{v+m+\ktotm}$, we have
    $\res{\mathcal{L}^T}{\fragmentco{v}{x}}=\{T\fragmentco{\ell}{r} \in \mathcal{L}^T :
    \fragmentco{\ell}{r} \subseteq \fragmentco{v}{v+m}\}$.
\end{fact}

Let $\A_T^w : T \onto Q^\infty\fragmentco{x_T^w}{y_T^w}$ be a $w$-optimal alignment of
cost $\edwl{T}{Q}$.
For each position~$v$ of~$T$, set $\rho(v)\in \fragment{x_T^w}{y_T^w}$ such that
$\A_T^w(T\fragmentco{v}{n})=\Q\fragmentco{\rho(v)}{y_T^w}$.
Observe that
\begin{align*}
    \edwl{T}{Q}&=
        \edw{T\fragmentco{0}{v}}{\Q\fragmentco{x_T^w}{\rho(v)}}
        +\edw{T\fragmentco{v}{n}}{\Q\fragmentco{\rho(v)}{y_T^w}}\\
               &=\edws{T\fragmentco{0}{v}}{\rot^{-\rho(v)}(Q)}
           +\edwp{T\fragmentco{v}{n}}{\rot^{-\rho(v)}(Q)}.
\end{align*}

\begin{lemma}[{\cite[see Lemma 6.3]{ckw22}}]
    \dglabel{fact:rot}
    For any two positions $v<x$ of $T$ such that
    $\res{\mathcal{L}^T}{\position{v}}=\res{\mathcal{L}^T}{\position{x}}=\emptyset$,
    we have
    \[
        \edwl{T\fragmentco{v}{x}}{Q}=\edw{T\fragmentco{v}{x}}{\Q\fragmentco{\rho(v)}{\rho(x)}}
    = \sum_{L \in \Lck^T_{\fragmentco{v}{x}}} \edwl{L}{Q}.\qedhere\]
\end{lemma}
\begin{proof}
    Set $\res{\mathcal{L}^T}{\fragmentco{v}{x}} = \{L^T_j \in \mathcal{L}^T: j \in
    \fragment{j_1}{j_2}\}$, and observe
    that all elements of this set are fragments of $T\fragmentco{v}{x}$.
    Now, we have
    \begin{align*}
        & \edw{T\fragmentco{v}{x}}{\Q\fragmentco{\rho(v)}{\rho(x)}}\\
        & \quad = \edw{T}{\Q\fragmentco{x_T^w}{y_T^w}} - \edw{T\fragmentco{0}{v}}{\Q\fragmentco{x_T^w}{\rho(v)}}\\
        & \qquad\qquad\qquad - \edw{T\fragmentco{x}{n}}{\Q\fragmentco{\rho(x)}{y_T^w}}\\
        & \quad \leq d_T - \sum_{j=1}^{j_1-1} \edwl{L^T_j}{Q} - \sum_{j=j_2+1}^{\ell^T} \edwl{L^T_j}{Q}\\
        & \quad = \sum_{j=j_1}^{j_2} \edwl{L^T_j}{Q}\\
        & \quad \leq \edwl{T\fragmentco{v}{x}}{Q}\\
        & \quad \leq \edw{T\fragmentco{v}{x}}{\Q\fragmentco{\rho(v)}{\rho(x)}}.\qedhere
    \end{align*}
\end{proof}

\begin{lemma}[{\cite[see Lemma 6.4]{ckw22}}]
    \dglabel{lem:ub}[lem:klocked,fact:restrlckd,fact:rot]
    Consider a light position $v$ of $T$.
    If $\edwp{L_1^P}{\rot^{-\rho(v)}(Q)} < \lpref$, then there is an $x\in
    \fragment{v}{n}$
    such that
    \[\edw{P}{T\fragmentco{v}{x}}\leq
    \edwp{L_1^P}{\rot^{-\rho(v)}(Q)}+\sum_{i=2}^{\ell^P}\edwl{L^P_i}{Q}+\sum_{L\in\res{\mathcal{L}^T}{\fragmentco{v}{v+m}}
    }\edwl{L}{Q}.\qedhere\]
\end{lemma}
\begin{proof}
    We first prove two auxiliary claims.

    \begin{claim}
        \label{claim:rotate}
        We have
        $\edwp{P}{\rot^{-\rho(v)}(Q)} =
        \edwp{L_1^P}{\rot^{-\rho(v)}(Q)}+\sum_{i=2}^{\ell^P}\edwl{L^P_i}{Q} <
        \lpref+d^w_P$.
    \end{claim}
    \begin{claimproof}
        The inequality $\edwp{P}{\rot^{-\rho(v)}(Q)} \geq
        \edwp{L_1^P}{\rot^{-\rho(v)}(Q)}+\sum_{i=2}^{\ell^P}\edwl{L^P_i}{Q}$ holds
        trivially.

        By~\cref{lem:klocked}, we have $P=L_1^P Q^{\alpha_1} L_2^P Q^{\alpha_2} \cdots
        Q^{\alpha_{\ell^P-1}} L_{\ell^P}^P$ for some non-negative integers $\alpha_i$.
        Since $L_1^P$ is $\lpref$-locked, we have
        $\edwp{L_1^P}{\rot^{-\rho(v)}(Q)}=\edw{L_1^P}{Q^\infty\fragmentco{\rho(v)}{jq}}$ for some non-negative integer $j$.
        Further, for each $i \in \fragmentoo{1}{\ell^P}$,
        we have $\edwl{L^P_i}{Q}=\edw{L^P_i}{Q^{\beta_i}}$ for some non-negative integer $\beta_i$.
        Finally, we have
        $\edwl{L^P_{\ell^P}}{Q}=\edw{L^P_{\ell^P}}{Q^\infty\fragmentco{0}{y}}$ for some
        non-negative integer $y$.

        Set $\gamma = \alpha_1+\sum_{i=2}^{\ell^P-1}(\alpha_i+\beta_i)$.
        The above discussion implies that there is an alignment of $P$ with the prefix
        $Q\fragmentco{\rho(v)}{jq}Q^\gamma Q^\infty\fragmentco{0}{y}$
        of $Q\fragmentco{\rho(v)}{jq}Q^\infty$ that costs
        $\edwp{L_1^P}{\rot^{-\rho(v)}(Q)}+\sum_{i=2}^{\ell^P}\edwl{L^P_i}{Q}$, thus
        proving
        that $\edwp{P}{\rot^{-\rho(v)}(Q)} \leq
        \edwp{L_1^P}{\rot^{-\rho(v)}(Q)}+\sum_{i=2}^{\ell^P}\edwl{L^P_i}{Q}$.
        This concludes the proof of the claimed equality.

        The claimed inequality holds since $\edwp{L_1^P}{\rot^{-\rho(a)}(Q)}< \lpref$ and
        $\sum_{i=2}^{\ell^P}\edwl{L^P_i}{Q}\leq d^w_P$.
    \end{claimproof}

    \begin{claim}
        \label{claim:rhow}
        There is an $x\in \fragmentoo{v+m-\ktotm}{v+m+\ktotm}\cap \fragmentco{0}{n}$ such
        that $\edwp{P}{\rot^{-\rho(v)}(Q)}=\edw{P}{Q^\infty
        \fragmentco{\rho(v)}{\rho(x)}}$.
    \end{claim}
    \begin{claimproof}
        Recall that $\Lck^T_{\fragmentco{v+m-\ktotm}{v+m+\ktotm}} = \emptyset$
        holds because $v$ is a light position of $T$.
        Since $\Lck^T$ contains a suffix of $T$, we conclude that either $n \le
        v+m-\ktotm$ or $v+m+\ktot \le n$.
        The former case contradicts $v \le n-m+k < n-m+\ktotm$, so
        $T\fragmentco{v+m-\ktotm}{v+m+\ktotm}$ must be a fragment of $T$ disjoint with all
        locked fragments in $\Lck^T$.
        This means that $\A_T^w$ matches $T\fragmentco{v+m-\ktotm}{v+m+\ktotm}$ without
        any edits to $\A_T^w(T\fragmentco{v+m-\ktotm}{v+m+\ktotm})$.
        Since $(v,\rho(v))\in \A_T^w$ and the cost of $\A_T^w$ is $d^w_T$, the fragment
        $\A_T^w(T\fragmentco{v+m-\ktotm}{v+m+\ktotm})$ must contain
        $\Q\fragmentco{\rho(v)+m-\ktotm+d^w_T}{\rho(v)+m+\ktotm-d^w_T}$.
        Consequently, for every $y\in
        \fragmentoo{\rho(v)+m-\ktotm+d^w_T}{\rho(v)+m+\ktotm-d^w_T}$, there exists $x\in
        \fragmentoo{v+m-\ktotm}{v+m+\ktotm}\cap \fragmentco{0}{n}$ such that $y=\rho(x)$.
        Since $\edwp{P}{\rot^{-\rho(v)}(Q)}\le \lpref + d^w_P < \ktotm - d^w_T$ holds by
        \cref{claim:rotate},
        such $x$ exists in particular for $y$ chosen so that
        $\edwp{P}{\rot^{-\rho(v)}(Q)}=\edw{P}{Q^\infty \fragmentco{\rho(v)}{y}}$.
    \end{claimproof}

    We get the desired bound by combining
    \cref{claim:rhow,claim:rotate,fact:restrlckd,fact:rot} via the triangle inequality;
    observe that \cref{fact:rot} is applicable because $\Lck^T_{\position{x}}\subseteq
    \Lck^T_{\fragmentco{v+m-\ktotm}{v+m+\ktotm}}=\emptyset$.
    \begin{align*}
    &\edw{P}{T\fragmentco{v}{x}}\\
    &\quad \leq \edw{P}{Q^\infty \fragmentco{\rho(v)}{\rho(x)}} +
        \edw{T\fragmentco{v}{x}}{Q^\infty \fragmentco{\rho(v)}{\rho(x)}}\\
    &\null\hspace{25em} (\text{symmetry, triangle inequality})\\
    &\quad = \edwp{P}{\rot^{-\rho(v)}(Q)} + \edw{T\fragmentco{v}{x}}{Q^\infty
    \fragmentco{\rho(v)}{\rho(x)}}\\
    &\null\hspace{25em} (\text{\cref{claim:rhow}})\\
    &\quad = \edwp{L_1^P}{\rot^{-\rho(v)}(Q)}+\sum_{i=2}^{\ell^P}\edwl{L^P_i}{Q} +
        \sum_{L\in \res{\mathcal{L}^T}{\fragmentco{v}{x}}}\edwl{L}{Q}\\
    &\null\hspace{25em} (\text{\cref{claim:rotate,fact:rot}})\\
    &\quad = \edwp{L_1^P}{\rot^{-\rho(v)}(Q)}+\sum_{i=2}^{\ell^P}\edwl{L^P_i}{Q} +
    \sum_{L\in \res{\mathcal{L}^T}{\fragmentco{v}{v+m}}}\edwl{L}{Q}.\\
    &\null\hspace{25em} (\text{\cref{fact:restrlckd}})
    \end{align*}
    This completes the proof of the lemma.
\end{proof}

For a position $v$ and a locked fragment $L \in \mathcal{L}^P \cup \mathcal{L}^T$,
we write $\markf(v,L)$ for the number of marks placed in
$v$ due to pairs of locked fragments that contain $L$; formally:
\begin{align*}
\markf(v,L)= \begin{cases}
    \sum_{L^T \in \Lck^T} \markf(v,\ktotm ,L,L^T) & \text{if }L\in \Lck^P;\\
    \sum_{L^P \in \Lck^P} \markf(v,\ktotm ,L^P,L) & \text{if }L\in \Lck^T.\\
\end{cases}
\end{align*}

\begin{definition}[{\cite[see Definition 6.7]{ckw22}}]
    For a light position $v$ of $T$, set
    \[\mathcal{D}(v) \coloneqq \{L^P_1\} \cup \{L \in \mathcal{L}^P\cup \mathcal{L}^T_{\fragmentco{v}{v+m}} : \markf(v,L) < \edwl{L}{Q}\}.
    \qedhere\]
\end{definition}

Let us now provide some intuition on what follows.
Consider an alignment $\mathcal{A}: P \onto T\fragmentco{v}{w}$, where $v$ is a light position of $T$
and the cost of $\mathcal{A}$ with respect to $w$ is not larger than $k$.
In the next lemma, we essentially lower-bound the cost of the restriction of such an alignment $\A$
to each locked fragment. For instance, we lower-bound $\edw{L}{\A(L)}$ for each $L\in \mathcal{L}^P$.
Our lower bound is positive only for elements of $\mathcal{D}(v)$.
Consider some locked fragment $L \in \mathcal{D}(v) \cap \mathcal{L}^P$ other than $L^P_1$.
Roughly speaking, at most $\markf(v,L)$ errors of $L$ with a fragment $Q'$ of $Q^\infty$
cancel out with errors between $\mathcal{A}(L)$ and $Q'$, yielding a lower bound
$\edw{L}{\A(L)} \ge \edwl{L}{Q} - \markf(v,L)$.
Then, the definition of $\mathcal{D}(v)$ guarantees that,
for any $L \in \mathcal{D}(v) \cap \mathcal{L}^P$, $\A(L)$ is disjoint from all $L' \in \mathcal{D}(v) \cap \mathcal{L}^T$.
One can exploit this property to obtain a lower bound for the cost of~$\A$ by showing that we can
sum over the lower bounds for individual locked fragments in $\mathcal{D}(v)$.
In fact, we use this reasoning to lower bound the cost of an alignment between two other strings,
obtained from $P$ and $T$, respectively, via the deletion of some fragments;
this happens in the proof of \cref{lem:nofakelight} in \cref{sec:redundantsqrtk}.

\begin{lemma}[{\cite[see Lemma 6.8]{ckw22}}]
    \dglabel{lem:lb_local}[fact:rot,lem:klocked]
    Consider a light position $v$ of $T$
    and a locked fragment $L \in \mathcal{L}^P \cup \mathcal{L}^T_{\fragmentco{v}{v+m}}$.
    \begin{enumerate}
        \item If $L = T\fragmentco{v+\ell}{v+r} \in \mathcal{L}^T$, then every $U\in
            \{P\fragmentco{i}{j} : i \geq \ell-k \text{ and } j \leq r+k \}$ satisfies
            \[ \edw{L}{U}\geq \edwl{L}{Q} - \markf(v,L).\]
        \item If $L = P\fragmentco{\ell}{r} \in \mathcal{L}^P\setminus \{L^P_1\}$, then
            every $U\in \{T\fragmentco{i}{j} : i \geq v+\ell-k \text{ and } j \leq v+r+k
            \}$ satisfies
            \[\edw{L}{U} \geq \edwl{L}{Q} - \markf(v,L).\]
        \item If $L$ is $L^P_1$, then for every $U\in  \{ T\fragmentco{v}{v+j} : j \in
            \fragment{|L|-k}{|L|+k}\}$ satisfies
            \[\edw{L}{U}\geq \edwp{L}{\rot^{-\rho(v)}(Q)} - \markf(v,L).\]
    \end{enumerate}
\end{lemma}
\begin{proof}
    Consider some $L \in \mathcal{L}^P\cup \mathcal{L}^T_{\fragmentco{v}{v+m}}$.
    If $L\in \mathcal{L}^T$, let $Y=P$; otherwise, let $Y=T$.
    Further, let $U \coloneqq Y\fragmentco{u_1}{u_2}$ denote
    any of the fragments specified in the statement of the lemma, and let
    \begin{align*}
        \zeta \coloneqq
        \begin{cases}
            \edwp{L_1^P}{\rot^{-\rho(v)}(Q)} & \text{if } L = L^P_1,\\
            \edwl{L}{Q} & \text{otherwise.}
        \end{cases}
    \end{align*}

    Let $x$ and $y$ denote integers that satisfy
    $\edw{U}{Q^\infty\fragmentco{x}{y}}=\edwl{U}{Q}$.
    In the case where $L=L^P_1$,
    we choose $x=\rho(v)$ and $y=\rho(u_2)$ so that $Q^\infty\fragmentco{x}{y}$ is a
    prefix of $\rot^{-\rho(v)}(Q)^\infty$.
    This is allowed by~\cref{fact:rot} because
    $\res{\mathcal{L}^T}{\fragmentco{v-\ktotm}{v+\ktotm}}=\res{\mathcal{L}^T}{\fragmentco{v+|L_1^P|-\ktotm}{v+|L_1^P|+\ktotm}}=\emptyset$,
    and it ensures that the following inequality holds in all three cases (the inequality
    holds trivially if $L\ne L^P_1$):
    \begin{equation}\label{eq:rot}
        \edw{L}{Q^\infty\fragmentco{x}{y}} \geq \zeta.
    \end{equation}
    Further, our marking scheme
    implies
    \begin{equation}\label{eq:dom}
        \markf(v,L) \geq \sum_{K \in \mathcal{L}^Y_{\fragmentco{u_1}{u_2}}} \edwl{K}{Q}.
    \end{equation}
    Let $Y\fragmentco{u_1'}{u_2'}$ be the fragment that is covered by
    $Y\fragmentco{u}{w}$ and the elements of $\mathcal{L}^Y_{\fragmentco{u_1}{u_2}}$.
    We have
    \begin{equation}\label{eq:extlock}
        \edw{U}{Q^\infty\fragmentco{x}{y}} = \edwl{U}{Q} \leq
        \edwl{Y\fragmentco{u_1'}{u_2'}}{Q} =\sum_{K \in
        \mathcal{L}^Y_{\fragmentco{u_1}{u_2}}} \edwl{K}{Q},
    \end{equation}
    where the last equality follows from the properties of locked fragments as computed by
    \cref{lem:klocked}.

    We are now ready to prove the claimed inequality.
    \begin{align*}
        \edw{L}{U} & \geq |\edw{L}{
        Q^\infty\fragmentco{x}{y}}-\edw{U}{Q^\infty\fragmentco{x}{y}}|
                   & \text{(symmetry, triangle inequality)}\\
                 &  \geq    \edw{L}{Q^\infty\fragmentco{x}{y}}-\edw{U}{Q^\infty\fragmentco{x}{y}}\\
                 &  \geq  \zeta -\edw{U}{Q^\infty\fragmentco{x}{y}}
                 &\text{(due to~\eqref{eq:rot})}\\
                 &  \geq  \zeta - \sum_{K \in \mathcal{L}^Y_{\fragmentco{i}{j}}}
                 \edwl{K}{Q}
                 &\text{(due to~\eqref{eq:extlock})}\\
                 &  \geq  \zeta - \markf(v,L).
                 &\text{(due to \eqref{eq:dom})}
    \end{align*}
    This completes the proof of the lemma.
\end{proof}

\paragraph*{Shrinking Runs of Plain Pairs}\label{sec:redundantsqrtk}

We set $\textsf{Red}(P)$ to be equal to
\[\textsf{Special}(P)\cup \{P_i : i\in\fragmentoo{1}{z} \text{ and }
\exists_{P\fragmentco{\ell}{r} \in \mathcal{L}^P}
\fragmentco{\ell-\tbd\ktotm}{r+\tbd\ktotm} \cap \fragmentco{p_i}{p_{i+1}+\Delta} \neq
\emptyset\}.\]
Similarly, we set $\textsf{Red}(T)$ to be equal to
\[\textsf{Special}(T)\cup \{T_i : i\in\fragmentoo{\min J + 1}{\max J + z} \text{ and }
\exists_{T\fragmentco{\ell}{r} \in \mathcal{L}^T}
\fragmentco{\ell-\tbd\ktotm}{r+\tbd\ktotm} \cap \fragmentco{t_i}{t_{i+1}+\Delta} \neq
\emptyset\}.\]
Next, we upper-bound the sizes of these two sets and show how to construct them
efficiently.

\begin{lemma}[{\cite[see Lemma 6.9]{ckw22}}]
    \dglabel{lem:fewreds}[lem:specialbound-simpl,lem:klocked,fct:rj]
    Given $P$, $T$, $\A_P$, $\A_T$, $\mathcal{L}^P$, and $\mathcal{L}^T$,
    the sets $\textsf{Red}(P)$ and $\textsf{Red}(T)$ are of size $\cO(kW^2)$ and can be
    constructed in time $\cOtilde(kW^2)$ in the \pillar model.
\end{lemma}
\begin{proof}
    We start with upper-bounding the size of each of the sets $\textsf{Red}(P)$ and
    $\textsf{Red}(T)$.
    $\textsf{Special}(P)$ and $\textsf{Special}(T)$ are of size $\cO(d)=\Oh(k)$ due to
    \cref{lem:specialbound-simpl}.
    The task is therefore to bound, for each of $P$ and~$T$, the number of internal pieces
    that are within $\tbd \ktotm=\Oh(kW)$ positions of a locked fragment.

    Let $(S,\beta_S)$ denote either of $(P,\beta_P)$ or $(T,\beta_T)$ and
    $\mathcal{L}^S=\{L\fragmentco{\ell_j}{h_j} : j \in \fragment{1}{\ell^S}\}$.
    It suffices to upper bound the number of tiles $S\fragmentco{s_{i-1}}{s_i}$ with $i
    \in \fragmentoo{1}{\beta_S}$
    in the $\tau$-tile partition of~$S$ (with respect to $\A_S$) such that
    $\fragmentco{s_{i-1}}{s_i}$ overlaps $\fragmentco{\ell_j-\Delta - \tbd \ktotm}{h_j+
    \tbd \ktotm}$ for some $j\in \fragment{1}{\ell^S}$;
    that is,
    \[|\{i \in \fragmentoo{1}{\beta_S} : \exists_{j \in \fragment{1}{\ell^S}}
    \fragmentco{\ell_j-\Delta -\tbd \ktotm}{h_j + \tbd \ktotm} \cap
\fragmentco{s_{i-1}}{s_{i}} \neq \emptyset\}|.\]

    Intuitively, we extend each locked fragment in $\mathcal{L}^S$ by $\Delta + \tbd
    \ktotm$ characters to the left and by $\tbd \ktotm$ characters to the right,
    thus covering at most $\|\mathcal{L}^S\|+ \ell^S \cdot
    (2\cdot\tbd+6)\ktotm=\|\mathcal{L}^S\|+ \tbc\ell^S \ktotm$ positions of $S$, since
    $\Delta=6 \ktot \leq 6 \ktotm$.
    Let us first upper bound the size $\phi$ of the set $\Phi$ of tiles (other than the
    first and last ones) that are fully covered by these ``extended'' locked fragments,
    that is, $\Phi \coloneqq \{i \in \fragmentoo{0}{\beta_S} : \fragmentco{s_{i-1}}{s_{i}}
    \subseteq \bigcup_{j  \in \fragment{1}{\ell^S}} \fragmentco{\ell_j-\Delta-\tbd
\ktotm}{h_j+\tbd \ktotm}\}$.
    We have
    \[\sum_{i \in \Phi} (\tau - |S\fragmentco{s_{i-1}}{s_i}|) \leq \sum_{i \in
    \Phi}\edu{S\fragmentco{s_{i-1}}{s_i}}{Q\fragmentco{(i-1)\tau}{i\tau}} \leq
\edl{S}{Q}\]
    and hence
    \[\phi \cdot \tau - \edl{S}{Q} \leq \sum_{i \in \Phi} |S\fragmentco{s_{i-1}}{s_i}|
    \leq \|\mathcal{L}^S\|+\tbc\ell^S \ktotm,\]
    which is equivalent to
    \[\phi \leq \frac{\|\mathcal{L}^S\|+\tbc\ell^S\ktotm +\edl{S}{Q}}{\tau}.\]
    Finally, we have to account for the at most $2\ell^S$ tiles that overlap ``extended''
    locked fragments but are not fully contained in them; we have at most two such tiles
    for each $L \in \mathcal{L}^S$.

    Since $\ell^S = \cO(\edwl{S}{Q})$ and $\|\mathcal{L}^S\|=\cO(\edwl{S}{Q} \cdot q)$
    by~\cref{lem:klocked},
    $\edwl{S}{Q} = \cO(kW)$, and $\tau = \Theta(\max\{k,q\})$,
    we have
    \[2\ell^S + \frac{\|\mathcal{L}^S\|+\tbc \ell^S\ktotm + \edl{S}{Q}}{\tau}  =
        \cO\left(kW + \frac{kWq + (kW)^2 + k}{\tau}\right)
                                                                            = \cO(kW^2).\]

    Let us now show how to efficiently construct the sets in scope.
    First, recall that $\textsf{Special}(P)$ and $\textsf{Special}(T)$ can be constructed
    in $\cO(kW^2)$ time in the \pillar model due to \cref{lem:specialbound-simpl}.
    The remaining elements of $\rred{P}$ (resp.~$\rred{T}$) can be computed in a
    simultaneous left-to-right scan of:
    \begin{itemize}
        \item the representation of starting positions of internal pieces $p_i$
            (resp.~$t_i$) as $\cO(k)$ arithmetic progressions, which can be computed in
            $\cO(k)$ time due 				to \cref{fct:rj};
        \item the $\cO(kW^2)$ locked fragments of $P$ (resp.~$T$) sorted with respect to their starting positions.
    \end{itemize}
    The $\Ohtilde(kW^2)$ time required for sorting the starting positions of the locked
    fragments is the bottleneck of the algorithm in the \pillar model.
\end{proof}

\begin{definition}[{\cite[see Definition 6.10]{ckw22}}]\dglabel{def:fjgj}
    For $j \in J$, set $\I'_j
    \coloneqq (F_{j,1},G_{j,1})\cdots(F_{j,z_j}, G_{j,z_j})
    \coloneqq \Comp(\I_j, \rred{P}, \rred{T}, 2\tradeoff + \slack)$,
    and \(F_j \coloneqq \val_{\Delta}( F_{j,1},\ldots, F_{j,z_j} )\) and
    \(G_j \coloneqq \val_{\Delta}( G_{j,1},\ldots, G_{j,z_j} )\).

    Further, let $F_{j,i}$ and $G_{j,i}$ correspond to the fragments
    $F_j\fragmentco{f_{j,i}}{f_{j,i+1}+\Delta}$ and $G_j\fragmentco{g_{j,i}}{g_{j,i+1}+\Delta}$,
    respectively, that is, we set $f_{j,1}=g_{j,1}=0$ and, for $i\in \fragment{2}{z_j+1}$,
    $f_{j,i}=\left(\sum_{x<i} |F_{j,x}| \right)-\Delta$ and $g_{j,i}=\left(\sum_{x<i}
    |G_{j,x}| \right)-\Delta$.
\end{definition}

The focus of the remainder of~\cref{sec:redundantsqrtk} is to prove the following lemma.

\begin{restatable*}[{\cite[see Lemma 6.11]{ckw22}}]{lemma}{correctlight}
    \dglabel{lem:correctlight}[cor:supset,lem:nofakelight]
    \[\OccW_k(P,T) \cap \light
    = \big(\bigcup_{j \in J} (r_j +  \OccW_k(\I'_j))\big) \cap \light
    =
    \big(\bigcup_{j \in J} (r_j +  \OccW_k(F_j,G_j))\big) \cap \light.\qedhere\]
\end{restatable*}

The combination of \cref{fct:rj} and \cref{lem:redundantk} directly yields the following.

\begin{corollary}\dglabel{cor:supset}[fct:rj,lem:redundantk]
    We have $\big(\bigcup_{j \in J} (r_j +  \OccW_k(F_j,G_j))\big) \supseteq \big(\bigcup_{j
    \in J} (r_j +  \OccW_k(P,R_j))\big)=\OccW_k(P,T)$.
    \lipicsEnd
\end{corollary}

See \cite[Example 6.13]{ckw22} an example that illustrates that,
for some $j\in J$ and $p\in \OccE_k(F_j,G_j)$,
we might have $r_j + p \not\in \OccE_k(P,T)$ if $r_j+p\in \Hv$.

It remains to show that, for any light position $p \in \big(\bigcup_{j \in J} (r_j +
\OccW_k(F_j,G_j))\big)$,
we have $p \in \OccW_k(P,T)$.
The following lemma demonstrates that, for each $j \in J$, we can restrict our attention
to a subset of the positions of $R_j$.

\begin{lemma}[{\cite[see Lemma 6.14]{ckw22}}]
    \dglabel{lem:focus3k}[fact:lengths_diff,lem:lb_local]
    Consider some $j \in J$ and a position $p$ of $R_j$ such that $r_j+p$ is a light position of $T$.
    If $\edwp{L_1^P}{\rot^{-\rho(r_j+p)}(Q)} \geq \lpref$, then $p \not\in \OccW_k(F_j,G_j)$.
\end{lemma}
\begin{proof}
    By the definition of $\rred{T}$, it follows that $|F_j| \geq |L^P_1|$ and
    any prefix of $F_j$ of length at most $|L^P_1|+\tbd \ktotm$ is also a prefix of $P$.
    Combined with the upper bound on the sum of length-differences from \cref{fact:lengths_diff},
    this also implies that any prefix of $G_j$ of length at most $|L^P_1| + \tbd \ktotm - 3\ktot$
    is also a prefix of $R_j$.
    Further, recall that we have $|R_j|\leq |P| + 3\ktot - k$ due to \cref{fact:lengths_diff}.
    Consequently, since $|R_j|-|F_j| = |P| - |G_j|$, we have
    \[
        \OccW_k(F_j,G_j) \subseteq \fragment{0}{|G_j|-|F_j|+k} \subseteq \fragment{0}{3\ktot}.
    \]

    It thus suffices to consider $p \in \fragment{0}{3\ktot}$.
    Let $W$ denote a prefix of $G_j\fragmentco{p}{|G_j|}$ that satisfies
    \[
        \edw{F_j\fragmentco{0}{|L^P_1|}}{W} = \min_t \edw{F_j\fragmentco{0}{|L^P_1|}}{G_j\fragmentco{p}{t}}.
    \]
    We distinguish between two cases.
    \begin{itemize}
        \item If $|W| \in \fragment{|L^P_1|-k}{|L^P_1|+k}$, then
            $W$ is a prefix of $T\fragmentco{r_j+p}{n}$ since
            \[p+|W| \leq 3\ktot+|L^P_1|+k \leq |L^P_1| + \tbd \ktotm - 3\ktot.\]
            Thus, by~\cref{lem:lb_local} and since $\lpref = \Theta(kW)$, we have
            \begin{align*}
                \edw{L^P_1}{W} &\geq \edwp{L_1^P}{\rot^{-\rho(r_j+p)}(Q)} -
                \markf(r_j+p,L^P_1)\\
                               &\geq \lpref - \markf(r_j+p,L^P_1)\\
                               &= 2kW - \tradeoff > kW  \ge k.
            \end{align*}
        \item Otherwise, we have $\edwp{F_j\fragmentco{0}{|L^P_1|}}{W}\geq \big||W|-|L^P_1|\big| > k$.
    \end{itemize}

    To conclude the proof, it suffices to observe that \[
        \min_{t} \edw{F_j}{G_j\fragmentco{p}{t}} \geq \min_{t} \edw{L^P_1}{G_j\fragmentco{p}{t}} > k,
    \]
    and hence $p \not\in \OccW_k(F_j, G_j)$.
\end{proof}

In what follows, for convenience, we assume that, for each run of plain pairs of $\I_j$
that has been trimmed, the deleted pairs correspond to a suffix of this run.
This yields a natural mapping from pairs
of $\I'_j$ to pairs of $\I_j$.

\begin{definition}[{\cite[see Definition 6.15]{ckw22}}]
    For each $j\in J$ and each $i \in \fragment{1}{z_j}$, let $\orig(j,i)$ denote the
    number of pairs to the left of pair $(F_{j,i}, G_{j,i})$
    that were  deleted in the process of obtaining $\I'_j$ from $\I_j$.
    We say that pair $(F_{j,i}, G_{j,i})$ \emph{originates} from pair $(P_{i+\orig(j,i)}, T_{j,i+\orig(j,i)})$.
\end{definition}

Let each internal pair of pieces $(F_{j,i}, G_{j,i})$ inherit the color of
$(P_{i+\orig(j,i)}, T_{j, i+\orig(j,i)})$. In addition, mark the first and the last pairs
as not plain.

\begin{definition}[{\cite[see Definition 6.16]{ckw22}}]
    For $j\in J$ and $i\in \fragment{1}{z_j}$, we say that $F_j\fragmentco{x_1}{x_2}$
    (or~$G_j\fragmentco{x_1}{x_2}$) has an \emph{overlap} with a pair $(F_{j,i},G_{j,i})$
    if and only if
    $\fragmentco{x_1}{x_2} \cap \fragmentco{f_{j,i}}{f_{j,i+1}+\Delta} \neq \emptyset$
    (or ~$\fragmentco{x_1}{x_2} \cap \fragmentco{g_{j,i}}{g_{j,i+1}+\Delta} \neq \emptyset$).
\end{definition}

\begin{definition}[{\cite[see Definition 6.17]{ckw22}}]
    Consider some $j \in J$ and a contiguous sequence
    $\mathcal{M}=(F_{j,i_1},G_{j,i_1})\cdots (F_{j,i_2},G_{j,i_2})$ of pairs in $\I'_j$
    that are either all plain or all not plain (that is, \(\mathcal{M}\) is
    \emph{monochromatic}).
    Fix an $X \in \{F_j,G_j\}$.
    For $i \in \fragment{1}{z_j}$, if $X=F$, we set
    $x_{j,i}=f_{j,i}$; otherwise, we set
    $x_{j,i}=g_{j,i}$.

    For a non-negative integer $\gamma$, we say that a fragment
    $X\fragmentco{x_1}{x_2}$ of $X \in \{F_j,G_j\}$ is
    $\gamma$-\emph{contained} by $\mathcal{M}$ when
    the following two conditions are satisfied: (a)
    $x_1 \geq x_{j,i_1} + \gamma$ or $i_1 = 1$
    and (b)
    $x_2 \leq x_{j,i_2+1}+\Delta - \gamma$ or $i_2 = z_j$.
\end{definition}

\begin{fact}[{\cite[see Fact 6.18]{ckw22}}]\dglabel{obs:Dcont}
    If a fragment $U$ of $F$ or $G$ is $\Delta$-contained by a monochromatic
    sequence~$\mathcal{M}$ of contiguous pairs in $\I'_j$,
    then $U$ only overlaps pairs in $\mathcal{M}$.
\end{fact}

The proof of the following lemma is identical to that of \cite[Lemma 6.19]{ckw22}.

\begin{lemmaq}[{\cite[see Lemma 6.19]{ckw22}}]\dglabel{lem:monochrom}
    For $j\in J$, consider a monochromatic sequence $\mathcal{M}=(F_{j,i_1},G_{j,i_1})\cdots
    (F_{j,i_2},G_{j,i_2})$ of contiguous pairs in $\I'_j$.
    Let $\{X, Y\} = \{F_j, G_j\}$ and $X\fragmentco{x_1}{x_2}$ be a fragment of $X$ that is
    $\gamma$-contained by $\mathcal{M}$ for $\gamma \geq 13 \ktot$.
    All the pairs of $\I'_j$ that overlap $Y\fragmentco{y_1}{y_2} \coloneqq
    Y\fragmentco{\max\{0,x_1-4\ktot\}}{\min\{|Y|,x_2+4\ktot\}}$ belong to $\mathcal{M}$.
\end{lemmaq}

\begin{lemma}[{\cite[see Lemma 6.19]{ckw22}}]
    \dglabel{lem:nofakelight}[def:marking,lem:lb_local,lem:monochrom,lem:focus3k,lem:ub,obs:Dcont]
    Consider some $j\in J$ and a position $p \in \OccW_k(F_j,G_j)$ such that $r_j+p$ is a light position of $T$.
    Then, $r_j + p \in \OccW_k(P,T)$.
\end{lemma}
\begin{proof}
    First, observe that we may assume that $\I'_j \neq \I_j$; otherwise, the statement follows trivially.
    To avoid clutter, we drop the subscript $j$ when referring to $F_j$ and $G_j$ and
    simply call them $F$ and~$G$, respectively.

    Let $b$ denote a position of $G$ and let
    $\mathcal{B}: F \onto G\fragmentco{p}{b}$ denote an alignment of cost $\min_t \edw{F}{G\fragmentco{p}{t}} \leq k$.
    We intend to show that, in this case, there exist a position $c$ of $R_j$ and an
    alignment $\mathcal{C}: P \onto R_j\fragmentco{p}{c}$ of the same cost.

    Set \[
        \res{\mathcal{L}^T}{\fragmentco{r_j+p}{r_j+p+m}}=\{L^T_i : i \in \fragment{i_1}{i_2}\}.
    \]
    Observe that there is a natural mapping of each locked fragment $L^P_i \in
    \mathcal{L}^P$ to a fragment of $F$, which we denote by~$L^{F}_i$;
    we denote the set of fragments in the image of this mapping by~$\mathcal{L}^F$.
    Similarly, there is a natural mapping of each locked fragment $L^T_i \in \{L^T_i : i
    \in \fragment{i_1}{i_2}\}$ to a fragment of~$G$, which we denote by~$L^{G}_i$;
    we denote the set of fragments in the image of this mapping by~$\mathcal{L}^G$.
    Let $\mathcal{D}'(p)$ consist of the images of the locked fragments in $\D(r_j+p)$ under these mappings.
    Observe that each fragment $L \in \mathcal{L}^F \cup \mathcal{L}^G$ only overlaps
    non-plain pairs.
    For a fragment $L^X_y \in \mathcal{L}^X$, where $X \in \{P,T,F,G\}$, let $L^X_y=X\fragmentco{\ell^X_y}{r^X_y}$.
    The following claim follows instantly.

    \begin{claim}[{\cite[see Claim 6.21]{ckw22}}]\label{claim:lockcont}
        Each fragment $L \in \mathcal{L}^F \cup \mathcal{L}^G$ is $\tbd \ktotm$-contained by a sequence
        of contiguous non-plain pairs in~$\I'_j$.
        \lipicsClaimEnd
    \end{claim}

    We next essentially show that, if we were to mark positions of $G_j$
    based on overlaps of pairs of fragments in $\mathcal{L}^F \times \mathcal{L}^G$, consistently with \cref{def:marking},
    position $p$ of $G_j$ would get the same number of marks as position $r_j+p$ of $T$.

    \begin{claim}[{\cite[see Claim 6.22]{ckw22}}]\label{claim:wip}
        For all $x \in \fragment{1}{|\mathcal{L}^P|}$ and $y \in \fragment{i_1}{i_2}$, we
        have
        \[\big|\fragmentco{\ell^G_y-\ktotm}{r^G_y+\ktotm} \cap \fragmentco{p + \ell^F_x}{p
            + r^F_x}\big| =
        \big|\fragmentco{\ell^T_y-\ktotm}{r^T_y+\ktotm} \cap \fragmentco{r_j + p +
    \ell^P_x}{r_j + p + r^P_x}\big|.\]
    \end{claim}
    \begin{claimproof}
        Consider sequences
        \[
            \mathcal{M}_x = (F_{j,w_1}, G_{j,w_1}) \cdots (F_{j,w_2}, G_{j,w_2})
            \quad\text{and}\quad
            \mathcal{M}_y = (F_{j,w_3}, G_{j,w_3}) \cdots (F_{j,w_4}, G_{j,w_4})
        \] of
        contiguous non-plain pairs of $\I'_j$
        that $\tbd \ktotm$-contain $L^F_x$ and $L^G_y$, respectively,
        and are maximal in the sense that they cannot be extended and remain monochromatic.
        (Such sequences exist by \cref{claim:lockcont}.)

        First, consider the case where $\mathcal{M}_x$ and $\mathcal{M}_y$ do not
        coincide.
        We treat the case where $\mathcal{M}_x$ lies to the left of~$\mathcal{M}_y$; the
        other case can be handled analogously.
        In this case, $w_2<z_j$ and $w_3>1$.
        Recall that $p \in \fragment{0}{3\ktot}$.
        By \cref{lem:monochrom}, we have
        that $G\fragmentco{p + \ell^F_x}{p + r^F_x}$
        only overlaps pairs in $\mathcal{M}_x$ and hence it is disjoint
        from $\fragmentco{\ell^G_y-\ktotm}{r^G_y+\ktotm}$ since
        $p + r^F_x \leq g_{j,w_3} \leq \ell^G_y-\ktotm$.
        Now, observe that $\orig(j,w_1) \leq \orig(j,w_3)$
        and hence
        \[r_j + p + r^P_x = r_j + \orig(j,w_1)\tau + p + r^F_x < r_j + \orig(j,w_3)\tau +
        \ell^G_y-\ktotm = \ell^T_y - \ktotm.\]
        Thus, in this case, both considered intersections are empty.

        Otherwise, $\mathcal{M}_x$ and $\mathcal{M}_y$ coincide.
        We then have
        \[r_j + \orig(j,w_1)\tau + \fragmentco{p+\ell^F_x}{p+r^F_x} = r_j +
        \fragmentco{p+\ell^P_x}{p+r^P_x} = \fragmentco{r_j + p+\ell^P_x}{r_j + p+r^P_x}\]
        and
        \[r_j + \orig(j,w_1)\tau + \fragmentco{\ell^G_y - \ktotm}{r^G_y + \ktotm} = r_j +
        \fragmentco{\ell^T_y - r_j - \ktotm}{r^T_y - r_j + \ktotm} = \fragmentco{\ell^T_y
    - \ktotm}{r^T_y + \ktotm}.\]
        The statement readily follows in the considered case.
    \end{claimproof}

    Let $\mathcal{B} = (f_v,g_v)_{v=0}^u$.
    For each $L_i^F\in \D'(p)$, let $\fragmentco{a^F_i}{b^F_i}\subseteq \fragment{0}{u}$
    so that $L^F_i = F\fragmentco{f_{a^F_i}}{f_{b^F_i}}$ and $\mathcal{B}(L^F_i) =
    G\fragmentco{g_{a^F_i}}{g_{b^F_i}}$.
    Symmetrically, for each  $L^G_i\in \D'(p)$, let $\fragmentco{a^G_i}{b^G_i}\subseteq
    \fragment{0}{u}$ so that
    $L^G_i = G\fragmentco{g_{a^G_i}}{g_{b^G_i}}$ and $\mathcal{B}^{-1}(L^G_i) =
    F\fragmentco{f_{a^G_i}}{f_{b^G_i}}$.
    Further, let $\mathcal{E}$ denote the multiset union of the multisets
    \[\mathcal{E}^F = \{\fragmentco{a^F_i}{b^F_i} : L_i^F\in \mathcal{L}^F \cap
    \mathcal{D}'(p)\} \quad \text{and} \quad \mathcal{E}^G = \{\fragmentco{a^G_i}{b^G_i} :
L_i^G\in \mathcal{L}^G\cap \D'(p)\}.\]
    In fact, the following claim implies that the multiplicity of each element of
    $\mathcal{E}$ is one.

    \begin{claim}[{\cite[see Claim 6.23]{ckw22}}]
        \label{claim:disjoint}
        The intervals in $\mathcal{E}$ are pairwise disjoint.
    \end{claim}
    \begin{claimproof}
        First, observe that the elements of each of $\mathcal{L}^F$ and $\mathcal{L}^G$
        are pairwise disjoint and hence the elements of each of $\mathcal{E}^F$ and $\mathcal{E}^G$
        are pairwise disjoint.

        Now, consider any two fragments $L^F_x \in \D'(p)$ and
        $L^G_y\in \D'(p)$. Observe that, since $L^P_x, L^T_y \in \D(r_j+p)$, we have
        $\markf(r_j+p,\ktotm,L^P_x,L^T_y)=0$.
        Hence, $\fragmentco{\ell^T_y-\ktotm}{r^T_y+\ktotm} \cap \fragmentco{r_j + p +
        \ell^P_x}{r_j + p + r^P_x}=\emptyset$.
        By \cref{claim:wip}, we also have $\fragmentco{g_{a^G_y}-\ktotm}{g_{b^G_y}+\ktotm}
        \cap \fragmentco{p + f_{a^F_x}}{p + f_{b^F_x}}=\emptyset$.

        Since the cost of $\mathcal{B}$ is no more than $k$,
        $g_{i}\in \fragment{p+f_i-k}{p+f_i+k}$ holds for all $i\in \fragment{0}{u}$.
        In particular, we have $\fragment{g_{a^G_y}-k-1}{g_{b^G_y}+k} \supseteq \fragmentco{p+f_{a^G_y}-1}{p+f_{b^G_y}+1}$.
        Since $\fragmentco{g_{a^G_y}-\ktotm}{g_{b^G_y}+\ktotm} \supseteq \fragment{g_{a^G_y}-k-1}{g_{b^G_y}+k}$,
        we then have $\fragmentco{f_{a^G_y}-1}{f_{b^G_y}+1}\cap \fragmentco{f_{a^F_x}}{f_{b^F_x}}=\emptyset$,
        that is, $\fragment{f_{a^G_y}}{f_{b^G_y}}\cap \fragment{f_{a^F_x}}{f_{b^F_x}}=\emptyset$.
        This implies $\fragmentco{a^G_y}{b^G_y}\cap \fragmentco{a^F_x}{b^F_x}=\emptyset$;
        consequently, the intervals in $\mathcal{E}$ are pairwise disjoint.
    \end{claimproof}

    \[\text{Set } \Lambda \coloneqq
    \edwp{L_1^P}{\rot^{-\rho(r_j+p)}(Q)}+\sum_{i=2}^{\ell^P}\edwl{L^P_i}{Q}+\sum_{i=i_1}^{i_2}\edwl{L^T_i}{Q}.\]

    \begin{claim}[{\cite[see Claim 6.24]{ckw22}}]\label{claim:lb}
        \[\sum_{L^{F}_i\in \mathcal{D}'_j(p)}\edw{L^{F}_i}{\mathcal{B}(L^{F}_i)} +
        \sum_{L^{G}_i \in \mathcal{D}'_j(p)}\edw{L^{G}_i}{\B^{-1}(L^{G}_i)} \geq \Lambda -
    2\markf(r_j+p).\]
    \end{claim}
    \begin{claimproof}
        Using~\cref{lem:lb_local}, we lower-bound each individual term of the left-hand
        side of the proved inequality.
        We consider three cases.
        \begin{enumerate}
            \item Consider some $L^{G}_i=G\fragmentco{p+\ell}{p+r} \in \mathcal{D}'_j(p)$
                and let $L^T_i=T\fragmentco{r_j + p +\ell'}{r_j + p + r'}$.
                Since $\edw{F}{G\fragmentco{p}{b}} \leq k$, we have that
                $\mathcal{B}^{-1}(L^{G}_i)$ is a substring of
                $F\fragmentco{\max\{0,\ell-k\}}{\min\{|F|,r+k\}}$.
                By \cref{claim:lockcont}, $L^G_i$ is $\tbd \ktotm$-contained by a sequence
                $\mathcal{M}$ of contiguous non-plain pairs in $\I'_j$.
                Then, by \cref{lem:monochrom},
                $F\fragmentco{\max\{0,\ell-k\}}{\min\{|F|,r+k\}}$ only overlaps pairs of
                $\I'_j$ that are in $\mathcal{M}$ and
                hence it is a substring of
                $P\fragmentco{(\ell' - \ell) + \max\{0,\ell-k\}}{(r' - r) +
                \min\{|F|,r+k\}}$, which in turn is a substring of
                $P\fragmentco{\max\{0,\ell'-k\}}{\min\{|P|,r'+k\}}$.
                Thus, $\edw{L^G_i}{\mathcal{B}^{-1}(L^{G}_i)} \geq \edwl{L^T_j}{Q} - \markf(r_j+p,L^T_j)$
                holds by~\cref{lem:lb_local}.
            \item Consider some $L^F_i = F\fragmentco{\ell}{r} \in
                \mathcal{D}'_j(p)\setminus \{L^P_1\}$ and let $L^P_i =
                P\fragmentco{\ell'}{r'}$.
                Since $\edw{F}{G\fragmentco{p}{b}} \leq k$, we that have
                $\mathcal{B}(L^{F}_i)$ is a substring of $
                G\fragmentco{\max\{p,p+\ell-k\}}{\min\{|G|,p+r+k\}}$.
                As before, by combining \cref{claim:lockcont} and \cref{lem:monochrom} we
                get that $G\fragmentco{\max\{p,p+\ell-k\}}{\min\{|G|,p+r+k\}}$
                is a substring of
                $T\fragmentco{r_j + (\ell' - \ell) + \max\{p,p+\ell-k\}}{r_j + (r' - r) +
                \min\{|G|,p+r+k\}}$, which in turn is a substring of
                $T\fragmentco{r_j + p + \max\{0,\ell'-k\}}{r_j + p + \min\{|T|,r'+k\}}$.
                Thus,
                $\edw{L^F_i}{\mathcal{B}(L^F_i)}\geq \edwl{L^P_i}{Q} -
                \markf(r_j+p,L^P_i)$ holds by~\cref{lem:lb_local}.
           \item Lastly, since $\edw{F}{G\fragmentco{p}{b}} \leq k$, we have that
                $\mathcal{B}(L^F_1)$ is a prefix of $T\fragmentco{r_j+p}{r_j+p+|L^P_1|+k}$,
                and hence
                $\edw{L^F_1}{\mathcal{B}(L^F_1)}\geq \edwp{L_1^P}{\rot^{-\rho(r_j+p)}(Q)} -
                \markf(r_j+p,L_1^P)$ holds by~\cref{lem:lb_local}.
        \end{enumerate}

        We now put everything together.
        In the following inequalities, we use the fact that,
        for each $L \in (\mathcal{L}^P \cup
        \mathcal{L}^T_{\fragmentco{r_j+p}{r_j+p+m}})\setminus{\mathcal{D}(r_j+p,B)}$, we have
        $\edwl{L}{Q} \leq \markf(r_j+p,L)$,
        and hence the sum of $\edwl{L}{Q} - \markf(r_j+p,L)$ over all such $L$ is at most
        zero.
        \begin{align*}
    &		\sum_{L^{F}_i\in \mathcal{D}'_j(p)}\edw{L^{F}_i}{\mathcal{B}(L^{F}_i)} +
    \sum_{L^{G}_i \in \mathcal{D}'_j(p)}\edw{\B^{-1}(L^{G}_i)}{L^{G}_i}\\
    &\quad  \geq \edwp{L^{P}_1}{\rot^{-\rho(r_j+p)}(Q)} - \markf(r_j+p,L^P_1) + \sum_{L^P_i \in \mathcal{D}(r_j+p) \setminus \{L^P_1\}} (\edwl{L^P_i}{Q} - \markf(r_j+p,L^P_i))\\
    &\qquad\qquad\qquad\qquad\qquad\qquad\qquad\qquad\qquad\qquad + \sum_{L^T_i \in \mathcal{D}(r_j+p)} (\edwl{L^T_i}{Q} - \markf(r_j+p,L^T_i))\\
    &\quad  \geq \edwp{L_1^P}{\rot^{-\rho(r_j+p)}(Q)} + \sum_{i=2}^{\ell^P}\edwl{L^P_i}{Q}+\sum_{i=i_1}^{i_2}\edwl{L^T_i}{Q} - \sum_{i=1}^{\ell^P}\markf(r_j+p,L^P_i) \\
    &\qquad\qquad\qquad\qquad\qquad\qquad\qquad\qquad\qquad\qquad\qquad\qquad\qquad\qquad\;\, - \sum_{i=i_1}^{i_2}\markf(r_j+p,L^T_i)\\
    &\quad  = \Lambda - \sum_{i=1}^{\ell^P}\sum_{v=1}^{\ell^T} \markf(r_j+p,\ktotm,L^P_v,L^T_i) - \sum_{i=i_1}^{i_2}\sum_{v=1}^{\ell^P} \markf(r_j+p,\ktotm,L^P_v,L^T_i) \\
    &\quad  \geq \Lambda - 2\markf(r_j+p).
        \end{align*}
        This concludes the proof of the claim.
    \end{claimproof}

    Next, due to \cref{lem:focus3k}, we have
    $\edwp{L_1^P}{\rot^{-\rho(r_j+p)}(Q)} < \lpref$.
    By a direct application of \cref{lem:ub}, we then have
    that $\min_t \edw{P}{T\fragmentco{r_j+p}{t}}\leq \Lambda$.
    If $\Lambda\leq k$, then we are done. For the remainder of the proof we thus consider
    the case where $\Lambda >k$, which, combined with the fact
    that $r_j + p$ is a light position of $T$, means that
    $\Lambda-2\markf(r_j+p) > k-2\markf(r_j+p) \geq k-2 \tradeoff$.

    \begin{claim}[{\cite[see Claim 6.25]{ckw22}}]
        \label{claim:plain_exact}
        For any run $M = (F_{j,i_1},G_{j,i_1}) \cdots (F_{j,i_2},G_{j,i_2})$ of
        $2\tradeoff + \slack$ consecutive plain pairs in $\I'_j$,
        there exists some $i\in \fragmentoo{i_1}{i_2}$ for which
        $\edw{F\fragmentco{f_{j,i}}{f_{j,i+1}}}{\mathcal{B}(F\fragmentco{f_{j,i}}{f_{j,i+1}})}
        = 0$.
    \end{claim}
    \begin{claimproof}
        We have $|f_{j,i_1+14} - f_{j,i_1}| \geq 14 \tau - d_T \geq 14 \ktot/2 - d_T \geq
        6 \ktot = \Delta$. Similarly,
        $|f_{j,i_2-13} - f_{j,i_2+1}| \geq \Delta$.
        Thus, $U=F\fragmentco{f_{j,i_1+14}}{f_{j,i_2-13}}$ is $\Delta$-contained by $M$
        and hence only overlaps pairs in $M$ by \cref{obs:Dcont}.

        Now, each fragment $L_i^F \in \mathcal{L}^F$ is $13\ktotm$-contained by a sequence of
        contiguous non-plain pairs and hence only overlaps non-plain pairs.
        Further, for each fragment $L_i^G \in \mathcal{L}^G$, as shown in the proof of
        \cref{claim:lb}, $\mathcal{B}^{-1}(L_i^G)$ only overlaps non-plain pairs.
        Observe that two fragments of $F$ are necessarily disjoint if the sets of pairs
        that they overlap are disjoint, and hence
        $\fragmentco{f_{j,i_1+14}}{f_{j,i_2-13}}$
        is disjoint from all elements of $\mathcal{E}$.
        As the intervals in $\mathcal{E}$ are pairwise disjoint by \cref{claim:disjoint},
        using \cref{claim:lb}, we obtain
        \begin{align*}
            k &\geq \edw{F}{G\fragmentco{p}{b}}\\
              &\geq \edw{U}{\mathcal{B}(U)} + \sum_{L^{F}_i\in
              \mathcal{D}'_j(p)}\edw{L^{F}_i}{\mathcal{B}(L^{F}_i)} + \sum_{L^{G}_i \in
          \mathcal{D}'_j(p)}\edw{\B^{-1}(L^{G}_i)}{L^{G}_i}\\
              &\geq \edw{U}{\mathcal{B}(U)} + \Lambda - 2\markf(r_j+p) \\
              &> \edw{U}{\mathcal{B}(U)} + k - 2\tradeoff.
        \end{align*}
        Hence, we have $\edw{U}{\mathcal{B}(U)} < 2\tradeoff$.
        This implies that
        \[\sum_{i=i_1+14}^{i_2-14}
        \edw{F\fragmentco{f_{j,i}}{f_{j,i+1}}}{\mathcal{B}(F\fragmentco{f_{j,i}}{f_{j,i+1}})}
    < 2\tradeoff,\]
        and hence, since the number of summands in in the left-hand side of the inequality
        is $i_2 - 14 - (i_1+14) + 1 = 2\tradeoff + \slack - 28 > \tradeoff$
        and each of these summands is a non-negative integer, the claim follows.
    \end{claimproof}

    We are now ready to conclude the proof of the lemma.
    Let $\mathcal{M}_1, \ldots, \mathcal{M}_w$ denote the trimmed runs of $2\tradeoff + \slack$
    consecutive plain pairs in $\I'_j$ such that, for all $i$, $\mathcal{M}_i$
    originated from a run with $\chi(i)$ more plain pairs. For $i \in \fragment{1}{w}$,
    let $\mathcal{M}_i=(F_{j,s_i},G_{j,s_i}) \cdots (F_{j,e_i},G_{j,e_i})$
    let $\psi(i) \in \fragmentoo{s_i}{e_i}$ be such that
    and $\edw{F\fragmentco{f_{j,\psi(i)}}{f_{j,\psi(i)+1}}}{
    \mathcal{B}(F\fragmentco{f_{j,\psi(i)}}{f_{j,\psi(i)+1}})} = 0$; observe that $\psi(i)$
    exists due to \cref{claim:plain_exact}.
    Let us now show how to construct $\mathcal{C}$ given $\mathcal{B} = (f_v,g_v)_{v=0}^u$.
    We intuitively achieve this by inserting $\chi(i)$ copies of
    $Q^{\infty}\fragmentco{0}{\tau}$ after each of
    $F\fragmentco{f_{j,\psi(i)}}{f_{j,\psi(i)+1}}$
    and $\mathcal{B}(F\fragmentco{f_{j,\psi(i)}}{f_{j,\psi(i)+1}})$ and aligning them
    without
    errors, thus restoring the original length of each trimmed plain run,
    without changing the cost of the alignment.
    Initially, set $\mathcal{C} \coloneqq \mathcal{B}$.
    Then, for each $\psi(i)$, in decreasing order
    \begin{itemize}
        \item replace each pair $(f_v,g_v)$ of $\mathcal{C}$ that satisfies $f_v \geq f_{j,\psi(i)+1}$
            with $(\chi(i)\cdot\tau + f_v, \chi(i)\cdot\tau + g_v)$, and
        \item insert $(f_{j,\psi(i)+1}+\phi,
            g_{j,\psi(i)+1}+\phi)_{\phi=0}^{\chi(i)\cdot\tau-1}$ after
            $(f_{j,\psi(i)+1}-1,g_{j,\psi(i)+1}-1)$.\qedhere
    \end{itemize}
\end{proof}

We conclude with the proof of \cref{lem:correctlight} as promised.

\correctlight
\begin{proof}
    $(\subseteq)$: This direction is an immediate consequence of \cref{cor:supset}.

    $(\supseteq)$: This direction is an immediate consequence of \cref{lem:nofakelight}.
\end{proof}

\subsection{Combining the Partial Results: Faster \SM}

We are ready to prove the main result of this section---\cref{lem:solve_int_SM}---which we restate here for
convenience.

\solveintSM
\begin{proof}
    If $W > \sqrt{k}$, \cref{lem:solve_SM} yields a better running time, so we henceforth assume $W \le \sqrt{k}$.
    Consistently with all previous sections, set \[
        \ktot \coloneqq k+d_P+d_T=\Theta(k)\quad\text{and}\quad
        \ktotm \coloneqq 112(kW+d_P^w+d_T^w)=\Oh(kW)\quad\text{and}\quad \tradeoff
        \coloneqq \max\{1 ,\lfloor{\sqrt{k}W}\rfloor\} \leq \ktotm.
    \]

    We proceed in roughly four steps.
    \begin{itemize}
        \item First, in a preprocessing step, we identify the heavy positions \(\Hv\) and the light
            positions \(\light\) in \(T\), as well as a filter \(\mq\) (according to
            \cref{fct:rj}) for where potential \(k\)-error occurrences may start.
        \item Next, we compute all \(k\)-error occurrences starting at a position in \(\Hv
            \), using the algorithm from \cref{sus:heavy}.
            In particular, we obtain \(\Hv \cap \OccW_k(P,T)\) as a set of positions.
        \item Next, we use \(\swaps\) from \cref{lem:swaps} to obtain a candidate set
            \(\lmore\) of potential starting positions of \(k\)-error occurrences
            (represented as disjoint arithmetic progressions with difference
            \(\tau\)).
            From \cref{sus:light}, we have
            \((\lmore \cap \mq) \setminus (\Hv \cap \mq) = \OccW_k(P, T) \cap \light\);
            hence we proceed to compute \(\Hv \cap \mq\) (represented as a set), and
            \(\lmore \cap \mq\) (represented as disjoint
            arithmetic progressions with difference~\(\tau\)), and thereafter compute
            their set difference to obtain \(\OccW_k(P, T) \cap \light\)
            (represented as disjoint arithmetic progressions with
            difference~\(\tau\)).
        \item In a post-processing step, we union the two sets \(\Hv \cap \OccW_k(P,T)\) and
            \(\OccW_k(P, T) \cap \light\) and compute a representation as disjoint
            arithmetic progressions with difference \(\tau\).
    \end{itemize}

    \subparagraph{Preprocessing.}
    We compute the sets of locked fragments $\LP = \locked(P,Q,d^w_P,\lpref,w)$ and
    $\LT = \locked(T,Q,d^w_T,0,w)$ in $\cO(d^2)$ time using \cref{cor:klocked} and call
    $\protect\heavy(P,T,d,k,Q,\LP,\LT,\ktotm,\tradeoff)$ to obtain the set $\Hv$ of
    heavy positions,
    represented as the union of $\Oh(k^2W^2)$ disjoint integer ranges,
    in $\Ohtilde(k^2W^3)$ time (cf.~\cref{lem:heavy-alg}).

    Define $J$ and each of $r_j$ and $\I'_j$, for $j\in J$, as
    in \cref{fct:rj,def:fjgj}, and set \[
        \lmore \coloneqq \bigcup_{j \in J} (r_j +  \OccW_k( \I'_j)).
    \]
    Observe that, by \cref{lem:correctlight},
    \[
        \OccW_k(P,T) = (\OccW_k(P,T) \cap \Hv) \cup (\OccW_k(P,T) \cap \light) =
         (\OccW_k(P,T) \cap \Hv) \cup (\lmore \cap \light).
    \]
    Finally, for our filter, \cref{fct:rj,fact:simple} yield
    \[
        \mq \coloneqq \bigcup_{j\in \mathbb{Z}} \fragment{j q-x_T-\kappa-d_T}{j
        q-x_T+\kappa+d_T} \supseteq \OccE_k(P,T) \supseteq \OccW_k(P,T).
    \]
    Observe that we thus have
    \begin{equation}\label{eq:mq}
        \OccW_k(P,T) \cap \light = (\lmore \cap \mq) \setminus (\Hv \cap \mq).
    \end{equation}

    \subparagraph{Heavy occurrences.}

    By \cref{lem:heavy-total}, we can compute $\OccW_k(P,T) \cap \Hv$ in time
    $\Ohtilde((kW/\eta+1)k^3W^3)=\Ohtilde(k^{3.5}W^4)$.

    \subparagraph{Light occurrences.}

    Let us now proceed to computing the remaining $(k,w)$-error occurrences, namely those that
    start at light positions.
    To this end, we first compute the $\cO(kW^2)$-size sets $\rred{P}$ and $\rred{T}$,
    defined in the beginning of \cref{sec:redundantsqrtk},
    in $\Ohtilde(kW^2)$ time (cf.~\cref{lem:fewreds}).
    Then, we call \(\swaps(T, P, k, Q, \A_P, \A_T, \rred{P}, \rred{T}, 2\tradeoff +
    \slack)\) from \cref{lem:swaps}, which returns
    a representation of \(\lmore\)
    as $\cO(k^3 W^4 \cdot \tradeoff)=\cO(k^{3.5}W^4)$ arithmetic progressions with difference
    $\tau \coloneqq q\ceil{{\kappa}/{2q}}$.
    Plugging in \cref{xxdsyqwfxu} for the \DPM data structure, \swaps runs in time \[
        \Ohtilde(k^3W + k \cdot k^2 W^4 \tradeoff + k^2W^4 \cdot kW)=\Ohtilde(k^{3.5}W^4+k^3W^5)=\Ohtilde(k^{3.5}W^4).
    \]

    We move on to remove surplus positions from the candidate set \(\lmore\).
    We start by computing \(\Hv \cap \mq\).

    \begin{claim}\label{claim:hmq}
        The set $\Hv \cap \mq$ is of size $\cO(k^{2.5}W^3)$ and can be computed in time $\cOtilde(k^{2.5}W^3)$.
    \end{claim}
    \begin{claimproof}
        The proof is similar to a part of the proof of~\cref{lem:heavy-total}.

        Due to \cref{lem:heavy-alg}, our representation of $\Hv$ consists of $\Oh(k^2W^2)$ disjoint integer ranges.
        In addition, due to \cref{lem:heavy-bound}, we have \[
            |\Hv|=\cO((kW/\tradeoff + 1)kW(\ktotm+q)) = \cO(k^{1.5}W^2\cdot (kW+q)).
        \]

        Let us first upper-bound the size of \(|\Hv \cap \mq|\).
        We distinguish between two cases.
        \begin{itemize}
            \item First, if $q \leq kW$, we have $|\Hv \cap \mq| \leq |\Hv|=\cO(k^{2.5}W^3)$.
            \item As for the complementary case where $q > kW$,
                observe that only $\cO(k)$ out of any $\cO(q)$ consecutive integers are in
                $\mq$.
                Hence, $\Hv \cap \mq$ is of size $\cO(k^2W^2 + |\Hv|\cdot k/q) =
                \cO(k^2W^2+k^{2.5}W^2)=\cO(k^{2.5}W^3)$.
        \end{itemize}

        As for computing $\Hv \cap \mq$, we first
        sort the $\Oh(k^2W^2)$ integer ranges that comprise $\Hv$ with respect to their
        starting positions in $\Ohtilde(k^2W^2)$ time
        and then scan them from left to right, skipping positions in $\mathbb{Z} \setminus
        \mq$.
        The running time of the scan is proportional to the total number of input ranges
        and output positions, and hence we are done.
    \end{claimproof}

    We move on to compute \(\lmore \cap \mq\).
    In what follows, for any integer $x$, let us say that the \emph{residue modulo $x$} of
    a non-empty arithmetic progression whose difference is a multiple of $x$
    is the residue of any element of this arithmetic progression modulo $x$.

    \begin{claim}\label{claim:fmq}
        We can compute a representation of
        $\lmore \cap \mq$ as $\cO(k^{3.5}W^4)$ disjoint arithmetic progressions with
        difference~$\tau$,
        sorted according to their starting positions,
        in $\cOtilde(k^{3.5}W^4)$ time.
    \end{claim}
    \begin{claimproof}
        In a linear scan of the $\cO(k^{3.5}W^4)$ arithmetic progressions that comprise
        $\lmore$ as returned by the call to the algorithm $\swaps$,
        we delete any arithmetic progression whose elements are not in $\mq$.
        Then, we sort all arithmetic progressions according to their starting positions in
        time $\Ohtilde(k^{3.5}W^4)$
        and distribute them among $\cO(k/q\cdot \tau)=\cO(k^2)$ buckets according to their
        residues modulo~$\tau$.
        Finally, we scan linearly the arithmetic progressions in each bucket, greedily
        merging progressions that overlap (that is, we merge progressions if their union
        is also a valid arithmetic progression with difference \(\tau\)).
        The number of the resulting arithmetic progressions is clearly upper-bounded by the size
        of the input, that is, $\cO(k^{3.5}W^4)$;
        we sort them according to their starting positions in $\cO(k^{3.5}W^4)$ time.
    \end{claimproof}

    Finally, we compute the set difference \(
        (\lmore \cap \mq) \setminus (\Hv \cap \mq) =
        \OccW_k(P,T) \cap \light.
    \)

    \begin{claim}\label{claim:apdt}
        A representation of $\OccW_k(P,T) \cap \light$ as $\cO(k^{3.5}W^4)$ disjoint
        arithmetic progressions with difference~$\tau$,
        sorted according to their starting positions,
        can be computed in time $\cOtilde(k^{3.5}W^4)$ time the \pillar model.
    \end{claim}
    \begin{claimproof}
        We intend to use equation \eqref{eq:mq}, relying on \cref{claim:hmq,claim:fmq}
        to compute $\Hv \cap \mq$ and a representation of $\lmore \cap \mq$ as
        $\cO(k^{3.5}W^4)$ disjoint arithmetic progressions with difference $\tau$ in
        $\Ohtilde(k^{3.5}W^4)$ time in total.
        Then, we process the elements of these two sets in $\cO(d/q\cdot \tau)=\cO(k^2)$
        batches, where each batch contains all elements with a specific residue modulo
        $\tau$.
        For such a fixed residue, we scan in parallel the arithmetic progressions in
        $\lmore \cap \mq$ and
        the positions in $\Hv \cap \mq$, both of which are sorted in increasing order.
        When some element of $\Hv \cap \mq$ is contained in some arithmetic progression in $\lmore \cap \mq$,
        this arithmetic progression is split into two parts, either of which may be empty.
        The number of the resulting arithmetic progressions is upper-bounded by the total
        size of the input, that is, $\cO(k^{3.5}W^4)$;
        we sort them according to their starting positions in $\Ohtilde(k^{3.5}W^4)$ time.
    \end{claimproof}
    This concludes the proof of this lemma.
\end{proof}

\section{Fast Data Structures for (Bleach-Commit) Dynamic Puzzle Matching}\label{sec:dpm}

\subsection{On Matrix Multiplication}

Before we discuss how to efficiently compute ferns,
it is instructive to take a step back to abstractly describe two different schemes that we
use to efficiently \emph{multiply} several matrices.

To this end, we introduce the problem \PDPM, which is of a similar flavor compared to
\DPM.
The main difference of \PDPM and \DPM is that the former directly operates on
matrices, while the latter operates on puzzle pieces (that is, pairs of strings).

\begin{problem}{Dynamic\-Matrix\-Multiplication}
    \label{tvssrtdhgs}
    \PObject{A sequence of matrices \(\F \coloneqq F_1, \dots, F_z\in \Rp^{q\times q}\) that share a common matrix
    property \(P\); see \cref{def:fernprop}.}%

    \PInit{\begin{description}[left=0em..0em,itemindent=-.5em]
        \item {\tt DMM-Init($q$, $\F$)}:
            Given a dimension \(q\) and an initial sequence $\F$ of matrices that
            satisfy $P$, initialize the sequence of matrices to \(\F\).
    \end{description}}

    \PUpdate{\begin{description}[left=0em..0em,itemindent=-.5em]
        \item {\tt DMM-Delete(\(i\))}:
            Delete the $i$-th element of $\F$.
        \item {\tt DMM-Insert($F$, \(i\))}:
            Insert a matrix $F\in \Rp^{q\times q}$ with property $P$ after the $i$-th element of $\F$.
        \item {\tt DMM-Substitute($F$, \(i\)):}
            Substitute the $i$-th element of $\F$ with a matrix $F\in \Rp^{q\times q}$ with property $P$.
        \item {\tt DMM-Split(\(i\))}:
            Split the sequence after the \(i\)-th element and return an instance of the
            data structure for the prefix and an instance of the data structure for the
            suffix.
        \item {\tt DMM-Merge(\(\F'\))}:
            Given an instance of \PDPM representing \(\F'\) (a sequence of matrices with property $P$), append the sequence \(\F'\) to the sequence \(\F\).
    \end{description}}

    \PQuery{\begin{description}[left=0em..0em,itemindent=-.5em]
        \item {\tt DMM-Matrix()}:
            Return \(F_1 \oplus \dots
            \oplus F_z\).
        \item {\tt DMM-Vector($v$)}:
            Given a vector \(v\in \Rp^{q\times 1}\), return \(F_1 \oplus \dots \oplus F_z \oplus v\).
    \end{description}}%
\end{problem}

We discuss two different implementations for \PDPM; we start with
an implementation that allows for fast {\tt DMM-Matrix} operations at the cost of
only moderately fast update operations.

\begin{lemma}
    \dglabel{lemma1}
    Let us fix a matrix property $P$ with \(T_M(q)\)-time matrix multiplication and
    \(T_V(q)\)-time vector multiplication.

    Then, there is an implementation for \PDPM such that
    \begin{itemize}
        \item {\tt DMM-Init($q$, $\F'$)} takes time
            \(\Oh(|\F'| \cdot T_M(q))\).
        \item Any update operation takes time
            \(\Oh(\log |\F| \cdot T_M(q))\).
        \item {\tt DMM-Matrix()}
            takes time proportional to the time required to write
            down the output (that is, typically \(\Oh(q^2)\)).
        \item {\tt DMM-Vector($v$)}
            takes time \(\Oh(T_V(q))\) plus time proportional to the time required to write
            down the output (that is, typically \(\Oh(q)\)).
    \end{itemize}
    Additionally, we support an operation {\tt DMM-Root()} with constant-time read-only
    access to a pointer to a matrix corresponding to the result of {\tt DMM-Matrix()}.
\end{lemma}
\begin{proof}
    We maintain a balanced binary tree; we store \(F_i\) at leaf \(i\) and at an
    internal vertex we store the \((\min,+)\)-product of its (up to) two children.

    At initialization, we build the tree in a bottom-up fashion using a linear number of
    \((\min,+)\)-products in total.

    For updates, we insert, delete, or substitute the corresponding element of the
    sequence and then update its (former) parents and rebalance as necessary.
    Further, binary trees readily support both  split and merge operations using a logarithmic
    number of matrix multiplications.

    For the query operations, we either return the root of our tree (as we also do for the
    extra {\tt DMM-Root()}) or the root multiplied
    with the given vector.

    Associativity of the \((\min,+)\)-product ensures correctness;
    the running time guarantees follow from basic properties of balanced binary trees.
\end{proof}

Secondly, we discuss an even easier implementation with faster update and {\tt DMM-Vector}
operations at the cost of slower {\tt DMM-Matrix} operations.

\begin{lemma}
    \dglabel{lemma2}
    Let us fix a matrix property $P$ with \(T_M(q)\)-time matrix multiplication and
    \(T_V(q)\)-time vector multiplication.

    Then, there is an implementation for \PDPM such that
    \begin{itemize}
        \item {\tt DMM-Init($q$, $\F'$)} takes time
            \(\Oh(|\F'|)\)
            (where we assume that input matrices do not need to be copied).
        \item Any update operation takes time
            \(\Oh(1)\)
            (where we assume that input matrices do not need to be copied).
        \item {\tt DMM-Matrix()}
            takes time
            \(\Oh(|\F| \cdot T_M(q))\).
        \item {\tt DMM-Vector($v$)}
            takes time \(\Oh(|\F| \cdot T_V(q))\).
    \end{itemize}
\end{lemma}
\begin{proof}
    We maintain a (linked) list of pointers to the input matrices.
    For updates, we insert, delete, or substitute the corresponding element of the
    sequence by updating pointers or concatenating the (linked) lists of pointers.

    Now, for the query operations, we naively multiply the elements of the sequence.
    For the matrix operation, this involves \(|\F| - 1\) multiplications; for the vector
    operation, we exploit associativity and compute
    \(F_1 \oplus \cdots \oplus F_z \oplus v\)
    as
    \[F_1 \oplus (\cdots \oplus (F_z \oplus v) \cdots);\]
    that is, we start at the right and always \((\min,+)\)-multiply a matrix and a vector.

    Now, both correctness and running time guarantees are immediate.
\end{proof}

\subsection{Growing and Caring for Large Ferns}

In the setting of \BCDPM, we unfortunately have no direct guarantees on the lengths of the
input strings. Fortunately, we are able to show that the guaranteed bound on the summed
edit distance implies that there are large parts of the input strings that repeat in many
strings. This in turn allows us to split the input strings into shorter strings, compute
ferns for the shorter strings using \cref{lem:small_ed}, and then obtain ferns for the original
strings by applying \cref{lemma1}.

\begin{theorem}
    \dglabel{11-3-1}[dzefaaseda,cor:dproduct,lem:trim,fct:LValignment,lem:small_ed,gucbenbjuj,lm:build_mtx_dscomp,lem:etdyn,lemma1]
    For any family of strings $\S$ (of maximum length \(n\)) such that $\ed(\S)\leq d$, as
    well as a  positive integer $k$,
    there is an algorithm that given \(\S\), \(d\), \(k\), and
    oracle access to a normalized weight function
    \(w : \sqEsigma \to \intvl{0}{W}\), computes
    for each pair $(\dot{P},\dot{T}) \in \S \times \S$%
    \footnote{Recall from \cref{dzefaaseda} that \(\ed(\S) \le d\) implies \(|\S| = \Oh(d)\). }
    \begin{itemize}
        \item a Monge \((\min(|\dT|, d + k), k)\)-fern
            in total time $\Oh((d + k)^4 \log^2 (n(d+k)))$ in the \modelname model;
            and
        \item a \((W+1)\)-bounded difference Monge \((\min(|\dT|, d + k), k)\)-fern
            in total time
            $\Oh(W (d + k)^3 \log^2 (n(d+k)))$ in the \modelname model
            if all weights of \(w\) are integer.
    \end{itemize}
\end{theorem}
\begin{proof}
    Set $\nabla \coloneqq 8d +2k$.

    Now, write $\Sr$ for a 2-approximate center of $\S$.
    That is, we have $\sum_{\dot{S} \in \S} \ed(\Sr,\dot{S}) \leq 2\ed(\S)$.%
    \footnote{It is worth highlighting that the approximate center is with respect to
    the unweighted edit distance.}
    Further, for each \(\dot{S} \in \S\) write $\A_{\dot{S}}: \hat{S} \onto \dot{S}$
    for an optimal alignment.

    Next, for a fragment $V=U\fragment{i}{j}$ of a string $U$, by $V^{+\ell}$, we denote the
    fragment $U\fragmentco{i}{\min\{j+\ell\},|U|}$.

    We start with structural characterizations and then discuss how to turn said
    characterizations into efficient algorithms.

    \begin{claim}[Decomposing $\hat{S}$]\label{cl:decompose}
        There is a decomposition $\hat{S} = \bigcirc_{i=0}^{\tau-1} \hat{S}_i$
        such that all of the following hold.
        \begin{enumerate}
            \item We have $\tau \in \cO(d)$. \label{tmeqhnuais}
            \item For every $i\in \fragmentco{0}{\tau - 1}$, we have
                $|\hat{S}_i| \ge \nabla$.
            \item \label{jynijvbyrx}
                For every $i\in \fragmentco{0}{\tau}$,
                either $|\hat{S}_i|=\cO(\nabla)$ or,
                for all $\dot{S} \in \S$, we have
                $\A_{\dot{S}}(\hat{S}_i^{+\nabla}) = \hat{S}_i^{+\nabla}$.
        \end{enumerate}
    \end{claim}
    \begin{claimproof}
        If \(|\Sr| < 100 \nabla\), we return it as-is.

        Otherwise, we decompose \(\hat{S}\) into length-\(\nabla\) fragments (where the last fragment
        has a length in \(\fragment{2\nabla}{4\nabla}\));
        \(\hat{S} \coloneqq \bar{S}_0 \cdots \bar{S}_{\bar{\tau}-1}\).
        In particular, we have \(|\bar{S}_i| = \nabla\) for all
        \(i\in\fragmentco{0}{\bar{\tau} - 1}\).

        Observe that each \(\bar{S}_i\) satisfies
        \cref{jynijvbyrx} due to \(|\bar{S}_i| = \Oh(\nabla)\).
        However, \(\bar{\tau}\) might not satisfy \cref{tmeqhnuais}.

        Hence, let us say that an integer $i \in \fragmentco{0}{\bar{\tau}}$ is \emph{dirty} if and
        only if
        $\bar{S}_i \neq \A_{\dot{S}}(\bar{S}_i)$ for any $\dot{S} \in \S$.
        An integer is \emph{pure} if it is not dirty.
        Now, for each maximal interval $\fragment{i}{j} \subseteq \fragmentco{0}{\bar{\tau}}$
        of pure integers, merge tiles with indices in $\fragment{i}{j-1}$ in the natural
        order and re-index the resulting tiles.

        Now, for a dirty \(i\), the string \(\bar{S}_i\) remains
        unchanged (albeit with a potentially new index) and has all claimed properties.
        Further, we only ever merge a run of pure integers;
        which ensures that the second condition of \cref{jynijvbyrx} is satisfied.
        In total, we thus obtain the claim.
    \end{claimproof}

    If \(|\Sr| < 100 \nabla\), in the following, we construct \((\min(|\dT|, d + k),
    k)\)-ferns; only in this case we might have \(|\dT| < d + k\).
    For notational simplicity, we hence ignore this case in our notation from here onward.

    Next, we prove useful properties of the decomposition of \cref{cl:decompose}.
    To this end, for each $i\in \fragmentco{0}{\tau}$,
    for each $\dot{S} \in \S$, set $\dot{S}_i \coloneqq \A_{\dot{S}}(\hat{S}_i)$;
    and set $\S_i \coloneqq \{\dot{S}_i^{+\nabla} : \dot{S} \in \S\}$.

    We refer to \(\S_i\) as the \(i\)-th tile collection.
    Further, we say that an index \(i\) is \emph{pure} if \(|\S_i| = 1\).
    Otherwise, \(i\) is \emph{dirty}.

    It is instructive to highlight an easy property of the decomposition of
    \cref{cl:decompose}.

    \begin{claimq}\label{obs:tbd}
        For all dirty $i$, we have
        $|\dot{S}| = \cO(\nabla)$
        for all $\dot{S} \in \S_i$.
    \end{claimq}

    Next, we bound the total size of tile collections of dirty indices.

    \begin{claim}
        \label{cl:small_edit}
        If \(i\) is pure, then \(\ed(\S_i) = 0\).
        Further, we have $\sum_i \ed(\S_i) = \cO(d)$.
    \end{claim}
    \begin{claimproof}
        First, for pure \(i\), the corresponding tile collections contain only a single
        element each; hence the claim follows.

        For the second claim, we prove that for all $\dot{S} \in \S$, we have
        \[\sum_{i=0}^{\tau-1} \ed(\hat{S}_i^{+\nabla}, \dot{S}_i^{+\nabla}) \leq
        3\ed(\hat{S}, \dot{S}).\]

        To this end, observe that aligning all $\hat{S}_i$'s with $\dot{S}_i$'s costs
        \(\ed(\hat{S}, \dot{S})\) by construction as partitions of \(\hat{S}\) and
        \(\dot{S}\).

        Next, we bound the contribution of the overlaps. To this end, observe that
        aligning $\hat{S}_i^{+\nabla} \setminus \hat{S}_i$ with
        $\dot{S}_i^{+\nabla} \setminus \dot{S}_i$ costs \(2\ed(\hat{S}, \dot{S})\), as
        in addition to the original edits (which we have to account for again),
        the alignment $\A_S$ might change the lengths of the overlaps by at most \(\ed(\hat{S}, \dot{S})\)
        characters in total.
    \end{claimproof}

    In particular, \cref{cl:small_edit} implies that \(\sum_{i=0}^{\tau-1} |\S_i| =
    \Oh(d)\) (recall \cref{dzefaaseda}).

    Next, we argue that we obtain good puzzles out of our tile collections.

    \begin{claim}[Bounded torsion]
        \label{cl:small_torsion}
        For every $\dot{P}, \dot{T} \in \S$,
        the torsion of the $\nabla$-puzzle
        $(P_0^{+\nabla},T_0^{+\nabla}), \dots$,
        $(P_{\tau-1}^{+\nabla},T_{\tau-1}^{+\nabla})$
        is at most $4d = \nabla/2 - k$.
    \end{claim}
    \begin{claimproof}
        For each $i \in \fragmentco{0}{\tau}$, we have
        \[ ||P_i^{+\nabla}|-|T_i^{+\nabla}|| = ||P_i|-|T_i|| \leq \ed(P_i, T_i) \le  \ed(\hat{S}_i,P_i) +
        \ed(\hat{S}_i,T_i). \]

        Hence, as \(\Sr\) is a 2-approximate center of \(\S\), we obtain
        \[ \tor((P_0^{+\nabla},T_0^{+\nabla}), \dots,
        (P_{\tau-1}^{+\nabla},T_{\tau-1}^{+\nabla}))= \sum_{i=0}^{\tau-1}
            ||P_i^{+\nabla}|-|T_i^{+\nabla}|| \leq 2d + 2d = 4d
        .\claimqedhere\]
    \end{claimproof}
    \begin{claim}
        Fix \(\dot{P}, \dot{T} \in \S\) and write
        \(\fml^{P_0,T_0}_k, \fmi^{P_1, T_1} \dots, \fmi^{P_{\tau-2},T_{\tau-2}}_k, \fmt^{P_{\tau-1},T_{\tau-1}}_k\)
        for \((\nabla,k)\)-ferns for
        the $\nabla$-puzzle
        $\P\P \coloneqq (P_0^{+\nabla},T_0^{+\nabla}), \dots$,
        $(P_{\tau-1}^{+\nabla},T_{\tau-1}^{+\nabla})$.
        Then,
        \[
            F_k^{\dot{P},\dot{T}} \coloneqq \fml^{P_0,T_0}_k
            \oplus \fmi^{P_1,T_1}_k
            \oplus
            \dots
            \oplus \fmi^{P_{\tau-2},T_{\tau-2}}_k
            \oplus
            \fmt^{P_{\tau-1},T_{\tau-1}}_k
        \] is a \((\nabla,k)\)-fern (and thus a \((d + k,k)\)-fern after suitable trimming) for \((\dot{P},\dot{T})\).
    \end{claim}
    \begin{claimproof}
        We wish to employ \cref{cor:dproduct}.
        To this end, it suffices to observe that indeed \(\nabla /2 - \tor(\P\P) \ge
        4d + k - 4d = k\).
        Finally, for the trimming we use \cref{lem:trim}.
    \end{claimproof}

    With our structural insights gathered, it is time to implement them.
    We start by recalling an algorithm to compute a 2-approximate center from
    \cite[Lemma~9.9]{ckw22}.

    \begin{claim}[{\cite[Lemma~9.9]{ckw22}}]
        \label{uqzrrhcclx}
        Given a string family $\S$, we can construct a string $\Sr \in \S$ with
        $\sum_{S\in S}\ed(S,\Sr)\le 2\ed(\S)$
        in $\Oh(1+\ed(\S)^3)$ time in the \modelname model.
        Further, in the same time, we can compute alignments \(\A_S : \Sr \onto S\) for
        every \(S\in\S\).
    \end{claim}
    \begin{claimproof}[Proof sketch.]
        For each pair of strings \(S, S' \in \S\), we use \cref{fct:LValignment}
        to compute \(\ed(S,S')\), as well as a corresponding alignment.
        We return the string \(\Sr\) that minimizes \(\sum_{S\in \S} \ed(S,\Sr)\),
        as well as the corresponding alignments.

        We obtain the correctness from the triangle inequality.

        For the running time, we compute
        \[
            \Oh\big(\sum_{S, S'\in\S} \ed(S,S')^2\big)
            = \Oh\big(\sum_{S,S'\in\S} \ed(S,\Sr)^2 + \ed(S',\Sr)^2\big)
            = \Oh(|\S| + \ed(\S)^2)
            = \Oh(1 + \ed(\S)^3).
            \claimqedhere
        \]
    \end{claimproof}

    Now, we run the algorithm from  \cref{uqzrrhcclx}: we compute a 2-approximate center
    \(\Sr\) and corresponding alignments \(\Sr \onto S\) for every \(S \in \S\).
    Using the alignments obtained from \cref{uqzrrhcclx} (which runs in time \(\Oh(1 +
    d^3)\) in our case),
    we then compute the decomposition of \cref{cl:decompose}, again
    in time \(\Oh(1 + d^3)\).
    This in turn then allows to compute the tile collections in
    the same time.

    Next, for each tile collection \(\S_i\), we compute ferns for each pair of strings
    \(\dP, \dT \in \S_i\); the details differ depending on the guarantees that we have for
    the weight function \(w\).

    \begin{description}
        \item[General weights.]
            For each pure index \(i\), we use {\tt\apmsmalled} of \cref{lem:small_ed}, on
            (suitably trimmed versions of)
            \(\hat{S}_i, \hat{S}_i, \nabla\), and \(w\) to obtain a Monge
            \((\nabla,\nabla)\) fern for this index.
            As there are at most \(\Oh(d)\) pure indices, this costs in total
            \(\Oh(d \nabla^3 \log^2 (n\nabla)) = \Oh((d + k)^4 \log^2(n(d + k)))\) time.

            For each dirty index \(i\), we use \cref{gucbenbjuj}, on each pair of
            \(\dot{P}, \dot{T} \in \S_i\). In total, we call
            \cref{gucbenbjuj} at most \(\Oh(d^2)\) times; each time on a pair of strings of total
            length \(\Oh(\nabla)\). Hence, this step takes in total \(\Oh(d^2 \nabla^2 \log
            \nabla) = \Oh((d + k)^4 \log(d + k))\) time.
        \item[Integer weights.]
            For each pure index \(i\), we use {\tt\apmsmalled} of \cref{lem:small_ed}, on
            (suitably trimmed versions of)
            \(\hat{S}_i, \hat{S}_i, \nabla\), and \(w\) to obtain a \((W+1)\)-bounded
            difference, Monge \((\nabla,\nabla)\) fern for this index.
            As there are at most \(\Oh(d)\) pure indices, this costs in total
            \(\Oh(d W \nabla^2 \log^2 (n\nabla) = \Oh(W (d + k)^3 \log^2(n(d + k)))\) time.
            After computing the ferns, we use \cref{lm:build_mtx_dscomp} to construct the
            corresponding core-based matrix oracles.

            For the dirty indices, we proceed akin to \cite{ckw22}, except that we use the
            dynamic weighted edit distance data structure from \cite[Theorem~5.10]{gk24} (see
            \cref{lem:etdyn}).
            We initialize the data structure with the piece \((\Sr_i, \Sr_i)\).
            Next, for each pair \((\dP_i,\dT_i) \in \S_i \times \S_i\), we use the update
            operations to make the data structure represent \((\dP_i,\dT_i)\),
            followed by a query operation to obtain a corresponding fern.
            Observe that the data structure already outputs the ferns as a core-based
            matrix oracle.

            Across all at most \(\Oh(d)\) dirty indices, the initializations take time
            \(\Oh(dW\nabla^2 \log \nabla) = \Oh(W(d+k)^3 \log (d+k))\).
            Further, in total, we query all data structures at most \(\Oh(d^2)\) times;
            each such query is preceded by a constant number of calls to update
            operations. In total, this takes time \(\Oh(W\nabla d^2 \log^2 \nabla) =
            \Oh(W(d+k)^3 \log^2 (d + k))\).
            For constructing the core-based matrix oracles (for pure indices),
            we spend \(\Oh((d+k)^2 + W(d+k)\log(W(d+k)))\)
            time per pure index (of which there are at most \(\Oh(d)\) in total), which
            is dominated by the running time for constructing the ferns.
    \end{description}

    Finally, to compute a fern for each pair \(\dot{P}, \dot{T} \in \S \times \S\), we
    use \cref{lemma1} and proceed as follows.
    We initialize the data structure  with ferns for%
    \footnote{Technically, we consider a sequence of ferns \(F_0, \dots, F_{\tau-1}\) with
        \[
            F_0 \in \fmatsl{\hat{S}_{0}^{+\nabla}}{\hat{S}_{0}^{+\nabla}}{w}{\nabla}{k},
            \quad
            F_i \in \fmatsi{\hat{S}_{i}^{+\nabla}}{\hat{S}_{i}^{+\nabla}}{w}{\nabla}{k},
            \quad
            F_{\tau-1} \in
            \fmatst{\hat{S}_{\tau-1}^{+\nabla}}{\hat{S}_{\tau-1}^{+\nabla}}{w}{\nabla}{k}.
    \]}
    \[
        (\hat{S}_{0}^{+\nabla},\hat{S}_{0}^{+\nabla}), \dots,
        (\hat{S}_{\tau-1}^{+\nabla},\hat{S}_{\tau-1}^{+\nabla}).
    \]
    Next, we iterate over each pair \(\dot{P}, \dT  \in \S \times \S\); substitute
    the corresponding ferns, call {\tt DMM-Matrix} and undo the substitution by
    substituting the original ferns.
    We then use \cref{lem:trim} to trim the obtained ferns to \((d+k, k)\)-ferns.

    We obtain different running times depending on the type of ferns that we multiply.

    \begin{description}
        \item[General weights.]
            As all of our ferns are Monge, they have \(\Oh((d+k)^2)\)-time multiplication;
            hence our in total \(\Oh(d^2)\) substitution and updates cost \(\Oh((d + k)^4 \log d)\)
            time in total.

            In total, we obtain all desired ferns in time
            \(
                \Oh((d + k)^4 \log^2 (n(d + k)) );
            \) thereby completing the proof.
        \item[Integer weights.]
            As all of our ferns are Monge and \((W+1)\)-bounded difference, they have
            \(\Oh(W (d+k) \log(d + k))\)-time multiplication;
            hence our in total \(\Oh(d^2)\) substitution and updates cost \(\Oh(W (d +
            k)^3 \log^2(d + k))\)
            time in total.

            In total, we obtain all desired ferns in time
            \(
                \Oh(W (d + k)^3 \log^2 (n(d + k)) );
            \) thereby completing the proof.
            \qedhere
    \end{description}
\end{proof}

\subsection{Using Ferns: (Bleach-Commit) Dynamic Puzzle Matching}

\cref{11-3-1} allows us to efficiently compute ferns for a given collection of tiles
(out of which then the puzzle pieces are built).
What remains is to show how to efficiently maintain a sequence of said ferns under the
operations to be supported in \BCDPM. This, we tackle next.

\rdpmds
\begin{proof}
    We maintain an instance of the data structure of \cref{lemma2} on ferns that we
    maintain using multiple copies of the data structure from \cref{lemma1}.

    \subparagraph*{Initialization.}

    Given \(k\), \(\Delta\), and the families $\Sb$, $\Sm$, $\Sf$,
    we use \cref{11-3-1} with \( d \coloneqq \Delta - k\) (and then \cref{lem:trim} to cut
    out the required subferns) to obtain%
    \footnote{The torsion bound indeed implies \(\Delta \ge 2k\) and thus \(\ed(\S) \le k
        \le \Delta - k = d.\)}
    \begin{itemize}
        \item a leading \((\Delta,k)\)-fern for each pair of strings from \(\Sb\),
        \item an internal \((\Delta,k)\)-fern for each pair of strings from \(\Sm\), and
        \item a trailing \((\Delta,k)\)-fern for each pair of strings from \(\Sf\).
    \end{itemize}

    Further, we compute the \emph{neutral fern} \(\N\) that
    corresponds to the canonical pair \(\QQ\) using \cref{lem:small_ed}.
    Observe that \(\N\) is Monge, and
    \((W+1)\)-bounded difference if \(w\) is integer.
    We then create an instance \(I_{\N}\) of the data structure of \cref{lemma1} for a sequence of
    \(\alpha\) pointers to \(\N\).%
    \footnote{One might say that we grow a neutral fern tree of size \(\alpha\).}

    In the following, whenever we need a fern (or an instance of the data structure from
    \cref{lemma1}) for a blank piece of order \(\ell \le \alpha\), we use
    a corresponding {\tt DMM-Split} operation on \(I_{\N}\) to split of the first \(\ell\)
    ferns.
    As we never change the underlying ferns of \(I_\N\), we are able to precompute such a
    split for each \(\ell \le \alpha\), and may thus assume that accessing a blank piece
    of order \(\ell\) is a constant-time operation.

    Next, we split the input sequence \(\Y'\) into consecutive subsequences of pairs of
    strings and blank pieces.
    For each subsequence of pairs of strings, we initialize a data structure of
    \cref{lemma1} (using the ferns computed earlier).

    Finally, we initialize an instance \(I_{\Y}\) of the data structure of \cref{lemma2} on the sequence of pointers
    {\tt DMM-Root()} of the data structures for the subsequences of pairs of strings and
    the blank pieces; so that the sequence stored in \(I_{\Y}\) corresponds to the
    sequence \(\Y'\).

    As for the running time, we distinguish between the different types of weights that we
    consider.
    \begin{description}
        \item[General weights.]
            As \(d = \Oh(k)\), the calls to \cref{11-3-1} take time
            \(\Oh(k^4 \log^2 (\lambda k))\) where $\lambda$ is an upper bound on the lengths of the strings
            in $\Sb \cup \Sm \cup \Sf$; we obtain Monge \((d+k,k)\)-ferns in this case,
            which have \(\Oh((d+k)^2) = \Oh(k^2)\) multiplication and \(\Oh(d + k) =
            \Oh(d)\)-time vector multiplication (see \cref{uogiuvnbph}).

            Next, precomputing (the powers of) the neutral fern for all blank
            pieces of all orders up to \(\alpha =
            \Oh(k)\) does not exceed  \(\Oh(k^4 \log^2 (\lambda k))\).

            Next, the initialization of the data structures of \cref{lemma1} in total takes
            time \(\Oh(|\Y'| k^2)\); the initialization of the
            data structure for \cref{lemma2} runs in time \(\Oh(|\Y'|)\).

            In total, {\tt BCDPM-Init} thus takes time \(\Oh(k^4 \log^2 (\lambda k) + |\Y'| k^2 )\),
            as claimed.
        \item[Integer weights.]
            As \(d = \Oh(k)\), the calls to \cref{11-3-1} take time
            \(\Oh(k^3W \log^2 (\lambda k))\) where $\lambda$ is an upper bound on the lengths of the strings
            in $\Sb \cup \Sm \cup \Sf$; we obtain \((W+1)\)-bounded difference and
            Monge \((d+k,k)\)-ferns in this case,
            which have \(\Oh((W+1)(d+k) \log(d+k)) = \Oh(Wk \log k)\) multiplication and \(\Oh(d + k) =
            \Oh(d)\)-time vector multiplication (see \cref{azjqbhozka}).

            Next, precomputing (the powers of) the neutral fern for all blank
            pieces of all orders up to \(\alpha =
            \Oh(k)\) does not exceed  \(\Oh(Wk^3 \log^2 (\lambda k))\).

            Next, the initialization of the data structures of \cref{lemma1} in total takes
            time \(\Oh(|\Y'| kW \log k)\); the initialization of the
            data structure for \cref{lemma2} runs in time \(\Oh(|\Y'|)\).

            In total, {\tt BCDPM-Init} thus takes time \(\Oh(k^3W \log^2 (\lambda k) +
            |\Y'| k W \log k )\),
            as claimed.
    \end{description}

    \subparagraph*{The {\tt BCDPM-Set} operation.}

    For the  {\tt BCDPM-Set} operation, we swap out the pointer to the old neutral fern for
    the pointer to the new neutral fern. This takes constant time.

    \subparagraph*{The {\tt BCDPM-Commit} operation.}

    We use the {\tt DMM-Merge} operation of the data structure corresponding to the blank
    piece and its two neighbors. This takes time
    \(\Oh(k^2 \log(|\Y| + \alpha)) = \Oh(k^2 \log(|\Y| + k))\) for general weights and
    time
    \(\Oh(kW \log k \log(|\Y| + \alpha)) = \Oh(kW \log^2(|\Y| + k))\) for integer weights.

    \subparagraph*{The {\tt BCDPM-Bleach} operation.}

    We use twice the {\tt DMM-Split} operation of the data structure corresponding to the
    subsequence of pairs of strings to cut out the blank piece.
    In \(I_{\Y}\), we use one {\tt DMM-Delete} to delete the pointer to the old instance of
    the data structure from \cref{lemma1} and then three calls to {\tt DMM-Insert} to
    insert pointers to the prefix, the corresponding neutral fern, and the suffix.

    The only part that has a significant time cost is the {\tt DMM-Split} operation, which
    again takes time
    \(\Oh(k^2 \log(|\Y| + \alpha)) = \Oh(k^2 \log(|\Y| + k))\) for general weights and
    time
    \(\Oh(kW \log k \log(|\Y| + \alpha)) = \Oh(kW \log^2(|\Y| + k))\) for integer weights.

    \subparagraph*{The remaining update operations.}
    We forward the update to the corresponding instance of the data structure from
    \cref{lemma1}; the claimed running times follow.%
    \footnote{While not explicitly forbidden, we typically assume that such operations do
    not target a blank piece. If they do, we potentially first split the corresponding
    neutral fern.}

    \subparagraph*{The {\tt BCDPM-Query} operation.}
    We call {\tt DMM-Vector} of \(I_{\Y}\) with the all-zeroes vector and return any index
    of the resulting vector whose entry is at most \(k\).

    The correctness of this operation follows from \cref{cor:dproduct}:
    we use \((\Delta,k)\)-ferns of a \(\Delta\) puzzle with \(k \le \Delta/2 -
    \tor(\unpack(\Y))\).

    For the running time,
    as at any point, \(I_{\Y}\) stores an alternating sequence of ordinary ferns and a
    (power of the) neutral fern, we have \(|I_\Y| = \rho\), where \(\rho\) is the number
    of blank pieces in \(\Y\).
    Hence, this operation takes time \(\Oh(k \rho)\) for general weights and \(\Oh(k \rho
    \log \Delta)\) for integer weights (observe that for integer weights we have to pay an
    extra logarithmic factor for explicitly producing the resulting vector).

    In total, this completes the proof.
\end{proof}

\section{Fast Algorithms in Important Settings}\label{sec:concl}

In this section, we obtain efficient algorithms for \PMWED in important settings
by combining the \pillar model algorithm from \cref{thm:main} with efficient
implementations of the \pillar model in those settings.

\pillark*

\subsection*{The Standard Setting}

The \pillar model admits an optimal implementation in the standard
setting~\cite{F97,Bender2000,IPM}; see \cite[Theorem 7.2]{ckw20}.

\begin{theoremq}[{\cite{F97,Bender2000,IPM}}]\dglabel{thm:pilis}
    After an $\cO(n)$-time preprocessing of~a collection of~strings of~total length $n$,
    each \pillar operation can be performed in $\Oh(1)$ time.
\end{theoremq}

Combined, the standard trick, \cref{thm:main,thm:pilis} yield \cref{thm:stalgmaink4,thm:stalgmaink35}.

\stalgmainkfour

\stalgmainkthreehalf

Recall that for the standard setting, we also have a non-\pillar-based algorithm that
outperforms \cref{thm:stalgmaink4,thm:stalgmaink35} for large values of \(k\).

\stalgmainnk*

We compliment our algorithms for the standard setting with an easy generalization of the
conditional lower bound of \cite{ckw23}.

\lbwed
\lbpmwed
\begin{proof}
    Let us first consider the case of $0.5 \le \kappa \le 1$, assuming $\delta \le 0.25$
    without loss of generality.
    Suppose that we are given strings $X,Y$ of length at most $M$ over an alphabet
    $\Sigma$, a threshold $d\in \mathbb{R}_{\ge 1}$
    satisfying $d\le M^\kappa$, and oracle access to a normalized weight function $w$.
    Define $k=\lceil d \rceil + 1$ and $N = \lceil k^{1/\kappa}\rceil = \Oh(d^{1/\kappa})=
    \Oh(M)$.
    We construct an instance of \PMWED with pattern $P = \$ X \$^{N+1}$, text $T = \$ Y
    \$^N \#$, where $\$, \# \not\in \Sigma$, integer threshold $k$, and the weight
    function $w$ extended so that $\w{\$}{\$} = \w{\#}{\#} = 0$, $\w{\$}{\#} = k - d \geq
    1$, and other edits involving $\$$ or $\#$ cost $10k$.

    If $\edw{X}{Y}\le d$, then $\edw{P}{T}\le
    \edw{\$}{\$}+\edw{X}{Y}+\edw{\$^{N+1}}{\$^N\#} \le {0 + d + k - d = k}$, so
    $\OccW_k(P,T)\ne \emptyset$.
    Conversely, if $\edw{P}{T\fragmentco{i}{j}}\le k$ holds for some fragment of~$T$, then
    every $\$$ in $P$ has to be matched with a $\$$ or a $\#$ in $T$ (otherwise, it would
    contribute at least $10k$), so ${\edw{X}{Y}} \le k - \edw{\$}{\#} = d$.
    Hence, to decide $\edw{X}{Y} \le d$, it suffices to check $\OccW_k(P,T)\ne \emptyset$
    using the hypothetical \PMWED algorithm, which is applicable because $m > N \ge
    k^{1/\kappa}$.
    This approach takes $\Oh(m+m^{\min(0.5+1.5\kappa,
    2.5\kappa)-\delta})=\Oh(M+M^{0.5+1.5\kappa-\delta})=\Oh(M^{0.5+1.5\kappa-\delta})$
    time, contradicting \cref{thm:lb_wed}.

    For $0.4 < \kappa \le 0.5$, we assume $\delta \le 1.25-0.5/\kappa$ without loss of
    generality and, instead of $d \le M^\kappa$, pick $d\le M^{0.5}$ so that $N =
    \Oh(d^{1/\kappa})=\Oh(M^{0.5/\kappa})$.
    The hypothetical approach takes $\Oh(m+m^{\min(0.5+1.5\kappa,
    2.5\kappa)-\delta})=\Oh(M^{0.5/\kappa}+M^{0.5/\kappa \cdot
    (2.5\kappa-\delta)})=\Oh(M^{1.25-\delta})$ time, contradicting \cref{thm:lb_wed}.

    In the final case of $0\le \kappa \le 0.4$, the hypothetical running time for \PMWED
    is $\Oh(m^{2.5\kappa-\delta})=\Oh(m^{1-\delta})$, so it suffices to observe that
    solving \PMWED requires reading the entire pattern.
\end{proof}

\subsection*{The Compressed Setting}

A straight-line program (SLP) is a context-free grammar $\G$ that consists of a set
$\Sigma$ of terminals and a set $N_\G \coloneqq \{A_1,\dots,A_n\}$ of non-terminals such that each
$A_i \in N_\G$ is associated with a unique production rule
$A_i \rightarrow f_\G(A_i) \in (\Sigma \cup \{A_j : j < i\})^*$.
We may assume, without loss of generality, that each production rule is of the form $A
\rightarrow BC$ for some symbols $B$ and~$C$.

Every symbol $A \in S_G:=N_G \cup\Sigma$ generates a unique string, which we denote by
$\gen(A) \in \Sigma^*$;
we obtain this string from $A$ by repeatedly replacing each non-terminal with its
production. We say that~$\G$ generates $\gen(A_n)$.

Combining a recent result of Duyster and Kociumaka~\cite[Theorem
1]{DBLP:conf/spire/DuysterK24} on answering $\ipmOpName$ queries over run-length straight-line
programs into~\cite[Theorem 7.13]{ckw20}, which summarizes results
from~\cite{BilleLRSSW15,I17,ckw20}, we obtain the following implementation of the \pillar
model in the compressed setting.%
\footnote{See also the discussion just
after~\cite[Corollary~2]{DBLP:conf/spire/DuysterK24}.}

\begin{theoremq}[{\cite{BilleLRSSW15,I17,ckw20,DBLP:conf/spire/DuysterK24}}]
    \dglabel{3-30-1}
    Given a collection of SLPs of total size $n$, generating
    strings of total length $N$, each PILLAR operation can be performed in $\cO(\log^2 N)$ time after an
    $O(n \log N)$-time preprocessing.
\end{theoremq}

Combining \cref{thm:main,3-30-1} in a non-trivial, but by now standard procedure
(see~\cite{bkw19,ckw20,ckw22}), we obtain the following result.

\begin{corollaryq}[Compressed Setting]
    \dglabel{cor:compalgmain}[thm:main,3-30-1]
    Let $\G_T$ denote an SLP of~size~$n$ generating a string~$T$,
    let~$\G_P$ denote an SLP  of~size~$m$ generating a string~$P$,
    and set $N \coloneqq |T|$ and $M \coloneqq |P|$.

    Given $\G_T$, $\G_P$, an integer threshold $k$, and oracle access to a normalized
    weight function $w: \sqEsigma \to \intvl{0}{W}$,
    we can compute $|\OccW_k(P, T)|$
    \begin{itemize}
        \item in time $\cOtilde(m + n\, k^{4})$ (for general \(w\)), and
        \item in time $\cOtilde(m + n\, k^{3.5} W^4)$ if $w$ is an integer metric weight function.
    \end{itemize}
    Further, we can report the elements of $\OccW_k(P, T)$ within $\cO(|\OccW_k(P, T)|)$
    extra time.
\end{corollaryq}

\subsection*{The Dynamic Setting}

Write $\mathcal{X}$ for an initially empty collection of strings that is subject
to the following updates.
\begin{itemize}
    \item $\texttt{Makestring}(U)$: Insert a non-empty string $U$ to $\mathcal{X}$.
    \item $\texttt{Concat}(U,V)$: Insert string $UV$ to $\mathcal{X}$, for $U,V\in \mathcal{X}$.
    \item $\texttt{Split}(U,i)$: Insert $U\fragmentco{0}{i}$ and $U\fragmentco{i}{|U|}$ to
        $\mathcal{X}$, for $U\in\mathcal{X}$ and $i\in \fragmentco{0}{|U|}$.
\end{itemize}
Let $N$ denote an upper bound on the cumulative length of all strings in $\mathcal{X}$
throughout the execution of the algorithm.

The collection $\mathcal{X}$ can be maintained with operations $\mathtt{Makestring}(U)$,
$\mathtt{Concat}(U,V)$, $\mathtt{Split}(U,i)$ requiring time $\Oh(\log N+|U|)$, $\Oh(\log
N)$ and $\Oh(\log N)$, respectively, so that \pillar operations can be performed in time
$\Oh(\log^2 N)$; see
\cite{DBLP:conf/soda/GawrychowskiKKL18,DBLP:conf/spire/DuysterK24}.\footnote{All stated
time complexities hold with probability $1-1/N^{\Omega(1)}$.}

Observe that Kempa and Kociumaka~\cite[Section~8]{DBLP:conf/stoc/KempaK22} presented an
alternative deterministic implementation at the cost of extra $\mathrm{poly}(\log \log
N)$-factor overheads in the running time.

By combining the data structures
of~\cite{DBLP:conf/soda/GawrychowskiKKL18,DBLP:conf/spire/DuysterK24,DBLP:conf/stoc/KempaK22},
\cref{thm:main}, and the standard trick (recall the first paragraph of \cref{sec:reduction})
we obtain the following result.

\begin{corollaryq}[Dynamic Setting]
    \dglabel{cor:dynalgmain}[rem:std,thm:main]
    We can maintain a collection $\X$ of non-empty persistent strings of~total length $N$
    subject to
        $\makestring(U)$ in time $\cO(\log N +|U|)$,
        $\concat(U,V)$ in time $\cO(\log N)$,
            and
        $\splitOp(U,i)$ in time $\cO(\log N)$,
    such that given two strings $P, T
    \in \X$, an integer threshold $k>0$, and oracle access to a weight function $w:
    \sqEsigma \to \intvl{0}{W}$, we can compute $\OccW_k(P, T)$%
    \footnote{All running time bounds hold with
        probability $1-1/N^{\Omega(1)}$. This result can be derandomized at the cost of a
        $\mathrm{poly}(\log \log N)$-factor overhead.}
    \begin{itemize}
        \item in time $\cOtilde(|T|/|P| \cdot k^{4})$ (for general \(w\)); and
        \item in time $\cOtilde(|T|/|P| \cdot k^{3.5} W^4)$ if $w$ is an integer metric weight function.
            \qedhere
    \end{itemize}
\end{corollaryq}

\subsection*{The Quantum Setting}
We say an algorithm on an input of size $n$ succeeds \emph{with high probability} if the
success probability can be made at least $1-1/n^c$ for any desired constant $c>1$.

Let us turn our attention to the quantum setting, where we assume that the input strings
can be accessed in a quantum query model~\cite{AMB04,DBLP:journals/tcs/BuhrmanW02}.
We are interested in the time complexity of quantum algorithms~\cite{BBCplus}.
The near-optimal quantum algorithm of Hariharan and
Vinay~\cite{DBLP:journals/jda/HariharanV03} for the decision version of exact pattern
matching, encapsulated in the following statement,
implies that \pillar operations on strings of length $\cO(m)$ can be performed in
$\cOtilde(\sqrt{m})$ time without any preprocessing. See \cite{kCPM} for details.

\begin{theoremq}[\cite{DBLP:journals/jda/HariharanV03}]\dglabel{the:quantumPM}
    Consider a string $T$ of length $n$ and a string $P$ of length $m \le n$.
    We can decide whether $P$ occurs in $T$ in $\cOtilde(\sqrt{n})$ time in the quantum
    model with high probability.
    If the answer is YES, then the algorithm returns a witness occurrence.
\end{theoremq}

Combined, the standard trick, \cref{thm:main,the:quantumPM} yield the following result.

\begin{corollaryq}[Quantum Setting]\dglabel{cor:quantum}[rem:std,thm:main,the:quantumPM]
    In the quantum model,
    \PMWED can be solved w.h.p.
    \begin{itemize}
        \item in time $\cOtilde(k^4 \cdot n/\sqrt{m})$ for general weights; and
        \item in time $\cOtilde(k^{3.5}W^4 \cdot n/\sqrt{m})$ if $w$ is an integer metric weight function.
            \qedhere
    \end{itemize}
\end{corollaryq}

An open question is whether faster quantum algorithms can be designed for \PMWED as in~\cite{KNW24,KNW25} for \PMED.

\subsection*{The Packed Setting}

In the Real RAM model of computation with word size $\Theta(\log n)$, we may store up to
$\Omega(\log_{|\Sigma|} n)$ characters in a single machine word.
The \emph{packed representation} of a string $S$ over an integer alphabet
$\fragmentco{0}{\sigma}$ is a list obtained by storing
$\Theta(\log_\sigma n)$ characters per machine word---thus allowing us to store $S$ in
just $O(|S|/\log_\sigma n)$ machine words.

\begin{theoremq}[{\cite{DBLP:conf/stoc/KempaK19,IPM}}]\dglabel{thm:packed}
    Given a packed representation of a string $S$ of length $n$, after an $\cO(n /
    \log_\sigma n)$-time preprocessing, we can perform
    \pillar operations in $\cO(1)$ time.
\end{theoremq}

Combined, the standard trick, \cref{thm:main,thm:packed} yield the following result.

\begin{corollaryq}[Packed Setting]\dglabel{cor:packed}[rem:std,thm:main,thm:packed]
    Given packed representations of a text $T$ of~length $n$ and a pattern $P$ of~length
    $m$ over alphabet $\fragmentco{0}{\sigma}$,
    an integer threshold $k$, and oracle access to a normalized weight function $w:
    \sqEsigma \to \intvl{0}{W}$,
    we can compute $\OccW_k(P,T)$
    \begin{itemize}
        \item in time $\cO(n / \log_{\sigma} n + k^4 \cdot n/m \cdot \log^2(mk))$ (for
            general \(w\)); and
        \item in time $\Oh(n / \log_{\sigma} n) + \Ohtilde(k^{3.5} W^4 \cdot n/m)$ if $w$ is an integer metric weight function.
            \qedhere
    \end{itemize}
\end{corollaryq}

\subsection*{The Read-only Setting}

Finally, we consider the read-only setting, where we have random access to $P$ and $T$,
but our extra working space is limited.

\begin{theoremq}[{\cite{DBLP:conf/cpm/BathieCS24,DBLP:conf/cpm/KosolobovS24}}]\dglabel{thm:small_space}
    Suppose that we have read-only random access to a $n$-length string $S$ of length $n$
    over an integer alphabet.
    For any $\tau \in \fragment{1}{n}$, after a preprocessing step that uses $\cOtilde(n)$
    time and $\cOtilde(n/\tau)$ extra space,
    \pillar operations can be performed in time $\cOtilde(\tau)$.
\end{theoremq}

Combined, the standard trick, \cref{thm:main,thm:small_space} yield the following result.

\begin{corollaryq}[Read-only
    Setting]\dglabel{cor:readonly}[rem:std,thm:main,thm:small_space]
    Suppose that we have read-only random access to a text~$T$ of~length $n$, a pattern
    $P$ of~length $m$, and .
    Given an integer threshold $k$, and oracle access to a normalized weight function $w:
    \sqEsigma \to \intvl{0}{W}$, for any $\tau \in \fragment{1}{n}$,
    we can compute $\OccW_k(P,T)$
    \begin{itemize}
        \item in time $\cOtilde(n+k^{4}\tau \cdot n/m)$
            using $\cOtilde(m/\tau+k^{4})$ extra space
            (for general \(w\)); and
        \item in time $\cOtilde(n+k^{3.5}W^4\tau \cdot n/m)$
            using $\cOtilde(m/\tau+k^{3.5}W^4)$ extra space
            if $w$ is an integer metric weight function.
            \qedhere
    \end{itemize}
\end{corollaryq}

\bibliographystyle{alphaurl}
\bibliography{main}
\tableofresults

\end{document}